\documentclass[10pt, final, oneside]{IEEEtran} \twocolumn %\linespread{1.3}
\usepackage{epsfig,float,epstopdf,url,subfigure}
\usepackage{amsmath, amssymb, amsthm, algorithm, algorithmic,graphicx}
\usepackage{color}

\newcommand{\rank}{\operatorname{rank}}

\newcommand{\calb}{\mathcal{B}}
\newcommand{\ecalb}{\mathcal{B}^e}
\newcommand{\calc}{\mathcal{C}}
\newcommand{\ecalc}{\mathcal{C}^e}

\newtheorem{theorem}{Theorem}[section]
\newtheorem{lem}[theorem]{Lemma}
\newtheorem{sigmodel}[theorem]{Signal Model}
\newtheorem{corollary}[theorem]{Corollary}

\newtheorem{definition}[theorem]{Definition}
\newtheorem{remark}[theorem]{Remark}

\newtheorem{ass}[theorem]{Assumption}

\newtheorem{fact}[theorem]{Fact}

\newcommand{\vect}[2]{\left[\begin{array}{cccccc}
     #1 \\
     #2
   \end{array}
  \right]
  }

\newcommand{\ds}{\displaystyle}

\newcommand{\beq}{\begin{equation}}
\newcommand{\eeq}{\end{equation}}
\newcommand{\bdm}{\begin{displaymath}}
\newcommand{\edm}{\end{displaymath}}

\newcommand{\eps}{\epsilon}
\newcommand{\bd}{\begin{definition}}
\newcommand{\ed}{\end{definition}}

\newcommand{\bv}{\begin{vugraph}}
\newcommand{\ev}{\end{vugraph}}
\newcommand{\bi}{\begin{itemize}}
\newcommand{\ei}{\end{itemize}}
\newcommand{\ben}{\begin{enumerate}}
\newcommand{\een}{\end{enumerate}}

\newcommand{\bean}{\begin{eqnarray*} }
\newcommand{\eean}{\end{eqnarray*} }
\newcommand{\bea}{\begin{eqnarray} }
\newcommand{\eea}{\end{eqnarray} }
\newcommand{\nn}{\nonumber}
\newcommand{\ba}{\begin{array} }
\newcommand{\ea}{\end{array} }

\newcommand{\E}{\mathbf{E}}

\renewcommand\thetheorem{\arabic{section}.\arabic{theorem}}

\newcommand{\new}{\mathrm{new}}
\newcommand{\cs}{\text{cs}}

\newcommand{\old}{\mathrm{old}}
\newcommand{\Shat}{\hat{S}}
\newcommand{\Lhat}{\hat{L}}

\newcommand{\Phat}{\hat{P}}

\newcommand{\Span}{\operatorname{span}}
\newcommand{\del}{\text{del}}

\newcommand{\add}{\text{add}}

\newcommand{\train}{\text{train}}

\newcommand{\That}{\hat{T}}

\newcommand{\SE}{\text{SE}}
\newcommand{\egam}{\Gamma^e}

\title{Recursive Robust PCA or Recursive Sparse Recovery in Large but Structured Noise}
\author{
Chenlu~Qiu,~\IEEEmembership{Member,~IEEE,} Namrata~Vaswani,~\IEEEmembership{Senior Member,~IEEE,} Brian~Lois~\IEEEmembership{Graduate~Student~Member,~IEEE}, and Leslie Hogben
\thanks{
C. Qiu, N. Vaswani are with the ECE dept at Iowa State University. B. Lois, L. Hogben are with the Mathematics dept. Email: \{chenlu,namrata,blois,lhogben\}@iastate.edu. Some ideas of this paper were presented at Allerton 2010, ICASSP 2013 and ISIT 2013 \cite{rrpcp_allerton,rrpcp_perf,rrpcp_perf2}. Copyright (c) 2014 IEEE. Personal use of this material is permitted.  However, permission to use this material for any other purposes must be obtained from the IEEE by sending a request to pubs-permissions@ieee.org. This research was supported in part by NSF grants CCF-0917015 and CCF-1117125.
}
}

\begin{document}

\maketitle

%\input{edits_needed}

%\input{discuss_nv}

  %We do not assume any model on the sequence of sparse vectors. Their support sets and their nonzero element values may be either independent or correlated over time (usually in many applications they are correlated). The only thing required is that there be {\em some} support change every so often.
%Here, ``robust" refers to robustness to both independent and correlated sparse outliers.

\begin{abstract}
This work studies the recursive robust principal components analysis (PCA) problem. If the outlier is the signal-of-interest, this problem can be interpreted as one of recursively recovering a time sequence of sparse vectors, $S_t$, in the presence of large but structured noise, $L_t$. The structure that we assume on $L_t$ is that $L_t$ is dense and lies in a low dimensional subspace that is either fixed or changes ``slowly enough." A key application where this problem occurs is in video surveillance where the goal is to separate a slowly changing background ($L_t$) from moving foreground objects ($S_t$) on-the-fly. To solve the above problem, in recent work, we introduced a novel solution called Recursive Projected CS (ReProCS). In this work we develop a simple modification of the original ReProCS idea and analyze it. This modification assumes knowledge of a subspace change model on the $L_t$'s. Under mild assumptions and a denseness assumption on the unestimated part of the subspace of $L_t$ at various times, we show that, with high probability (w.h.p.), the proposed approach can exactly recover the support set of $S_t$ at all times; and the reconstruction errors of both $S_t$ and $L_t$ are upper bounded by a time-invariant and small value. In simulation experiments, we observe that the last assumption holds as long as there is some support change of $S_t$ every few frames.
%
%%and an assumption that depends on an algorithm estimate (that holds, in simulation experiments, for correlated support changes as long as there is some support change every few frames),
%To the best of our knowledge, this is the first rigorous analysis of recursive robust PCA and equivalently of recursive sparse recovery in large but low-dimensional and dense noise.%
%The background sequence is well modeled as lying in a low dimensional subspace, that can gradually change over time, while the moving foreground objects constitute the correlated sparse outliers. In this and many other applications, the foreground is an outlier for PCA but is actually the ``signal of interest" for the application; where as the background is the corruption or noise. Thus our problem can also be interpreted as one of recursively recovering a time sequence of sparse signals in the presence of large but spatially correlated noise.
%\vspace{-1mm}
\end{abstract}


%This version has the following changes: (a) Algorithm 1 and the proof of the main result, Theorem \ref{thm1}, have been reorganized to make them shorter and easier to follow; (b) the undersampled case has been removed.

\section{Introduction}
\label{intro}
%Most high dimensional data often approximately lies in a lower dimensional subspace. Principal Components' Analysis (PCA) is a widely used dimension reduction technique that finds a small number of orthogonal basis vectors (principal components), along which most of the variability of the dataset lies. To be precise, for a given dimension, $r$, PCA finds the $r$-dimensional subspace that minimizes the mean squared error between data vectors and their $r$-dimensional projections \cite{PCA}. The subspace spanned by the principal components (PCs) is called the principal components' space (PC space). Often, for time series data, the PC space changes gradually over time. Updating the PC space recursively as more data comes in, without re-solving the entire PCA problem, is referred to as recursive PCA \cite{sequentialSVD}.

%If we want to update the PC space on-the-fly (causally), recursive PCA is computationally a lot less expensive than redoing the full PCA again every time \cite{sequentialSVD}.
%
%Notice that to find an $r$ dimensional PC space, one needs at least $r$ data vectors, usually more. Thus, even for recursive PCA, the %initial step needs to be a batch one or the initial PC space needs to be pre-specified.%Once an initial PC directions' estimate,

%does not follow the data model, or
Principal Components Analysis (PCA) is a widely used dimension reduction technique that finds a small number of orthogonal basis vectors, called principal components (PCs), along which most of the variability of the dataset lies. It is well known that PCA is sensitive to outliers. Accurately computing the PCs in the presence of outliers is called robust PCA \cite{Roweis98emalgorithms,Torre03aframework,rpca,rpca2}.
Often, for time series data, the PCs space changes gradually over time. Updating it on-the-fly (recursively) in the presence of outliers, as more data comes in is referred to as online or recursive robust PCA \cite{sequentialSVD,ipca_weightedand,Li03anintegrated}.
%Doing recursive PCA in the presence of outliers is referred to as recursive robust PCA \cite{sequentialSVD,ipca_weightedand,Li03anintegrated}.
``Outlier" is a loosely defined term that refers to any corruption that is not small compared to the true data vector and that occurs occasionally. As suggested in \cite{error_correction_PCP_l1,rpca}, an outlier can be nicely modeled as a sparse vector whose nonzero values can have any magnitude. %We will use this definition in this paper. %, i.e. a vector whose most elements are zero, while the few that are

%In  \cite{error_correction_PCP_l1,rpca}, the outlier is modeled as being spatially and temporally independent. In most real applications, the time at which the outliers begin to occur is often random and independent of all past times. But once outliers begin to occur, for some time after that they are both spatially and temporally correlated. In this work, we focus on this case, i.e. on {\em recursive robust PCA that is robust to sparse outliers that may be spatially and/or temporally correlated}.

A key application where the robust PCA problem occurs is in video analysis where the goal is to separate a slowly changing background from moving foreground objects \cite{Torre03aframework,rpca}. %We show an example in Fig. \ref{realvideo}.  which change in a correlated fashion over time  that change in a correlated fashion over time that usually do not move around randomly
If we stack each frame as a column vector, the background is well modeled as being dense and lying in a low dimensional subspace that may gradually change over time, while the moving foreground objects constitute the sparse outliers \cite{error_correction_PCP_l1,rpca}. Other applications include detection of brain activation patterns from functional MRI (fMRI) sequences (the ``active" part of the brain can be interpreted as a sparse outlier), detection of anomalous behavior in dynamic social networks and  sensor networks based detection and tracking of abnormal events such as forest fires or oil spills. Clearly, in all these applications, an online solution is desirable. %In this work, we focus on this case, i.e. on {\em recursive robust PCA that is robust to both independent and correlated sparse outliers}.

The moving objects or the active regions of the brain or the oil spill region may be ``outliers" for the PCA problem, but in most cases, these are actually the signals-of-interest whereas the background image is the noise. Also, all the above signals-of-interest are sparse vectors. Thus, this problem can also be interpreted as one of recursively recovering a time sequence of sparse signals, $S_t$, from measurements $M_t: = S_t + L_t$ that are corrupted by (potentially) large magnitude but dense and structured noise, $L_t$.  The structure that we require is that $L_t$ be dense and lie in a low dimensional subspace that is either fixed or changes ``slowly enough" in the sense quantified in Sec \ref{slowss}. %If the goal is recursive robust PCA, then $L_t$ is the signal of interest while $S_t$ is the noise (outlier).
\subsection{Related Work}
There has been a large amount of work on robust PCA, e.g. \cite{Torre03aframework,rpca,rpca2,Roweis98emalgorithms,novel_m_estimator,outlier_pursuit, mccoy_tropp11}, and recursive robust PCA e.g. \cite{sequentialSVD,ipca_weightedand,Li03anintegrated}. In most of these works, either the locations of the missing/corruped data points are assumed known \cite{sequentialSVD} (not a practical assumption); or they first detect the corrupted data points and then replace their values using nearby values \cite{ipca_weightedand}; or weight each data point in proportion to its reliability (thus soft-detecting and down-weighting the likely outliers) \cite{Torre03aframework,Li03anintegrated}; or just remove the entire outlier vector \cite{outlier_pursuit, mccoy_tropp11}. %Approaches like \cite{sequentialSVD} can be adapted to the case where the missing/corrupted data points are unknown by using the outlier detection approach from other works, e.g. from \cite{Li03anintegrated}.
Detecting or soft-detecting outliers ($S_t$) as in \cite{ipca_weightedand,Torre03aframework,Li03anintegrated} is easy when the outlier magnitude is large, but not otherwise. When the signal of interest is $S_t$, the most difficult situation is when nonzero elements of $S_t$ have small magnitude compared to those of $L_t$ and in this case, these approaches do not work.

In recent works \cite{rpca,rpca2}, a new and elegant solution to robust PCA called Principal Components' Pursuit (PCP)  has been proposed, that does not require a two step outlier location detection/correction process and also does not throw out the entire vector. It redefines batch robust PCA as a problem of separating a low rank matrix, ${\cal L}_t := [L_1,\dots,L_t]$, from a sparse matrix, ${\cal S}_t := [S_1,\dots,S_t]$, using the measurement matrix, ${\cal M}_t := [M_1,\dots,M_t] = {\cal L}_t+ {\cal S}_t$. Other recent works that also study batch algorithms for recovering a sparse ${\cal S}_t$ and a low-rank ${\cal L}_t$ from ${\cal M}_t := {\cal L}_t+ {\cal S}_t$ or from undersampled measurements include \cite{rpca_tropp, linear_inverse_prob, rpca_hu,SpaRCS,rpca_vayatis,rpca_zhang, rpca_Giannakis,compressivePCP,rpca_reduced,noisy_undersampled_yuan}.

Let $\|A\|_*$ be the nuclear norm of $A$ (sum of singular values of $A$) while $\|A\|_1$ is the $\ell_1$ norm of $A$ seen as a long vector.
It was shown in \cite{rpca} that, with high probability (w.h.p.), one can recover ${\cal L}_t$ and ${\cal S}_t$ exactly by solving PCP: %$\underset{{\cal L},{\cal S}} {\min}\|{\cal L}\|_* + \lambda\|{\cal S}\|_1 \ \text{subject to}  \ \ {\cal L} + {\cal S} = {\cal M}_t$
\beq
\underset{{\cal L},{\cal S}} {\min}\|{\cal L}\|_* + \lambda\|{\cal S}\|_1 \ \text{subject to}  \ \ {\cal L} + {\cal S} = {\cal M}_t
\label{pcp_prob}
\vspace{-2mm}
\eeq
provided that (a) the left and right singular vectors of ${\cal L}_t$ are dense; (b) any element of the matrix ${\cal S}_t$ is nonzero w.p. $\varrho$, and zero w.p. $1-\varrho$, independent of all others;
% (in particular, this means that the support sets of the different $S_t$'s are independent over time)
%(b) the support of ${\cal S}_t$ is selected out of a uniform distribution over all possible sets of a given size so that with high probability (w.h.p.) ${\cal S}_t$ is full rank (in particular, this means that the support sets of the different $S_t$'s need to be independent; each element of the matrix ${\cal S}_t$ is included into the support set with probability $\varrho$ (and not included with probability $1-\rho$) independent of all others
%
and (c) the rank of ${\cal L}_t$ is bounded by a small enough value. %(for a given support size of the sparse part).
%and the support size of ${\cal S}_t$ are small enough

%As explained earlier, a key application where the robust PCA problem occurs is in video layering where the foreground sequence is sparse while the background sequence is approximately low dimensional.
%Moreover, while it is fair to assume that the background changes are dense (i.e. ${\cal L}_t$ is dense), the assumption that the foreground support is independent over time is often not valid. Foreground objects typically move in a correlated fashion.

As described earlier, many applications where robust PCA is required, such as video surveillance, require an online (recursive) solution. Even for offline applications, a recursive solution is typically faster than a batch one. In recent work \cite{rrpcp_allerton,rrpcp_allerton11,han_tsp}, we introduced a novel solution approach, called Recursive Projected Compressive Sensing (ReProCS), that recursively recovered $S_t$ and $L_t$ at each time $t$. In simulation and real data experiments (see \cite{han_tsp} and \url{http://www.ece.iastate.edu/~chenlu/ReProCS/ReProCS_main.htm}), it was faster than batch methods such as PCP and also significantly outperformed them in situations where the support changes were correlated over time (as long as there was some support change every few frames) or when the background subspace dimension was large (for a given support size).
%The only thing that was needed was that there be {\em some} support changes every few frames. 
In this work we develop a simple modification of the original ReProCS idea and analyze it. This modification assumes knowledge of the subspace change model on the $L_t$'s.

\subsection{Our Contributions} %The work of Candes, Wright et al
%Under mild assumptions and an assumption on an algorithm estimate (that holds in simulations as long as there is some support change every few frames), we show that,
%and (v) a few other mild assumptions hold,
We show that (i) if an estimate of the subspace of $L_t$ at the initial time is available; (ii) if $L_t$, lies in a slowly changing low dimensional subspace as defined in Sec \ref{slowss}, (iii) if this subspace is dense, if (iv) the unestimated part of the changed subspace is dense at all times, and (v) if the subspace change model is known to the algorithm, then, w.h.p., ReProCS can exactly recover the support set of $S_t$ at all times; and the reconstruction errors of both $S_t$ and $L_t$ are upper bounded by a time invariant and small value. Moreover, after every subspace change time, w.h.p., the subspace error decays to a small enough value within a finite delay. Because (iv) depends on an algorithm estimate, our result, in its current form, cannot be interpreted as a correctness result but only a useful step towards it. From simulation experiments, we have observed that (iv) holds for correlated support changes as long as the support changes every few frames. This connection is being quantified in ongoing work. Assumption (v) is also restrictive and we explain in Sec \ref{discuss_add} how it can possibly be removed in future work.

We also develop and analyze a generalization of ReProCS called ReProCS with cluster-PCA (ReProCS-cPCA) that is designed for a more general subspace change model, and that needs an extra clustering assumption. Its main advantage is that it does not require a bound on the number of subspace changes, $J$, as long as the separation between the change points is allowed to grow logarithmically with $J$. Equivalently, it does not need a bound on the rank of ${\cal L}_t$.
%As a result there is also no bound needed on the rank of ${\cal L}_t$.

If $L_t$ is the signal of interest, then ReProCS is a solution to recursive robust PCA in the presence of sparse outliers. To the best of our knowledge, this is the first analysis of any recursive (online) robust PCA approach. If $S_t$ is the signal of interest, then ReProCS is a solution to recursive sparse recovery in large but low-dimensional noise. To our knowledge, this work is also the first to analyze any recursive (online) sparse plus low-rank recovery algorithm. Another online algorithm that addresses this problem is given in \cite{grass_undersampled}, however, it does not contain any performance analysis.
Our results directly apply to the recursive version of the matrix completion problem \cite{matrix_com_candes,Bresler_matrix} as well since it is a simpler special case of the current problem (the support set of $S_t$ is the set of indices of the missing entries and is thus known) \cite{rpca}. %missing data case as well or, equivalently, to the case where the outlier locations are known.

% mccoy_tropp11, outlier_pursuit(the works of \cite{mccoy_tropp11,outlier_pursuit} also study a different case: an entire vector is either an outlier or an inlier)
The proof techniques used in our work are very different from those used to analyze other recent batch robust PCA works \cite{rpca,rpca2, novel_m_estimator,mccoy_tropp11, outlier_pursuit,rpca_tropp, linear_inverse_prob, rpca_reduced, rpca_Giannakis,rpca_zhang,compressivePCP,noisy_undersampled_yuan}. The works of \cite{mccoy_tropp11,outlier_pursuit} also study a different case: that where an entire vector is either an outlier or an inlier. Our proof utilizes (a) sparse recovery results \cite{candes_rip}; (b) results from matrix perturbation theory that bound the estimation error in computing the eigenvectors of a perturbed Hermitian matrix with respect to eigenvectors of the original Hermitian matrix (the famous sin $\theta$ theorem of Davis and Kahan \cite{davis_kahan}) and (c) high probability bounds on eigenvalues of sums of independent random matrices (matrix Hoeffding inequality \cite{tail_bound}).
 %and that bound the perturbed eigenvalues (Weyl's theorem \cite{hornjohnson})

A key difference of our approach to analyzing the subspace estimation step compared with most existing work analyzing finite sample PCA, e.g. \cite{nadler} and references therein, is that it needs to provably work in the presence of error/noise that is correlated with $L_t$. Most existing works, including \cite{nadler} and the references it discusses, assume that the noise is independent of (or at least uncorrelated with) the data. However, in our case, because of how the estimate $\Lhat_t$ is computed, the error $e_t := L_t - \Lhat_t$ is correlated with $L_t$. As a result, the tools developed in these earlier works cannot be used for our problem. This is also the reason why simple PCA cannot be used and we need to develop and analyze projection-PCA based approaches for subspace estimation (see Appendix \ref{projpca} for details).

The ReProCS approach is related to that of \cite{decodinglp,rpca_regression, rpca_regression_sparse} in that all of these first try to nullify the low dimensional signal by projecting the measurement vector into a subspace perpendicular to that of the low dimensional signal, and then solve for the sparse ``error" vector (outlier). However, the big difference is that in all of these works the basis for the subspace of the low dimensional signal is {\em perfectly known.} Our work studies {\em the case where the subspace is not known}. We have an initial approximate estimate of the subspace, but over time it can change significantly. %The only thing we require is that the changes per unit time are ``slow" in a sense quantified in Sec \ref{slowss}.
In this work, to keep things simple, we use $\ell_1$ minimization done separately for each time instant (also referred to as basis pursuit denoising (BPDN)) \cite{candes_rip,bpdn}. However, this can be replaced by any other sparse recovery algorithm, either recursive or batch, as long as the batch algorithm is applied to $\alpha$ frames at a time, e.g. one can replace BPDN by modified-CS or support-predicted modified-CS \cite{rrpcp_isit}.

\subsection{Paper Organization}
The rest of the paper is organized as follows. We give the notation and background required for the rest of the paper in Sec \ref{bgnd}. The problem definition and the model assumptions are given in Sec \ref{probdef}. We explain the ReProCS algorithm and give its performance guarantees (Theorem \ref{thm1}) in Sec \ref{reprocs_sec}. The terms used in the proof are defined in Sec \ref{detailed}. The proof is given in Sec \ref{mainlemmas}.
%main lemmas needed to prove Theorem \ref{thm1} and its proof as well as the proofs of the two main lemmas are given in Section \ref{mainlemmas}.  %In Sections \ref{pfoflem1} and \ref{pfoflem} we prove the two main lemmas.
A more general subspace change model and ReProCS-cPCA which is designed to handle this model are given in Sec. \ref{Del_section}. We also give the main result for ReProCS-cPCA in this section and discuss it. A discussion with respect to the result for PCP \cite{rpca} is also provided here. Section \ref{thmproof} contains the proof of this theorem.
 In Sec \ref{model_verify}, we show that our slow subspace change model indeed holds for real videos. In Sec \ref{sims}, we show numerical experiments demonstrating Theorem \ref{thm1}, as well as comparisons of ReProCS with PCP. Conclusions and future work are given in Sec \ref{conc}.

\section{Notation and Background}
\label{bgnd}
\subsection{Notation}

For a set $T \subset \{1,2,\dots, n\}$, we use $|T|$ to denote its cardinality, i.e., the number of elements in $T$. We use $T^c$ to denote its complement w.r.t. $\{1,2,\dots n\}$, i.e. $T^c:= \{i \in \{1,2,\dots n\}: i \notin T \}$.

%For a scalar, $d$, $d^+: = \max(0, d)$.
We use the interval notation, $[t_1, t_2]$, to denote the set of all integers between and including $t_1$ to $t_2$, i.e. $[t_1, t_2]:=\{t_1, t_1+1, \dots, t_2\}$. For a vector $v$, $v_i$ denotes the $i$th entry of $v$ and $v_T$ denotes a vector consisting of the entries of $v$ indexed by $T$. We use $\|v\|_p$ to denote the $\ell_p$ norm of $v$. The support of $v$, $\text{supp}(v)$, is the set of indices at which $v$ is nonzero, $\text{supp}(v) := \{i : v_i\neq 0\}$. We say that $v$ is s-sparse if $|\text{supp}(v)| \leq  s$.

For a matrix $B$, $B'$ denotes its transpose, and $B^{\dag}$ its pseudo-inverse. For a matrix with linearly independent columns, $B^{\dag} = (B'B)^{-1}B'$.
%We use $\Span(M)$ to denote the space spanned by the columns of $M$.
We use $\|B\|_2:= \max_{x \neq 0} \|Bx\|_2/\|x\|_2$ to denote the induced 2-norm of the matrix. Also, $\|B\|_*$ is the nuclear norm (sum of singular values) and $\|B\|_{\max}$ denotes the maximum over the absolute values of all its entries.
%The Frobenius norm of matrix M is denoted by $\|M\|_F$ , i.e., $\|M\|_F :=\sqrt{\sum_j \sum_j M_{i,j}^2}$.  %when $\sigma_i(M)$s' are arranged in non-increasing order  and $\lambda_i(M)$s' are arranged in non-increasing order For a square matrix, $\lambda_{\max}(M)$ and $\lambda_{\min}(M)$ denote the maximum and minimum eigenvalues.
We let $\sigma_i(B)$ denotes the $i$th largest singular value of $B$. For a Hermitian matrix, $B$, we use the notation $B \overset{EVD}{=} U \Lambda U'$ to denote the eigenvalue decomposition of $B$. Here $U$ is an orthonormal matrix and $\Lambda$ is a diagonal matrix with entries arranged in decreasing order. Also, we use $\lambda_i(B)$ to denote the $i$th largest eigenvalue of a Hermitian matrix $B$ and we use $\lambda_{\max}(B)$ and $\lambda_{\min}(B)$ denote its maximum and minimum eigenvalues. If $B$ is Hermitian positive semi-definite (p.s.d.), then $\lambda_i(B) = \sigma_i(B)$. For Hermitian matrices $B_1$ and $B_2$, the notation $B_1 \preceq B_2$ means that $B_2-B_1$ is p.s.d. Similarly, $B_1 \succeq B_2$ means that $B_1-B_2$ is p.s.d.

For a Hermitian matrix $B$, $\|B\|_2 = \sqrt{\max( \lambda_{\max}^2(B), \lambda_{\min}^2(B) )}$ and thus, $\|B\|_2 \le b$ implies that $-b \le \lambda_{\min}(B) \le \lambda_{\max}(B) \le b$.

%Let $I_n$ be an $n\times n$ identity matrix.  We use $I$ without identifying its dimensionality to denote an identity matrix of proper size. Given a index set $T \subseteq \{1, 2, \cdots, n\}$, $(I_n)_T$ or just $I_T$ denotes an $n \times |T |$ matrix composed by the column vectors of $I$ indexed by $T$. Similarly,  $(I_{T})'$ is a $|T| \times n$ matrix composed by the row vectors of $I_n$ indexed by $T$. Suppose $B$ is of size $n_1 \times n_2$. Given two index sets $T_{\text{row}} \subseteq \{1, 2, \cdots, n_1\}$ and $T_{\text{col}} \subseteq \{1, 2, \cdots, n_2\}$, $B I_{T_{\text{col}}}$ is an $n_1 \times |T_{\text{col}}|$ matrix composed of the columns of $B$ indexed by $T_{\text{col}}$ and ${I_{T_{\text{row}}}}' B$ is a $|T_{\text{row}}| \times n_2$ matrix composed of the rows of $B$ indexed by $T_{\text{row}}$. Also, $B_{T_{\text{col}}}:= B  I_{T_{\text{col}}}$. Given the complement set $T_{\text{col}}^c$, we use $B \setminus B_{T_{\text{col}}}$ to denote $B_{(T_{\text{col}}^c)}$. Given another matrix $N$ of size $n_1\times n_3$, $[B \ B_2]$ constructs a new matrix by concatenating matrices $B$ and $B_2$ in a horizontal direction. Thus, $[B\  B_2] \setminus B_{T_{\text{col}}} := [(B)_{T_{\text{col}}^c} \  B_2]$. For any matrix $B$ and sets $T_1, T_2$, $(B)_{T_1,T_2}$ denotes the sub-matrix containing the rows with indices in $T_1$ and columns with indices in $T_2$.

We use $I$ to denote an identity matrix of appropriate size. For an index set $T$ and a matrix $B$, $B_T$ is the sub-matrix of $B$ containing columns with indices in the set $T$. Notice that $B_T = B I_T$. Given a matrix $B$ of size $m \times n$ and $B_2$ of size $m \times n_2$, $[B \ B_2]$ constructs a new matrix by concatenating matrices $B$ and $B_2$ in the horizontal direction. Let $B_{\text{rem}}$ be a matrix containing some columns of $B$. Then $B \setminus B_{\text{rem}}$ is the matrix $B$ with columns in $B_{\text{rem}}$ removed.
%We use the notation $B \setminus B_T$ to denote the matrix $B_{T^c}$. 

For a tall matrix $P$, $\Span(P)$ denotes the subspace spanned by the column vectors of $P$.

The notation $[.]$ denotes an empty matrix.

\begin{definition}
We refer to a tall matrix $P$ as a {\em basis matrix} if it satisfies $P'P=I$.
\end{definition}

\begin{definition}
We use the notation $Q = \mathrm{basis}(B)$ to mean that $Q$ is a basis matrix and $\Span(Q) = \Span(B)$.  In other words, the columns of $Q$ form an orthonormal basis for the range of $B$.
\end{definition}

%\begin{definition}
%\label{defcontain}
%For a basis matrix $P$ and any other matrix $B$, we say that ``$\Span(P)$ approximately contains $\Span(B)$" if $\|(I-P P') B\|_2/\|B\|_2$  is small.
%\end{definition}

\begin{definition}\label{defn_delta}
The {\em $s$-restricted isometry constant (RIC)} \cite{decodinglp}, $\delta_s$, for an $n \times m$ matrix $\Psi$ is the smallest real number satisfying $(1-\delta_s) \|x\|_2^2 \leq \|\Psi_T x\|_2^2 \leq (1+\delta_s) \|x\|_2^2$
for all sets $T \subseteq \{1,2,\dots n \}$ with $|T| \leq s$ and all real vectors $x$ of length $|T|$.%
\end{definition} %($n \leq m$),

It is easy to see that $\max_{T:|T| \le s} \|({\Psi_T}'\Psi_T)^{-1}\|_2 \le \frac{1}{1-\delta_s(\Psi)}$ \cite{decodinglp}.

%\begin{remark}
%According to Definition \ref{defn_delta}, an alternative way of defining $\delta_s(\Psi)$ is
%$$\delta_s(\Psi) = \max\{ \max_{|T| \leq s} (1-\lambda_{\min} ({\Psi_T}'\Psi_T)), \max_{|T| \leq s} (\lambda_{\max} ({\Psi_T}'\Psi_T) - 1)\}$$
%\end{remark}

\begin{definition}
We give some notation for random variables in this definition.
\ben
\item We let $\E[Z]$ denote the expectation of a random variable (r.v.) $Z$ and $\E[Z|X]$ denote its conditional expectation given another r.v. $X$.
\item Let $\calb$ be a set of values that a r.v. $Z$ can take. We use $\ecalb$ to denote the {\em event} $Z \in \calb$, i.e. $\ecalb:= \{Z \in \calb\}$.
\item The probability of any event $\ecalb$ can be expressed as \cite{grimmett},
$$\mathbf{P}(\ecalb) := \E[\mathbb{I}_\calb(Z)].$$
where
\bea
\mathbb{I}_\calb(Z): = \left\{ \begin{array}{cc}
                                 1 \ & \ \text{if} \ Z \in \calb \nn \\
                                 0  \ & \ \text{otherwise}
                                 \end{array} \right. \nn
\eea
is the indicator function on the set $\calb$.

\item For two events $\ecalb, \tilde\ecalb$,  $\mathbf{P}(\ecalb|\tilde\ecalb)$ refers to the conditional probability of $\ecalb$ given $\tilde\ecalb$, i.e. $\mathbf{P}(\ecalb|\tilde\ecalb) := \mathbf{P}(\ecalb, \tilde\ecalb)/\mathbf{P}(\tilde\ecalb)$.

\item For a r.v. $X$, and a set $\calb$ of values that the r.v. $Z$ can take, the notation $\mathbf{P}(\ecalb|X)$ is defined as %should be interpreted as %and the corresponding event $\ecalb$, we define
$$\mathbf{P}(\ecalb|X) := \E[\mathbb{I}_{\calb}(Z)|X].$$
Notice that $\mathbf{P}(\ecalb|X)$ is a r.v. (it is a function of the r.v. $X$) that always lies between zero and one.
%
%\item Notice that $\mathbf{P}(\ecalb)$ and $\mathbf{P}(\ecalb|\tilde\ecalb)$ are probabilities, i.e. they take values only between zero and one. However  $\mathbf{P}(\ecalb|X)$ is a r.v. -- it is a function of the r.v. $X$. %not a ``probability" (it is always greater than or equal to zero, but it need not be less than one).
\een
\label{probdefs}
\end{definition}

Finally, RHS refers to the right hand side of an equation or inequality; w.p. means ``with probability"; and w.h.p. means ``with high probability".
%Also we use $a \lesssim b$ to indicate (in a non-rigorous sense) that the dominant term in the upper bound on $a$ is $b$.

\subsection{Compressive Sensing result}
The error bound for noisy compressive sensing (CS) based on the RIC is as follows \cite{candes_rip}.
%\begin{definition}\label{defn_delta}
%The $s$-restricted isometry constant\cite{decodinglp}, $\delta_s$, for an $n \times m$ matrix $\Psi$ ($n \leq m$), is the smallest real number satisfying
%\beq
%(1-\delta_s) \|x\|_2^2 \leq \|\Psi_T x\|_2^2 \leq (1+\delta_s) \|x\|_2^2 \nn
%\eeq
%for all subsets $T \subset [1:m]$ of cardinality $|T| \leq s$ and all real vectors $x$ of length $|T|$.
%\end{definition}
%\begin{remark}
%According to Definition \ref{defn_delta}, an alternative way of defining $\delta_s(\Psi)$ is
%$$\delta_s(\Psi) = \max\{ \max_{|T| \leq s} (1-\lambda_{\min} ({\Psi_T}'\Psi_T)), \max_{|T| \leq s} (\lambda_{\max} ({\Psi_T}'\Psi_T) - 1)\}$$
%\end{remark}
%
\begin{theorem}[\cite{candes_rip}]
\label{candes_csbound}
Suppose we observe
\beq
y := \Psi x + z \nn
\eeq
where $z$ is the noise. Let $\hat{x}$ be the solution to following problem
\beq
\min_{x} \|x\|_1 \ \text{subject to} \  \|y - \Psi x\|_2 \leq \xi  \label{*}
\eeq
Assume that $x$ is $s$-sparse, $\|z\|_2 \leq \xi$, and $\delta_{2s}(\Psi) < b (\sqrt{2}-1)$ for some $0 \le b < 1$. Then the solution of (\ref{*}) obeys
$$\|\hat{x} - x\|_2 \leq C_1 \xi$$
with $\ds C_1 = \frac{4\sqrt{1+\delta_{2s}(\Psi)}}{1-(\sqrt{2}+1)\delta_{2s}(\Psi)} \le \frac{4\sqrt{1+b (\sqrt{2}-1)}}{1-b}$.
\end{theorem}

\begin{remark}
Notice that if $b$ is small enough, $C_1$ is a small constant but $C_1 >1$. For example, if $\delta_{2s}(\Psi) \le 0.15$, then $C_1 \le 7$.
If $C_1 \xi  > \|x\|_2$, the normalized reconstruction error bound would be greater than $1$, making the result useless. Hence, (\ref{*}) gives a small reconstruction error bound only for the small noise case, i.e., the case where $\|z\|_2 \leq \xi \ll \|x\|_2$. %In fact this is true for a lot of existing literature on CS and sparse recovery, with the exception of work that studies sparse recovery in sparse noise, e.g. \cite{error_correction_PCP_l1,Laska_exactsignal,trac_tran}, or sparse plus low-rank recovery, e.g. \cite{rpca,rpca2, novel_m_estimator,mccoy_tropp11, outlier_pursuit,rpca_tropp, linear_inverse_prob, rpca_reduced, rpca_Giannakis,rpca_zhang,compressivePCP,noisy_undersampled_yuan}.
%\cite{error_correction_PCP_l1,Laska_exactsignal,trac_tran} (focus on large but sparse noise) and \cite{rpca,rpca2}.% (batch solutions for large but low dimensional noise).
% most CS works focus on the small noise case while we focus on the large noise case.
\end{remark}

\subsection{Results from linear algebra}

Davis and Kahan's $\sin \theta$ theorem \cite{davis_kahan} studies the rotation of eigenvectors by perturbation.

\begin{theorem}[$\sin \theta$ theorem \cite{davis_kahan}] \label{sin_theta}
Given two Hermitian matrices $\mathcal{A}$ and $\mathcal{H}$ satisfying
\begin{align}
 \label{sindecomp}
\mathcal{A} &= \left[ \begin{array}{cc} E & E_{\perp} \\ \end{array} \right]
\left[\begin{array}{cc} A\ & 0\ \\ 0 \ & A_{\perp} \\ \end{array} \right]
\left[ \begin{array}{c} E' \\ {E_{\perp}}' \\ \end{array} \right]\nn, \\
\mathcal{H} &= \left[ \begin{array}{cc} E & E_{\perp} \\ \end{array} \right]
\left[\begin{array}{cc} H \ & B'\ \\ B \ & H_{\perp} \\ \end{array} \right]
\left[ \begin{array}{c} E' \\ {E_{\perp}}' \\ \end{array} \right]
\end{align}
where $[E \ E_{\perp}]$ is an orthonormal matrix. The two ways of representing $\mathcal{A}+\mathcal{H}$ are
\begin{align*}
\mathcal{A} + \mathcal{H}  &= \left[ \begin{array}{cc} E & E_{\perp} \\ \end{array} \right]
\left[\begin{array}{cc} A + H \ & B'\ \\ B \ & A_{\perp} + H_{\perp} \\ \end{array} \right]
\left[ \begin{array}{c} E' \\ {E_{\perp}}' \\ \end{array} \right] \\
&= \left[ \begin{array}{cc} F & F_{\perp} \\ \end{array} \right]
\left[\begin{array}{cc} \Lambda\ & 0\ \\ 0 \ & \Lambda_{\perp} \\ \end{array} \right]
\left[ \begin{array}{c} F' \\ {F_{\perp}}' \\ \end{array} \right] \nn
\end{align*}
where $[F\ F_{\perp}]$ is another orthonormal matrix. Let $\mathcal{R} := (\mathcal{A}+\mathcal{H}) E - \mathcal{A}E = \mathcal{H} E $. If $ \lambda_{\min}(A) >\lambda_{\max}(\Lambda_{\perp})$, then
\beq
\|(I-F F')E \|_2 \leq \frac{\|\mathcal{R}\|_2}{\lambda_{\min}(A) - \lambda_{\max}(\Lambda_{\perp})} \nn
\eeq

\end{theorem}
The above result bounds the amount by which the two subspaces $\Span(E)$ and $\Span(F)$ differ as a function of the norm of the perturbation $\|\mathcal{R}\|_2$ and of the gap between the minimum eigenvalue of $A$ and the maximum eigenvalue of $\Lambda_{\perp}$. Next, we state Weyl's theorem which bounds the eigenvalues of a perturbed Hermitian matrix, followed by Ostrowski's theorem.

\begin{theorem}[Weyl \cite{hornjohnson}]\label{weyl}
Let $\mathcal{A}$ and $\mathcal{H}$ be two  $n \times n$ Hermitian matrices. For each $i = 1,2,\dots,n$ we have
$$\lambda_i(\mathcal{A}) + \lambda_{\min}(\mathcal{H}) \leq \lambda_i(\mathcal{A}+\mathcal{H}) \leq \lambda_i(\mathcal{A}) + \lambda_{\max}(\mathcal{H})$$
\end{theorem}

\begin{theorem}[Ostrowski \cite{hornjohnson}]\label{ost}
Let $H$ and $W$ be $n \times n$ matrices, with $H$ Hermitian and $W$ nonsingular. For each $i=1,2 \dots n$, there exists a positive real number $\theta_i$ such that $\lambda_{\min} (WW') \leq \theta_i \leq \lambda_{\max}(W{W}')$ and $\lambda_i(W H {W}') = \theta_i \lambda_i(H)$. Therefore,
$$\lambda_{\min}(W H {W}') \geq \lambda_{\min} (W{W}') \lambda_{\min} (H)$$
\end{theorem}

The following lemma proves some simple linear algebra facts.
\begin{lem} \label{lemma0}\label{hatswitch}
Suppose that $P$, $\Phat$ and $Q$ are three basis matrices. Also, $P$ and $\Phat$ are of the same size, ${Q}'P = 0$ and $\|(I-\Phat{\Phat}')P\|_2 = \zeta_*$. Then,
\ben
  \item $\|(I-\Phat{\Phat}')PP'\|_2 =\|(I - P{P}')\Phat{\Phat}'\|_2 =  \|(I - P P')\Phat\|_2 = \|(I - \Phat \Phat')P\|_2 =  \zeta_*$
  \item $\|P{P}' - \Phat {\Phat}'\|_2 \leq 2 \|(I-\Phat{\Phat}')P\|_2 = 2 \zeta_*$
  \item $\|{\Phat}' Q\|_2 \leq \zeta_*$ \label{lem_cross}
  \item $ \sqrt{1-\zeta_*^2} \leq \sigma_i((I-\Phat \Phat')Q)\leq 1 $
\een
Further, if $P$ is an $n \times r_1$ basis matrix and $\Phat$ is an $n \times r_2$ basis matrix with $r_2 \geq r_1$,
then $\|(I-\Phat{\Phat}')PP'\|_2 \leq \|(I - P{P}')\Phat{\Phat}'\|_2$
\end{lem}
The proof is in the Appendix.

\subsection{High probability tail bounds for  sums of independent random matrices}

The following lemma follows easily using Definition \ref{probdefs}. We will use this at various places in the paper.
\begin{lem}
Suppose that $\calb$ is the set of values that the r.v.s $X,Y$ can take. Suppose that $\calc$ is a set of values that the r.v. $X$ can take.
For a $0 \le p \le 1$, if $\mathbf{P}(\ecalb|X) \ge p$ for all $X \in \calc$,  then $\mathbf{P}(\ecalb|\ecalc) \ge p$ as long as $\mathbf{P}(\ecalc)> 0$.
\label{rem_prob}
\end{lem}
The proof is in the Appendix.

The following lemma is an easy consequence of the chain rule of probability applied to a contracting sequence of events.
\begin{lem} \label{subset_lem}
For a sequence of events $E_0^e, E_1^e, \dots E_m^e$ that satisfy $E_0^e \supseteq E_1^e  \supseteq E_2^e \dots  \supseteq E_m^e$, the following holds
$$\mathbf{P}(E_m^e|E_0^e) = \prod_{k=1}^{m} \mathbf{P}(E_k^e | E_{k-1}^e).$$
\end{lem}
\begin{proof}
$\mathbf{P}(E_m^e|E_0^e) = \mathbf{P}(E_m^e, E_{m-1}^e, \dots E_0^e | E_0^e) = \prod_{k=1}^{m} \mathbf{P}(E_k^e | E_{k-1}^e, E_{k-2}^e, \dots E_0^e) = \prod_{k=1}^{m} \mathbf{P}(E_k^e | E_{k-1}^e)$.
\end{proof}

Next, we state the matrix Hoeffding inequality \cite[Theorem 1.3]{tail_bound} which gives tail bounds for sums of independent random matrices.

\begin{theorem}[Matrix Hoeffding for a zero mean Hermitian matrix \cite{tail_bound}]\label{hoeffding}
Consider a finite sequence $\{Z_t\}$ of independent, random, Hermitian matrices of size $n\times n$, and let $\{A_t\}$ be a sequence of fixed Hermitian matrices. Assume that each random matrix satisfies (i) $\mathbf{P}(Z_t^2 \preceq A_t^2)=1$ and (ii) $\E(Z_t) = 0$. Then, for all $\epsilon >0$, % almost surely
\[
\mathbf{P} \left(\lambda_{\max}\left(\sum_t Z_t \right) \leq \epsilon \right) \geq 1 - n \exp\left(\frac{-\epsilon^2}{8 \sigma^2}\right),
\]
where $\sigma^2 = \Big\| \sum_t A_t^2 \Big\|_2$.
\end{theorem}

The following two corollaries of Theorem \ref{hoeffding} are easy to prove. The proofs are given in Appendix \ref{appendix prelim}.

\begin{corollary}[Matrix Hoeffding conditioned on another random variable for a nonzero mean Hermitian matrix]\label{hoeffding_nonzero}
Given an $\alpha$-length sequence $\{Z_t\}$ of random Hermitian matrices of size $n\times n$, a r.v. $X$, and a set ${\cal C}$ of values that $X$ can take. Assume that, for all $X \in \calc$, (i) $Z_t$'s are conditionally independent given $X$; (ii) $\mathbf{P}(b_1 I \preceq Z_t \preceq b_2 I|X) = 1$ and (iii) $b_3 I \preceq \frac{1}{\alpha}\sum_t \E(Z_t|X) \preceq b_4 I $. Then for all $\epsilon > 0$,
\begin{multline}
\mathbf{P} \left( \lambda_{\max}\left(\frac{1}{\alpha}\sum_t Z_t \right) \leq b_4 + \epsilon \Big | X \right) \\
\geq 1- n \exp\left(\frac{-\alpha \epsilon^2}{8(b_2-b_1)^2}\right) \ \text{for all} \ X \in \calc
\nn 
\end{multline}
\begin{multline}
\mathbf{P} \left(\lambda_{\min}\left(\frac{1}{\alpha}\sum_t Z_t \right) \geq b_3 -\epsilon \Big| X \right) \\
\geq  1- n \exp\left(\frac{-\alpha \epsilon^2}{8(b_2-b_1)^2} \right) \ \text{for all} \ X \in \calc \nn
\end{multline}
\end{corollary}
The proof is in Appendix \ref{appendix prelim}.

\begin{corollary}[Matrix Hoeffding conditioned on another random variable for an arbitrary nonzero mean matrix]\label{hoeffding_rec}
Given an $\alpha$-length sequence $\{Z_t\}$ of random matrices of size $n_1 \times n_2$, a r.v. $X$, and a set ${\cal C}$ of values that $X$ can take. Assume that, for all $X \in \calc$, (i) $Z_t$'s are conditionally independent given $X$; (ii) $\mathbf{P}(\|Z_t\|_2 \le b_1|X) = 1$ and (iii) $\|\frac{1}{\alpha}\sum_t \E( Z_t|X)\|_2 \le b_2$. Then, for all $\epsilon >0$,
\begin{multline*}
\mathbf{P} \left(\bigg\|\frac{1}{\alpha}\sum_t Z_t \bigg\|_2 \leq b_2 + \epsilon \Big| X \right) \\
\geq 1-(n_1+n_2) \exp\left(\frac{-\alpha \epsilon^2}{32 b_1^2}\right)  \ \text{for all} \ X \in \calc
\end{multline*}
\end{corollary}
The proof is in Appendix \ref{appendix prelim}.

\section{Problem Definition and Model Assumptions}
\label{probdef}
We give the problem definition below followed by the model and then describe the two key assumptions.

\subsection{Problem Definition}
\label{model}
The measurement vector at time $t$, $M_t$, is an $n$ dimensional vector which can be decomposed as
\beq
M_t = L_t + S_t \label{problem_defn}
\eeq
Here $S_t$ is a sparse vector with support set size at most $s$ and minimum magnitude of nonzero values at least $S_{\min}$. $L_t$ is a dense but low dimensional vector, i.e. $L_t = P_{(t)} a_t$ where $P_{(t)}$ is an $n \times r_{(t)}$ basis matrix with $r_{(t)} < n$, that changes every so often according to the model given below. We are given an accurate estimate of the subspace in which the initial $t_\train$ $L_t$'s lie, i.e. we are given a basis matrix $\Phat_0$ so that $\|(I-\Phat_0 \Phat_0')P_0 \|_2$ is small.
Here $P_0$ is a basis matrix for $\Span({\cal L}_{t_{\train}})$, i.e. $\Span(P_0) = \Span({\cal L}_{t_{\train}})$. Also,  for the first $t_{\train}$ time instants, $S_t$ is zero.
The goal is
\ben
\item to estimate both $S_t$ and $L_t$ at each time $t > t_\train$, and
\item to estimate $\Span({\cal L}_t)$ every so often, i.e. compute $\Phat_{(t)}$ so that  the subspace estimation error, $\SE_{(t)}:=\|(I- \hat{P}_{(t)} \hat{P}_{(t)}')P_{(t)}\|_2$ is small. \label{item2}
\een

We assume a subspace change model that allows the subspace to change at certain change times $t_j$ rather than continuously at each time. It should be noted that this is only a model for reality. In practice there will typically be some changes at every time $t$; however this is difficult to model in a simple fashion. Moreover the analysis for such a model will be a lot more complicated. However, we do allow the variance of the projection of $L_t$ along the subspace directions to change continuously. The projection along the new directions is assumed to be small initially and allowed to gradually increase to a large value (see Sec \ref{slowss}). %??? in deletion model: need to say this: allow the new directions to increase only until $tj+1 - \vartheta \alpha$ or actually need not say it explicitly. %better: do clustering at $t= t_{j+1} - \vartheta \alpha$ and make the same change to

%The use of a piecewise constant model i
\begin{sigmodel}[Model on $L_t$] \label{Ltmodel} \
\ben
%\item $L_t$ lies in a low dimensional subspace that changes every-so-often.  Then the following holds.
%Let $t_j$ denote the subspace change times.
\item We assume that $L_t = P_{(t)} a_t$ with $P_{(t)} = P_j$ for all $t_j \leq t <t_{j+1}$, $j=0,1,2 \cdots J$. Here  $P_j$ is an $n \times r_j$ basis matrix with $r_j  < \min(n,(t_{j+1} - t_j))$ that changes as
    $$P_j = [P_{j-1} \ P_{j,\new}]$$
    where $P_{j,\new}$ is a $n \times c_{j,\new}$ basis matrix with $P_{j,\new}'P_{j-1} = 0$.   Thus
    $$r_j = \rank(P_j) =  r_{j-1} + c_{j,\new}.$$
We let $t_0=0$. Also $t_{J+1}$ can be the length of the sequence or $t_{J+1} = \infty$.     This model is illustrated in Figure \ref{model_fig}.

\item The vector of coefficients, $a_t:={P_{(t)}}'L_t$, is a zero mean random variable (r.v.) with mutually uncorrelated entries, i.e. $\E[a_t]=0$ and $\E[(a_t)_i (a_t)_j]=0$ for $i \neq j$. %???
.%, i.e. $\E(a_t) = 0$ for all times $t$. %  is bounded, and the $a_t$'s are mutually independent over time, $t$.
\een
\end{sigmodel}

\begin{figure}
\centerline{
\epsfig{file = 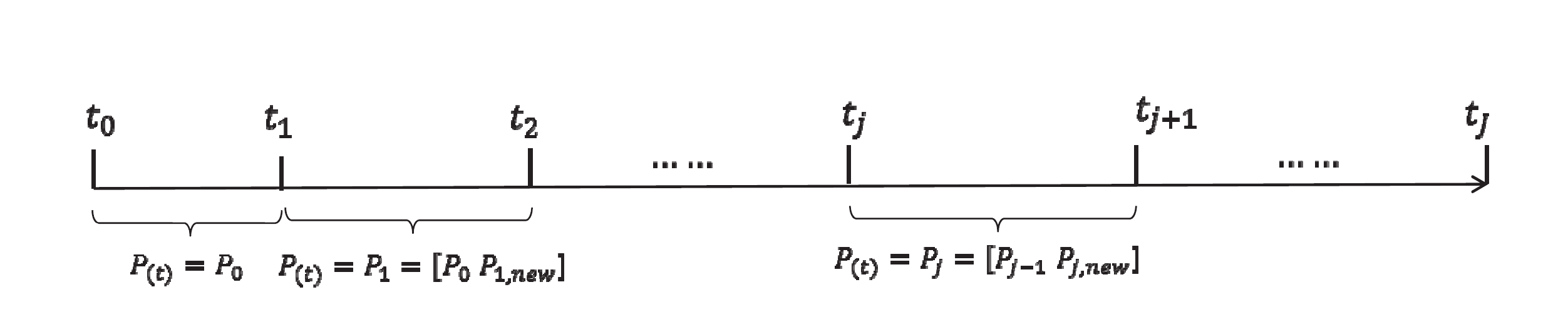, width = \columnwidth}
}
\caption{\small{The subspace change model explained in Sec \ref{model}. Here $t_0=0$ and $0 < t_\train < t_1$.}}
\label{model_fig}
\end{figure}

%Here we define some parameters for the above model.
\begin{definition}
Define the covariance matrix of $a_t$ to be the diagonal matrix
$$\Lambda_t: = \mathrm{Cov}[a_t] = \E(a_ta_t').$$
%Without loss of generality we can assume that $\Lambda_t$ is diagonal (the directions $P_j$ are obtained as the eigenvectors with nonzero eigenvalues of $\text{Cov}(L_t)$ and $\Lambda_t$ is the matrix of its eigenvalues). ???
Define
For $t_j \le t < t_{j+1}$, $a_t$ is an $r_j$ length vector which can be split as
  $$a_t ={P_j}'L_t = \vect{a_{t,*}}{a_{t,\new}}$$
where $a_{t,*}: = {P_{j-1}}'L_t$ and $a_{t,\new}: = {P_{j,\new}}'L_t$. % is an $r_{j-1}$ length vector of coefficients for the existing directions %is a $c_{j,\new}$ length vector of coefficients for the new directions.
Thus, for this interval, $L_t$ can be rewritten as
\[
L_t = \left[ P_{j-1} \ P_{j,\new}\right] \vect{a_{t,*}}{a_{t,\new}} = P_{j-1} a_{t,*} + P_{j,\new} a_{t,\new}
\]
Also, $\Lambda_t$ can be split as
$$\Lambda_t = \left[ \begin{array}{cc}
(\Lambda_t)_* &  0 \nn \\
0  & (\Lambda_t)_\new  \nn \\
\end{array}
\right]
$$
where $(\Lambda_t)_* = \mathrm{Cov}[a_{t,*}] $ and $(\Lambda_t)_\new = \operatorname{Cov}[a_{t,\new}]$ are diagonal matrices.
Define
\begin{align*}
\lambda^- &: = \inf_t \lambda_{\min}(\Lambda_t), \quad \lambda^+:=\sup_t \lambda_{\max} (\Lambda_t),\\
&\text{and}\\
\lambda_\new^- &: = \inf_t \lambda_{\min}((\Lambda_t)_\new), \quad
\lambda_\new^+ :=\sup_t \lambda_{\max} ( (\Lambda_t)_\new).
\end{align*}
Also let,
\[
f : = \frac{\lambda^+}{\lambda^-}
\]
and
\[
g : = \frac{\lambda_\new^+}{\lambda_\new^-}.
\]
\end{definition}

The above simple model only allows new additions to the subspace and hence the rank of $P_j$ can only grow over time. The ReProCS algorithm designed for this model can be interpreted as a recursive algorithm for solving the robust PCA problem studied in \cite{rpca} and other batch robust PCA works. At time $t$ we estimate the subspace spanned by $L_1,L_2, \dots L_t$. For the above model, the subspace dimension is bounded by $r_0+J c_{\max}$. Thus a bound on $J$ is needed to keep the subspace dimension small at all times. We remove this limitation in Sec \ref{Del_section} where we also allow for subspace deletions and correspondingly design a ReProCS algorithm that does the same thing. For that algorithm, as we will see, we will not need a bound on the number of changes, $J$, as long as the separation between the subspace change times is allowed to grow logarithmically with $J$ and a clustering assumption holds.
%This can be modeled as $P_j = [(P_{j-1}R_j \setminus P_{j,\old}) \ P_{j,\new}] $ %the subspace basis matrix

Define the following quantities for the sparse part.
\begin{definition}
Let $T_t :=\{i: \  (S_t)_i \neq 0 \}$ denote the support of $S_t$. Define
\[
S_{\min}: = \min_{t> t_\train} \min_{i \in T_t} |(S_t)_i |, \ \text{and} \ s: = \max_t |T_t|
\]
%In words, $S_{\min}$ is a lower bound on the magnitude of a non-zero entry of $S_t$ for all $t$, and $s$ is an upper bound on the support size of $S_t$ for all $t$.
\end{definition}

\subsection{Slow subspace change}
\label{slowss}
By slow subspace change we mean all of the following.

First, the delay between consecutive subspace change times is large enough, i.e., for a $d$  large enough,
\bea
t_{j+1} - t_j \ge d
\label{delay}
\eea
Second, the magnitude of the projection of $L_t$ along the newly added directions, $a_{t,\new}$, is initially small but can increase gradually. We model this as follows. Assume that for an $\alpha > 0$ \footnote{As we will see in the algorithm $\alpha$ is the number of previous frames used to get a new estimate of $P_{j,\new}$.} the following holds
\bea
\|a_{t,\new}\|_{\infty} \le \min\Big(v^{\frac{t-t_j}{\alpha}-1} \gamma_\new,\gamma_*\Big)
\label{atnew_inc}
\eea
when $ t \in [t_j, t_{j+1}-1]$
 for a $v>1$ but not too large and with $\gamma_\new < \gamma_* \ \text{and} \  \gamma_\new < S_{\min}$. Clearly, the above assumption implies that %$\|a_{t,\new}\|_{\infty} \le \gamma_\new$ in the first $\alpha$ frames after $t_j$ and in the $k^{th}$ such interval it satisfies $\|a_{t,\new}\|_{\infty} \le \min(v^{k-1} \gamma_\new,\gamma_*)$.
\[
\|a_{t,\new}\|_{\infty} \le \gamma_{\new,k}:= \min(v^{k-1} \gamma_\new,\gamma_*)
\]
for all $t \in [t_j + (k-1) \alpha, t_{j}+k\alpha-1]$.
 This assumption is verified for real video data in Sec. \ref{model_verify}.

Third, the number of newly added directions is small, i.e.  $c_{j,\new} \le c_{\max} \ll r_{0}$. This is also verified in Sec. \ref{model_verify}.% for real video data

\begin{remark}[Large $f$] \label{large_f}
Since our problem definition allows large noise, $L_t$, but assumes slow subspace change, thus the maximum condition number of $\operatorname{Cov}[L_t]$, which is bounded by $f$, cannot be bounded by a small value. The reason is as follows. Slow subspace change implies that the projection of $L_t$ along the new directions is initially small, i.e. $\gamma_\new$ is small. Since $\lambda^- \le \gamma_\new$, this means that $\lambda^-$ is small.
Since $\E[\|L_t\|^2] \le r_{\max} \lambda^+ $ and $r_{\max}$ is small (low-dimensional), thus, large $L_t$ means that $\lambda^+$ needs to be large. As a result $f=\lambda^+/\lambda^-$ cannot be upper bounded by a small value.
%
%Notice from Sec \ref{probdef} that $\lambda^- \le \gamma_\new$ and $\E[\|L_t\|_2^2] \le r \lambda^+$. Slow subspace change implies that $\gamma_\new$ is small. Thus, $\lambda^-$ is small. However, to allow $L_t$ to have large average magnitude, $\lambda^+$ needs to be large. Thus,  $f = \lambda^+ / \lambda^-$ cannot be small unless we require that $L_t$ has small magnitude for all times $t$. This, however, violates our problem formulation where say that we allow large but strucutured noise $L_t$.
\end{remark}

%Notice that $\lambda^- \leq \gamma_{\new}$ by the definition of $\lambda^-$. chenlu
%for $t_j \le t < t_j + 1$, $\|a_{t,\new}\|_{\infty} \le \gamma_{\new,(t)} = \min(v^{(t-t_j - \alpha-1)^+} \gamma_\new,\gamma_*)$ for a $v > 1$. Thus, $a_{t,\new}$ satisfies $\max_{t_j + (k-1) \alpha \le t < t_j + k \alpha} \|a_{t,\new}\|_{\infty} \le \min((v^{\alpha})^{k-1} \gamma_\new,\gamma_*)$ for any $k \ge 1$. ?? edit this ??

\subsection{Measuring denseness of a matrix and its relation with RIC}
\label{denseness}
%For a tall $n \times r$ matrix, $B$, or for a $n \times 1$ vector, $B$, we define the the denseness coefficient as follows: $$\kappa_s(B) := \max_{|T| \le s} \frac{\|{I_T}' B\|_2}{\|B\|_2}.$$ where $\|.\|_2$ is the matrix or vector 2-norm respectively \footnote{In future work \cite{long_undersamp} we define $\kappa_s(B):=\max_{|T| \le s} \|{I_T}' Q(B)\|_2$ where $Q(B)$ is a basis matrix for $\Span(B)$ ie the columns of $Q(B)$ form an orthonormal basis for $\Span(B)$. With this definition item 4 of Remark \ref{kapparemark} is immediate. Also, all of our results will still hold with this new definition.}.
Before we can state the denseness assumption, we need to define the denseness coefficient.
\begin{definition}[denseness coefficient]\label{subspace_kappa}
For a matrix or a vector $B$, define
\beq %T \subseteq \{1,2,\cdots, n\},
\kappa_s(B)=\kappa_s(\Span(B)) : = \max_{|T| \le s} \|{I_T}' \mathrm{basis}(B)\|_2
\eeq
where $\|.\|_2$ is the vector or matrix $\ell_2$-norm.
\end{definition}
%Recall that $\text{basis}(B)$ is short for $\text{basis}(\Span(B))$. Similarly $\kappa_s(B)$ is short for $\kappa_s(\Span(B))$.

Clearly, $\kappa_s(B) \le 1$. First consider an $n$-length vector $B$. Then $\kappa_s$ measures the denseness (non-compressibility) of $B$. A small value indicates that the entries in $B$ are spread out, i.e. it is a dense vector. A large value indicates that it is compressible (approximately or exactly sparse). The worst case (largest possible value) is $\kappa_s(B)=1$ which indicates that $B$ is an $s$-sparse vector. The best case is $\kappa_s(B) = \sqrt{s/n}$ and this will occur if each entry of $B$ has the same magnitude. %, $\|B\|/\sqrt{n}$.
Similarly, for an $n \times r$ matrix $B$, a small $\kappa_s$ means that most (or all) of its columns are dense vectors.
\begin{remark}\label{kapparemark}
The following facts should be noted about $\kappa_s(.)$:
\ben
\item For a given matrix $B$, $\kappa_s(B)$ is an non-decreasing function of $s$.

\item $\kappa_s([B_1]) \leq \kappa_s([B_1 \ B_2])$ i.e. adding columns cannot decrease $\kappa_s$.
\item A bound on $\kappa_s(B)$ is $\kappa_s(B) \le \sqrt{s} \kappa_1(B)$.
This follows because $\|B\|_2 \leq \big\| \big[\|b_1\|_2 \dots \|b_r\|_2 \big] \big\|_2$ where $b_i$ is the $i^{\text{th}}$ column of $B$.

%A similar loose bound is $\kappa_s(B) \le r \max_{i} \kappa_s(B_i)$ where $B_i$ is the $i^{th}$ column of $B$. ???
%kappa(B) = kappa(B^r) where B^r contains the r columns of B that span span(B). Then using triangle inequality we get the above.

%\item From the definition it is clear that $\kappa_s$ is a property of the subspace spanned by the columns of $B$ and not of the entries of $B$ itself.
%For a basis matrix $P$, $\|P\|_2 =1$ and hence $\kappa_s(P)  = \max_{|T| \le s} {\|I_T' P\|_2}$ and $\kappa_s(PP') = \kappa_s(P)$. Thus, for any other basis matrix $Q$ for which $\Span(Q) = \Span(P)$, $\kappa_s(P) = \kappa_s(Q)$. Thus, $\kappa_s(P)$ is a property of $\Span(P)$, which is the subspace spanned by the columns of $P$, and not of the actual entries of $P$.
%    In particular, this means that if most columns of ${\cal L}_t$ are dense, then $P_j$ is dense too, i.e. $\kappa_s(P_j)$ is small. Recall that $\Span({\cal L}_t) = \Span(P_j)$ for $t_j \le t < t_{j+1}-1$.
% In particular, this means that the assumption $P_{j}$ is dense is equivalent to assuming that any basis for $\Span({\cal L}_t)$ is dense for $t_j \le t < t_{j+1}-1$.
\een
\label{dense_remark}
\end{remark}
%
%\footnote{This follows because $\Span(Q) = \Span(P)$ is equivalent to $QQ' = PP'$.} ($RR'=I=R'R$) \footnote{i.e. $Q$ satisfies $Q=PR$ where $R$ is orthonormal}

The lemma below relates the denseness coefficient of a basis matrix $P$ to the RIC of $I-PP'$. The proof is in the Appendix. %Its proof is in \cite{long}.
\begin{lem}\label{delta_kappa}
For an $n \times r$ basis matrix $P$  (i.e $P$ satisfying $P'P=I$), %the RIP constant of $I-PP'$ and the diversity constant of $P$ satisfies
$$\delta_s(I-PP') = \kappa_s^2 (P).$$
\end{lem}
In other words, if $P$ is dense enough (small $\kappa_s$), then the RIC of $I-PP'$ is small.
%Thus, using Theorem \ref{candes_csbound}, all $s$-sparse vectors, $S_t$ can be exactly recovered from $y_t := (I-PP')S_t$ by $\ell_1$ minimization. %\cite[Theorem 1.2]{candes_rip}.

% PCP assumes upper bounds on both $\kappa_1(U)$ and $\kappa_1(V)$ where $U$ and $V$ are the left and right singular vectors of ${\cal L}_t$. In addition, it also requires an upper bound on $\|UV'\|_{\max}$. % where  $\|B\|_{\max} := \max_{i_1,i_2} |(B)_{i_1,i_2}|$.

In this work, we assume an upper bound on $\kappa_{2s}(P_{j})$ for all $j$, and a tighter upper bound on $\kappa_{2s}(P_{j,\new})$, i.e., there exist $\kappa_{2s,*}^+<1$ and a $\kappa_{2s,\new}^+ < \kappa_{2s,*}^+$ such that
\begin{align}
\max_{j} \kappa_{2s}(P_{j-1}) \leq \kappa_{2s,*}^+ \label{kappa plus}\\
\max_{j} \kappa_{2s}(P_{j,\new}) \leq \kappa_{2s,\new}^+ \label{kappa new plus}
\end{align}
Additionally, we also assume denseness of another matrix, $D_{j,\new,k}$, whose columns span the currently unestimated part of $\Span(P_{j,\new})$ (see Theorem \ref{thm1}).

The denseness coefficient $\kappa_s(B)$ is related to the denseness assumption required by PCP \cite{rpca}. That work uses $\kappa_1(B)$ to quantify denseness.

%Recall that, for $t_{j} \le t < t_{j+1}$, $\Span([P_{j-1}, P_{j,\new}]) = \Span({\cal L}_t) = \Span(U)$. Hence this assumption is only a subset of the assumptions required by PCP \cite{rpca} (we assume nothing about $\kappa_s(V)$ or $\|UV'\|_{\max}$).

%and its the estimates.% This fact is partly verified by simulations.  the newly added subspace,

%We also introduce ReProCS with deletion. This relaxes the denseness requirement even further. It only requires bounds on $\kappa_s(P_j)$ and $\kappa_s(P_{j,\new})$.

%\begin{ass}
%We assume that the matrices $P_{j-1}$ and $P_{j,\new}$ are dense.  That is there exist $\kappa_{2s,*}^+$ and $\kappa_{2s,\new}^+$ such that
%\begin{align}
%\max_{j} \kappa_{2s}(P_{j-1}) \leq \kappa_{2s,*}^+ \label{kappa plus}\\
%\max_{j} \kappa_{2s}(P_{j,\new}) \leq \kappa_{2s,\new}^+ \label{kappa new plus}
%\end{align}
%\end{ass}
%

\section{Recursive Projected CS (ReProCS) and its Performance Guarantees}
\label{reprocs_sec}
In this section we introduce the ReProCS algorithm and state the performance guarantee for it.
We begin by first stating the result in \ref{result}, and then
describe and explain the algorithm in Section \ref{basic_rep}.
In Section \ref{proj PCA} we describe the projection-PCA algorithm that is used in the ReProCS algorithm.
The assumptions used by the result are discussed in Section \ref{discuss_add}.
%A qualitative discussion w.r.t. the result for PCP is given in Section \ref{discuss_pcp}.
%Practical parameter setting for ReProCS is discussed later in Section \ref{prac_rep}. %An extension of ReProCS that includes a deletion step is briefly described in Sec \ref{rep_del}.

\subsection{Performance Guarantees}
\label{result}
We state the main result here and then discuss it in Section \ref{discuss_add}. Definitions needed for the proof are given in Section \ref{detailed} and the actual proof is given in Section \ref{mainlemmas}.

%From the theorem, $K(J+1) \alpha$ is a lower bound on the duration of the entire sequence whereas $n$ is the length of one vector. Thus the above says that the dependence of $\alpha$ on both the sequence duration and the vector length is only logarithmic.

\begin{definition}\label{defn_alpha}
We define here the parameters that will be used in Theorem \ref{thm1}.
\ben
\item Let $c:= c_{\max}$ and $r:= r_0 + (J-1)c$.
\item  Define $K=K(\zeta) := \left\lceil\frac{\log(0.6c\zeta)}{\log {0.6}} \right\rceil$
\item Define $\xi_0(\zeta)  :=  \sqrt{c} \gamma_{\new} + \sqrt{\zeta}(\sqrt{r}  +  \sqrt{c})$
%\item Define  $\rho  :=  \max_{t}\{\kappa_1 (\hat{S}_{t,\text{cs}} - S_t)\}$. Notice that $\rho \le 1$.
\item Define
\begin{multline*}
 \alpha_\add(\zeta)  :=  \left\lceil (\log 6KJ + 11 \log n) \frac{8 \cdot 24^2} {\zeta^2 (\lambda^-)^2}\cdot \right. \\
 \max\left(\min(1.2^{4K} \gamma_{\new}^4, \gamma_*^4), \frac{16}{c^2}, \right. \\  4(0.186 \gamma_\new^2 + 0.0034 \gamma_\new + 2.3)^2 \Big)
\Bigg\rceil
\end{multline*}
%$$ \alpha_\add  = \lceil (\log 6 K(\zeta)J + 11 \log n) \frac{4608 \ \min(1.2^{4K(\zeta)} \gamma_{\new}^4, \gamma_*^4)}{\zeta^2 (\lambda^-)^2} \rceil$$
We note that $\alpha_{\text{add}}$ is the number of data points, $\alpha$, used for one projection PCA step and is chosen to ensure that the conclusions of Theorem \ref{thm1} hold with probability at least $(1 -   n^{-10})$.
If $\gamma_*$ is large enough (${\gamma_*}^4>16$), a simpler but larger value for $\alpha_\add(\zeta)$ is
$$\alpha_\add(\zeta) = \left\lceil (\log 6KJ + 11 \log n) \frac{ 8 \cdot 24^2 \gamma_*^4}{\zeta^2 (\lambda^-)^2} \right\rceil$$
%For most probability distributions on the $a_t$'s, $\gamma_*^2$ is proportional to $\lambda^+$, e.g. in the case of $(a_t)_i$'s being i.i.d. zero mean uniform as in our simulations (see Sec \ref{sims}),  $\gamma_*^2 = 4\lambda^+$. In this case, for $\gamma_*$ large enough, a simpler but larger value for $\alpha_\add(\zeta)$ is $\alpha_\add(\zeta) = \left\lceil C(\log 6KJ + 11 \log n) \frac{f^2}{\zeta^2} \right\rceil$ for a large constant $C$, e.g. in the uniform case, $C=8 \cdot 24^2 \cdot 4$ for a large constant $C$, e.g. in the uniform case, $C=8 \cdot 24^2 \cdot 4$.
\een
%For most distributions, $\gamma_*^2$ is proportional to $\lambda^+$, e.g. in the case of certain $(a_t)_i$'s being $uniform(-\gamma_*, \gamma^*)$, $\gamma_*^2 = 4\lambda^+$. Thus, in this case, $\alpha_\add = ??$.
\end{definition}

\begin{theorem} \label{thm1} %= \max_j c_{j,\new} = \max_j r_{j-1}
Consider Algorithm \ref{reprocs}. Pick a $\zeta$ that satisfies
\[
\zeta  \leq  \min\left(\frac{10^{-4}}{r^2},\frac{1.5 \times 10^{-4}}{r^2 f},\frac{1}{r^{3}\gamma_*^2}\right)
%\ \text{where} \ f := \frac{\lambda^+}{\lambda^-}
\]
Assume that the initial subspace estimate is accurate enough, i.e. $\|(I - \Phat_0 \Phat_0') P_0\| \le r_0 \zeta$. If the following conditions hold:

\ben
\item The algorithm parameters are set as
$\xi = \xi_0(\zeta), \   7 \xi \leq \omega \leq S_{\min} - 7 \xi,  \ K = K(\zeta), \ \alpha \ge \alpha_{\text{add}}(\zeta)$
\item$L_t$ satisfies Signal Model \ref{Ltmodel} with
\ben
\item  $0 \le c_{j,\new} \leq c_{\max}$ for all $j$ (thus $r_j \le r_{\max}:=r_0 + J c_{\max}$),
\item the $a_t$'s mutually independent over $t$,
\item  $\|a_t\|_{\infty} \leq \gamma_*$ for all $t$ ($a_t$'s bounded);,
\item $0 < \lambda^- \le \lambda^+ < \infty$,and
\item $g \le g^+ = \sqrt{2}$;
\een

\item \label{slow}slow subspace change holds:  %for a given value of $S_{\min}$, the subspace change is slow enough, i.e.
(\ref{delay}) holds with $d=K\alpha$; (\ref{atnew_inc}) holds with $v = 1.2$; and $c$ and $\gamma_\new$ are small enough so that $14 \xi_0 (\zeta) \le S_{\min}$.
%\bea
%&& \min_{j} (t_{j+1} -t_j) > K  \alpha, \nn \\
%&& \max_{j} \max_{t_j + (k-1)\alpha \leq t < t_j + k\alpha} \|a_{t,\new}\|_{\infty} \leq \gamma_{\new,k}:= \min (1.2^{k-1} \gamma_{\new}, \gamma_*), \ \text{for all} \ k=1,2, \dots K, \nn \\
%&& 14 \rho \xi_0 (\zeta) \le S_{\min}, \nn
%\eea

\item denseness holds: equation \eqref{kappa plus} holds with  $\kappa_{2s,*}^+ = 0.3$ and equation \eqref{kappa new plus} holds with $\kappa_{2s,\new}^+ = 0.15$
%for all $j = 1, \dots, J - 1$ and $k = 0, \dots, K$, with $\Phat_{j,\new,0} = [.]$ (empty matrix).

%\item the maximum condition number of $\mathrm{Cov}[a_{t,\new}]$, $g \le g^+ = \sqrt{2}$ %the condition number of the covariance matrix of $a_{t,\new}$ averaged over $t \in \mathcal{I}_{j,k}$, is bounded, i.e.

\item the matrices
\begin{align*}
D_{j,\new,k} &:= (I - \Phat_{j-1} \Phat_{j-1}'-\Phat_{j,\new,k} \Phat_{j,\new,k}')P_{j,\new} \\
\text{and} \\
Q_{j,\new,k} &: = (I-P_{j,\new}{P_{j,\new}}')\Phat_{j,\new,k}
\end{align*}
 satisfy
\begin{align*}
\max_j \max_{1 \le k \le K}  \kappa_{s}(D_{j,\new,k}) & \le \kappa_{s}^+ := 0.152 \\
\max_j \max_{1 \le k \le K} \kappa_{2s}(Q_{j,\new,k}) & \leq \tilde{\kappa}_{2s}^+ := 0.15
\end{align*}

%$$g_{j,k} \le g^+ = \sqrt{2}$$
%where $g_{j,k}$ is defined in Definition \ref{defn_alpha} and $g^+$ is defined in Definition \ref{kappaplus}.

%\item and the initial subspace estimate is accurate enough: $\|(I - \Phat_0 \Phat_0') P_0\| \le r_0 \zeta$,
\een

then, with probability at least $(1 -  n^{-10})$, at all times, $t$, all of the following hold:
\ben
\item at all times, $t$, $$\That_t = T_t \ \ \text{and}$$
\begin{multline*}
\|e_t\|_2 = \|L_t - \hat{L}_t\|_2 = \|\hat{S}_t - S_t\|_2  \le  \\0.18 \sqrt{c} \gamma_{\new} + 1.2\sqrt{\zeta}(\sqrt{r} + 0.06 \sqrt{c}).
\end{multline*}
 %$0.18=1.2*0.15$, $0.06 = 0.4*0.15$..

\item  the subspace error $\SE_{(t)} := \|(I - \Phat_{(t)} \Phat_{(t)}') P_{(t)} \|_2$ satisfies
\begin{align*}
\SE_{(t)}   &\le    \left\{  \begin{array}{ll}
(r_0 + (j-1)c) \zeta + 0.4 c \zeta + 0.6^{k-1}  & \ \\
\hspace{1in} \text{if}  \ \    t \in \mathcal{I}_{j,k}, \ k=1,2 \dots K \nn \\  %\ t_j + (k-1) \alpha \le t \le t_j + k \alpha -1, \ 1 \le k \le K \nn \\
(r_0 + jc) \zeta    \qquad \text{if} \ \   t \in \mathcal{I}_{j,K+1}  % t_j + K \alpha \le t_{j+1} -1 \nn
\end{array} \right. \nn \\
  &\le   \left\{  \begin{array}{ll}
 10^{-2} \sqrt{\zeta} +  0.6^{k-1}  \ \\
\hspace{.8in} \text{if}  \ \ t \in \mathcal{I}_{j,k}, \ k=1,2 \dots K \nn \\  % t_j + (k-1) \alpha \le t \le t_j + k \alpha -1, \ 1 \le k \le K \nn \\
10^{-2} \sqrt{\zeta}    \qquad \text{if} \ \    t \in \mathcal{I}_{j,K+1}  %t_j + K \alpha \le t_{j+1} -1 \nn
\end{array} \right.
\end{align*}

\item the error $e_t = \hat{S}_t - S_t = L_t - \hat{L}_t$ satisfies the following at various times
\begin{align*}
\|e_t\|_2  & \le   \left\{  \begin{array}{ll}
0.18 \sqrt{c}0.72^{k-1}\gamma_{\new} + \\
\qquad 1.2 (\sqrt{r} + 0.06 \sqrt{c})  (r_0+(j-1)c)\zeta  \gamma_*      \\
 \hspace{1 in} \text{if} \ \ t \in \mathcal{I}_{j,k}, \ k=1,2 \dots K  \nn \\ %t_j + (k-1) \alpha \le t \le t_j + k \alpha -1, \ 1 \le k \le K \nn \\
1.2(r_0+ j c) \zeta \sqrt{r} \gamma_*    \quad \ \text{if} \ \ t \in \mathcal{I}_{j,K+1} % t_j + K \alpha \le t \le t_{j+1} -1
\end{array} \right. \nn \\
 & \le   \left\{  \begin{array}{ll}
0.18 \sqrt{c}0.72^{k-1}\gamma_{\new} + 1.2(\sqrt{r} + 0.06 \sqrt{c}) \sqrt{\zeta}   \\
\hspace{1in} \text{if} \ \ t \in \mathcal{I}_{j,k}, \ k=1,2 \dots K \nn \\ %t_j + (k-1) \alpha \le t \le t_j + k \alpha -1, \ 1 \le k \le K \nn \\
1.2 \sqrt{r} \sqrt{\zeta}   \quad \text{if} \ \  t \in \mathcal{I}_{j,K+1} %t_j + K \alpha \le t \le t_{j+1} -1
\end{array} \right.
\end{align*}

\een
%\bea
%\hat{T}_t & = & T_t, \ \text{and} \nn \\
%\|L_t - \hat{L}_t\|_2 = \|\hat{S}_t - S_t\|_2 %& \le &  \sqrt{c} \gamma_{\new} + \zeta \gamma_* 1.2 ( (r_0 + (j-1) c) \sqrt{r} +  0.001c \sqrt{c} ), \ \text{for all} \ t_j \le t < t_{j+1}  \nn \\
%                                               & \le &   \sqrt{c} \gamma_{\new} + 1.2\sqrt{\zeta}(\sqrt{r} + 0.4 \sqrt{c})  \nn
%\eea
\end{theorem}

\begin{remark} \label{Dnew0_rem} %??? fix the proof in the appendix to use this fact
Consider the last assumption. We actually also need a similar denseness of $\kappa_s(D_{j,\new})$ where $D_{j,\new} =D_{j,\new,0}= (I - \Phat_{j-1} \Phat_{j-1}')P_{j,\new}$. Conditioned on the fact that $\Span(P_{j-1})$ has been accurately estimated, this follows easily from the denseness of $P_{j,\new}$ (see Lemma \ref{Dnew0_lem}). %\zeta_{j,*}^+$.  conditioned on the fact that $\|(I - \Phat_{j-1} \Phat_{j-1}') P_{j-1}\|_2 \le r\zeta$
%$||\Phat_{j-1}'P_{j,\new}||_2 \le \zeta_*^+$  %(r_0+(j-1)c)
\end{remark} 

The above result says the following. Consider Algorithm \ref{reprocs}. Assume that the initial subspace error is small enough.  If the algorithm parameters are appropriately set, if slow subspace change holds, if the subspaces are dense, if the condition number of $\mathrm{Cov}[a_{t,\new}]$ is small enough, and if the currently unestimated part of the newly added subspace is dense enough (this is an assumption on the algorithm estimates), then, w.h.p., we will get exact support recovery at all times. Moreover, the sparse recovery error will always be bounded by $0.18\sqrt{c} \gamma_\new$ plus a constant times $\sqrt{\zeta}$. Since $\zeta$ is very small, $\gamma_\new < S_{\min}$, and $c$ is also small, the normalized reconstruction error for recovering $S_t$ will be small at all times. %The same is true for the error in recovering $L_t$.
In the second conclusion, we bound the subspace estimation error, $\SE_{(t)}$. When a subspace change occurs, this error is initially bounded by one. The above result shows that, w.h.p., with each projection PCA step, this error decays exponentially and falls below $0.01\sqrt{\zeta}$ within $K$ projection PCA steps. The third conclusion shows that, with each projection PCA step, w.h.p., the sparse recovery error as well as the error in recovering $L_t$ also decay in a similar fashion.

As we explain in Section \ref{discuss_add}, the most important limitation of our result is that it requires an assumption on $D_{\new,k}$ and $Q_{\new,k}$ which depend on algorithm estimates. Moreover, it studies an algorithm that requires knowledge of model parameters.

\subsection{Projection-PCA algorithm for ReProCS}
\label{proj PCA}
Given a data matrix $\mathcal{D}$, a basis matrix $P$ and an integer $r$, projection-PCA (proj-PCA) applies PCA on $\mathcal{D}_{\text{proj}}:=(I-PP')\mathcal{D}$, i.e., it computes the top $r$ eigenvectors (the eigenvectors with the largest $r$ eigenvalues) of $\frac{1}{\alpha} \mathcal{D}_{\text{proj}} {\mathcal{D}_{\text{proj}}}'$. Here $\alpha$ is the number of column vectors in $\mathcal{D}$. This is summarized in Algorithm \ref{algo_pPCA}.

If $P =[.]$, then projection-PCA reduces to standard PCA, i.e. it computes the top $r$ eigenvectors of $\frac{1}{\alpha} \mathcal{D} {\mathcal{D}}'$.

The reason we need projection PCA algorithm in step 3 of Algorithm \ref{reprocs} is because the error $e_t = \Lhat_t - L_t = S_t - \Shat_t$ is correlated with $L_t$; and the maximum condition number of $\operatorname{Cov}(L_t)$, which is bounded by $f$, cannot be bounded by a small value (see Remark \ref{large_f}). This issue is explained in detail in Appendix \ref{projpca}. Most other works that analyze standard PCA, e.g. \cite{nadler} and references therein, do not face this issue because they assume uncorrelated-ness of the noise/error and the true data vector. With this assumption, one only needs to increase the PCA data length $\alpha$ to deal with the larger condition number.
%the large condition number does not cause a problem because they assume that the error ($e_t$ in our case) in the observed data vector ($\Lhat_t$) is uncorrelated with the true data vector ($L_t$). Under this assumption, one only needs to increase the PCA data length $\alpha$ to deal with larger condition numbers $f$. However, in our case, because $e_t$ is correlated with $L_t$, this strategy does not work. This issue is explained in detail in Appendix \ref{projpca}.

%We explain this point in detail in Appendix \ref{projpca}. We also discuss there some other alternatives.
We should mention that the idea of projecting perpendicular to a partly estimated subspace has been used in other different contexts in past work \cite{PP_PCA_Li_Chen, mccoy_tropp11}.

\begin{algorithm}
\caption{projection-PCA: $Q \leftarrow \text{proj-PCA}(\mathcal{D},P,r)$}\label{algo_pPCA}
\ben
\item Projection: compute $\mathcal{D}_{\text{proj}} \leftarrow (I - P P') \mathcal{D}$
\item PCA: compute $\frac{1}{\alpha}  \mathcal{D}_{\text{proj}}{\mathcal{D}_{\text{proj}}}' \overset{EVD}{=}
\left[ \begin{array}{cc}Q & Q_{\perp} \\\end{array}\right]
\left[ \begin{array}{cc} \Lambda & 0 \\0 & \Lambda_{\perp} \\\end{array}\right]
\left[ \begin{array}{c} Q' \\ {Q_{\perp}}'\\\end{array}\right]$
where $Q$ is an $n \times r$ basis matrix and  $\alpha$ is the number of columns in $\mathcal{D}$.
\een
\end{algorithm}

\subsection{Recursive Projected CS (ReProCS)}
\label{basic_rep}

% This allows us to handle much larger ranks of ${\cal L}_t$. Its details are provided in \cite{long}.

\begin{algorithm*}[ht]
\caption{Recursive Projected CS (ReProCS)}\label{reprocs}
{\em Parameters: } algorithm parameters: $\xi$, $\omega$, $\alpha$, $K$, model parameters: $t_j$, $c_{j,\new}$ %$c_{\max}$ %$r_0$,
\\ (set as in Theorem \ref{thm1} )
%thresholds $\xi$, $\omega$, subspace change times, $t_j$, number of frames used for one projection PCA step, $\alpha$, number of times projection PCA is done in one change period, $K$
\\
{\em Input: } $M_t$, {\em Output: } $\Shat_t$, $\Lhat_t$, $\Phat_{(t)}$ %{\em Feed-back: } $\Phat_{(t-1)}$, $\Lhat_{t-1}, \Lhat_{t-2}, \dots \Lhat_{t-\alpha}$ (at most $\alpha$ previous estimates of $L_t$)
\\
Initialization: Compute $\Phat_0 \leftarrow$ proj-PCA$\left( [L_{1},L_{2},\cdots,L_{t_{\train}}], [.], r_0 \right)$ where $r_0 = \rank([L_{1},L_{2},\cdots,L_{t_{\train}}])$.
\\
Set $\Phat_{(t)} \leftarrow \Phat_0$,  $j \leftarrow 1$, $k\leftarrow 1$.

For $t > t_{\train}$, do the following:
\ben
\item Estimate $T_t$ and $S_t$ via Projected CS:
\ben
\item \label{othoproj} Nullify most of $L_t$: compute $\Phi_{(t)} \leftarrow I-\Phat_{(t-1)} {\Phat_{(t-1)}}'$, compute $y_t \leftarrow \Phi_{(t)} M_t$
\item \label{Shatcs} Sparse Recovery: compute $\hat{S}_{t,\cs}$ as the solution of $\min_{x} \|x\|_1 \ s.t. \ \|y_t - \Phi_{(t)} x\|_2 \leq \xi$
\item \label{That} Support Estimate: compute $\hat{T}_t = \{i: \ |(\hat{S}_{t,\cs})_i| > \omega\}$
\item \label{LS} LS Estimate of $S_t$: compute $(\hat{S}_t)_{\hat{T}_t}= ((\Phi_t)_{\hat{T}_t})^{\dag} y_t, \ (\hat{S}_t)_{\hat{T}_t^{c}} = 0$
\een
\item Estimate $L_t$: $\hat{L}_t = M_t - \hat{S}_t$.
\item \label{PCA} %Projection PCA
Update $\Phat_{(t)}$: K Projection PCA steps.
\ben
%\item $k \leftarrow 1$
\item If $t = t_j + k\alpha-1$,
\ben
%\item $\Phi_{j,0} = I-\Phat_{j-1}{\Phat_{j-1}}'$ %$\Phi_j = I-\Phat_{j-1}{\Phat_{j-1}}'$ and %\sum_{t=t_j+(k-1)\alpha}^{t_j+ k \alpha-1}

\item $\Phat_{j,\new,k} \leftarrow$ proj-PCA$\left(\left[\hat{L}_{t_j+(k-1)\alpha}, \dots, \hat{L}_{t_j+k\alpha-1}\right],\Phat_{j-1},c_{j,\new}\right)$.

\item set $\Phat_{(t)} \leftarrow [\Phat_{j-1} \ \Phat_{j,\new,k}]$; increment $k \leftarrow k+1$.
\een
Else
\ben
\item set $\Phat_{(t)} \leftarrow \Phat_{(t-1)}$.
\een
%end if
\item If $t = t_j + K \alpha - 1$, then set $\Phat_{j} \leftarrow [\Phat_{j-1} \ \Phat_{j,\new,K}]$. Increment $j \leftarrow j + 1$. Reset $k \leftarrow 1$.
%\ben
%\item
%\een
%end if
\een
%\item \label{deletionPCA} (Optional) If $t=t_j + K \alpha +\alpha_{del}-1$, delete $P_{j,\old}$ from $\Phat_j$ by Algorithm \ref{prostDel}.
\item Increment $t \leftarrow t + 1$ and go to step 1.
\een
%end for
\end{algorithm*}

%%%%%%%%%%%%%%%%%%%%%%%%%%%%%%%%%

\begin{figure*}[!t]
\centerline{
\includegraphics[width =15cm]{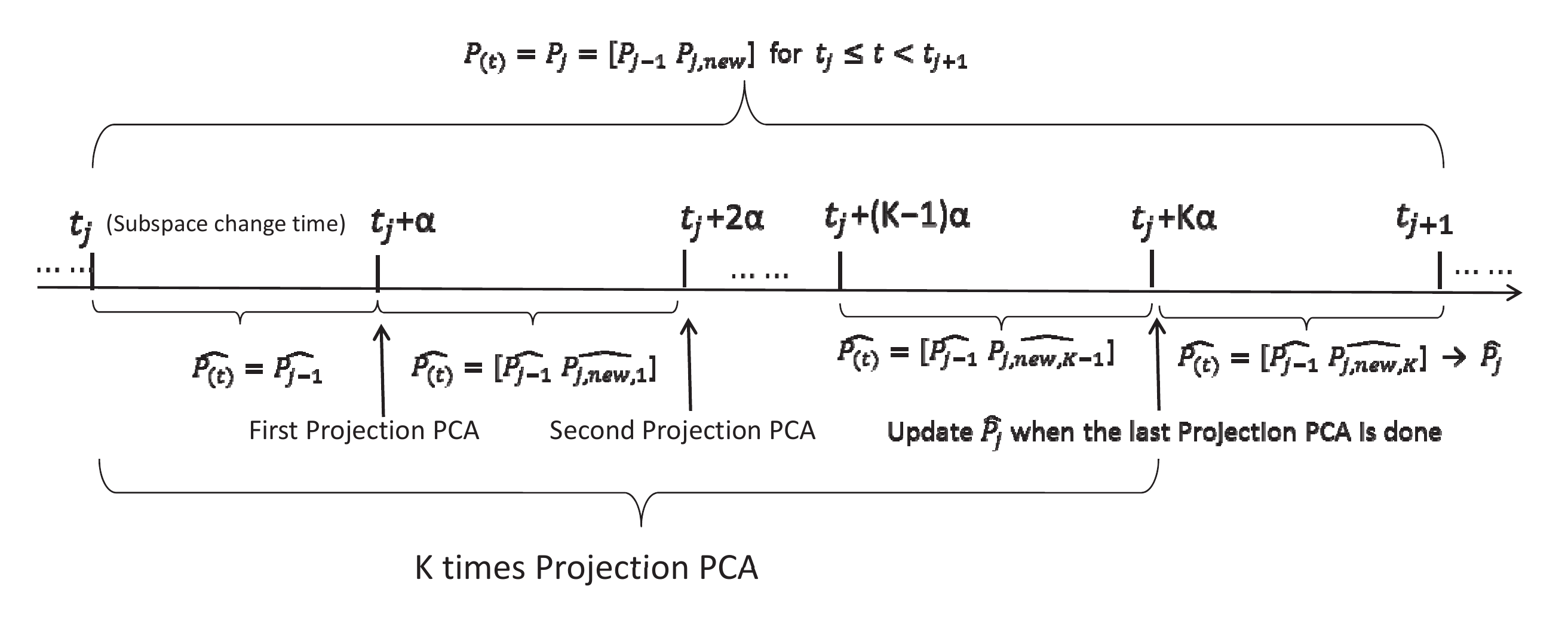}
}
\caption{\small{The K projection PCA steps.
}}
\label{algo_fig}
\end{figure*}

%for a given value of $S_{\min}$,

We summarize the Recursive Projected CS (ReProCS) algorithm in Algorithm \ref{reprocs}. It uses the following definition.

\begin{definition}\label{Ijk}
Define the time interval $\mathcal{I}_{j,k}: = [t_j + (k-1) \alpha, t_j + k \alpha - 1]$ for $k=1, \dots K$ and $\mathcal{I}_{j,K+1}:= [t_j + K \alpha, t_{j+1}-1]$. %Here, $K$ is the algorithm parameter in Algorithm \ref{reprocs}.
\end{definition}

The key idea of ReProCS is as follows.
First, consider a time $t$ when the current basis matrix $P_{(t)}=P_{(t-1)}$ and this has been accurately predicted using past estimates of $L_t$, i.e. we have $\Phat_{(t-1)}$ with $\|(I -  \Phat_{(t-1)} \Phat_{(t-1)}') P_{(t)}\|_2$ small. We project the measurement vector, $M_t$, into the space perpendicular to $\Phat_{(t-1)}$ to get the projected measurement vector $y_t:= \Phi_{(t)} M_t$ where $\Phi_{(t)} = I -\Phat_{(t-1)} \Phat_{(t-1)}'$ (step 1a). Since the $n \times n$ projection matrix, $\Phi_{(t)}$ has rank $n- r_*$ where $r_*= \rank(\Phat_{(t-1)})$, therefore $y_t$ has only $n-r_*$ ``effective" measurements\footnote{i.e. some $r_*$ entries of $y_t$ are linear combinations of the other $n-r_*$ entries}, even though its length is $n$. Notice that $y_t$ can be rewritten as $y_t = \Phi_{(t)} S_t + \beta_t$ where $\beta_t: = \Phi_{(t)} L_t$. Since %$\Span(\Phat_{(t-1)})$ approximately contains $\Span(P_{(t)})$,
$\|(I -  \Phat_{(t-1)} \Phat_{(t-1)}') P_{(t-1)}\|_2$ is small, the projection nullifies most of the contribution of $L_t$ and so the projected noise $\beta_t$ is small.
Recovering the $n$ dimensional sparse vector $S_t$ from $y_t$ now becomes a traditional sparse recovery or CS problem in small noise \cite{feng_bresler,gorod_rao,bpdn,decodinglp,candes,donoho}. We use $\ell_1$ minimization to recover it (step 1b). If the current basis matrix $P_{(t)}$, and hence its estimate, $\Phat_{(t-1)}$, is dense enough, then, by Lemma \ref{delta_kappa}, the RIC of $\Phi_{(t)}$ is small enough. Using Theorem \ref{candes_csbound}, this  ensures that $S_t$ can be accurately recovered from $y_t$.
%Thus $S_t$ will not be nullified by $\Phi$, and it can be accurately recovered from $y_t$ \cite{candes_rip}. =P_{(t-1)}
%as long as the projected ``noise", $\beta_t: =(I -\Phat_{(t-1)} \Phat_{(t-1)}')L_t$ is small \cite{candes_rip}.
%(see Theorem \ref{cs_result}). %, assuming that most of $L_t$ has been nullified (as stated in Theorem \ref{cs_result}).
%
By thresholding on the recovered $S_t$, one gets an estimate of its support (step 1c). By computing a least squares (LS) estimate of $S_t$ on the estimated support and setting it to zero everywhere else (step 1d), we can get a more accurate final estimate, $\Shat_t$, as first suggested in \cite{dantzig}. This $\Shat_t$ is used to estimate $L_t$ as $\Lhat_t = M_t-\Shat_t$.  As we explain in the proof of Lemma \ref{cslem}, if $S_{\min}$ is large enough and the support estimation threshold, $\omega$, is chosen appropriately, we can get exact support recovery, i.e. $\That_t = T_t$. In this case, the error $e_t: = \Shat_t - S_t = L_t - \Lhat_t$ has the following simple expression:
\beq
e_t = I_{T_t} {(\Phi_{(t)})_{T_t}}^{\dag} \beta_t = I_{T_t} [ (\Phi_{(t)})_{T_t}'(\Phi_{(t)})_{T_t}]^{-1}  {I_{T_t}}' \Phi_{(t)} L_t
%I_{T_t} [ (\Phi)_{T_t}'(\Phi)_{T_t}]^{-1} ( {I_{T_t}}' \Phi P_* a_{t,*} + {I_{T_t}}' \Phi P_{\new} a_{t,\new})
\label{etdef0}
\eeq
The second equality follows because ${(\Phi_{(t)})_T}' \Phi_{(t)} ={(\Phi_{(t)} I_T)}' \Phi_{(t)} = {I_T}' \Phi_{(t)}$ for any set $T$.

Now consider a time $t$ when $P_{(t)} = P_j = [P_{j-1}, P_{j,\new}]$ and $P_{j-1}$ has been accurately estimated but $P_{j,\new}$ has not been estimated, i.e. consider a $t \in \mathcal{I}_{j,1}$.  At this time, $\Phat_{(t-1)} = \Phat_{j-1}$ and so $\Phi_{(t)} = \Phi_{j,0}:=I - \Phat_{j-1} \Phat_{j-1}'$.  Let $r_{*}:=r_0 + (j-1)c_{\max}$ (We remove subscript $j$ for ease of notation.)
, and $c:= c_{\max}$. Assume that the delay between change times is large enough so that by $t=t_j$, $\Phat_{j-1}$ is an accurate enough estimate of $P_{j-1}$, i.e. $\|\Phi_{j,0} P_{j-1}\|_2 \le r_{*} \zeta \ll 1$. %for a very small $\zeta$.
It is easy to see using Lemma \ref{hatswitch} that $\kappa_s(\Phi_{0} P_\new) \le  \kappa_s(P_\new) + r_* \zeta $, i.e. $\Phi_{0} P_\new$ is dense because $P_\new$ is dense and because $\Phat_{j-1}$ is an accurate estimate of $P_{j-1}$ (which is perpendicular to $P_\new$).
Moreover, using Lemma \ref{delta_kappa}, it can be shown that $\phi_0: = \max_{|T| \le s} \|[ (\Phi_{0})_{T}'(\Phi_{0})_{T}]^{-1}\|_2 \le \frac{1}{1-\delta_s(\Phi_0)} \le  \frac{1}{1- (\kappa_s(P_{j-1}) + r_* \zeta)^2}$. The error $e_t$ still satisfies (\ref{etdef0}) although its magnitude is not as small.
Using the above facts in (\ref{etdef0}), we get that
%$\|e_t\|_2 \le \phi_0 r_* \zeta \sqrt{r_*} \gamma_* + \phi_0 (\kappa_{2s,\new}^+ + r_* \zeta) \sqrt{c} \gamma_\new$. Thus,
\[
\|e_t\|_2 \le \frac{\kappa_{s}(P_\new) \sqrt{c} \gamma_\new +  r_* \zeta(\sqrt{r_*} \gamma_*  + \sqrt{c} \gamma_\new)}{1- (\kappa_s(P_{j-1}) + r \zeta)^2}
\]
If $\sqrt\zeta < 1/\gamma_*$, all terms containing $\zeta$ can be ignored and we get that the above is approximately upper bounded by $\frac{\kappa_{s}(P_\new)}{1- \kappa_s^2(P_{j-1})}  \sqrt{c} \gamma_\new$. Using the denseness assumption, this quantity is a small constant times $\sqrt{c} \gamma_\new$, e.g. with the numbers assumed in Theorem \ref{thm1} we get a bound of $0.18 \sqrt{c} \gamma_\new$.
Since $\gamma_\new \ll S_{\min}$ and $c$ is assumed to be small, thus, $\|e_t\|_2 =  \|S_t - \Shat_t\|_2$ is small compared with $\|S_t\|_2$, i.e. $S_t$ is recovered accurately. With each projection PCA step, as we explain below, the error $e_t$ becomes even smaller.% than in the previous interval.

Since $\Lhat_t = M_t - \Shat_t$ (step 2), $e_t$ also satisfies $e_t = L_t  - \Lhat_t$. Thus, a small $e_t$ means that $L_t$ is also recovered accurately. The estimated $\Lhat_t$'s are used to obtain new estimates of $P_{j,\new}$ every $\alpha$ frames for a total of $K \alpha$ frames via a modification of the standard PCA procedure, which we call projection PCA (step 3). We illustrate the projection PCA algorithm in Figure \ref{algo_fig}.
In the first projection PCA step, we get the first estimate of $P_{j,\new}$, $\Phat_{j,\new,1}$. For the next $\alpha$ frame interval, $\Phat_{(t-1)} = [\Phat_{j-1}, \Phat_{j,\new,1}]$  and so $\Phi_{(t)} = \Phi_{j,1} = I - \Phat_{j-1} \Phat_{j-1}' - \Phat_{\new,1} \Phat_{\new,1}'$. Using this in the projected CS step reduces the projection noise, $\beta_t$, and hence the reconstruction error, $e_t$, for this interval, as long as $\gamma_{\new,k}$ increases slowly enough.
%\footnote{For $t \in \mathcal{I}_{j,k}$,  $\|e_t\|_2 \le \phi_{k-1} r_* \zeta \sqrt{r_*} \gamma_* + \phi_{k-1} \kappa_{s,k-1} \zeta_{j,k-1} \sqrt{c} \gamma_{\new,k}$. By Lemma \ref{RIC_bnd}, $\phi_{k-1}$ is only a little more than one; by assumption, $\kappa_{s,k-1}$ is significantly less than one; and as we explain in Sec \ref{reason_expodecay}, $\zeta_{j,k-1}$ can be shown to decrease exponentially with $k$. If $\gamma_{\new,k}$ increases at a slower rate than the rate of decrease of $\zeta_{j,k-1}$, then the second term in the bound on $\|e_t\|_2$ also decreases expoentially with $k$.}.
Smaller $e_t$ makes the perturbation seen by the second projection PCA step even smaller, thus resulting in an improved second estimate $\Phat_{j,\new,2}$. Within $K$ updates ($K$ chosen as given in Theorem \ref{thm1}), it can be shown that both $\|e_t\|_2$ and the subspace error drop down to a constant times $\sqrt{\zeta}$. At this time, we update $\Phat_{j}$ as $\Phat_j = [\Phat_{j-1}, \Phat_{j,\new,K}]$.

\subsection{Discussion}
\label{discuss_add}
First consider the choices of $\alpha$ and of $K$. Notice that $K = K(\zeta)$ is larger if $\zeta$ is smaller. Also, $\alpha_\add$ is inversely proportional to $\zeta$. Thus, if we want to achieve a smaller lowest error level, $\zeta$, we need to compute projection PCA over larger durations $\alpha$ and we need more number of projection PCA steps $K$. This means that we also require a larger delay between subspace change times, i.e. larger $t_{j+1}-t_j$.

% We have not defined g_{j,k} here ??  Also, notice that $(\lambda_{j,\new,k}^-)^2  \ge (\lambda^-)^2$ and from the assumptions in the theorem given below, $(\lambda_{j,\new,k}^+)^2 \le 1.44^K \gamma_{\new}^4$. Thus, $\frac{1.44^K \gamma_{\new}^4}{(\lambda^-)^2}$ is an upper bound on $g_{j,k}^2$. Here $g_{j,k}$ is the condition number of the average covariance matrix of $a_{t,\new}$ over an $\alpha$ period. Thus, $\alpha_\add$ is directly proportional to an upper bound on the condition number of the covariance of the newly added directions.

Now consider the assumptions used in the result.  %The most limiting assumption is that on $D_{\new,k}$ since this is an assumption on the algorithm estimates and the model paramters. We discuss this after discussing the others.
We assume slow subspace change, i.e. the delay between change times is large enough,  $\|a_{t,\new}\|_\infty$ is initially below $\gamma_\new$ and increases gradually, and $14 \xi_0 \le S_{\min}$ which holds if $c_{\max}$ and $\gamma_\new$ are small enough. Small $c_{\max}$, small initial $a_{t,\new}$ (i.e. small $\gamma_\new$) and its gradual increase are verified for real video data in Section \ref{model_verify}. As explained there, one cannot estimate the delay between change times unless one has access to an ensemble of videos of a given type and hence the first assumption cannot be verified. %with just one video sequence of a particular type (need an ensemble)

We also assume denseness of $P_{j-1}$ and $P_{j,\new}$. 
This is a subset of the denseness assumptions used in earlier work \cite{rpca}. 
As explained there, this is valid for the video application because typically the changes of the background sequence are global, e.g. due to illumination variation affecting the entire image or due to textural changes such as water motion or tree leaves' motion etc. 
We quantify this denseness using the parameter $\kappa_s$. 
The way it is defined, bounds on $\kappa_s$ simultaneously place restrictions on denseness of $L_t$, 
$r = \rank(P_J)$, and $s$ (the maximum sparsity of any $S_t$).
To compare our assumptions with those of Cand\`{e}s et. al. in \cite{rpca}, we could assume 
$\kappa_1(P_{J}) \leq \sqrt{\frac{\mu r}{n}}$,
where $\mu$ is any value between $1$ and $\frac{n}{r}$.
Using the bound $\kappa_s(P) \leq \sqrt{s}\kappa_1(P)$, we see that if
$\frac{2 s r}{n}\leq \mu^{-1}(0.3)^2$,
then our assumption of $\kappa_{2s}(P_{J}) \leq 0.3$ will be satisfied.
Up to differences in the constants, this is the same requirement found in \cite{hsu2011robust},
even though \cite{hsu2011robust} studies a batch approach (PCP) while we study an online algorithm.
From this we can see that if $s$ grows linearly with $n$, then $r$ must be constant.  Similarly, if $r$ grows linearly with $n$, then $s$ must be constant.
This is a stronger assumption than required by \cite{rpca} where $s$ is allowed to grow linearly with $n$, and $r$ is simultaneously allowed to grow as $\frac{n}{\log(n)^2}$.
However, the comparison with \cite{rpca} is not direct because we do not need denseness of the right singular vectors or a bound on the vector infinity norm of $UV'$.
The reason for the stronger requirement on the product $sr$ is because we study an online algorithm that recovers the sparse vector $S_t$ at each time $t$  rather than in a batch or a piecewise batch fashion. 
Because of this the sparse recovery step does not use the low dimensional structure of the new (and still unestimated) subspace.
%Thus, most columns of the matrix ${\cal L}_t$ are dense. This along with its low rank ensures that ${\cal L}_t$, and hence $P_j$, is dense.

%Besides the above, we also need a few other assumptions.
%We assume that the condition number of  $\mathrm{Cov}[a_{t,\new}]$ is small. Since the number of newly added directions is assumed small, this is not very limiting. In the simplest case, if $c_{\max}=1$, then the condition number is always one. Moreover, this assumption can be removed if the newly added eigenvalues can be separated into clusters so that the condition number of each cluster is small (even though the overall condition number may be large).  Under this assumption, we can use idea similar to the cluster-PCA approach described in Section  \ref{Del_section} to recover one cluster of new directions at a time.

We assume the independence of $a_t$'s, and hence of $L_t$'s, over time. This is typically not valid in practice; however, it allows us to simplify the problem and hence the derivation of the performance guarantees. In particular it allows us to use the matrix Hoeffding inequality to bound the terms in the subspace error bound. In ongoing work by Zhan and Vaswani \cite{reprocs_cor}, we are seeing that, with some more work, this can be replaced by a more realistic assumption: an autoregressive model on the $a_t$'s, i.e. assume $a_t = b a_{t-1} + \nu_t$ where $\nu_t$'s are independent over time and $b<1$. We can work with this model in two ways. If we assume $b$ is known, then a simple change to the algorithm (in the subspace update step, replace $\Lhat_t$ by $\Lhat_t - b \Lhat_{t-1}$ everywhere) allows us to get a result that is almost the same as the current one using exactly the same approach. Alternatively if $b$ is unknown, as long as $b$ is bounded by a $b_* <1$, we can use the matrix Azuma inequality to still get a result similar to the current one. It will require a larger $\alpha$ though and some other changes.

The most limiting assumption is the assumption on $D_{j,\new,k}$ and $Q_{j,\new,k}$ because these are functions of algorithm estimates. The denseness assumption on $Q_{j,\new,k}$ is actually not essential, it is possible to prove a slightly more complicated version of Theorem \ref{thm1} without it.
We use this assumption only in Lemma \ref{RIC_bnd}. However, if we use tighter bounds on other quantities such as $g$ and $\kappa_s(P_{j,\new})$, and if we analyze the first projection-PCA step differently from the others, we can get a tighter bound on $\zeta_{j,1}$ (and hence $\zeta_{j,k}$ for $k \ge 1$) and then we will not need this assumption. %For example if $c_{\max}=1$ and so $g=1$ and if we assume a tighter bound on  $\kappa_s(P_{j,\new})$  (which implies a tighter bound on $\kappa_s(D_{\new,0})$), we can remove this assumtpion. %If we do not assume its denseness, we will need to analyze the first projection-PCA step differently from the later ones. This will require simple changes to Lemmas \ref{expzeta} and \ref{RIC_bnd}.
%
% whose columns span the currently unestimated part of the newly added subspace  (the number of changes needed increased with $c_{\new})$

Consider denseness of $D_{j,\new,k}$. Our proof actually only needs smallness of $\max_{t \in \mathcal{I}_{j,k+1}} d_t$ where $d_t = \|{I_{T_t}}' D_{j,\new,k}\|_2 / \|D_{j,\new,k}\|_2$ for $t \in \mathcal{I}_{j,k+1}$ for $k=1,2 \dots K$. Since this quantity is upper bounded by $\kappa_s(D_{j,\new,k})$, we have just assumed a bound on this for simplicity. Note also that densenss of $D_{j,\new,0}$ does not need to be assumed, this follows from denseness of $P_{j,\new}$ conditioned on the fact that $P_{j-1}$ has been accurately estimated.
We attempted to verify the smallness of $d_t$ in simulations done with a dense $P_j$ and $P_{j,\new}$ and involving correlated support change of $S_t$'s. We observed that, as long as there was a support change every few frames, this quantity was small. %We also saw that, even if $d_t$ was close to one for most times in a given interval, that just meant that, for that interval, the subspace error decreased only slightly. As a result, a larger $K$ was required for the subspace error to become small enough. It did not mean that the algorithm became unstable.
For example, with $n=2048$, $s=20$, $r_0=36$, $c_\new=1$, support change by one index every 2 frames was sufficient to ensure a small $d_t$ at all times (see Sec \ref{sims}). Even one index change every 50 frames was enough to ensure that the errors decayed down to small enough values, although in this case $d_t$ was large at certain times and the decay of the subspace error was not exponential. It should be possible to use a similar idea to modify our result as well.
The first thing to point out is that the max of $d_t$ can be replaced by its average over $t \in \mathcal{I}_{j,k}$ with a minor change to the proof of Lemma \ref{termbnds}. Moreover, if we try to show linear decay of the subspace error (instead of exponential decay), and if we analyze the first projection-PCA interval differently from the others, we will need a looser bound on the $d_t$'s, which will be easier to obtain under a certain support change assumption. In the first interval, the subspace error is large since $P_\new$ has not been estimated but  $D_{\new,0}$ is dense (see Remark \ref{Dnew0_rem}). In the later intervals, the subspace error is lower but $D_{\new,k}$ may not be as dense.%

Finally, Algorithm \ref{reprocs} assumes knowledge of certain model parameters and these may not always be available. % in practice. %However, as we explain here, some of these requirements are redundant and can be removed. If the initial training data is noise-free, then one does not need $r_0$, one can just keep all left singular vectors with nonzero singular values. In the algorithm as it is written, we assume $c_{j,\new}$ is known and we select this many top singular vectors in the projection PCA step. However, in the analysis, everywhere we bound $c_{j,\new}$ by $c_{\max}$. Thus, even in the algorithm, we could replace $c_{j,\new}$ by  $c_{\max}$. Requiring knowledge of just $c_{\max}$ is not very restrictive. Finally
It needs to know $c_{j,\new}$, which is the number of new directions added at subspace change time $j$, and it needs knowledge of $\gamma_\new$ (in order to set $\xi$ and $\omega$), which is the bound on the infinity norm of the projection of $a_t$ along the new directions for the first $\alpha$ frames. %Neither of these is very restrictive. However 
It also needs to know the subspace change times $t_j$, and this is the most restrictive. 

A practical version of Algorithm \ref{reprocs} (that provides reasonable heuristics for setting its parameters without model knowledge) is given in \cite{han_tsp}. As explained there, $\hat{t}_j + \alpha-1$ can be estimated by taking the last set of $\alpha$ estimates $\Lhat_t$, projecting them perpendicular to $\Phat_{j-1}$ and checking if any of the singular values of the resulting matrix is above $\sqrt{\hat\lambda^-}$. It should be possible to prove in future work that this happens only after an actual change and within a short delay of it.

%In summary, the most limiting aspect of the above result is the fact that it needs an assumption that depends on the algorithm estimates (denseness of $D_{\new,k}$). A second limiting aspect is that the algorithm needs knowledge of the subspace change times. We expect to be able to also remove this as explained above.% and hence it cannot be interpreted as a correctness result. In future work, we hope to be able to replace it with an assumption on support change of $S_t$'s

Lastly, note that, because the subspace change model only allows new additions to the subspace, the rank of the subspace basis matrix $P_j$ can only grow over time. The same is true for its ReProCS estimate. Thus, $\max_j \kappa_{2s}(P_j) = \kappa_{2s}(P_J)$ and a bound on this imposes a bound on the number of allowed subspace change times, $J$, or equivalently on the maximum rank of ${\cal L}_t$ for any $t$. A similar bound is also needed by PCP \cite{rpca} and all batch approaches. In Sec \ref{Del_section}, we explain how we can remove the bound on $J$ and hence on the rank of ${\cal L}_t$ if an extra clustering assumption holds. %by allowing for subspace deletion on the model as well as introducing a deletion step in ReProCS. However to do that we need an extra clustering assumption. %bound on the rank of ${\cal L}$ %, $r_0+Jc_{\max}$

\section{Definitions needed for proving Theorem \ref{thm1}}
\label{detailed} %and Other Key Ideas of the Proof Strategy

%In Sec \ref{defins}, we define the various quantities that will be used in the lemmas leading to the proof of Theorem \ref{thm1}. Next, we give the proof outline in Sec. \ref{outline}. The reason for exponential decay of the subspace error with every projection PCA step is briefly explained in Sec \ref{reason_expodecay}. The lemmas leading to the proof of the theorem, and their proofs, are given in the subsequent sections - Sec \ref{bound_cs_subspace_err} and  Sec \ref{high_prob_bnd}. The theorem's proof follows easily using these lemmas. This is given in Sec \ref{thmproof}.

%\subsection{Definitions}\label{defins}
A few quantities are already defined in the model (Section \ref{model}), Definition \ref{Ijk}, Algorithm \ref{reprocs}, and Theorem \ref{thm1}. Here we define more quantities needed for the proofs.

\begin{definition}
\label{kappaplus}
In the sequel, we let
\ben
\item $r := r_{\max}= r_0 + Jc_{\max}$ and $c: = c_{\max} = \max_j c_{j,\new}$,
\item $\kappa_{s,*} := \max_j \kappa_s(P_{j-1})$, $\kappa_{s,\new} := \max_j \kappa_s(P_{j,\new})$, $\kappa_{s,k}:= \max_j \kappa_s (D_{j,\new,k})$, $\tilde{\kappa}_{s,k} := \max_j \kappa_s((I-P_{j,\new}{P_{j,\new}}') \Phat_{j,\new,k})$,
\item $\kappa_{2s,*}^+ := 0.3$, $\kappa_{2s,\new}^+ := 0.15$, ${\kappa}_{s}^+ := 0.152$, $\tilde{\kappa}_{2s}^+ := 0.15$ and $g^+ := \sqrt{2}$ are the upper bounds assumed in Theorem \ref{thm1} on $\max_j \kappa_{2s}(P_j)$, $\max_j \kappa_{2s}(P_{j,\new})$, $\max_j \max_k \kappa_{s}(D_{j,\new,k})$, $\max_j \kappa_{2s}(Q_{j,\new,k})$ and $g$ respectively.
\item $\phi^+ := 1.1735$ %We see later that this is an upperbound on $\phi_k$ under the assumptions of Theorem \ref{thm1}.
\item $\gamma_{\new,k} := \min (1.2^{k-1} \gamma_{\new}, \gamma_*)$ (recall that this is defined in Sec \ref{slowss}).
%\\ (recall that assumption \ref{slow} of the theorem implies that  $\max_{j}\max_{t\in\mathcal{I}_{j,k}} \|a_{t,\new}\|_{\infty} \leq \gamma_{\new,k}$)
%\item $P_{j,*}: = P_{j-1}$ and $\Phat_{j,*}:= \Phat_{j-1}$ (see Remark \ref{remove_j}).
\een
\end{definition}

\begin{definition}
\label{zetakplus}
Define the following:
\begin{enumerate}
\item $\zeta_{j,*}^+ := (r_0 + (j-1)c)\zeta$

\item Define the sequence $\{{\zeta_{j,k}}^+\}_{k=0,1,2,\dots, K}$ recursively as follows:
\begin{align}
\zeta_{j,0}^+ & := 1 \nn  \\
 \zeta_{j,k}^+ & :=\frac{b + 0.125 c \zeta}{1 - (\zeta_{j,*}^+)^2 - (\zeta_{j,*}^+)^2 f - 0.125 c \zeta - b} \; \text{ for} \ k \geq 1,
\end{align}
\end{enumerate}
where
\begin{align*}
& b :=  C \kappa_s^+ g^+  \zeta_{j,k-1}^+ + \tilde{C} (\kappa_s^+)^2 g^+ (\zeta_{k-1}^+)^2 + C' f (\zeta_{j,*}^+)^2 \\
& C := \frac{2\kappa_s^+ \phi^+}{\sqrt{1-(\zeta_{j,*}^+)^2}} + \phi^+ ,  \\
& C' :=  (\phi^+)^2 +  \frac{2\phi^+ }{\sqrt{1-(\zeta_{j,*}^+)^2}}
+ 1 + \\
&\hspace{.8in} \phi^+ + \frac{\kappa_s^+ \phi^+}{\sqrt{1-(\zeta_{j,*}^+)^2}} + \frac{\kappa_s^+(\phi^+)^2 }{\sqrt{1-(\zeta_{j,*}^+)^2}} , \\
& \tilde{C} :=  (\phi^+)^2  +  \frac{\kappa_s^+ (\phi^+)^2 }{\sqrt{1-(\zeta_{j,*}^+)^2}} .
\end{align*}
As we will see, $\zeta_{j,*}^+$ and $\zeta_{j,k}^+$ are the high probability upper bounds on $\zeta_{j,*}$ and $\zeta_{j,k}$ (defined in Definition \ref{def_SEt}) under the assumptions of Theorem \ref{thm1}.
\end{definition}

\begin{definition}
We define the noise seen by the sparse recovery step at time $t$ as
$$\beta_t: = (I - \Phat_{(t-1)} \Phat_{(t-1)}') L_t.$$
Also define the reconstruction error of $S_t$ as
$$e_t:= \Shat_t - S_t.$$
Here $\Shat_t$ is the final estimate of $S_t$ after the LS step in Algorithm \ref{reprocs}. Notice that $e_t$ also satisfies $e_t = L_t - \Lhat_t$.
\end{definition}

\begin{definition}
\label{def_SEt}
We define the subspace estimation errors as follows. Recall that $\Phat_{j,\new,0}=[.]$ (empty matrix).%
\bea  %\|(I - \Phat_{(t)} \Phat_{(t)}') P_{(t)} \|_2, \   t_j + k \alpha \le t < t_j + (k+1) \alpha -1, \ k=0,1, \dots K \nn \\
&& \SE_{(t)} := \|(I - \Phat_{(t)} \Phat_{(t)}') P_{(t)} \|_2, \nn \\ %\   t_j + k \alpha \le t < t_j + (k+1) \alpha -1, \ k=0,1, \dots K \nn \\
&& \zeta_{j,*} := \|(I - \Phat_{j-1} \Phat_{j-1}') P_{j-1}\|_2 \nn \\
&& \zeta_{j,k} := \|(I - \Phat_{j-1} \Phat_{j-1}' - \Phat_{j,\new,k} \Phat_{j,\new,k}') P_{j,\new}\|_2 \nn %,\zeta_{j,0} := \|(I - \Phat_{j-1} \Phat_{j-1})' P_{j,\new}\|_2, \ \  %\ \text{for} \ k \ge 1 \nn
\eea
\end{definition}

\begin{remark}\label{zetastar}
Recall from the model given in Sec \ref{model} and from Algorithm \ref{reprocs} that
\begin{enumerate}
\item $\Phat_{j,\new,k}$ is orthogonal to $\Phat_{j-1}$, i.e. $\Phat_{j,\new,k}'\Phat_{j-1}=0$
\item $\Phat_{j-1} := [\Phat_{0}, \Phat_{1,\new,K}, \dots \Phat_{j-1,\new,K}]$ and $P_{j-1}: = [P_0, P_{1,\new}, \dots P_{j-1,\new}]$
\item  for $t \in \mathcal{I}_{j,k+1}$, $\Phat_{(t)} = [\Phat_{j-1}, \Phat_{j,\new,k}]$ and $P_{(t)} = P_j = [P_{j-1}, P_{j,\new}]$. %$t_j + k \alpha \le t < t_j + (k+1) \alpha -1$
\item  $\Phi_{(t)} := I - \Phat_{(t-1)} \Phat_{(t-1)}'$
\end{enumerate}
Then it is easy to see that
\begin{enumerate}
\item $\zeta_{j,*} \le \zeta_{j-1,*} + \zeta_{j,K} =
\zeta_{1,*} + \sum_{j'=1}^{j-1} \zeta_{j',K}$
\item $\SE_{(t)}  \le \zeta_{j,*} + \zeta_{j,k} \le \zeta_{1,*} + \sum_{j'=1}^{j-1} \zeta_{j',K} + \zeta_{j,k}$ \; for \ $t \in \mathcal{I}_{j,k+1}$.
\end{enumerate}
\end{remark}

\begin{definition}\label{defn_Phi}
%For $t_j \leq t < t_{j+1}$, define the following
Define the following
\ben
\item $\Phi_{j,k}$, $\Phi_{j,0}$ and $\phi_k$ %Recall that $\Phat_{j,\new,0}=[]$ (empty matrix).
\ben
\item $\Phi_{j,k} := I-\Phat_{j-1} {\Phat_{j-1}}' - \Phat_{j,\new,k} {\Phat_{j,\new,k}}'$ is the CS matrix for $t \in \mathcal{I}_{j,k+1}$, i.e. $\Phi_{(t)} = \Phi_{j,k}$ for this duration.

%the CS matrix used to recover the sparse vector $S_t$ from the projected measurements $y_t$ for $t \in \mathcal{I}_{j,k+1}$, i.e. $\Phi_{(t)} = \Phi_{j,k}$ for this duration.
%used to cancel most of  $L_t$ during $t \in \mathcal{I}_{j,k+1}$ and it serves as the measurement matrix for sparse recovery for this period.
\item $\Phi_{j,0} := I-\Phat_{j-1} {\Phat_{j-1}}'$ is the CS matrix for $t \in \mathcal{I}_{j,1}$, i.e. $\Phi_{(t)} = \Phi_{j,0}$ for this duration. $\Phi_{j,0}$ is also the projection matrix used in all of the projection PCA steps for $t \in [t_j, t_{j+1}-1]$.

\item $\phi_k := \max_j \max_{T:|T|\leq s}\|({(\Phi_{j,k})_T}'(\Phi_{j,k})_T)^{-1}\|_2$. It is easy to see that $\phi_k \le \frac{1}{1-\max_j \delta_s(\Phi_{j,k})}$ \cite{decodinglp}.

\een
\item $D_{j,\new,k}$, $D_{j,\new}$, $D_{j,*,k}$ and $D_{j,*}$
\ben
\item $D_{j,\new,k} := \Phi_{j,k} P_{j,\new}$. $\Span(D_{j,\new,k})$ is the unestimated part of the newly added subspace for any $t \in \mathcal{I}_{j,k+1}$. %at this time %denotes the unestimated part of the newly added directions at any $t \in \mathcal{I}_{j,k+1}$, thus $\Span(D_{j,\new,k})$ is the unestimated subspace at this time,

\item $D_{j,\new} := D_{j,\new,0} = \Phi_{j,0} P_{j,\new}$. $\Span(D_{j,\new})$ is interpreted similarly for any $t \in \mathcal{I}_{j,1}$.

\item $D_{j,*,k} := \Phi_{j,k} P_{j-1}$. $\Span(D_{j,*,k})$ is the unestimated part of the existing subspace for any $t \in \mathcal{I}_{j,k}$
\item $D_{j,*} := D_{j,*,0} = \Phi_{j,0} P_{j-1}$. $\Span(D_{j,*,k})$ is interpreted similarly for any $t \in \mathcal{I}_{j,1}$
\item Notice that $\zeta_{j,0} = \|D_{j,\new}\|_2$, $\zeta_{j,k} = \|D_{j,\new,k}\|_2$, $\zeta_{j,*} = \|D_{j,*}\|_2$. Also, clearly, $\|D_{j,*,k}\|_2 \le \zeta_{j,*}$. %Recall from Definition \ref{def_SEt}
\een
\een
\end{definition}

\begin{definition}
\label{defHk}\
\begin{enumerate}
\item Let $D_{j,\new} \overset{QR}{=} E_{j,\new} R_{j,\new}$ denote its reduced QR decomposition, i.e. let $E_{j,\new}$ be a basis matrix for $\Span(D_{j,\new})$ and let $R_{j,\new} = E_{j,\new}'D_{j,\new}$. %Conditioned on $\zeta_{j,*} \le \zeta_{j,*}^+$ (previous subspace accurately estimated), $D_{j,\new}$ is full rank.
%Conditioned on $\zeta_{j,*} \le \zeta_{j,*}^+$ (previous subspace accurately estimated), it is easy to see using Lemma \ref{hatswitch} that $R_{j,\new}$ is invertible and $E_{j,\new}$ is a basis matrix for $\Span(D_{j,\new})$.
%\big( conditioned on $\Gamma_{j,0}$ (see def \ref{Gamma_def}) by Lemma \ref{hatswitch} since $\sigma_i(D_{j,\new}) = \sigma_i(R_{j,\new})$ \big) and $E_{j,\new}$ is a basis matrix for $\Span(D_{j,\new})$.

\item Let $E_{j,\new,\perp}$ be a basis matrix for the orthogonal complement of $\Span(E_{j,\new})=\Span(D_{j,\new})$. To be precise, $E_{j,\new,\perp}$ is a $n \times (n-c_{j,\new})$ basis matrix that satisfies $E_{j,\new,\perp}'E_{j,\new}=0$. %Consider $t \in \mathcal{I}_{j,k}$.

\item Using $E_{j,\new}$ and $E_{j,\new,\perp}$, define $A_{j,k}$, $A_{j,k,\perp}$, $H_{j,k}$, $H_{j,k,\perp}$ and $B_{j,k}$ as
\bea
A_{j,k} &:=& \frac{1}{\alpha} \sum_{t \in \mathcal{I}_{j,k}} {E_{j,\new}}' \Phi_{j,0} L_t {L_t}' \Phi_{j,0} E_{j,\new} \nn \\
A_{j,k,\perp} &:=& \frac{1}{\alpha} \sum_{t \in \mathcal{I}_{j,k}} {E_{j,\new,\perp}}' \Phi_{j,0} L_t {L_t}' \Phi_{j,0} E_{j,\new,\perp} \nn \\
H_{j,k} &:=& \frac{1}{\alpha}\sum_{t \in \mathcal{I}_{j,k}} {E_{j,\new}}' \Phi_{j,0} \nn \\
&&\hspace{.5in}(e_t {e_t}' -L_t {e_t}' - e_t {L_t}') \Phi_{j,0} E_{j,\new} \nn\\
H_{j,k,\perp} &:=& \frac{1}{\alpha} \sum_{t \in \mathcal{I}_{j,k}} {E_{j,\new,\perp}}'\Phi_{j,0} \nn \\
&&\hspace{.4in} (e_t {e_t}' - L_t {e_t}' - e_t {L_t}') \Phi_{j,0} E_{j,\new,\perp} \nn \\
%chenlu
B_{j,k} &:=& \frac{1}{\alpha}\sum_{t \in \mathcal{I}_{j,k}} {E_{j,\new,\perp}}'\Phi_{j,0} \Lhat_t \Lhat_t' \Phi_{j,0} E_{j,\new}\nn \\
&=& \frac{1}{\alpha}\sum_{t \in \mathcal{I}_{j,k}} {E_{j,\new,\perp}}'\Phi_{j,0} (L_t-e_t) \nn \\ &&\hspace{1.2in}({L_t}'-{e_t}')\Phi_{j,0} E_{j,\new} \nn
\eea

\item Define
\bea
&&\mathcal{A}_{j,k} := \left[ \begin{array}{cc} E_{j,\new} & E_{j,\new,\perp} \\ \end{array} \right]
\left[\begin{array}{cc} A_{j,k} \ & 0 \ \\ 0 \ & A_{j,k,\perp}  \\ \end{array} \right]
\left[ \begin{array}{c} {E_{j,\new}}' \\ {E_{j,\new,\perp}}' \\ \end{array} \right]\nn\\
&&\mathcal{H}_{j,k} := \left[ \begin{array}{cc} E_{j,\new} & E_{j,\new,\perp} \\ \end{array} \right]
\left[\begin{array}{cc} H_{j,k} \ & {B_{j,k}}' \ \\ B_{j,k} \ &  H_{j,k,\perp} \\ \end{array} \right]
\left[ \begin{array}{c} {E_{j,\new}}' \\ {E_{j,\new,\perp}}' \\ \end{array} \right] \nn
\eea
\end{enumerate}
\end{definition}

\begin{remark}
\begin{enumerate}
\item From the above, it is easy to see that $$\mathcal{A}_{j,k} + \mathcal{H}_{j,k} =\frac{1}{\alpha} \sum_{t \in \mathcal{I}_{j,k}} \Phi_{j,0} \hat{L}_t {\hat{L}_t}' \Phi_{j,0}.$$

\item Recall from Algorithm \ref{reprocs} that
\begin{align*}
\mathcal{A}_{j,k} +& \mathcal{H}_{j,k} \overset{EVD}{=}\\
 &\left[ \begin{array}{cc} \Phat_{j,\new,k} & \Phat_{j,\new,k,\perp} \\ \end{array} \right]
\left[\begin{array}{cc} \Lambda_k \ & 0 \ \\ 0 \ & \ \Lambda_{k,\perp} \\ \end{array} \right]
\left[ \begin{array}{c} \Phat_{j,\new,k}' \\ \Phat_{j,\new,k,\perp}' \\ \end{array} \right]
\end{align*}
is the EVD of $\mathcal{A}_{j,k} + \mathcal{H}_{j,k}$.
%and $\Lambda_k$ is a $c_{j,\new} \times c_{j,\new}$ diagonal matrix.

\item Using the above, $\mathcal{A}_{j,k} + \mathcal{H}_{j,k}$ can be decomposed in two ways as follows:
\begin{align*}
&\mathcal{A}_{j,k} + \mathcal{H}_{j,k} \\
&= \left[ \begin{array}{cc} \Phat_{j,\new,k} & \Phat_{j,\new,k,\perp} \\ \end{array} \right]
\left[\begin{array}{cc} \Lambda_k \ & 0 \ \\ 0 \ & \ \Lambda_{k,\perp} \\ \end{array} \right] \left[ \begin{array}{c} \Phat_{j,\new,k}' \\ \Phat_{j,\new,k,\perp}' \\ \end{array} \right]
\\
&= \left[ \begin{array}{cc} E_{j,\new} & E_{j,\new,\perp} \\ \end{array} \right] \\
&\hspace{.4in}
\left[\begin{array}{cc} A_{j,k} + H_{j,k} \ & B_{j,k}' \ \\ B_{j,k} \ & A_{j,k,\perp} + H_{j,k,\perp}  \\ \end{array} \right]
\left[ \begin{array}{c} {E_{j,\new}}' \\ {E_{j,\new,\perp}}' \\ \end{array} \right]
\end{align*}
\end{enumerate}

\end{remark}

%\begin{remark}
%Thus, from the above definition,
%$\mathcal{H}_{j,k}  = \frac{1}{\alpha} [\Phi_0 \sum_t (-L_t e_t' - e_t L_t' +  e_t e_t') \Phi_0 + F + F']$ where $F:=E_{\new,\perp} E_{\new,\perp}'\Phi_0 \sum_t L_t L_t' \Phi_0 E_\new E_\new' =E_{\new,\perp} E_{\new,\perp}' (D_{*,k-1} a_{t,*})(D_{*,k-1} a_{t,*} + D_{\new,k-1} a_{t,\new})' E_\new E_\new'$. Since $\E[a_{t,*} a_{t,\new}']=0$, $\|\frac{1}{\alpha}F\|_2 \lesssim r^2 \zeta^2 \lambda^+$ w.h.p. Recall $\lesssim$ means (in an informal sense) that the RHS contains the dominant terms in the bound.
%\end{remark}

\begin{definition}
Define the random variable $X_{j,k} := \{a_1,a_2,\cdots,a_{t_j+k\alpha-1}\}$.
\end{definition}
Recall that the $a_t$'s are mutually independent over $t$, hence $X_{j,k}$ and $\{ a_{t_j+k\alpha}, \dots,a_{t_j+(k+1)\alpha -1} \}$ are mutually independent.

\begin{definition}
Define the set $\check{\Gamma}_{j,k}$ as follows:
\begin{align*}
\check{\Gamma}_{j,k} &:= \{ X_{j,k} : \zeta_{j,k} \leq \zeta_{k}^+ \text{ and } \hat{T}_t = T_t  \text{ for all } t \in \mathcal{I}_{j,k} \} \\
\check{\Gamma}_{j,K+1} &:= \{ X_{j+1,0} :  \hat{T}_t = T_t  \text{ for all } t \in \mathcal{I}_{j,K+1} \}
\end{align*}

\end{definition}

\begin{definition}\label{Gamma_def}
Recursively define the sets $\Gamma_{j,k}$ as follows:
%\begin{align*}
%\Gamma_{1,0} &:= \{ X_{1,0} : \zeta_{1,*} \leq \zeta_{1,*}^+ \text{ and } \hat{T}_t = T_t \text{ for all } t \in [ t_{\mathrm{train} +1} : t_1 -1 ] \} \\
%\Gamma_{j,k} &:= \Gamma_{j,k-1} \cap \check{\Gamma}_{j,k} \quad k = 1, 2, \dots , K, j = 1, 2, \dots, J \\
%\Gamma_{j+1, 0} &:= \Gamma_{j,K} \cap \check{\Gamma}_{j,K+1} \quad j = 1, 2, \dots, J
%\end{align*}
%
\begin{align*}
\Gamma_{1,0} & := \{ X_{1,0}: \zeta_{1,*} \leq r\zeta \\ 
&\hspace{.5in} \text{and} \ \hat{T}_t = T_t  \ \text{for all} \ t\in [t_{\mathrm{train}}+1: t_1 -1]\} \\
\Gamma_{j,0} & :=  \{X_{j,0}: \zeta_{j',*} \le \zeta_{j',*}^+ \ \text{for all} \ j' = 1, 2, \dots, j \\
 &\hspace{1.2in} \text{and} \ \hat{T}_t = T_t \ \text{for all} \ t \le t_{j-1} \}\\
\Gamma_{j,k} &:=\Gamma_{j,k-1} \cap \check{\Gamma}_{j,k} \ k=1,2,\dots K+1
\end{align*}

\end{definition}

\begin{remark} \label{etdef_rem}
Whenever $\hat{T}_t = T_t$ we have an exact expression for $e_t$:
\begin{equation}\label{et expression}
e_t =  I_{T_t} [ (\Phi_{(t)})_{T_t}'(\Phi_{(t)})_{T_t}]^{-1}  {I_{T_t}}' \Phi_{(t)} L_t
\end{equation}
Recall that $L_t = P_j a_t = P_{j-1} a_{t,*} + P_{j,\new} a_{t,\new}$.
\end{remark}

\begin{definition}
Define $P_{j,*}: = P_{j-1}$ and $\Phat_{j,*}: = \Phat_{j-1}$.
\end{definition}

\begin{remark}
Notice that the subscript $j$ always appears as the first subscript, while $k$ is the last one. At many places in the rest of the paper, we remove the subscript $j$ for simplicity,  %. Whenever there is only one subscript, it refers to the value of $k$,
e.g., $\Phi_0$ refers to $\Phi_{j,0}$, $\Phat_{\new,k}$ refers to $\Phat_{j,\new,k}$, $P_*$ refers to $P_{j,*}: = P_{j-1}$ and so on.  %Also, define $P_{j,*}: = P_{j-1}$ and $\Phat_{j,*}: = \Phat_{j-1}$.  %Finally, for $t\in\mathcal{I}_{j,k}$, $\zeta_{*} := \zeta_{j,*}$ and $\zeta_{*}^+ := \zeta_{j,*}^+$.
\label{remove_j}
\end{remark}

\section{Proof of Theorem \ref{thm1}} \label{mainlemmas}
%Recall that when there is only one subscript, it refers to the value of $k$ (i.e. $\zeta_k = \zeta_{j,k}$).

\subsection{Two Main Lemmas and Proof of Theorem \ref{thm1}}
The proof of Theorem \ref{thm1} essentially follows from two main lemmas that we state below.  Lemma \ref{expzeta} gives an exponentially decaying upper bound on $\zeta_k^+$ defined in Definition \ref{zetakplus}. $\zeta_k^+$ will be shown to be a high probability upper bound for $\zeta_k$ under the assumptions of the Theorem.  Lemma \ref{mainlem} says that conditioned on $X_{j,k-1}\in\Gamma_{j,k-1}$, $X_{j,k}$ will be in $\Gamma_{j,k}$ w.h.p..  In words this says that if, during the time interval $\mathcal{I}_{j,k-1}$, the algorithm has worked well (recovered the support of $S_t$ exactly and recovered the background subspace with subspace recovery error below $\zeta_{k-1}^+ + \zeta_*^+$), then it will also work well in $\mathcal{I}_{j,k}$ w.h.p..

\begin{lem}[Exponential decay of $\zeta_k^+$] \label{expzeta}
Assume that the bounds on $\zeta$ from Theorem \ref{thm1} hold. Define the sequence $\zeta_k^+$ as in Definition \ref{zetakplus}. Then
\begin{enumerate}
\item $\zeta_0^+ = 1$ and $\zeta_k^+ \leq 0.6^k + 0.4c\zeta$ for all $k = 1, 2, \dots, K,$
\item the denominator of $\zeta_k^+$ is positive for all $k = 1, 2, \dots, K$.
\end{enumerate}
\end{lem}

We prove this lemma in Section \ref{pfoflem1}.

\begin{lem}\label{mainlem}
Assume that all the conditions of Theorem \ref{thm1} hold.  Also assume that $\mathbf{P}(\Gamma^e_{j,k-1})>0.$
Then
\[
\mathbf{P}(\Gamma^e_{j,k}|\Gamma^e_{j,k-1}) \geq p_k(\alpha,\zeta) \geq p_K(\alpha,\zeta) 
\]
for all $ k = 1, 2, \ldots, K$,
and
\[
\mathbf{P}(\Gamma^e_{j,K+1}|\Gamma^e_{j,K}) = 1
\]
where $p_k(\alpha,\zeta)$ is defined in equation \eqref{pk}.
\end{lem}

We prove this lemma in Section \ref{pfoflem2}.

%\begin{remark}\label{Gamma_rem}
%?? Remove this as it is now built into the definition of $\Gamma$.  I'm leaving it in for now to be able to fix references to this remark.
%
%
%Under the assumptions of Theorem \ref{thm1}, it is easy to see that the following holds.
% For any $k=1,2 \dots K$, $\Gamma_{j,k}^e$ implies that  $\zeta_{j,*} \le \zeta_{*}^+$
%\end{remark}
%
%From the definition of $\Gamma_{j,k}^e$, $\zeta_{j',K}\leq \zeta_{K}^+$ for all $j'\leq j-1$. By Lemma \ref{expzeta} and the definition of $K$ in Definition \ref{defn_alpha}, $\zeta_{K}^+\leq 0.6^K + 0.4c\zeta\leq c\zeta$ for all $j'\leq j-1$. Using Remark \ref{zetastar}, $\zeta_{j,*} \leq \zeta_1^* + \sum_{j'=1}^{j-1} \zeta_{j',K} \leq r_0\zeta + (j-1)c\zeta \leq \zeta_*^+$.

\begin{remark}\label{Gamma_rem2}
Using Lemma \ref{expzeta} and Remark \ref{zetastar} and the value of $K$ given in the theorem, it is easy to see that, under the assumptions of Theorem \ref{thm1},
$$\Gamma_{j,0} \cap (\cap_{k=1}^{K+1} \check\Gamma_{j,k}) \subseteq \Gamma_{j+1,0}.$$
Thus $\mathbf{P}(\Gamma_{j+1,0}^e|\Gamma^e_{j,0}) \ge \mathbf{P}(\check\Gamma^e_{j,1}, \dots \check\Gamma^e_{j,K+1} | \Gamma^e_{j,0})$.
%reason: from lemma \ref{expzeta} and the choice of $K$, $\zeta_K^+ \le c \zeta$. From Remark \ref{zetastar}, $\zeta_{j,*} \le \zeta_{j-1,*} + \zeta_{j,K}$. Since $\zeta_{j+1,*}^+ = r_0 \zeta + jc \zeta =  \zeta_{j,*}^+ + c \zeta$, above remark follows. ?? check this argument, but not needed in actual paper.
\end{remark}

\vspace{.4 cm}

\begin{proof}[\textbf{Proof of Theorem \ref{thm1}}]

The theorem is a direct consequence of Lemmas \ref{expzeta}, \ref{mainlem}, and Lemma \ref{subset_lem}.
From Remark \ref{Gamma_rem2}, $\mathbf{P}(\Gamma_{j+1,0}^e|\Gamma_{j,0}^e) \ge \mathbf{P}(\check\Gamma_{j,1}^e, \dots \check\Gamma_{j,K+1}^e | \Gamma_{j,0}^e) = \prod_{k=1}^{K+1} P(\check\Gamma_{j,k}^e|\Gamma_{j,k-1}^e)$. Also, since $\Gamma_{j+1,0} \subseteq \Gamma_{j,0}$, using Lemma \ref{subset_lem}, $\mathbf{P}(\Gamma_{J+1,0}^e | \Gamma_{1,0}^e) = \prod_{j=1}^{J}  \mathbf{P}(\Gamma_{j+1,0}^e | \Gamma_{j,0}^e)$.
Thus,
\[
\mathbf{P}(\Gamma_{J+1,0}^e | \Gamma_{1,0}^e) \ge \prod_{j=1}^{J}  \prod_{k=1}^{K+1} \mathbf{P}(\check\Gamma_{j,k}^e|\Gamma_{j,k-1}^e)
\]
Using Lemma \ref{mainlem}, and the fact that $p_k(\alpha,\zeta) \geq p_K(\alpha,\zeta)$ (see their respective definitions in Lemma \ref{termbnds} and equation \eqref{pk} and observe that $p_k(\alpha,\zeta)$ is decreasing in $k$), we get 
$$\mathbf{P}(\Gamma_{J+1,0}^e| \Gamma_{1,0}) \geq {p}_K(\alpha,\zeta)^{KJ}.$$
Also, $\mathbf{P}(\Gamma_{1,0}^e)=1$. This follows by the assumption on $\hat{P}_0$ and Lemma \ref{cslem}. Thus, $\mathbf{P}(\Gamma_{J+1,0}^e) \geq {p}_K(\alpha,\zeta)^{KJ}$.

Using the definition of $\alpha_\add$, and $\alpha \geq \alpha_{\add}$, we get that
$$\mathbf{P}(\Gamma_{J+1,0}^e) \geq {p}_K(\alpha,\zeta)^{KJ} \geq 1- n^{-10}$$

The event $\Gamma_{J+1,0}^e$ implies that $\That_t=T_t$ and $e_t$ satisfies (\ref{etdef0}) for all $t < t_{J+1}$. Using Remarks \ref{zetastar} and \ref{Gamma_rem2}, $\Gamma_{J+1,0}^e$ implies that all the bounds on the subspace error hold. Using these, $\|a_{t,\new}\|_2 \le \sqrt{c} \gamma_{\new,k}$, and $\|a_t\|_2 \le \sqrt{r} \gamma_*$, $\Gamma_{J+1,0}^e$ implies that all the bounds on $\|e_t\|_2$ hold (the bounds are obtained in Lemma \ref{cslem}).

Thus, all conclusions of the the result hold w.p. at least $1- n^{-10}$.
\end{proof}

\subsection{Proof of Lemma \ref{expzeta} } \label{pfoflem1}

\begin{proof}
First recall the definition of $\zeta_k^+$ (Definition \ref{zetakplus}).  Recall from Definition \ref{kappaplus} that $\kappa_s^+ := 0.15$ , $\phi^+ := 1.1735$, and $g^+ := \sqrt{2}$.  So we can make these substitutions directly.  Notice that $\zeta_k^+$ is an increasing function of $\zeta_*^+, \zeta, c$, and $f$. Therefore we can use upper bounds on each of these quantities to get an upper bound on $\zeta_k^+$.  From the definition of $\zeta$ in Theorem \ref{thm1} and $\zeta_{j,*}^+ := (r_0 + (j-1)c)\zeta$ we get
\begin{itemize}
\item $\zeta_{j,*}^+ \leq 10^{-4} $
\item $\zeta_{j,*}^+ f \leq 1.5 \times 10^{-4}$
\item $c\zeta \leq 10^{-4}$
\item $\ds \frac{\zeta_{j,*}^+}{c\zeta} = \frac{(r_0+(j-1)c)\zeta}{c\zeta} \leq \frac{r_0 +(J-1)c}{c} = \frac{r}{c}\leq r$ (Without loss of generality we can assume that $c = c_{\max} \geq 1$ because if $c=0$ then there is no subspace estimation problem to be solved. $c=0$ is the trivial case where all conclusions of Theorem \ref{thm1} will hold just using Lemma \ref{cslem}.) % Recall that $c$ is the upper bound on the number of new directions added
\item $\zeta_{j,*}^+ f r \leq r^2 f \zeta \leq 1.5 \times 10^{-4}$
\end{itemize}
First we prove by induction that $\zeta_k^+ \leq \zeta_{k-1}^+  \leq 0.6$ for all $k\geq 1$.  Notice that $\zeta_0^+ =1$ by definition.
\begin{itemize}
\item Base case ($k=1$):  Using the above bounds we get that $\zeta_1^+ < 0.5985 < 1 = \zeta_0^+$.
\item For the induction step, assume that $\zeta_{k-1}^+ \leq \zeta_{k-2}^+$.  Then because $\zeta_k^+$ is increasing in $\zeta_{k-1}^+$ (denote the increasing function by $f_{inc}$) we get that $\zeta_{k}^+ = f_{inc}(\zeta_{k-1}^+) \leq f_{inc}(\zeta_{k-2}^+) = \zeta_{k-1}^+$.
\end{itemize}
\begin{enumerate}
\item To prove the first claim, first rewrite $\zeta_k^+$ as
\begin{multline*}
\zeta_k^+ = \zeta_{k-1}^+ \frac{ C \kappa_s^+ g^+  + \tilde{C} (\kappa_s^+)^2 g^+ (\zeta_{k-1}^+) }{1 - (\zeta_*^+)^2 - (\zeta_*^+)^2 f - 0.125 c \zeta - b} + \\
c\zeta\frac{ C(\zeta_*^+ f)\frac{ (\zeta_*^+)}{c\zeta} + .125}{1 - (\zeta_*^+)^2 - (\zeta_*^+)^2 f - 0.125 c \zeta - b}
\end{multline*}
where $C, \tilde{C},$ and $b$ are as in Definition \ref{zetakplus}.
Using the above bounds including $\zeta_{k-1}^+ \leq .6$ we get that
\begin{align*}
\zeta_k^+  &\leq \zeta_{k-1}^+(0.6)  + c\zeta(0.16) \\
&= \zeta_0^+ (0.6)^{k} + \sum_{i = 0}^{k-1}(0.6)^k(0.16) c\zeta \\
&\leq \zeta_0^+ (0.6)^{k} + \sum_{i = 0}^{\infty}(0.6)^k(0.16) c\zeta \\
&\leq 0.6^k + 0.4 c\zeta
\end{align*}

%\begin{align*}
%\zeta_k^+ & \leq \zeta_{k-1}^+(0.6)  + c\zeta(0.16) \\
%& = \zeta_0^+ (0.6)^{k} + \sum_{i = 0}^{k-1}(0.6)^k(0.16) c\zeta \\
%& \leq \zeta_0^+ (0.6)^{k} + \sum_{i = 0}^{\infty}(0.6)^k(0.16) c\zeta \\
%& \leq 0.6^k + 0.4 c\zeta
%\end{align*}

\item To see that the denominator is positive, observe that the denominator is decreasing in all of its arguments: $\zeta_{j,*}^+, \zeta_{j,*}^+ f, c\zeta$, and $b$.  Using the same upper bounds as before, we get that the denominator is greater than or equal to $0.78 > 0$.

\end{enumerate}

\end{proof}

\subsection{Proof of Lemma \ref{mainlem} }\label{pfoflem2}
The proof of Lemma \ref{mainlem} follows from two lemmas. The first, Lemma \ref{cslem}, is the final conclusion for the projected CS step for $t\in \mathcal{I}_{j,k}$. Its proof follows using Lemmas \ref{expzeta}, \ref{delta_kappa}, \ref{hatswitch}, the CS error bound (Theorem \ref{candes_csbound}) and some straightforward steps.
The second, Lemma \ref{zetak}, is the final conclusion for one projection PCA step, i.e. for $t\in \mathcal{I}_{j,k}$. Its proof is much longer. It first uses a lemma based on the $\sin\theta$ and Weyl theorems (Theorems \ref{sin_theta} and \ref{weyl}) to get a bound on $\zeta_k$. This is Lemma \ref{zetakbnd}. Next we bound $\kappa_s(D_\new)$ in Lemma \ref{Dnew0_lem}. Finally in Lemma \ref{termbnds}, we use the expression for $e_t$ from Lemma \ref{cslem}, the matrix Hoeffding inequalities (Corollaries \ref{hoeffding_nonzero} and \ref{hoeffding_rec}) and the bound from Lemma \ref{Dnew0_lem} to bound each of the terms in the bound on $\zeta_k$ to finally show that, conditioned on $\Gamma_{j,k-1}^e$, $\zeta_k \le \zeta_k^+$ w.h.p..
We state the two lemmas first and then proceed to prove them in order.

\begin{lem}[Projected CS Lemma]\label{cslem}
Assume that all conditions of Theorem \ref{thm1} hold.
\ben
\item For all $t \in  \mathcal{I}_{j,k}$, for any $k=1,2,\dots K$, if $X_{j,k-1} \in \Gamma_{j,k-1}$,
\ben
\item the projection noise $\beta_t$ satisfies $\|\beta_t\|_2 \leq \zeta_{k-1}^+ \sqrt{c} \gamma_{\new,k} + \zeta_{*}^+ \sqrt{r} \gamma_* \le \sqrt{c} 0.72^{k-1} \gamma_{\new} + 1.06 \sqrt{\zeta} \le \xi_0$.
\item the CS error satisfies $\|\hat{S}_{t,\cs} - S_t\|_2 \le7 \xi_0$.
\item  $\hat{T}_t = T_t$
\item $e_t$ satisfies \eqref{etdef0} and $\|e_t\|_2 \leq \phi^+ [\kappa_s^+ \zeta_{k-1}^+ \sqrt{c} \gamma_{\new,k} + \zeta_{*}^+ \sqrt{r} \gamma_*] \le
 0.18 \cdot 0.72^{k-1} \sqrt{c}\gamma_{\new} + 1.17 \cdot 1.06 \sqrt{\zeta}$.  Recall that \eqref{etdef0} is
\[
 I_{T_t} {(\Phi_{(t)})_{T_t}}^{\dag} \beta_t = I_{T_t} [ (\Phi_{(t)})_{T_t}'(\Phi_{(t)})_{T_t}]^{-1}  {I_{T_t}}' \Phi_{(t)} L_t
\]

\een
\item For all $k=1,2,\dots K$, $\mathbf{P}(\That_t = T_t \ \text{and} \ e_t \ \text{satisfies (\ref{etdef0})}  \text{ for all } t \in \mathcal{{I}}_{j,k}  | X_{j,k-1} ) = 1$   for all $X_{j,k-1} \in \Gamma_{j,k-1}$.
%\item For all $k=1,2,\dots K$, $\mathbf{P}(\That_t = T_t  \ \text{ for all } t \in \mathcal{{I}}_{j,k}  | \Gamma_{j,k-1}^e)=1$.
\een
\end{lem}

\begin{lem}[Projection PCA Lemma] \label{zetak} % Let $\zeta_*^+ = r \zeta$.
Assume that all the conditions of Theorem \ref{thm1} hold. Then, for all $k=1,2, \dots K$,
\[
\mathbf{P}(\zeta_{k} \le \zeta_k^+|\egam_{j,k-1}) \ge p_k(\alpha,\zeta)
\]
where $\zeta_k^+$ is defined in Definition \ref{zetakplus} and $p_k(\alpha,\zeta)$ is defined in \eqref{pk}.
\end{lem}

\vspace{.5 cm}

\begin{proof}[Proof of Lemma \ref{mainlem}]
Observe that $\mathbf{P}(\Gamma_{j,k}|\Gamma_{j,k-1}) = \mathbf{P}( \check{\Gamma}_{j,k} | \Gamma_{j,k-1})$. The lemma then follows by combining Lemma \ref{zetak} and item 2 of Lemma \ref{cslem} and Lemma \ref{rem_prob}.
\end{proof}

\vspace{.4 cm}

\subsection{Proof of Lemma \ref{cslem}}
We begin by first bounding the RIC of the CS matrix $\Phi_k.$

\begin{lem}[Bounding the RIC of $\Phi_k$] \label{RIC_bnd}
Recall that $\zeta_*:= \|(I-\Phat_*{\Phat_*}')P_*\|_2$.  The following hold.
\ben
\item Suppose that a basis matrix $P$ can be split as $P = [P_1, P_2]$ where $P_1$ and $P_2$ are also basis matrices. Then $\kappa_s^2 (P) = \max_{T: |T| \le s} \|I_T'P\|_2^2 \le \kappa_s^2 (P_1) + \kappa_s^2 (P_2)$.
\item $\kappa_s^2(\Phat_*) \leq \kappa_{s,*}^2 + 2\zeta_*$
\item $\kappa_s (\Phat_{\new,k}) \leq \kappa_{s,\new} + \tilde{\kappa}_{s,k} \zeta_k + \zeta_*$
\item $\delta_{s} (\Phi_0) = \kappa_s^2 (\Phat_*) \leq  \kappa_{s,*}^2 + 2 \zeta_*$
\item $\delta_{s}(\Phi_k)  = \kappa_s^2 ([\Phat_* \ \Phat_{\new,k}]) \leq \kappa_s^2 (\Phat_*) + \kappa_s^2 (\Phat_{\new,k})\leq \kappa_{s,*}^2 + 2\zeta_* + (\kappa_{s,\new} + \tilde{\kappa}_{s,k} \zeta_k + \zeta_*)^2$ for $k \ge 1$
\een
\end{lem}

\begin{proof}
\ben
\item Since $P$ is a basis matrix, $\kappa_s^2 (P) = \max_{|T| \leq s} \|{I_T}' P\|_2^2$. Also, $\|{I_T}' P\|_2^2 = \|{I_T}' [P_1, P_2] [P_1, P_2]' I_T \|_2 =  \|{I_T}' (P_1P_1' + P_2 P_2') I_T \|_2 \le \|{I_T}' P_1 P_1' I_T\|_2 + \|{I_T}' P_2 P_2' I_T\|_2$. Thus, the inequality follows.

\item For any set $T$ with $|T| \le s$, $\|{I_T}' \Phat_*\|_2^2  = \|{I_T}' \Phat_* {\Phat_*}'I_T\|_2 =\|{I_T}'( \Phat_* {\Phat_*}' -P_* {P_*}' + P_*{P_*}')I_T\|_2 \leq \|{I_T}'( \Phat_* {\Phat_*}' -P_* {P_*}')I_T\|_2 + \|{I_T}'P_* {P_*}' I_T\|_2 \leq 2\zeta_* + \kappa_{s,*}^2$. The last inequality follows using Lemma \ref{lemma0} with $P=P_*$ and $\hat{P} = \hat{P}_*$.

\item By Lemma \ref{lemma0} with $P = P_*$, $\Phat = \Phat_*$ and $Q = P_\new$, $\|{P_{\new}}' \Phat_*\|_2 \leq \zeta_*$. By Lemma \ref{lemma0} with $P = P_\new$ and $\hat{P} = \hat{P}_{\new,k}$, $\|(I-P_\new P_\new')\Phat_{\new,k}\|_2  = \|(I-\Phat_{\new,k}{\Phat_{\new,k}}')P_{\new}\|_2$.
For any set $T$ with $|T| \leq s$, $\|{I_T}'\Phat_{\new,k}\|_2 \leq \|{I_T}'(I-P_{\new}P_{\new}') \Phat_{\new,k}\|_2 + \|{I_T}'P_{\new}P_{\new}' \Phat_{\new,k}\|_2 \leq \tilde{\kappa}_{s,k} \|(I- P_{\new}{P_{\new}}')\Phat_{\new,k}\|_2 + \|{I_T}'P_{\new}\|_2 = \tilde{\kappa}_{s,k} \|(I-\Phat_{\new,k}{\Phat_{\new,k}}')P_{\new}\|_2 + \|{I_T}'P_{\new}\|_2 \leq \tilde{\kappa}_{s,k} \|D_{\new,k}\|_2 + \tilde{\kappa}_{s,k}\| \Phat_* {\Phat_*}'P_{\new}\|_2 + \|{I_T}'P_{\new}\|_2 \le \tilde{\kappa}_{s,k}\zeta_{k} + \tilde{\kappa}_{s,k} \zeta_* + \kappa_{s,\new} \leq  \tilde{\kappa}_{s,k}\zeta_{k} + \zeta_* + \kappa_{s,\new}$. Taking $\max$ over $|T| \le s$ the claim follows.

\item This follows using Lemma \ref{delta_kappa} and the second claim of this lemma.

\item This follows using Lemma \ref{delta_kappa} and the first three claims of this lemma.
\een

\end{proof}

\begin{corollary}\label{RICnumbnd}
If the conditions of Theorem \ref{thm1} are satisfied, and $X_{j,k-1}\in \Gamma_{j,k-1}$,  then
\ben
\item $\delta_s(\Phi_0) \leq \delta_{2s}(\Phi_0)  \leq {\kappa_{2s,*}^+}^2 + 2\zeta_*^+ <0.1 < 0.1479$
\item $\delta_s(\Phi_{k-1}) \leq \delta_{2s}(\Phi_{k-1}) \leq {\kappa_{2s,*}^+}^2 + 2\zeta_*^+ +(\kappa_{2s,\new}^+ + \tilde{\kappa}_{2s,k-1}^+ \zeta_{k-1}^+ + \zeta_*^+)^2 <0.1479$
\item $\phi_{k-1} \le \frac{1}{1-\delta_s(\Phi_{k-1})} < \phi^+$
\een
\end{corollary}

\begin{proof}
This follows using Lemma \ref{RIC_bnd}, the definition of $\Gamma_{j,k-1}$, and the bound on $\zeta_{k-1}^+$ from Lemma \ref{expzeta}.
\end{proof}

The following are straightforward bounds that will be useful for the proof of Lemma \ref{cslem} and later.

\begin{fact}\label{constants}
Under the assumptions of Theorem \ref{thm1}:
\ben
\item $\zeta \gamma_* \le \frac{\sqrt{\zeta}}{(r_0 + (J-1) c)^{3/2}} \le  \sqrt{\zeta}$
\item $\zeta_{j,*}^+ \leq  \frac{10^{-4}}{(r_0 + (J-1) c)} \le 10^{-4}$
\item $ \zeta_{j,*}^+ \gamma_*^2 \leq \frac{1}{(r_0 + (J-1)c)^2} \le 1$
\item $ \zeta_{j,*}^+ \gamma_* \le \frac{\sqrt{\zeta}}{\sqrt{r_0 + (J-1)c}} \le \sqrt{\zeta}$
\item $\zeta_{j,*}^+ f \leq \frac{1.5 \times 10^{-4}}{r_0 + (J-1)c} \le 1.5 \times 10^{-4}$
\item $\zeta_{k-1}^+ \leq 0.6^{k-1} + 0.4 c\zeta$ (from Lemma \ref{expzeta})
\item $\zeta_{k-1}^+ \gamma_{\new,k} \leq (0.6 \cdot 1.2)^{k-1} \gamma_{\new} + 0.4c\zeta \gamma_*  \le 0.72^{k-1}\gamma_{\new} + \frac{0.4\sqrt{\zeta}}{\sqrt{r_0 + (J-1)c}} \le 0.72^{k-1}\gamma_{\new} + 0.4\sqrt{\zeta}$
\item $\zeta_{k-1}^+ \gamma_{\new,k}^2 \leq (0.6 \cdot 1.2^2)^{k-1} \gamma_{\new}^2 + 0.4 c\zeta \gamma_*^2 \le 0.864^{k-1}\gamma_{\new}^2 + \frac{0.4}{{(r_0 + (J-1)c})^2} \le 0.864^{k-1}\gamma_{\new}^2 + 0.4$
\een
\end{fact}

\begin{proof}[Proof of Lemma \ref{cslem}]
Recall that $X_{j,k-1} \in \Gamma_{j,k-1}$ implies that $\zeta_{j,*} \leq \zeta_{j,*}^+$ and $\zeta_{k-1}\leq \zeta_{k-1}^+$.
\ben
\item
\begin{enumerate}
\item For $t \in \mathcal{I}_{j,k}$, $\beta_t := (I-\Phat_{(t-1)} {\Phat_{(t-1)} }') L_t = D_{*,k-1} a_{t,*} + D_{\new,k-1} a_{t,\new} $. Thus, using Fact \ref{constants}
\begin{align*}
\|\beta_t\|_2 & \leq \zeta_{j,*} \sqrt{r} \gamma_* + \zeta_{k-1} \sqrt{c} \gamma_{\new,k} \\
& \leq \sqrt{\zeta}\sqrt{r} + (0.72^{k-1}\gamma_{\new} + .4\sqrt{\zeta})\sqrt{c} \\
& = \sqrt{c} 0.72^{k-1} \gamma_{\new} + \sqrt{\zeta} (\sqrt{r} + 0.4\sqrt{c}) \leq \xi_0.
\end{align*}

\item By Corollary \ref{RICnumbnd}, $\delta_{2s} (\Phi_{k-1}) < 0.15< \sqrt{2}-1$. Given $|T_t| \leq s$, $\|\beta_t\|_2 \leq \xi_0 = \xi$, by Theorem \ref{candes_csbound}, the CS error satisfies
\[
\|\hat{S}_{t,\cs} - S_t\|_2 \leq  \frac{4\sqrt{1+\delta_{2s}(\Phi_{k-1})}}{1-(\sqrt{2}+1)\delta_{2s}(\Phi_{k-1})} \xi_0 < 7 \xi_0.
\]

\item Using the above, $\|\hat{S}_{t,\cs} - S_t\|_{\infty} \leq 7  \xi_0$. Since $\min_{i\in T_t} |(S_t)_{i}| \geq S_{\min}$ and $(S_t)_{T_t^c} = 0$, $\min_{i\in T_t} |(\hat{S}_{t,cs})_i| \geq S_{\min} - 7 \xi_0$ and $\min_{i\in T_t^c} |(\hat{S}_{t,\cs})_i| \leq 7 \xi_0$. If $\omega < S_{\min} - 7 \xi_0$, then $\hat{T}_t \supseteq T_t$. On the other hand, if $\omega > 7 \xi_0$, then $\hat{T}_t \subseteq T_t$. Since $S_{\min} > 14 \xi_0$ (condition 3 of the theorem) and $\omega$ satisfies $7 \xi_0 \leq \omega \leq S_{\min} -7 \xi_0$ (condition 1 of the theorem), then the support of $S_t$ is exactly recovered, i.e. $\hat{T}_t = T_t$.

\item Given $\hat{T}_t = T_t$, the LS estimate of $S_t$ satisfies $(\hat{S}_t)_{T_t} = [(\Phi_{k-1})_{T_t}]^{\dag} y_t =[(\Phi_{k-1})_{T_t}]^{\dag} (\Phi_{k-1} S_t + \Phi_{k-1}L_t)$ and $(\hat{S}_t)_{T_t^c} = 0$ for $t \in \mathcal{I}_{j,k}$. Also,  ${(\Phi_{k-1})_{T_t}}' \Phi_{k-1} = {I_{T_t}}' \Phi_{k-1}$ (this follows since $(\Phi_{k-1})_{T_t} = \Phi_{k-1} I_{T_t}$ and $\Phi_{k-1}'\Phi_{k-1} = \Phi_{k-1}$).  Using this, the LS error $e_t := \hat{S}_t - S_t$ satisfies (\ref{etdef0}).
Thus, using Fact \ref{constants} and condition 2 of the theorem,
\begin{align*}
\|e_t\|_2 & \le \phi^+ (\zeta_{j,*}^+ \sqrt{r}\gamma_* + \kappa_{s,k-1} \zeta_{k-1}^+ \sqrt{c}\gamma_{\new,k}) \\
& \le 1.2 \left(\sqrt{r}\sqrt{\zeta}+ \sqrt{c} 0.15 (0.72)^{k-1} + \right. \\
& \hspace{1.5in} \left. \sqrt{c} 0.06\sqrt{\zeta}\right) \\
& = 0.18 \sqrt{c}0.72^{k-1}\gamma_{\new} + \\
&\hspace{.8in}  1.2 \sqrt{\zeta}(\sqrt{r} + 0.06 \sqrt{c}).
\end{align*}
\end{enumerate}

\item The second claim is just a restatement of the first.
%\item The third claim follows from the second by Lemma \ref{rem_prob}.

\een

\end{proof}

\subsection{Proof of Lemma \ref{zetak}}

The proof of Lemma \ref{zetak} will use the next three lemmas %(\ref{zetakbnd}, and \ref{termbnds}).

\begin{lem}\label{zetakbnd}
If $\lambda_{\min}(A_k) - \|A_{k,\perp}\|_2 - \|\mathcal{H}_k\|_2 >0$, then
\begin{align*} \label{zetakbound}
\zeta_k &\leq  \frac{\|\mathcal{R}_k\|_2}{\lambda_{\min} (A_k) - \|A_{k,\perp}\|_2 - \|\mathcal{H}_k\|_2} \\
& \leq  \frac{\|\mathcal{H}_k\|_2}{\lambda_{\min} (A_k) - \|A_{k,\perp}\|_2 - \|\mathcal{H}_k\|_2}
\end{align*}
where $\mathcal{R}_k := \mathcal{H}_k E_\new$ and $A_k$, $A_{k,\perp}$, $\mathcal{H}_k$ are defined in Definition \ref{defHk}.
\end{lem}

\begin{proof}
Since $\lambda_{\min}(A_k) - \|A_{k,\perp}\|_2 - \|\mathcal{H}_k\|_2 >0$, so $\lambda_{\min}(A_k) > \|A_{k,\perp}\|_2$. Since $A_k$ is of size $c_\new \times c_\new$ and $\lambda_{\min}(A_k) > \|A_{k,\perp}\|_2$, $\lambda_{c_\new+1} (\mathcal{A}_k) = \|A_{k,\perp}\|_2$. By definition of EVD, and since $\Lambda_k$ is a $c_\new \times c_\new$ matrix, $\lambda_{\max}(\Lambda_{k,\perp}) = \lambda_{c_\new+1}(\mathcal{A}_k + \mathcal{H}_k)$. By Weyl's theorem (Theorem \ref{weyl}), $\lambda_{c_\new+1}(\mathcal{A}_k + \mathcal{H}_k) \leq \lambda_{c_\new+1} (\mathcal{A}_k) + \|\mathcal{H}_k\|_2 = \|A_{k,\perp}\|_2 + \|\mathcal{H}_k\|_2$. Therefore, $\lambda_{\max}(\Lambda_{k,\perp})\leq \|A_{k,\perp}\|_2 + \|\mathcal{H}_k\|_2$ and hence $\lambda_{\min}(A_k) - \lambda_{\max}(\Lambda_{k,\perp})\geq \lambda_{\min}(A_k) - \|A_{k,\perp}\|_2 - \|\mathcal{H}_k\|_2 > 0$. Apply the $\sin \theta$ theorem (Theorem \ref{sin_theta}) with $\lambda_{\min}(A_k) - \lambda_{\max}(\Lambda_{k,\perp})> 0$, we get
\begin{align*}
\|(I- \Phat_{\new,k} {\Phat_{\new,k}}') E_{\new} \|_2 &\leq \frac{\|\mathcal{R}_k\|_2}{ \lambda_{\min}(A_k) - \lambda_{\max}(\Lambda_{k,\perp})}\\
& \leq \frac{\|\mathcal{H}_k\|_2}{\lambda_{\min}(A_k) - \|A_{k,\perp}\|_2 - \|\mathcal{H}_k\|_2}
\end{align*}
Since $\zeta_k  = \|(I- \Phat_{\new,k} {\Phat_{\new,k}}') D_{\new}\|_2 = \|(I- \Phat_{\new,k} {\Phat_{\new,k}}') E_{\new} R_{\new} \|_2 \leq \|(I- \Phat_{\new,k} {\Phat_{\new,k}}') E_{\new}\|_2$, the result follows.
The last inequality follows because  $\|R_\new\|_2 = \|E_\new' D_\new\|_2 \le 1$.
\end{proof}

\begin{lem} \label{Dnew0_lem}
Assume that the assumptions of Theorem \ref{thm1} hold.
Conditioned on $\Gamma_{j,k-1}^e$, 
\begin{align*}
\kappa_s(D_{\new}) &\le \frac{\kappa_s(P_{\new}) +  \zeta_*^+}{\sqrt{1-\zeta_*^+}} \\
& \le \frac{\kappa_{2s,\new}^+ + 0.0015}{\sqrt{1 - 0.0015}} \approx 0.1516 \le \kappa_s^+.
\end{align*}
\end{lem}
\begin{proof}
Recall that $D_\new = D_{\new,0} = (I - \Phat_{j-1}\Phat_{j-1}') P_\new$. Also $D_{\new} \overset{\mathrm{QR}}{=} E_{\new}R_{\new}$.   By lemma \ref{hatswitch} $\|{R_{\new}}^{-1}\|_2 \leq \frac{1}{\sqrt{1-\zeta_*^+}}$. 
$\kappa_s(D_{\new}) = \kappa_s(E_{\new}) = \max_{|T|\leq s} \|I_T'D_{\new}{R_{\new}}^{-1}\|_2 \leq \max_{|T|\leq s} \|I_T'D_{\new}\|_2\|{R_{\new}}^{-1}\|_2 \leq \frac{\kappa_s(P_{\new}) + \zeta_*}{\sqrt{1-\zeta_*^+}}$.
%By the triangle inequality and  Lemma \ref{hatswitch}, $\kappa_s(D_{\new,0}) \le \max_{|T|\le s} \|I_T' P_\new\|_2 + \|\Phat_{j-1}'P_\new\|_2 \le \kappa_s(P_{\new}) + \zeta_*$. 
The event $\Gamma_{j,k-1}^e$ implies that $\zeta_* \le \zeta_*^+ \le 0.0015$. Thus, the lemma follows.
\end{proof}

\begin{lem}[High probability bounds for each of the terms in the $\zeta_k$ bound (\ref{zetakbnd})]\label{termbnds}
Assume the conditions of Theorem \ref{thm1} hold.  Also assume that $\mathbf{P}(\Gamma_{j,k-1}^e)>0$ for all $1\leq k \leq K+1$. Then, for all $1 \le k \le K$
\begin{enumerate}

\item $\mathbf{P} \left(\lambda_{\min} (A_{k}) \geq  \lambda_{\new,k}^- \left(1 -(\zeta_{j,*}^+)^2  - \frac{c \zeta}{12}\right) \big|\egam_{j,k-1}\right) > 1- p_{a,k}(\alpha,\zeta)$
where
\begin{align*}
p_{a,k}(\alpha,\zeta) &:= c \exp \left(\frac{-\alpha \zeta^2 (\lambda^-)^2}{8 \cdot 24^2 \cdot \min(1.2^{4k} \gamma_{\new}^4,\gamma_*^4)}\right) \\ 
&\quad + c \exp \left( \frac{-\alpha c^2 \zeta^2(\lambda^-)^2} {8 \cdot 24^2 \cdot 4^2}\right)
\end{align*}

\item $\mathbf{P}\left(\lambda_{\max}(A_{k,\perp}) \leq \lambda_{\new,k}^- \left( (\zeta_{j,*}^+)^2 f + \frac{c \zeta}{24}\right) \big| \egam_{j,k-1}\right) > 1- p_b(\alpha,\zeta)$
where
\[
p_b (\alpha,\zeta) := (n-c) \exp \left(\frac{-\alpha c^2 \zeta (\lambda^-)^2}{8 \cdot 24^2}\right)
\]

\item $\mathbf{P}\left(\|\mathcal{H}_{k}\|_2 \leq \lambda_{\new,k}^- (b + 0.125c\zeta) \ \big|\egam_{j,k-1}\right) \geq 1 - p_c(\alpha,\zeta)$
where $b$ is as defined in Definition \ref{zetakplus} and
\begin{align*}
& p_c(\alpha,\zeta) := \\
 & n \exp\left(\frac{-\alpha \zeta^2 (\lambda^-)^2}{8 \cdot 24^2 (.0324 \gamma_\new^2 + .0072 \gamma_\new + .0004)^2}\right) + \\
& n \exp\left(\frac{-\alpha \zeta^2 (\lambda^-)^2}{32\cdot 24^2 (.06  \gamma_\new^2 + .0006  \gamma_\new + .4)^2}\right)+ \\
&  n \exp\left(\frac{-\alpha  \zeta^2 (\lambda^-)^2 \epsilon^2}{32 \cdot 24^2 ( .186 \gamma_\new^2 + .00034 \gamma_\new + 2.3)^2}\right).
\end{align*}

\end{enumerate}

\end{lem}

\begin{proof}
The proof is quite long and hence is given in Appendix \ref{appendix termbnds}.  The first two claims are obtained by simplifying the terms and then appropriately  applying the Hoeffding corollaries. The third claim first uses Lemma \ref{cslem} to argue that conditioned on $X_{j,k-1} \in \Gamma_{j,k-1}$, $e_t$ satisfies (\ref{etdef0}). It then simplifies the resulting expressions and eventually uses the Hoeffding corollaries. The simplification also uses the bound on $\kappa_s(D_\new)$ from Lemma \ref{Dnew0_lem}.
%It is easy to see that all terms in the expansion of $A_k, A_{k,\perp}$ are of the form
%It uses the above facts and the matrix Hoeffding corollaries (Corollaries \ref{hoeffding_nonzero} and \ref{hoeffding_rec})
\end{proof}

\begin{proof}[Proof of Lemma \ref{zetak}]
Lemma \ref{zetak} now follows by combining Lemmas \ref{zetakbnd} and \ref{termbnds} and defining
\begin{equation}
p_k(\alpha,\zeta) := 1 - p_{a,k}(\alpha,\zeta) - p_b(\alpha,\zeta) - p_c(\alpha,\zeta). \label{pk}
\end{equation}
\end{proof}

%As above, we will start with some simple facts that will be used to prove Lemma \ref{termbnds}.

 %rename it after we have checked the edits %\input{lemmas_2}

\section{ReProCS with Cluster PCA}
\label{Del_section}

The ReProCS approach studied so far is designed under the assumption that the subspace in which $L_t$ lies can only grow over time.  In practice, usually, the dimension of this subspace typically remains roughly constant. A simple way to model this is to assume that at every change time, $t_j$, some new directions can get added and some directions from the existing subspace can get deleted and to assume an upper bound on the difference between the total number of added and deleted directions. We specify this model next.
\begin{sigmodel}
\label{del model}
Assume that $L_t =P_{(t)} a_t$ where $P_{(t)} = P_j$ for all $t_j \leq t <t_{j+1}$, $j=0,1,2 \cdots J$, $P_j$ is an $n \times r_j$ basis matrix with $r_j  \ll \min(n,(t_{j+1} - t_j))$. We let $t_0=0$ and $t_{J+1}$ equal the sequence length. This can be infinity also.
\begin{enumerate}
\item  At the change times, $t_j$, $P_j$ changes as
\[
P_j = [(P_{j-1}R_j\setminus P_{j,\old})  \quad P_{j,\new}]
\]
Here, $R_j$ is a rotation matrix, $P_{j,\new}$ is an $n \times c_{j,\new}$ basis matrix with $P_{j,\new}'P_{j-1} = 0$ and $P_{j,\old}$ contains $c_{j,\old}$ columns of $P_{j-1}R_j$. Thus $r_j = r_{j-1} + c_{j,\new} - c_{j,\old}$. Also, $0 < t_{\train} \le t_1$. This model is illustrated in Figure \ref{add_del_model}.

\item There exist constants $c_{\max}$ and $c_{\text{dif}}$ such that $0 \le c_{j,\new} \leq c_{\max}$ and $\sum_{i=1}^{j} (c_{i,\new} - c_{i,\old}) \leq c_{\text{dif}}$ for all $j$.
    Thus, $r_j = r_0+\sum_{i=1}^{j} (c_{i,\new} - c_{i,\old}) \leq r_{\max}: = r_0 + c_{\text{dif}}$, i.e., the rank of $P_j$ is upper bounded by $r_{\max}$.
\end{enumerate}
\end{sigmodel}

\begin{figure}[t!]
\centerline{
\epsfig{file = 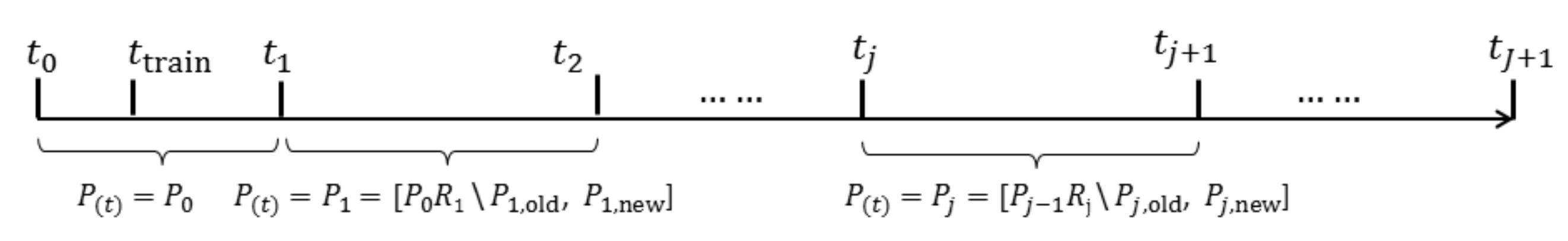, width =\columnwidth}
}
\caption{\small{The subspace change model given in Signal Model \ref{del model}. Here $t_0=0$. % and $0 < t_{\train} \le t_1$.
}}
\label{add_del_model}
\end{figure}

The ReProCS algorithm (Algorithm \ref{reprocs}) still applies for the above more general model. We can conclude the following for it.
\begin{corollary} \label{cor_rep}
Consider Algorithm \ref{reprocs} for the model given above. The result of Theorem \ref{thm1} applies with the following change: we also need $\kappa_{2s}([P_0, P_{1,\new}, \dots , P_{J-1,\new}]) \le 0.3$.
\end{corollary}
Because Algorithm \ref{reprocs} never deletes directions, the rank of $\Phat_{(t)}$ keeps increasing with every subspace change time (even though the rank of $P_{(t)}$ is now bounded by $r_0 + c_{\text{dif}}$). As a result, the performance guarantee above still requires a bound on $J$ that is imposed by the denseness assumption.  In this section, we address this limitation by  re-estimating the current subspace after the newly added directions have been accurately estimated. This helps to ``delete" $\Span(P_\old)$ from the subspaces estimate.  For the resulting algorithm, as we will see, we do not need a bound on the number of changes, $J$, as long as the separation between the subspace change times is allowed to grow logarithmically with $J$.
% $\kappa_{2s}([P_0, P_{1,\new}, \dots P_{J-1,\new}]) \le 0.3$

%introducing a novel approach called cluster-PCA that re-estimates the current subspace after the newly added directions have been accurately estimated. Or, in other words, it ``deletes" $\Span(P_\old)$ from the subspaces estimate.

One simple way to re-estimate the current subspace would be by standard PCA: at $t=\tilde{t}_j + \tilde{\alpha}-1$, compute $\Phat_j \leftarrow \text{proj-PCA}([\Lhat_t;  \tilde{\mathcal{I}}_{j,1}], [.], r_j)$ and let $\Phat_{(t)} \leftarrow \Phat_j$. Using the $\sin \theta$ theorem \cite{davis_kahan} and the matrix Hoeffding inequality \cite{tail_bound}, and using the procedure used earlier to analyze projection PCA, it can be shown that, as long as $f$, a bound on the maximum condition number of $\operatorname{Cov}[L_t]$, is small enough, doing this is guaranteed to give an accurate estimate of $\Span(P_j)$. However as explained in Remark \ref{large_f}, $f$ cannot be small because our problem definition allows large noise, $L_t$, but assumes slow subspace change. %the maximum condition number of $Cov[L_t]$, $f$, needs to be large
%being small is not compatible with the slow subspace change assumption. %Notice from Sec \ref{probdef} that $\lambda^- \le \gamma_\new$ and $\E[\|L_t\|_2^2] \le r \lambda^+$. Slow subspace change implies that $\gamma_\new$ is small. Thus, $\lambda^-$ is small. However, to allow $L_t$ to have large average magnitude, $\lambda^+$ needs to be large. Thus,  $f = \lambda^+ / \lambda^-$ cannot be small unless we require that $L_t$ has small magnitude for all times $t$. This, however, violates our problem formulation where say that we allow large but strucutured noise $L_t$.
%
In other works that analyze standard PCA, e.g. \cite{nadler} and references therein, the large condition number does not cause a problem because they assume that the error ($e_t$ in our case) in the observed data vector ($\Lhat_t$) is uncorrelated with the true data vector ($L_t$). Under this assumption, one only needs to increase the PCA data length $\alpha$ to deal with larger condition numbers. However, in our case, because $e_t$ is correlated with $L_t$, this strategy does not work. This issue is explained in detail in Appendix \ref{projpca}.

In this section, we introduce a generalization of the above strategy called cluster-PCA that removes the requirement that $f$ be small, but instead only requires that the eigenvalues of $\text{Cov}(L_t)$ be clustered for the times when the changed subspace has stabilized. Under this assumption, cluster-PCA recovers one cluster of entries of $P_j$ at a time by using an approach that generalizes the projection PCA step developed earlier. We first explain the clustering assumption in Sec \ref{eigencluster} below and then give the cluster-PCA algorithm.
%be sufficiently clustered as explained in Sec \ref{eigencluster}.

%We now consider a more general subspace change model where directions may also be deleted from the low dimensional subspace at the change times and as a result the subspace dimension does not need to keep increasing over time as modeled earlier. This is more practically valid. %For this model Algorithm \ref{reprocs} still applies and we will get a result si
%The algorithm used to handle this case is the same as Algorithm \ref{reprocs}  except that we add a cluster PCA step.

\subsection{Clustering assumption}\label{eigencluster}
%In order to obtain performance guarantees for the cluster PCA step, we need the following clustering assumption.

\newcommand{\group}{\mathcal{G}}
For positive integers $K$ and $\alpha$, let $\tilde{t}_j := t_j + K \alpha$. We set their values in, Theorem \ref{thm2}.  %where $K$ and $\alpha$ are introduced in Algorithm \ref{ReProCS_del} given in the next section. We set their values are set in our main result, Theorem \ref{thm2} which ensures that $\tilde{t}_j < t_{j+1}$.
Recall from the model on $L_t$ and the slow subspace change assumption that new directions, $P_{j,\new}$, get added at $t=t_j$ and initially, for the first $\alpha$ frames,  the projection of $L_t$ along these directions is small (and thus their variances are small), but can increase gradually. It is fair to assume that within $K \alpha$ frames, i.e. by $t=\tilde{t}_j$, the variances along these new directions have stabilized and do not change much for $t \in [\tilde{t}_j, t_{j+1}-1]$. It is also fair to assume that the same is true for the variances along the existing directions, $P_{j-1}$. In other words, we assume that the matrix $\Lambda_t$ is either constant or does not change much during this period. Under this assumption, we assume that we can cluster its eigenvalues (diagonal entries) into a few clusters such that the distance between consecutive clusters is large and the distance between the smallest and largest element of each cluster is small. % and (4) the number of clusters is small. We need these entries to be similar enough so that we can cluster them into a few groups that remain the same for all $t \in [\tilde{t}_j, t_{j+1}-1]$.
We make this precise below.

\begin{ass} \label{clusterass}
Assume the following.
\ben
\item %For a given $i$, for all $t \in [\tilde{t}_j, t_{j+1}-1]$, $\lambda_i(\Lambda_t)$ are similar enough so that
Either $\Lambda_t = \Lambda_{\tilde{t}_j}$ for all ${t \in [\tilde{t}_j, t_{j+1}-1]}$ or $\Lambda_t$ changes very little during this period so that for each $i=1,2,\cdots,r_j$, $\min_{t \in [\tilde{t}_j, t_{j+1}-1]} \lambda_i(\Lambda_t) \ge \max_{t \in [\tilde{t}_j, t_{j+1}-1]} \lambda_{i+1}(\Lambda_t)$. %A simple situation where this holds is when $\Lambda_t = \Lambda_{\tilde{t}_j}$ for all ${t \in [\tilde{t}_j, t_{j+1}-1]}$.

\item  Let $ \group_{j,(1)}, \group_{j,(2)}, \cdots, \group_{j,(\vartheta_j)} $ be a partition of the index set $\{1,2, \dots r_j\}$ so that $\min_{i \in \group_{j,(k)}} \min_{t \in [\tilde{t}_j, t_{j+1}-1]} \lambda_i(\Lambda_t)  > \max_{i \in \group_{j,(k+1)}} \max_{t \in [\tilde{t}_j, t_{j+1}-1]} \lambda_i(\Lambda_t)$, i.e. the first group/cluster contains the largest set of eigenvalues, the second one the next smallest set and so on (see Figure \ref{clustering_diag}). Let
\ben
\item $G_{j,k} := (P_j)_{ \group_{j,(k)} }$ be the corresponding cluster of eigenvectors, then $\operatorname{span}(P_j) = \operatorname{span}([G_{j,1},G_{j,2},\cdots,G_{j,\vartheta_j}])$;
\item $\tilde{c}_{j,k} := |\group_{j,(k)}|$ be the number of elements in $\group_{j,(k)}$, then $\sum_{k=1}^{\vartheta_j} \tilde{c}_{j,k} = r_j$; \\
$\tilde{c}_{\min} : = \min_j \min_{k = 1,2,\cdots,\vartheta_j} \tilde{c}_{j,k}$

\item  ${\lambda_{j,k}}^- := \min_{i\in \group_{j,(k)} } \min_{t\in [\tilde{t}_j, t_{j+1}-1]}  \lambda_i (\Lambda_t)$,  ${\lambda_{j,k}}^+ := \max_{i \in \group_{j,(k)} } \max_{t\in [\tilde{t}_j, t_{j+1}-1]}  \lambda_i (\Lambda_t)$ and ${\lambda_{j,\vartheta_j+1} }^+:= 0$;

\item $\tilde{g}_{j,k} := {\lambda_{j,k}}^+ / {\lambda_{j,k}}^- $ (notice that $\tilde{g}_{j,k} \ge 1$);
% be the upper bound on the condition number of $\text{Cov}(a_{t,k})$ where $a_{t,k}:= G_{j,k}' L_t$,

\item $\tilde{h}_{j,k} := {\lambda_{j,k+1}}^+ / {\lambda_{j,k}}^-$ (notice that $\tilde{h}_{j,k} < 1$);
 %(distance between consecutive clusters is large)
\item $\tilde{g}_{\max} := \max_j  \max_{k = 1,2,\cdots,\vartheta_j} \tilde{g}_{j,k}$, \\ $\tilde{h}_{\max} := \max_j  \max_{k = 1,2,\cdots,\vartheta_j} \tilde{h}_{j,k}$,
\item  $\vartheta_{\max}: = \max_j \vartheta_j$
\een
%Notice that $\tilde{g}_{j,k} \ge 1$, $\tilde{h}_{j,k} < 1$ and $\sum_{k=1}^{\vartheta_j} \tilde{c}_{j,k} = r_j$.

 We assume that $\tilde{g}_{\max}$ is small enough (the distance between the smallest and largest eigenvalues of a cluster is small)
and $\tilde{h}_{\max}$ is small enough (distance between consecutive clusters is large). We quantify this in Theorem \ref{thm2}.
\een
\end{ass}

\begin{remark}
In order to address a reviewer's concern, we should clarify the following point. The above assumption still allows the newly added eigenvalues to become large and hence still allows the subspace of $L_t$ to change significantly over time. The above requires the covariance matrix of $L_t$ to be constant or nearly constant only for the time between $\tilde{t}_j:= t_j + K \alpha$ and the next change time, $t_{j+1}$ and not for the first $K \alpha$ frames. Slow subspace change assumes that the projection of $L_t$ along the new directions is initially small for the first $\alpha$ frames but then can increase gradually over the next $K-1$ intervals of duration $\alpha$. The variance along the new directions can increase by as much as $1.2^{2K}$ times the initial variance. Thus by $t=\tilde{t}_j=t_j+K \alpha$, the variances along the new directions can have already increased to large enough values. %, i.e. the new directions can already have become prominent.
\\
We can allow the variances to increase for even longer with the following simple change: re-define $\tilde{t}_j$ as $\tilde{t}_j:= t_{j+1} - \vartheta_j \tilde\alpha$ in both the clustering assumption and the algorithm. With this redefinition, we will be doing cluster-PCA at the very end of the current subspace interval.
\\
Lastly, note that the projection along the new directions can further increase in the later subspace change periods also.
%Moreover, we can allow the variance along the new directions to increase for even longer by re-defining $\tilde{t}_j$ as $\tilde{t}_j:= t_{j+1} - \vartheta_j \tilde\alpha$ in both the clustering assumption and the algorithm. As we will see in Algorithm \ref{ReProCS_del}, cluster-PCA requires estimates of $L_t$ from $\vartheta_j \tilde\alpha$ time instants.
% because it requires the covariance matrix to be constant or nearly constant only for the time between $t_j + K \alpha$ and the next change time $t_{j+1}$. Slow subspace change assumes that the projection along the new directions is initially small for the first alpha frames but then can increase gradually (and hence the same is true for the variance along these directions). Thus by $t=t_j+K \alpha$, the variances can have increased to large enough values.
\end{remark}

\begin{figure}[t!]
\centerline{
\epsfig{file = 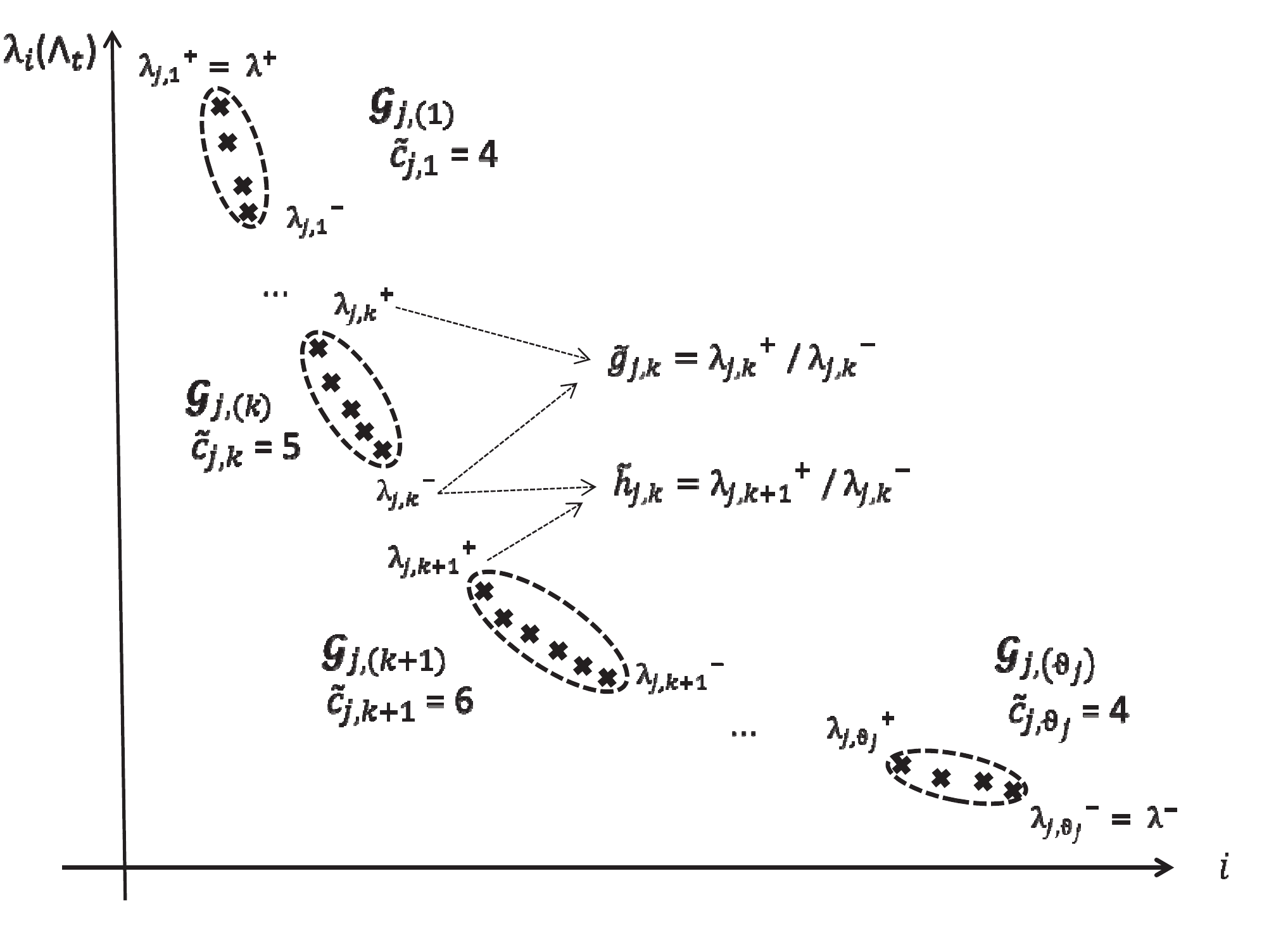, width = \columnwidth}
}
\caption{\small{We illustrate the clustering assumption. Assume $\Lambda_t = \Lambda_{\tilde{t}_j}$.
}}
\label{clustering_diag}
\end{figure}

\subsection{The ReProCS with Cluster PCA Algorithm}
ReProCS-cPCA is summarized in Algorithm \ref{ReProCS_del}. It uses the following definition.
\begin{definition}\label{defn_intervals}
%With $K$, $\alpha_\add$ and $\alpha_\del$ given in Definition \ref{defn_alpha} and with $\alpha \ge \alpha_\add$ and $\tilde{\alpha} \ge \alpha_\del$, we define the following time intervals for $[t_j, t_{j+1}-1]$:
Let  $\tilde{t}_j := t_j + K\alpha$. Define the following time intervals
\ben
%\item $\mathcal{I}_j := [t_j, t_{j+1}-1]$.
\item $\mathcal{I}_{j,k}:= [t_j + (k-1)\alpha, t_j + k\alpha-1]$ for $k=1,2,\cdots,K$.
%\item $\tilde{t}_j := t_j + K\alpha$.
\item $\tilde{\mathcal{I}}_{j,k} := [\tilde{t}_j + (k-1) \tilde{\alpha}, \tilde{t}_j + k \tilde{\alpha}-1]$ for $k = 1,2,\cdots, \vartheta_j$.
\item $\tilde{\mathcal{I}}_{j,\vartheta_j+1} := [\tilde{t}_j + \vartheta_j \tilde{\alpha}, t_{j+1}-1]$.
\een
\end{definition}

\begin{algorithm*}[t]
\caption{Recursive Projected CS with cluster-PCA (ReProCS-cPCA)}\label{ReProCS_del}
{\bf Parameters: } algorithm parameters: $\xi$, $\omega$, $\alpha$, $\tilde{\alpha}$, $K$, model parameters: $t_j$, $c_{j,\new}$, $\vartheta_j$ and $\tilde{c}_{j,i}$ %for $i=1,2,\cdots,\vartheta_j$.
%\\ (set as in Theorem \ref{thm2} or as in Sec \ref{prac_rep} when the model is not known)
\\
{\bf Input: } $n \times 1$ vector, $M_t$, and $n \times r_0$ basis matrix $\hat{P}_0$.
{\bf Output: } $n \times 1$ vectors $\Shat_t$ and $\Lhat_t$, and $n \times r_{(t)}$ basis matrix $\Phat_{(t)}$. %{\em Feed-back: } $\Phat_{(t-1)}$, $\Lhat_{t-1}, \Lhat_{t-2}, \dots \Lhat_{t-\alpha}$ (at most $\alpha$ previous estimates of $L_t$)
\\
{\bf  Initialization: } %For $t=0$ to $t=t_{\train}$, set $\Lhat_t= M_t$.  Compute $\Phat_0 \leftarrow \text{proj-PCA} ([\Lhat_t; t \in [0,1, \cdots, t_{\text{train}}], [.], r_0)$.
Let $\Phat_{(t_\train)} \leftarrow \Phat_0$.
Let $j \leftarrow 1$, $k\leftarrow 1$. %Do the following:
For $t > t_{\train}$, do the following:
\ben
\item {\bf Estimate $T_t$ and $S_t$ via Projected CS: }
\ben
\item \label{othoproj} Nullify most of $L_t$: compute $\Phi_{(t)} \leftarrow I-\Phat_{(t-1)} {\Phat_{(t-1)}}'$, $y_t \leftarrow \Phi_{(t)} M_t$
\item \label{Shatcs} Sparse Recovery: compute $\hat{S}_{t,\cs}$ as the solution of $\min_{x} \|x\|_1 \ s.t. \ \|y_t - \Phi_{(t)} x\|_2 \leq \xi$
\item \label{That} Support Estimate: compute $\hat{T}_t = \{i: \ |(\hat{S}_{t,\cs})_i| > \omega\}$
\item \label{LS} LS Estimate of $S_t$: compute $(\hat{S}_t)_{\hat{T}_t}= ((\Phi_t)_{\hat{T}_t})^{\dag} y_t, \ (\hat{S}_t)_{\hat{T}_t^{c}} = 0$
\een
\item {\bf Estimate $L_t$. } $\hat{L}_t = M_t - \hat{S}_t$.
\item \label{PCA} %Projection PCA
{\bf Update $\Phat_{(t)}$}: %by Addition and Deletion Projection PCA %: either use previous estimate or do projection PCA  %via Projection PCA: use a for loop for this ??
\ben
\item If $t \neq t_j + q\alpha-1$ for any $q=1,2, \dots K$ and $t \neq t_j + K \alpha + \vartheta_j \tilde{\alpha} -1$,
\ben
\item set $\Phat_{(t)} \leftarrow \Phat_{(t-1)}$
\een
%\item $k \leftarrow 1$
\item {\bf Addition: Estimate $\Span(P_{j,\new})$ iteratively using proj-PCA: } If $t = t_j + k\alpha-1$%, for any $k=1,2 \dots K$
\ben

\item $\Phat_{j,\new,k} \leftarrow \text{proj-PCA} ([\hat{L}_{t_j+(k-1)\alpha}, \dots , \hat{L}_{t_j + k\alpha - 1}], \Phat_{j-1}, c_{j,\new})$

\item set $\Phat_{(t)} \leftarrow [\Phat_{j-1} \ \Phat_{j,\new,k}]$.

\item If $k=K$, reset $k \leftarrow 1$; else increment $k \leftarrow k+1$.
\een

\item {\bf Deletion: Estimate $\Span(P_j)$ by cluster-PCA:} If $t= t_j + K \alpha + \vartheta_j \tilde{\alpha} -1$,

\ben
%\item $ \hat{P}_j \leftarrow \text{cluster-PCA} ([\hat{L}_{t_j + K\alpha},\cdots, \hat{L}_{t_j + K \alpha + \vartheta_j \tilde{\alpha} -1}], \tilde{\alpha}, \tilde{\alpha}, \cdots, \tilde{\alpha}, \tilde{c}_{j,1}, \tilde{c}_{j,2},\cdots, \tilde{c}_{j,\vartheta_j})$
%\item $\hat{G}_{\text{det}} \leftarrow [.]$.
\item set $\hat{G}_{j,0} \leftarrow [.]$
\item For $i = 1,2,\cdots, \vartheta_j$,
  \bi
  \item $\hat{G}_{j,i} \leftarrow \text{proj-PCA}( [\hat{L}_{\tilde{t}_j+(i-1)\tilde\alpha}, \dots , \hat{L}_{\tilde{t}_j+i\tilde\alpha -1}], [\hat{G}_{j,1},\hat{G}_{j,2}, \dots \hat{G}_{j,i-1}], \tilde{c}_{j,i})$
%  \item $\hat{G}_{\text{det}}  \leftarrow [\hat{G}_{\text{det}} \ \hat{G}_{j,i}]$
  \ei
  End for
\item set $\Phat_j \leftarrow [\hat{G}_{j,1},\cdots, \hat{G}_{j,\vartheta_j}]$ and set $\Phat_{(t)} \leftarrow \Phat_{j}$.
\item increment $ j \leftarrow j+1$.
\een

\een
%\item Increment $t \leftarrow t + 1$ and go to step 1.
\een
\end{algorithm*}

Steps 1, 2, 3a and 3b of ReProCS-cPCA are the same as Algorithm \ref{reprocs}.  As shown earlier, within $K$ proj-PCA updates ($K$ chosen as given in Theorem \ref{thm2}) $\|e_t\|_2$ and the subspace error,  $\SE_{(t)}$, drop down to a constant times ${\zeta}$. In particular, if at $t=t_{j}-1$,  $\SE_{(t)} \le r \zeta$, then at $t= \tilde{t}_j:=t_j + K \alpha$, we can show that $\SE_{(t)} \le (r + c_{\max}) \zeta$. Here $r:=r_{\max} = r_0 + c_{\text{dif}}$.
To bring $\SE_{(t)}$ down to $r \zeta$ before $t_{j+1}$, we proceed as follows. %need a step to estimate only $\Span(P_j)$, i.e. we have ``deleted" $\Span(P_{j,\old})$. We do this at as follows.
%
%One simple way to do this is by standard PCA: at $t=\tilde{t}_j + \tilde{\alpha}-1$, compute $\Phat_j \leftarrow \text{proj-PCA}([\Lhat_t;  \tilde{\mathcal{I}}_{j,1}], [.], r_j)$ and let $\Phat_{(t)} \leftarrow \Phat_j$. Using the $\sin \theta$ theorem and the Hoeffding corollaries, it can be shown that, as long as $f$ is small enough, doing this is guaranteed to give an accurate estimate of $\Span(P_j)$. However $f$ being small is not compatible with the slow subspace change assumption. Notice from Sec \ref{probdef} that $\lambda^- \le \gamma_\new$ and $\E[||L_t||_2^2] \le r \lambda^+$. Slow subspace change implies that $\gamma_\new$ is small. Thus, $\lambda^-$ is small. However, to allow $L_t$ to have large magnitude, $\lambda^+$ needs to be large. Thus,  $f = \lambda^+ / \lambda^-$ cannot be small unless we require that $L_t$ has small magnitude for all times $t$. Here we introduce a generalization of the strategy from Algorithm \ref{reprocs} called cluster-PCA, that removes the bound on $f$, but instead requires that the eigenvalues of $\text{Cov}(L_t)$ be sufficiently clustered as explained in Sec \ref{eigencluster}.
%
%%The first iteration uses $P \leftarrow [.]$, i.e. it computes standard PCA to estimate the first cluster, $\Span(G_{j,1})$.  or $\Phat_j$%$[\Lhat_t; t \in \tilde{I}_{j,k}]$
The main idea is to recover one cluster of entries of $P_j$ at a time. For each batch we use a new set of $\tilde\alpha$ frames. The entire procedure is done at $t=\tilde{t}_j + \vartheta_j \tilde\alpha -1$ (since we cannot update $\Phat_{(t)}$ until all clusters are recovered). We proceed as follows. In the first iteration, we use standard PCA to estimate the first cluster, $\Span(G_{j,1})$. In the $k^{th}$ iteration, we apply proj-PCA on $[\hat{L}_{\tilde{t}_j+(k-1)\tilde\alpha}, \dots , \hat{L}_{\tilde{t}_j+k\tilde\alpha -1}]$ with $P \leftarrow [\hat{G}_{j,1}, \hat{G}_{j,2}, \dots \hat{G}_{j,k-1}]$ to estimate $\Span(G_{j,k})$. By modifying the approach used to prove Theorem \ref{thm1}, we can show that since $\tilde{g}_{j,k}$ and $\tilde{h}_{j,k}$ are small enough, $\Span(G_{j,k})$ will be accurately recovered, i.e. $\|(I - \sum_{i=1}^{k} \hat{G}_{j,i} \hat{G}_{j,i}')G_{j,k}\|_2 \le \tilde{c}_{j,k} \zeta$. We do this $\vartheta_j$ times and finally we set $\Phat_j \leftarrow [\hat{G}_{j,1}, \hat{G}_{j,2} \dots \hat{G}_{j,\vartheta_j}]$ and $\Phat_{(t)} \leftarrow \Phat_j$.
Thus, at $t=\tilde{t}_j + \vartheta_j \tilde{\alpha}-1$, $\SE_{(t)} \le \sum_{k=1}^{\vartheta_j}  \|(I - \sum_{i=1}^{k} \hat{G}_{j,i} \hat{G}_{j,i}') G_{j,k} \|_2 \le  \sum_{k=1}^{\vartheta_j} \tilde{c}_{j,k} \zeta =  r_j \zeta \le r \zeta$. Under the assumption that $t_{j+1} - t_j \ge K \alpha + \vartheta_{\max} \tilde{\alpha}$, this means that before the next subspace change time, $t_{j+1}$, $\SE_{(t)}$ is below $r \zeta$.

\begin{figure*}[t!]
\centerline{
\epsfig{file = 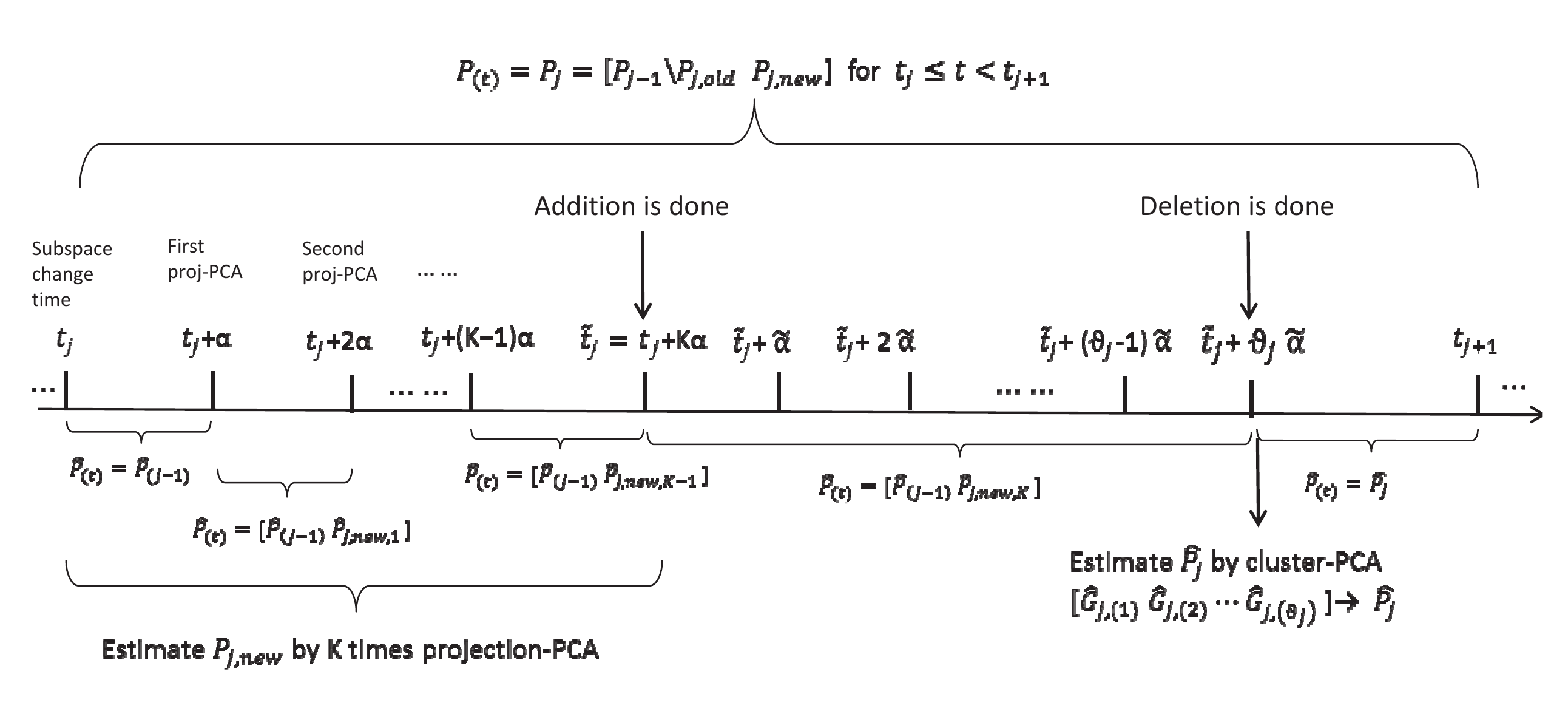, width =16cm, height = 5cm}
}
\caption{\small{A diagram illustrating subspace estimation by ReProCS-cPCA
%Subspace addition and deletion by projection PCA (proj-PCA).
}}
\label{add_del_proj_pca_diag2}
\end{figure*}

We illustrate the ideas of subspace estimation by addition proj-PCA and cluster-PCA in Fig. \ref{add_del_proj_pca_diag2}. The connection between proj-PCA done in the addition step and for the cluster-PCA (in deletion) step is given in Table \ref{tab_diff}.

\subsection{Performance Guarantees} \label{perf_g_del}

%We state the result first and then discuss it in the next subsection. We give its corollary for the case where $f$ is small in Sec \ref{f_small_sec}. The proof is given in Sec \ref{thmproof}.% proof outline is given in Sec \ref{detailed} and the

%\subsection{Cluster PCA Result}

\begin{definition} \label{def alpha del}
We need the following definitions for stating the main result.
\ben
\item
We define $\alpha_{\del}(\zeta)$ as %the smallest value of $\alpha$ so that $\tilde{p}(\tilde{\alpha}, \zeta)^{\vartheta_{\max} J} \ge 1-  n^{-10}$ where $\tilde{p}(\tilde{\alpha},\zeta)$ is defined in Lemma \ref{tilde_zeta}.  We can compute an explicit value for it by using the fact that for any $x \le 1$ and $r \ge 1$, $(1-x)^r \ge 1-rx$ and that $\sum_{i=1}^6 e^{-\frac{\alpha}{d_i^2}} \le 6 e^{-\frac{\alpha}{\max_{i=1,2\dots 6} d_i^2}}$. We get
\begin{multline*}
\alpha_\del(\zeta) : = \left\lceil (\log 6 \vartheta_{\max} J + 11 \log n)\cdot \right.\\
\left.\frac{8 \cdot 10^2}{( \zeta \lambda^-)^2} \max( 4.2^2, 4 b_7^2 ) \right\rceil \nn % \max(r^2 \gamma_*^4, 4.2^2, b_1^2, b_3^2, 4 b_7^2) \rceil \nn
\end{multline*}
where $b_7 := (\sqrt{r} \gamma_* + \phi^+ \sqrt{\zeta})^2$ and $\phi^+=1.1732$.
%where $b_1 = {\phi^+}^2 \zeta$, $b_3 := \sqrt{r\zeta} \phi^+ \gamma_*$, $b_7 := (\sqrt{r} \gamma_* + \phi^+ \sqrt{\zeta})^2$ and $\phi^+=1.1732$.
We choose $\alpha_{\text{del}}$ so that if , $\tilde\alpha \geq \alpha_{\text{del}}$, then the conclusions of the theorem will hold wth probability at least $(1 -   2n^{-10})$.
\item Define
%$$f_{inc}(\tilde{g},\tilde{h},\kappa_{s,e}^+): =  (r+c) \zeta\Big[ 3\kappa_{s,e}^+  \phi^+ \tilde{g} +  [\kappa_{s,e}^+ \phi^+ + \kappa_{s,e}^+ (1+ 2\phi^+)\frac{r^2\zeta^2}{\sqrt{1-r^2\zeta^2}}] \tilde{h} +  [\frac{r^2}{r+c}\zeta + 4r \zeta \kappa_{s,e}^+ \phi^+ + 2(r+c) \zeta(1+ {\kappa_{s,e}^+}^2) {\phi^+}^2] f+ 0.2 \frac{1}{r+c}\Big],$$
%
\begin{align*}
&f_{inc}(\tilde{g},\tilde{h},\kappa_{s,e}^+,\kappa_{s,D}^+) : =  \\
&(r+c) \zeta \Bigg[ \max( 3\kappa_{s,e}^+ \kappa_{s,D}^+  \phi^+ \tilde{g}, \kappa_{s,e}^+ \phi^+ \tilde{h})  \\
& +  \big[\kappa_{s,e}^+ \phi^+ + \kappa_{s,e}^+ (1+ 2\phi^+)\frac{r^2\zeta^2}{\sqrt{1-r^2\zeta^2}} \big] \tilde{h} \\
& + \big [\frac{r^2}{r+c}\zeta + 4r \zeta \kappa_{s,e}^+ \phi^+ + 2(r+c) \zeta(1+ {\kappa_{s,e}^+}^2) {\phi^+}^2\big] f \\
& \hspace{2in}+ 0.2 \frac{1}{r+c} \Bigg],
\end{align*}
\begin{align*}
f_{dec}(\tilde{g},\tilde{h},\kappa_{s,e}^+,\kappa_{s,D}^+) & := 1- \tilde{h} - 0.2 \zeta - r^2 \zeta^2 f - r^2 \zeta^2 \\
&\hspace{.5in} -  f_{inc}(\tilde{g},\tilde{h},\kappa_{s,e}^+,\kappa_{s,D}^+)
\end{align*}
%where is the term $0.2/(r+c)$ coming from ???
Notice that $f_{inc}(.)$ is an increasing function of $\tilde{g},\tilde{h}$ and $f_{dec}(.)$ is a decreasing function of $\tilde{g},\tilde{h}$.
\een
\end{definition}

\begin{theorem} \label{thm2}
Consider Algorithm \ref{ReProCS_del}.  Let $c:= c_{\max}$ and $r:= r_{\max} = r_0 + c_{\text{dif}}$.
Pick a $\zeta$ that satisfies
\[
\zeta  \leq  \min\left(\frac{10^{-4}}{r^2},\frac{1.5 \times 10^{-4}}{r^2 f},\frac{1}{r^{3}\gamma_*^2}\right)
%\ \text{where} \ f := \frac{\lambda^+}{\lambda^-}
\]
Assume that the initial subspace estimate is accurate enough, i.e. $\|(I - \Phat_0 \Phat_0') P_0\| \le r_0 \zeta$.
%
%Assume that $L_t$ obeys the model given in Assumption \ref{modelassum} and Assumption \ref{del model}. Also, assume that the initial subspace estimate is accurate enough, i.e. $\|(I - \Phat_0 \Phat_0') P_0\| \le r_0 \zeta$, for a $\zeta$ that satisfies
%\[
%\zeta  \leq  \min\left(\frac{10^{-4}}{(r+c)^2},\frac{1.5 \times 10^{-4}}{(r+c)^2 f},\frac{1}{(r+c)^{3}\gamma_*^2}\right) \ \text{where} \ f := \frac{\lambda^+}{\lambda^-}
%\]
%Let $\xi_0(\zeta), K(\zeta), \alpha_{\add}(\zeta), \alpha_{\del}(\zeta)$ be as defined in Definition \ref{defn_alpha} and Definition \ref{def alpha del}.
If the following conditions hold:
\ben
\item All of the conditions of Theorem \ref{thm1} hold with $L_t$ satisfying Signal model \ref{del model},
\item $\tilde{\alpha} \ge \alpha_{\del}(\zeta)$,
\item  $\min_{j} (t_{j+1} -t_j) > K  \alpha + \vartheta_{\max} \tilde{\alpha}$

\item algorithm estimates $\Phat_{j-1}$ and $\Phat_{j,\new,K}$ satisfy
\[
\max_j \kappa_s ((I-\hat{P}_{j-1} {\hat{P}_{j-1}}' - \hat{P}_{j,\new,K} {\hat{P}_{j,\new,K}}')P_j) \leq \kappa_{s,e}^+
\]

\item {\em (clustered eigenvalues) } Assumption \ref{clusterass} holds with $\tilde{g}_{\max},\tilde{h}_{\max}, \tilde{c}_{\min}$ satisfying $f_{dec}(\tilde{g}_{\max},\tilde{h}_{\max}, \kappa_{s,e}^+,\kappa_{s,*}^+ + r \zeta) - \frac{f_{inc}(\tilde{g}_{\max},\tilde{h}_{\max}, \kappa_{s,e}^+,\kappa_{s,*}^+ + r \zeta)}{\tilde{c}_{\min} \zeta} > 0$. %(also see Remark \ref{f_inc_rem} which weakens this requirement).
     %$f_{dec}(\tilde{g}_{\max},\tilde{h}_{\max})$ and $f_{inc}(\tilde{g}_{\max},\tilde{h}_{\max})$ are defined in Definition \ref{def alpha del} (also see Remark \ref{f_inc_rem} which weakens this requirement),

  %  $F(\tilde{g}_{\max},\tilde{h}_{\max}) > 0$ where $F(.)$ is defined in Lemma \ref{bnd_tzetakp}.
\een

then, with probability at least $1 -  2 n^{-10}$, at all times, $t$, %all of the following hold:
\ben
\item %at all times, $t$,
$
\That_t = T_t \  \text{and} \  \|e_t\|_2 = \|L_t - \hat{L}_t\|_2 = \|\hat{S}_t - S_t\|_2 \le 0.18 \sqrt{c} \gamma_{\new} + 1.24 \sqrt{\zeta}.
$ %(\sqrt{r} + 0.06 \sqrt{c}) \nn

\item the subspace error, $\SE_{(t)}$ satisfies
\bea
&&\SE_{(t)} \leq \nn  \\
&&\left\{  \begin{array}{ll}
0.6^{k-1} + r \zeta + 0.4 c \zeta  &  \ \text{if}  \    t \in \mathcal{I}_{j,k}, \ k=1,2,\cdots,K\\  %t_j + K\alpha
(r+c) \zeta                        & \  \text{if} \   t \in \tilde{\mathcal{I}}_{j,k}, \ k=1,2,\cdots,\vartheta_j  \\
r \zeta                            & \  \text{if} \  t \in \tilde{\mathcal{I}}_{j,\vartheta_j+1}
\end{array} \right. \nn
%\\
%%
% & \le  &   \left\{  \begin{array}{ll}
%0.6^{k-1} + 10^{-2} \sqrt{\zeta}  & \  \text{if}  \  t \in \mathcal{I}_{j,k}, \ k=1,2,\cdots,K \nn \\
%10^{-2} \sqrt{\zeta}   & \  \text{if}  \  t \in (\cup_{k=1}^{\vartheta_j} \tilde{\mathcal{I}}_{j,k}) \cup \tilde{\mathcal{I}}_{j,\vartheta_j+1}  %t_j + K \alpha \le t_{j+1} -1 \nn
%\end{array} \right.
\eea

%\item $e_t$ follows a trend similar to that of $\SE_{(t)}$ at various times (the bounds are available in \cite[Theorem 3.2]{rrpcp_del_perf}.
\item the error $e_t = \hat{S}_t - S_t = L_t - \hat{L}_t$ satisfies the following at various times
\[
\|e_t\|_2   \le  
\begin{cases}
 1.17 [ 0.15  \cdot 0.72^{k-1} \sqrt{c}\gamma_{\new} + \\
 \hspace{.5in}  0.15 \cdot 0.4 c \zeta \sqrt{c} \gamma_* + r \zeta \sqrt{r} \gamma_*]  \\
 \hspace{.5in}   \text{if}  \   t \in \mathcal{I}_{j,k}, \ k=1,2,\cdots,K \\
1.17(r+c) \zeta \sqrt{r} \gamma_*   \\
 \hspace{.5in} \text{if} \ \ t \in \tilde{\mathcal{I}}_{j,k}, \ k=1,2,\cdots,\vartheta_j \\
 1.17 r\zeta \sqrt{r} \gamma_*   \ \ \text{if} \ \ t \in \tilde{\mathcal{I}}_{j,\vartheta_j+1}
\end{cases} \nn
\]

\een
\end{theorem}
%
%move to correct place.
%Define
%$\kappa_{s,D}^+ : = \kappa_s( (I - \sum_{i=0}^{k-1} \hat{G}_{j,i} \hat{G}_{j,i}' ) G_{j,k})$. Conditioned on $\tilde\Gamma_{j,k-1}$, it is easy to see that
%$\kappa_{s,D}^+ \le \kappa_s(G_{j,k}) + r \zeta \le \kappa_s(P_j) + r \zeta \le \kappa_{s,*}^+ + r \zeta$

\subsection{Special Case when $f$ is small} \label{f_small_sec} %, which is the maximum condition number of $\text{Cov}(L_t)$ for any $t$,
If in a problem, $L_t$ has small magnitude for all times $t$ or if its subspace does not change, then $f$ can be small. In this case, the clustering assumption is not needed, or in fact it trivially holds with $\vartheta_j=1$, $\tilde{c}_{j,1} = r_j$, $\tilde{g}_{\max} =\tilde{g}_{j,1} = f$ and $\tilde{h}_{\max}={h}_{j,1} = 0$. Thus, $\vartheta_{\max} =1$. With this, the following corollary holds.
\begin{corollary} \label{f_small}
%Assume that the initial subspace estimate is accurate enough as given in Theorem \ref{thm2} with $\zeta$ as chosen there. Also assume that the first five conditions of Theorem \ref{thm2} hold.
Assume that all conditions of Theorem \ref{thm2} hold except the last one (clustering assumption).  If $f$ is small enough so that $f_{inc}(f,0,\kappa_{s,e}^+,\kappa_{s,*}^+ + r\zeta) \le f_{dec}(f,0,\kappa_{s,e}^+,\kappa_{s,*}^+ + r\zeta) r_j \zeta$, then, all conclusions of Theorem \ref{thm2} hold.
\end{corollary}
%Notice that the above corollary does not need the clustering assumption, Assumption \ref{clusterass}, to hold.

%The above result says the following. Assume that the initial subspace error is small enough. If the assumptions given in the theorem hold, then, w.h.p., we will get exact support recovery at all times. Moreover, the sparse recovery error (and the error in recovering $L_t$) will always be bounded by $0.18\sqrt{c} \gamma_\new$ plus a constant times $\sqrt{\zeta}$. Since $\zeta$ is very small, $\gamma_\new \ll S_{\min}$, and $c$ is also small, the normalized reconstruction error for $S_t$ will be small at all times, thus making this a meaningful result. In the second conclusion, we bound the subspace estimation error, $\SE_{(t)}$. When a subspace change occurs, this error is initially bounded by one. The above result shows that, w.h.p., with each adddition proj-PCA step, this error decays roughly exponentially and falls below $(r+c)\zeta$ within $K$ steps. After the cluster-PCA step, this error falls below $r\zeta$. By assumption, this occurs before the next subspace change time. Because of the choice of $\zeta$, both $(r+c)\zeta$ and $r \zeta$ are below $0.01 \sqrt{\zeta}$. The third conclusion shows that the sparse recovery error as well as the error in recovering $L_t$ decay in a similar fashion. % (which has $\vartheta_j$ iterations)

\begin{table*}
\caption{Comparing and contrasting the addition proj-PCA step and proj-PCA used in the deletion step (cluster-PCA)}
\begin{center}
%\hspace{-8mm}
\begin{tabular}{|l||l|}
  \hline  %estimating $P_{j,\new}$ by proj-PCA estimating $P_j$
 {\bf $k^\text{th}$ iteration of addition proj-PCA} & {\bf $k^\text{th}$ iteration of cluster-PCA in the deletion step}  \\ \hline
  done at $t= t_j+k \alpha-1$              &  done at $t=t_j + K \alpha + \vartheta_j \tilde\alpha-1$ \\ \hline
  goal: keep improving estimates of $\Span(P_{j,\new})$ &   goal: re-estimate $\Span(P_{j})$ and thus ``delete" $\Span(P_{j,\old})$ \\ \hline
  compute $\Phat_{j,\new,k}$ by proj-PCA on $[\hat{L}_t: t\in \mathcal{I}_{j,k}]$    &  compute $\hat{G}_{j,k}$ by proj-PCA on $[\hat{L}_t: t\in \tilde{\mathcal{I}}_{j,k}]$  \\
    with $P = \hat{P}_{j-1}$                                                         & with $P = \hat{G}_{j,\text{det},k} = [\hat{G}_{j,1}, \cdots, \hat{G}_{j,k-1}]$   \\ \hline
   start with $\|(I - \Phat_{j-1} {\Phat_{j-1}}')P_{j-1}\|_2 \leq r\zeta$  and $\zeta_{j,k-1} \leq \zeta_{k-1}^+ \le 0.6^{k-1} + 0.4 c \zeta $       & start with $\|(I - \hat{G}_{j,\text{det},k}{\hat{G}_{j,\text{det},k}}')G_{j,\text{det},k}\|_2 \leq r\zeta$ and $\zeta_{j,K} \leq c \zeta$     \\ \hline
%                    &   \\\hline
%
   need small $g$ which is the                                                   & need small $\tilde{g}_{\max}$ which is the \\
maximum condition number of $\text{Cov}(P_{j,\new}'L_t)$ &  maximum of the maximum condition number of $\text{Cov}(G_{j,k}'L_t)$ \\ \hline
% over $t\in \tilde{\mathcal{I}}_{j,k}$
no undetected subspace & extra issue: ensure perturbation due to $\Span(G_{j,\text{undet},k})$ is small; \\
%where $G_{j,undet,k}:=[G_{j,k+1},\cdots, G_{j,\vartheta_j}]$,
                      & need small $\tilde{h}_{j,k}$ to ensure the above   \\ \hline
  $\zeta_{j,k}$ is the subspace error in estimating $\text{span}(P_{j,\new})$ after the $k^{th}$ step & $\tilde{\zeta}_{j,k}$ is the subspace error in estimating $\text{span}(G_{j,k})$  after the $k^{th}$ step  \\ \hline
%  by the $k^{\text{th}}$ proj-PCA step  &  by the $k^{\text{th}}$ iteration of the cPCA step \\ \hline
%
end with $\zeta_{j,k} \leq  \zeta_k^+ \leq 0.6^k + 0.4 c\zeta$ w.h.p. & end with $\tilde{\zeta}_{j,k} \leq \tilde{c}_{j,k} \zeta$ w.h.p. \\   \hline
  stop when $k=K$ with $K$ chosen so that $\zeta_{j,K} \leq  c\zeta$  & stop when $k = \vartheta_j$ and $\tilde{\zeta}_{j,k} \leq \tilde{c}_{j,k}\zeta$ for all $k=1,2,\cdots,\vartheta_j$ \\ \hline
  after $K^{th}$ iteration: $\Phat_{(t)} \leftarrow [\Phat_{j-1} \ \Phat_{j,\new,K}]$ and $SE_{(t)} \leq (r+c)\zeta$ & after $\vartheta_j^{th}$ iteration: $\Phat_{(t)} \leftarrow  [\hat{G}_{j,1},\cdots, \hat{G}_{j,\vartheta_j}]$ and $SE_{(t)} \leq r\zeta$  \\ \hline
\end{tabular}
\end{center}
 \label{tab_diff}
\end{table*}

\subsection{Discussion} \label{discuss_del}
Notice from Definition \ref{defn_alpha} that $K = K(\zeta)$ is larger if $\zeta$ is smaller. Also, both $\alpha_\add(\zeta)$ and $\alpha_\del(\zeta)$ are inversely proportional to $\zeta$. Thus, if we want to achieve a smaller lowest error level, $\zeta$, we need to compute both addition proj-PCA and cluster-PCA's over larger durations, $\alpha$ and $\tilde\alpha$ respectively, and we will need more number of addition proj-PCA steps $K$. This means that we also require a larger delay between subspace change times, i.e. larger $t_{j+1}-t_j$.

Let us first compare the above result with that for ReProCS for the same subspace change model, i.e. the result from Corollary \ref{cor_rep}.  The most important difference is that ReProCS requires $\kappa_{2s}([P_0, P_{1,\new}, \dots P_{J,\new}]) \le 0.3$ whereas ReProCS-cPCA only requires $\max_j \kappa_{2s}(P_j) \le 0.3$. Moreover in case of ReProCS, the denominator in the bound on $\zeta$ also depends on $J$ whereas in case of ReProCS-cPCA, it only depends on $r_{max} + c_{\max}$.
%Moreover, ReProCS requires $\zeta$ to satisfy $\zeta  \leq  \min\left(\frac{10^{-4}}{(r_0+(J-1)c)^2},\frac{1.5 \times 10^{-4}}{(r_0+(J-1)c)^2 f},\frac{1}{(r_0+(J-1)c)^{3}\gamma_*^2}\right)$ whereas in case of ReProCS-cPCA the denominators in the bound on $\zeta$ only contain $r + c = r_0 + 2c$ (instead of $r_0+(J-1)c$).
Because of this, in Theorem \ref{thm2} for ReProCS-cPCA, the only place where $J$ appears is in the definitions of $\alpha_\add$ and $\alpha_\del$. These govern the delay  between subspace change times, $t_{j+1}-t_j$. Thus, with ReProCS-cPCA, $J$ can keep increasing, as long as $\min_j (t_{j+1}-t_j)$ also increases accordingly. Moreover, notice that the dependence of $\alpha_\add$ and $\alpha_\del$ on $J$ is only logarithmic and thus $\min_j (t_{j+1}-t_j)$ needs to only increase in proportion to $\log J$.
%
%On the other hand, for ReProCS (see \cite[Corollary 43]{rrpcp_perf}), $J$ appears in the denseness assumption, in the bound on $\zeta$ and in the definition of $\alpha_\add$. Thus, ReProCS needs a bound on $J$ that is indirectly imposed by the denseness assumption. %and does not depend on how large $t_{j+1}-t_j$ is.
%
The main extra assumptions that ReProCS-cPCA needs are the clustering assumption;  a longer delay between subspace change times; and  a denseness assumption similar to that on $D_{j,\new,k}$. We verify the clustering assumption in Sec \ref{model_verify}. The ReProCS-cPCA algorithm also needs to know the cluster sizes of the eigenvalues. These can, however, be estimated by computing the eigenvalues of the estimated covariance matrix at $t= \tilde{t}_j + \tilde\alpha$ and clustering them.

{\em Comparison with the PCP result from \cite{rpca}. }
Our results need many more assumptions compared with the PCP result \cite{rpca} which only assumes independent support change of the sparse part and a denseness assumption on the low-rank part. The most important limitation of our work is that both our results need an assumption on the algorithm estimates, thus neither can be called a correctness result. Moreover, both the results assume that the algorithms know the model parameters while the result for PCP does not. The key limiting aspect here is the knowledge of the subspace change times. The advantages of our results w.r.t. that for PCP are as follows.
(a) Both results are for online algorithms; and (b) both need weaker denseness assumptions on the singular vectors of ${\cal L}_t$ as compared to PCP.  PCP \cite{rpca} requires denseness of both the left and right singular vectors of ${\cal L}_t$ and it requires a bound on $\|UV'\|_{\infty}$ where $U$ and $V$ denote the left and right singular vectors. Denseness of only the left singular vectors is needed in our case (notice that $U=[P_{j-1}, P_{j,\new}]$). (c) Finally, the most important advantage of the ReProCS-cPCA result is that it does not need a bound on $J$ (number of subspace change times) as long as $\min_j (t_{j+1}-t_j)$ increases in proportion to $\log J$, and equivalently, does not need a bound on the rank of ${\cal L}_t$. However PCP needs a tight bound on the rank of  ${\cal L}_t$.

\section{Proof of Theorem \ref{thm2}} \label{thmproof}
%In this section, we remove the subscript $j$ at most places. The convention of Remark \ref{remove_j} applies.
We first give some new definitions next. We then give the key lemmas leading to the proof of the theorem and the proof itself. Finally we prove these lemmas.

%The proof is an easy consequence of the two lemmas stated below.

%\section{Definitions, Proof Outline and Connection between addition and deletion steps} \label{detailed}
%In Sec \ref{defs}, we define all the quantities that are needed for the proof. The proof outline is given in Sec \ref{outline}. We discuss how the proof strategy for the cluster-PCA (for deletion) step is related to that of addition proj-PCA in Sec \ref{connect}.

\subsection{Some New Definitions} \label{defs}
Unless redefined here, all previous definitions still apply.

\begin{definition}
\label{def_zeta}
Define the following:
\begin{enumerate}
\item $r = r_{\max} = r_0 + c_{\text{dif}}$ (Note that this is a redefinition from Definition \ref{kappaplus})
\item $\zeta_{j,*}^+ := r \zeta$ (Note that this is a redefinition from Definition \ref{zetakplus})

\item define the sequence $\{{\tilde\zeta_{k}}^+\}_{k=1,2,\cdots,\vartheta_j}$ as follows
\bea
{\tilde\zeta_{k}}^+: = \frac{f_{inc}(\tilde{g}_k,\tilde{h}_k, \kappa_{s,e}^+,\kappa_{s,*}^+ + r \zeta)}{f_{dec}(\tilde{g}_k,\tilde{h}_k, \kappa_{s,e}^+,\kappa_{s,*}^+ + r \zeta)} \nn
\eea
where $f_{inc}(.)$ and $f_{dec}(.)$ are defined in Definition \ref{def alpha del}.
%$f_{inc}(\tilde{g},\tilde{h}): =  (r+c) \zeta\Big[ 3\kappa_{s,e}^+  \phi^+ \tilde{g} +  [\kappa_{s,e}^+ \phi^+ + \kappa_{s,e}^+ (1+ 2\phi^+)\frac{r^2\zeta^2}{\sqrt{1-r^2\zeta^2}}] \tilde{h} +  [\frac{r^2}{r+c}\zeta + 4r \zeta \kappa_{s,e}^+ \phi^+ + 2(r+c) \zeta(1+ {\kappa_{s,e}^+}^2) {\phi^+}^2] f+ 0.2 \frac{1}{r+c}\Big],$
%and \label{defzetap}
%$f_{dec}(\tilde{g},\tilde{h}):= 1- \tilde{h} - 0.2 \zeta - r^2 \zeta^2 f - r^2 \zeta^2 -  f_{inc}(\tilde{g},\tilde{h})$. Notice that $f_{inc}(\tilde{g},\tilde{h})$ is an increasing function of $\tilde{g},\tilde{h}$ and $f_{dec}(\tilde{g},\tilde{h})$ is a decreasing function of $\tilde{g},\tilde{h}$.

%Define $\kappa_{s,D}^+ : = \kappa_s( (I - \sum_{i=0}^{k-1} \hat{G}_{j,i} \hat{G}_{j,i}' ) G_{j,k})$.

\end{enumerate}
\end{definition}

\begin{definition}
Define
\ben
\item $\Psi_{j,k} : =  I - \sum_{i=0}^{k} \hat{G}_{j,i} \hat{G}_{j,i}'$. %Notice that $\Psi_{j,k} = I - \hat{G}_{j,k+1,d} {\hat{G}_{j,k+1,d}}'$.

\item $G_{j,\text{det},k} := [G_{j,1} \cdots, G_{j,k-1}]$ and $\hat{G}_{j,\text{det},k} := [\hat{G}_{j,1} \cdots, \hat{G}_{j,k-1}]$. Notice that $\Psi_{j,k} =  I - \hat{G}_{j,\text{det},k+1}\hat{G}_{j,\text{det},k+1}'$.
\item $G_{j,\text{undet},k} := [G_{j,k+1} \cdots, G_{j,\vartheta_j}]$.

\item $D_{j,k} := \Psi_{j,k-1} G_{j,k}$,  $D_{j,\text{det},k} :=  \Psi_{j,k-1} G_{j,\text{det},k}$ and $D_{j,\text{undet},k} := \Psi_{j,k-1}G_{j,\text{undet},k} $. %$\Span(D_{j,\text{undet},k})$ is the unestimated part of  $\Span(G_{j,\text{undet},k})$ after the $k^{th}$ proj-PCA step.

\een
\end{definition}

\begin{definition}\label{defHk_del} \
\ben
\item Let $D_{j,k} \overset{QR}{=} E_{j,k} R_{j,k}$ denote its reduced QR decomposition, i.e. let $E_{j,k}$ be a basis matrix for $\Span(D_{j,k})$ and let $R_{j,k}:=E_{j,k}'D_{j,k}$.
%Thus $R_{j,k}$ is upper triangular and invertible \big( conditioned on $\Gamma_{j,0}$ (see def \ref{Gamma_def}) by Lemma \ref{hatswitch} since $\sigma_i(D_{j,\new}) = \sigma_i(R_{j,\new})$ \big) and $E_{j,k} $ is a basis matrix for $\Span(D_{j,k})$.
%Here, $E_{j,k}$ is a basis matrix while $R_{j,k}$ is upper triangular. \footnote{Notice that $0< \sqrt{1-r^2 \zeta^2} \leq \sigma_i(R_{j,k})$ by Lemma \ref{bound_R}, therefore, $R_{j,k}$ is invertible.}

\item Let $E_{j,k,\perp}$ be a basis matrix for the orthogonal complement of $\Span(E_{j,k}) = \Span(D_{j,k})$. To be precise, $E_{j,k,\perp}$ is a $n \times (n-\tilde{c}_{j,k})$ basis matrix that satisfies ${E_{j,k,\perp}}' E_{j,k} = 0$.

\item Using $E_{j,k}$ and $E_{j,k,\perp}$, define $\tilde{A}_{j,k}$, $\tilde{A}_{j,k,\perp}$, $\tilde{H}_{j,k}$, $\tilde{H}_{j,k,\perp}$ and $\tilde{B}_{j,k}$ as
    \bea
    \tilde{A}_{j,k} &:=& \frac{1}{\tilde{\alpha}} \sum_{t \in \tilde{I}_{j,k}} {E_{j,k}}' \Psi_{j,k-1} L_t{L_t}' \Psi_{j,k-1} E_{j,k} \nn \\
    \tilde{A}_{j,k,\perp} &:=& \frac{1}{\tilde{\alpha}} \sum_{t \in \tilde{I}_{j,k}} {E_{j,k,\perp}}' \Psi_{j,k-1} L_t{L_t}' \Psi_{j,k-1} E_{j,k,\perp} \nn \\
    \tilde{H}_{j,k} &:=& \frac{1}{\tilde{\alpha}} \sum_{t \in \tilde{I}_{j,k}} {E_{j,k}}'\Psi_{j,k-1} (e_t{e_t}' - L_t{e_t}' - e_t {L_t}') \Psi_{j,k-1} E_{j,k} \nn \\
    \tilde{H}_{j,k,\perp} &:=& \frac{1}{\tilde{\alpha}} \sum_{t \in \tilde{I}_{j,k}} {E_{j,k,\perp}}' \Psi_{j,k-1} \nn \\
&& \hspace{1in} (e_t{e_t}' - L_t{e_t}' - e_t {L_t}') \Psi_{j,k-1} E_{j,k,\perp} \nn \\
    \tilde{B}_{j,k} &:=& \frac{1}{\tilde{\alpha}} \sum_{t \in \tilde{I}_{j,k}} {E_{j,k,\perp}}' \Psi_{j,k-1} \hat{L}_t{\hat{L}_t}' \Psi_{j,k-1} E_{j,k} \nn \\
&=& \frac{1}{\tilde{\alpha}} \sum_{t \in \tilde{I}_{j,k}} {E_{j,k,\perp}}'\Psi_{j,k-1} (L_t-e_t)({L_t}' - {e_t}') \Psi_{j,k-1} E_{j,k} \nn
    \eea
\item Define
\bea
&&\tilde{\mathcal{A}}_{j,k} := \left[ \begin{array}{cc} E_{j,k} & E_{j,k,\perp} \\ \end{array} \right]
\left[\begin{array}{cc}\tilde{A}_{j,k} \ & 0 \ \\ 0 \ & \tilde{A}_{j,k,\perp}  \\ \end{array} \right]
\left[ \begin{array}{c} {E_{j,k}}' \\ {E_{j,k,\perp}}' \\ \end{array} \right]\nn\\
&&\tilde{\mathcal{H}}_{j,k} := \left[ \begin{array}{cc} E_{j,k} & E_{j,k,\perp} \\ \end{array} \right]
\left[\begin{array}{cc} \tilde{H}_{j,k} \ & {\tilde{B}_{j,k}}' \ \\ \tilde{B}_{j,k} \ &  \tilde{H}_{j,k,\perp} \\ \end{array} \right]
\left[ \begin{array}{c} {E_{j,k}}' \\ {E_{j,k,\perp}}' \\ \end{array} \right]\nn
\label{defn_tilde_Hk}
\eea

\item From the above, it is easy to see that $$\tilde{\mathcal{A}}_{j,k} + \tilde{\mathcal{H}}_{j,k} =\frac{1}{\tilde{\alpha}} \sum_{t \in \tilde{\mathcal{I}}_{j,k}} \Psi_{j,k-1} \hat{L}_t {\hat{L}_t}' \Psi_{j,k-1}.$$

\item Recall from Algorithm \ref{ReProCS_del} that
\begin{align*}
\tilde{\mathcal{A}}_{j,k} +& \tilde{\mathcal{H}}_{j,k} = \frac{1}{\tilde{\alpha}} \sum_{t \in \tilde{\mathcal{I}}_{j,k}} \Psi_{j,k-1} \hat{L}_t {\hat{L}_t}' \Psi_{j,k-1} \\
&\overset{EVD}{=} \left[ \begin{array}{cc} \hat{G}_{j,k} & \hat{G}_{j,k,\perp} \\ \end{array} \right]
\left[\begin{array}{cc} \Lambda_{j,k} \ & 0 \ \\ 0 \ & \ \Lambda_{j,k,\perp} \\ \end{array} \right]
\left[ \begin{array}{c} \hat{G}_{j,k}' \\ \hat{G}_{j,k,\perp}' \\ \end{array} \right]
\end{align*}
is the EVD of $\tilde{\mathcal{A}}_{j,k} + \tilde{\mathcal{H}}_{j,k}$. %Here $\hat{G}_{j,k}$ is a $n \times c_{j,k}$ basis matrix.%and
Here $\Lambda_k$ is a $\tilde{c}_{j,k} \times \tilde{c}_{j,k}$ diagonal matrix.

%\item Using the above, $\tilde{\mathcal{A}}_{j,k} + \tilde{\mathcal{H}}_{j,k}$ can be decomposed in two ways as follows:
%    $$\tilde{\mathcal{A}}_{j,k} + \tilde{\mathcal{H}}_{j,k}
%= \left[ \begin{array}{cc} \hat{G}_{j,k} & \hat{G}_{j,k,\perp} \\ \end{array} \right]
%\left[\begin{array}{cc} \Lambda_{j,k} \ & 0 \ \\ 0 \ & \ \Lambda_{j,k,\perp} \\ \end{array} \right] \left[ \begin{array}{c} \hat{G}_{j,k}' \\ \hat{G}_{j,k,\perp}' \\ \end{array} \right]
%= \left[ \begin{array}{cc} E_{j,k} & E_{j,k,\perp} \\ \end{array} \right]
%\left[\begin{array}{cc} \tilde{A}_{j,k} + \tilde{H}_{j,k} \ & \tilde{B}_{j,k}' \ \\ \tilde{B}_{j,k} \ & \tilde{A}_{j,k,\perp} + \tilde{H}_{j,k,\perp}  \\ \end{array} \right]
%\left[ \begin{array}{c} {E_{j,k}}' \\ {E_{j,k,\perp}}' \\ \end{array} \right] \nn $$

\een
\end{definition}

\begin{definition}

For $k=1,2,\cdots,\vartheta_j$, define
\[
\tilde{\zeta}_{j,k} : = \bigg\|\Big(I - \sum_{i=1}^{k} \hat{G}_{j,i} \hat{G}_{j,i}'\Big)G_{j,k}\bigg\|_2
\]
This is the error in estimating $\Span(G_{j,k})$ after the $k^{th}$ iteration of the cluster-PCA step. %after the $k^{th}$ deletion proj-PCA step.

\end{definition}

\begin{remark}\label{SE_rem} \  %$P_{j,*} = P_{j-1} \setminus P_{j,\old}$
\ben

\item Notice that $\zeta_{j,0} = \|D_{j,\new}\|_2$, $\zeta_{j,k} = \|D_{j,\new,k}\|_2$ and $\tilde{\zeta}_{j,k} = \|(I - \hat{G}_k \hat{G}_k') D_{j,k}\|_2 = \|\Psi_{j,k}G_{j,k}\|_2$.

\item Notice from the algorithm that (i) $\Phat_{j,\new,k}$ is perpendicular to $\Phat_{j,*}=\Phat_{j-1}$; and (ii) $\hat{G}_{j,k}$ is perpendicular to $[\hat{G}_{j,1},\hat{G}_{j,2},\dots \hat{G}_{j,k-1}]$.

\item For $t\in \mathcal{I}_{j,k}$, $P_{(t)} = P_j = [(P_{j-1}R_j \setminus P_{j,\old}), \ P_{j,\new}]$, $\hat{P}_{(t)} = [\hat{P}_{j-1} \ \hat{P}_{j,\new,k}]$ and
\begin{align*}
SE_{(t)} &= \|(I - \hat{P}_{j-1} {\hat{P}_{j-1}}' - \hat{P}_{j,\new,k}{\hat{P}_{j,\new,k}}')P_j\|_2\\
& \leq  \|(I - \hat{P}_{j-1} {\hat{P}_{j-1}}' - \hat{P}_{j,\new,k}{\hat{P}_{j,\new,k}}')\\
&\hspace{1.6in} [ P_{j-1} \ P_{j,\new}]\|_2 \\
&\leq \zeta_{j,*} + \zeta_{j,k}
\end{align*}
for $k=1,2 \dots K$. The last inequality uses the first item of this remark.
%The last inequality follows as $\hat{P}_{j-1}$ and $\hat{P}_{j,\new,k}$ are orthogonal and then $ \|(I - \hat{P}_{j-1} {\hat{P}_{j-1}}' - \hat{P}_{j,\new,k}{\hat{P}_{j,\new,k}}')P_{j,*}\|_2 =\|(I - \hat{P}_{j-1} {\hat{P}_{j-1}}' )(I- \hat{P}_{j,\new,k}{\hat{P}_{j,\new,k}}')P_{j,*}\|_2 \leq \|(I - \hat{P}_{j-1} {\hat{P}_{j-1}}' )P_{j,*}\|_2 \leq \|(I - \hat{P}_{j-1} {\hat{P}_{j-1}}' )P_{j-1}\|_2 = \zeta_{j,*}$.

\item For $t \in \tilde{\mathcal{I}}_{j,k}$, $P_{(t)} = P_j$, $\hat{P}_{(t)} = [\hat{P}_{j-1} \ \hat{P}_{j,\new,K}]$ and
$$SE_{(t)} = SE_{(t_j + K\alpha-1)} \leq \zeta_{j,*} + \zeta_{j,K}$$

\item For $t \in \tilde{\mathcal{I}}_{j,\vartheta_j+1}$, $P_{(t)} = P_j$, $\operatorname{span}(P_j) = \operatorname{span}([G_{j,1},\cdots,G_{j,\vartheta_j}])$, $\hat{P}_{(t)} = \hat{P}_j = [\hat{G}_{j,1},\cdots,\hat{G}_{j,\vartheta_j}]$, and
$$SE_{(t)} = \zeta_{j+1,*} \leq \sum_{k=1}^{\vartheta_j} \tilde{\zeta}_{j,k}$$
The last inequality uses the first item of this remark.

\een
\end{remark}

\begin{definition}
Recall the definition of $\Phi_{j,k}$ from Definition \ref{defn_Phi}. Define $\Phi_{(t)}$ as
\bea
\Phi_{(t)} := \left\{
\begin{array}{ll}
\Phi_{j,k-1} \ & \ \ t \in \mathcal{I}_{j,k}, \ k=1,2 \dots K \\
\Phi_{j,K} \ & \ \ t \in \mathcal{\tilde{I}}_{j,k}, \ k=1,2 \dots \vartheta_j \\
\Phi_{j+1,0} \ & \ \ t \in \mathcal{\tilde{I}}_{j,\vartheta_j+1}
\end{array}
\right. \nn
\eea

\end{definition}

\begin{definition}
Define the random variable
\begin{align*}
%X_{j,k} &:= \{ a_1,a_2,\cdots, a_{t_j + k \alpha - 1}\} \\
\tilde{X}_{j,k} &:= \{ a_1,a_2,\cdots, a_{t_j + K \alpha +k\tilde{\alpha}- 1}\}
\end{align*}
%Recall that $a_t$'s are mutually independent over $t$.
\end{definition}

\begin{definition}
Define the sets %$\check\Gamma_{j,k}$ and $\tilde{\check{\Gamma}}_{j,k}$ as follows.
\begin{align*}
%\check\Gamma_{j,k}& :=
%\{ X_{j,k} : \zeta_{j,k} \leq \zeta_{k}^+,  \ \text{and} \ \hat{T}_t = T_t \ \text{for all} \ t \in \mathcal{I}_{j,k}\}, \ k=1,2,\dots K, \ j=1,2,3,\dots J  \\  %
\tilde{\check{\Gamma}}_{j,k}& :=
\{ \tilde{X}_{j,k} :\tilde{\zeta}_{j,k} \leq \tilde{c}_{j,k}\zeta, \text{and} \  \hat{T}_t = T_t \ \text{for all} \ t \in \mathcal{\tilde{I}}_{j,k}\}, \\
&\hspace{.5in}  k=1,2,\dots \vartheta_j, \ j=1,2,3,\dots J  \\
\tilde{\check{\Gamma}}_{j, \vartheta_j+1}& :=
\{X_{j+1,0}: \hat{T}_t = T_t \  \text{for all} \  t \in \mathcal{\tilde{I}}_{j,\vartheta_j+1}\}, \\
&\hspace{1.2in} j=1,2,3,\dots J \nn
\end{align*}
Define the sets %$\Gamma_{j,k}$ and $\tilde{\Gamma}_{j,k}$ as follows.
\begin{align*}
%\Gamma_{1,0}& :=  \{X_{1,0}: \zeta_{1,*} \leq r \zeta, \ \text{and} \ \hat{T}_t = T_t \ \text{for all} \ t \in [t_{\text{train}},t_1-1] \},  \\
%\Gamma_{j,0} &:=  \{X_{j,0}: \zeta_{j',*} \le \zeta_{*}^+ \ \text{for all} \ j' = 1, 2, \dots, j \ \text{and} \ \hat{T}_t = T_t \  \text{for all} \ t \le t_{j-1} \} \\
%\Gamma_{j,k}& :=  \Gamma_{j,k-1} \cap \check\Gamma_{j,k}, \ k=1,2,\dots K, \ j=1,2,3,\dots J \nn \\
\tilde{\Gamma}_{j,0} &:= \Gamma_{j,K} \\
\tilde{\Gamma}_{j,k}& :=  \tilde{\Gamma}_{j,k-1} \cap \tilde{\check{\Gamma}}_{j,k}, \ k=1,2,\dots \vartheta_j, \ j=1,2,3,\dots J
\end{align*}

\end{definition}

\begin{definition} \label{def_kappa_D}
Define $\kappa_{s,D}: = \max_j \max_k \kappa_s( D_{j,k})$  %$\kappa_{s,D}: = \kappa_s( (I - \sum_{i=0}^{k-1} \hat{G}_{j,i} \hat{G}_{j,i}' ) G_{j,k})$ and
\end{definition}

\begin{remark} \label{rem_kappa_D}
Conditioned on $\tilde\Gamma_{j,k-1}^e$, it is easy to see that
\begin{align*}
\kappa_{s,D} & :=\max_j \max_k \kappa_s( D_{j,k}) \\
&\le \max_j \max_k (\kappa_s(G_{j,k}) + r \zeta) \\
& \le \max_j \kappa_s(P_j) + r \zeta \le \kappa_{s,D}^+: = \kappa_{s,*}^+ + r \zeta.
\end{align*}
 In the above we have used $ \kappa_s(G_{j,k}) \le \kappa_s(P_j)$ and the same idea as in Lemma \ref{Dnew0_lem}.%??? %This is a loose bound; in fact it would be valid to assume a tighter bound on $G_{j,k}$. ???
\end{remark}

\subsection{Two Main Lemmas}
In this and the following subsections we remove the subscript $j$ at most places. Also recall from earlier that $P_{*} = P_{j-1}$.

The theorem is a direct consequence of Lemmas \ref{lem_add} and \ref{lem_del} given below.
Lemma \ref{lem_add} is a restatement of Lemmas \ref{expzeta} and \ref{mainlem} with using the new definition of $\zeta_*^+$ and the new bound on $\zeta$ from Theorem \ref{thm2}. It summarizes the final conclusions of the addition step for ReProCS-cPCA.
%provides high probability bounds on the sparse recovery and subspace estimation errors in each of the addition proj-PCA steps.

\begin{lem}[Final lemma for addition step]\label{lem_add}
Assume that all the conditions in Theorem \ref{thm2} holds. Also assume that $\mathbf{P}(\Gamma_{j,k-1}^e ) > 0$.
Then
\ben
\item $\zeta_0^+=1$, $\zeta_k^+ \leq  0.6^{k}  + 0.4 c\zeta$ for all $k=1,2,\dots K$;
\item $\mathbf{P}(\Gamma_{j,k}^e \ | \Gamma_{j,k-1}^e ) \geq p_k(\alpha,\zeta) \ge p_K(\alpha,\zeta)$  for all $k=1,2,\dots K$.
\een
where $\zeta_k^+$ is defined in Definition \ref{zetakplus} and $p_k(\alpha,\zeta)$ is defined in equation \eqref{pk}.
\end{lem}

The lemma below summarizes the final conclusions for the cluster-PCA step. %It is proved using lemmas given in Sec \ref{keylems}.
\begin{lem}[Final lemma for deletion (cluster-PCA) step]\label{lem_del}
Assume that all the conditions in Theorem \ref{thm2} hold. Also assume that $\mathbf{P}(\tilde{\Gamma}_{j,k-1}^e ) > 0$. Then,
\ben
\item for all $k=1,2,\dots \vartheta_j$, $\mathbf{P} ( \tilde\Gamma_{j,k}^e \ | \ \tilde\Gamma_{j,k-1}^e) \geq \tilde{p}(\tilde{\alpha},\zeta)$ where $\tilde{p}(\tilde{\alpha},\zeta)$ is defined in Lemma \ref{lem_bound_terms}.
\item $\mathbf{P} ( \Gamma_{j+1,0}^e \ | \ \tilde\Gamma_{j,\vartheta_j}^e) = 1$.
\een
\end{lem}
\begin{proof} %, \ \text{and} \ e_t \ \text{satisfies} \ (\ref{etdef})
Notice that $\mathbf{P} ( \tilde\Gamma_{j,k}^e \ | \ \tilde\Gamma_{j,k-1}^e) = \mathbf{P} (\tilde{\zeta}_k \leq \tilde{c}_k \zeta \ \text{and} \ \That_t = T_t \ \text{for all} \ t \in \tilde{\mathcal{I}}_{j,k} \ | \ \tilde\Gamma_{j,k-1}^e)$ and $\mathbf{P} ( \Gamma_{j+1,0}^e \ | \ \tilde\Gamma_{j,\vartheta_j}^e)  = \mathbf{P} ( \hat{T}_t = T_t \ \text{for all} \ t \in \mathcal{I}_{j,\vartheta_j+1} )$.
The first claim of the lemma follows by combining Lemma \ref{tilde_zeta} and the last claim of Lemma \ref{cslem}. The second claim follows using the last claim of Lemma \ref{cslem}.
%and the fact that $\Gamma_{j+1,0,0} \supseteq \Gamma_{j,K,\vartheta_j} \cap \{ \That_t = T_t, \ e_t \ \text{satisfies} \ (\ref{etdef_del}) \ \text{for all} \ t \in \tilde{I}_{j,\vartheta+1}$.
\end{proof}

\begin{remark}\label{Gamma_rem2_del}
Under the assumptions of Theorem \ref{thm2},
\[
\Gamma_{j,0} \cap (\cap_{k=1}^{K} \check{\Gamma}_{j,k}) \cap (\cap_{k=1}^{\vartheta_j} \tilde{\check{\Gamma}}_{j,k}) \subseteq \Gamma_{j+1,0}
\]
This follows easily using Remark  \ref{SE_rem} and the fact that $\sum_k \tilde{c}_k = r_j \le r$.
\end{remark}
%This follows because
%\[
%\zeta_{j+1,*}\leq \sum_{k=1}^{\vartheta_j} \tilde{\zeta}_{j,k} \leq \sum_{k=1}^{\vartheta_j} \tilde{c}_{j,k}\zeta = r_{j}\zeta \leq r\zeta =\zeta_*^+.
%\]
%Here the first inequality is from Remark \ref{SE_rem} and the second is from the definition of $\tilde{\check{\Gamma}}_{j,k}$.

\begin{remark}\label{Gamma_rem}
Under the assumptions of Theorem \ref{thm2}, the following hold.
\ben
%\item For any $k=1,2 \dots K$, $\Gamma_{j,k,0}^e$ implies that (i) $\zeta_{j,*} \le \zeta_{*}^+:= r \zeta$ and (ii) $\zeta_{j,k'} \le 0.6^{k'} + 0.4 c \zeta$ for all $k'=1,2,\dots k$
%\bi
%\item (i) follows from the definition of $\Gamma_{j,k,0}^e$ and $\zeta_{j,*} \le \sum_{k=1}^{\vartheta_{j-1}} \tilde{\zeta}_{j-1,k'} \le \sum_{k=1}^{\vartheta_{j-1}} \tilde{c}_{j-1,k'} \zeta =  r_{j-1} \zeta \le r \zeta = \zeta_*^+$; and (ii) follows from the definition of $\Gamma_{j,k,0}^e$ and the first claim of Lemma \ref{lem_add}. %implies that $\tilde{\zeta}_{j-1,k'} \le \tilde{c}_{j-1,k'} \zeta $ for all $k=1,2, \dots \vartheta_{j-1}$ and so
%\ei

\item For any $k=1,2 \dots \vartheta_j+1$, $\tilde{\Gamma}_{j,k}^e$ implies (i) $\zeta_{j,K} \le c \zeta$, (ii) $\|\Phi_{j,K} P_j\|_2 \le (r+c) \zeta$.
\bi
\item (i) follows from the first claim of Lemma \ref{lem_add} and the definition of $K$, (ii) follows using $\|\Phi_{j,K} P_j\|_2 \le \|\Phi_{j,K} [P_{*}, P_{\new}]\|_2 \le \zeta_{*} + \zeta_{K} \leq \zeta_*^+ +\zeta_{K}^+ \leq (r + c)\zeta$. %because $\Gamma_{j,K,k}^e$ implies that $\tilde\zeta_{j,k'} \le \tilde{c}_{j,k'} \zeta, \ \text{for all} \ k'=1,2,\dots k$
\ei

\item $\Gamma_{J+1,0}^e$ implies (i) $\zeta_{j,*} \le \zeta_{*}^+$ for all $j$, (ii) $\zeta_{j,k} \le 0.6^{k} + 0.4 c \zeta$ for all $k=1,\cdots,K$ and all $j$, (iii) $\zeta_{j,K} \le c \zeta$ for all $j$.% and (iii) $\zeta_{j+1,*} \le \sum_{k=1}^{\vartheta_{j}} \tilde{\zeta}_{j,k'} \le r \zeta$ for all $j$.
\een
\end{remark}

\subsection{Proof of Theorem \ref{thm2}}
\begin{proof}
From Remark \ref{Gamma_rem2_del}, %and Lemma \ref{subset_lem}
\begin{align*}
\mathbf{P}(\Gamma_{j+1,0}^e | \Gamma_{j,0}^e) &\geq \mathbf{P}(\check{\Gamma}_{j,1}^e,\dots,\check{\Gamma}_{j,K}^e, \tilde{\check{\Gamma}}^e_{j,1},\dots,\tilde{\check{\Gamma}}^e_{j,\vartheta_j} | \Gamma_{j,0}) \\
& = \prod_{k=1}^{K}\mathbf{P}(\check{\Gamma}_{j,k}^e | {\Gamma}_{j,k-1}^e)\prod_{k=1}^{\vartheta_j}\mathbf{P}(\tilde{\check{\Gamma}}_{j,k}^e | \tilde{\Gamma}_{j,k-1}^e)
\end{align*}
Also, since $\Gamma_{j+1,0} \subseteq \Gamma_{j,0}$ using Lemma \ref{subset_lem}, $\mathbf{P}(\Gamma_{J+1,0}^e|\Gamma_{1,0}^e) = \prod_{j=1}^{J} \mathbf{P}(\Gamma_{j+1,0}^e|\Gamma_{j,0}^e)$.
Thus
\begin{align*}
\mathbf{P}(\Gamma_{J+1,0}^e&|\Gamma_{1,0}^e)\geq \\
& \prod_{j=1}^J \left[ \prod_{k=1}^{K}\mathbf{P}(\check{\Gamma}_{j,k}^e | {\Gamma}_{j,k-1}^e)\prod_{k=1}^{\vartheta_j}\mathbf{P}(\tilde{\check{\Gamma}}_{j,k}^e | \tilde{\Gamma}_{j,k-1}^e) \right]
\end{align*}

Using Lemmas \ref{lem_add} and \ref{lem_del}, and the fact that $p_k(\alpha,\zeta) \ge p_K(\alpha,\zeta)$, we get $\mathbf{P}(\Gamma_{J+1,0}^e| \Gamma_{1,0}) \ge {p}_K(\alpha,\zeta)^{KJ} \tilde{p}(\tilde{\alpha},\zeta)^{\vartheta_{\max} J}$.
Also, $\mathbf{P}(\Gamma_{1,0}^e)=1$. This follows by the assumption on $\hat{P}_0$ and Lemma \ref{cslem}. Thus, $\mathbf{P}(\Gamma_{J+1,0}^e) \ge {p}_K(\alpha,\zeta)^{KJ} \tilde{p}(\tilde{\alpha},\zeta)^{\vartheta_{\max} J}$.

Using the definitions of $\alpha_\add(\zeta)$ and $\alpha_\del(\zeta)$ and $\alpha \ge \alpha_\add$ and $\tilde{\alpha} \ge \alpha_\del$, 
\begin{align*}
\mathbf{P}(\Gamma_{J+1,0}^e) &\ge {p}_K(\alpha,\zeta)^{KJ} \tilde{p}(\tilde{\alpha},\zeta)^{\vartheta_{\max} J} \\
& \ge (1-n^{-10})^2 \ge 1- 2n^{-10}
\end{align*}

%( \prod_k {p}_k(\tilde{\alpha},\zeta) \prod_k \tilde{p}(\tilde{\alpha},\zeta) )^J$.

%By Lemma \ref{lem_add}, $\zeta_k^+ \leq  0.6^{k-1}  + 0.4 c\zeta$.
The event $\Gamma_{J+1,0}^e$ implies that $\That_t=T_t$ for all $t < t_{J+1}$. Using Remark \ref{SE_rem} and the last claim of Remark \ref{Gamma_rem}, $\Gamma_{J+1,0}^e$ implies that all the bounds on the subspace error hold. Using these, Remark \ref{etdef_rem}, $\|a_{t,\new}\|_2 \le \sqrt{c} \gamma_{\new,k}$ and $\|a_t\|_2 \le \sqrt{r} \gamma_*$, $\Gamma_{J+1,0}^e$ implies that all the bounds on $\|e_t\|_2$ hold (the bounds are obtained in Lemma \ref{cslem}).

Thus, all conclusions of the the result hold w.p. at least $1- 2n^{-10}$.
\end{proof}

\subsection{A lemma needed for getting high probability bounds on the subspace error}

The following lemma is needed for bounding the subspace error, $\tilde\zeta_k$

% \leq  \frac{\|\tilde{\mathcal{R}}_k\|_2}{\lambda_{\min} (\tilde{A}_k) - \|\tilde{A}_{k,\perp}\|_2 - \|\tilde{\mathcal{H}}_k\|_2}

\begin{lem}\label{bound_R}
%Consider the time interval $\tilde{\mathcal{I}}_{j,k}$. :=\|\Psi_{k'} G_{k'}\|_2
Assume that $\tilde{\zeta}_{k'} \leq \tilde{c}_{k'} \zeta$ for $k'=1,\cdots, k-1$. Then
\ben
\item $\|D_{\text{det},k}\|_2 = \|\Psi_{k-1} G_{\text{det},k}\|_2 \leq r \zeta$.
\item $\|G_{\text{det},k} {G_{\text{det},k}}' - \hat{G}_{\text{det},k}{\hat{G}_{\text{det},k}}'\|_2 \leq 2 r\zeta$.
\item $0< \sqrt{1-r^2 \zeta^2} \leq \sigma_i(D_k) = \sigma_i(R_k) \leq 1$. Thus, $\|D_k\|_2 = \|R_k\|_2 \le 1$ and $\|D_k^{-1}\|_2 = \|R_k^{-1}\|_2 \le 1/\sqrt{1-r^2 \zeta^2} $.
\item $\|{D_{\text{undet},k}}'E_k \|_2 = \|{G_{\text{undet},k}}'E_k \|_2 \leq \frac{r^2 \zeta^2}{\sqrt{1-r^2\zeta^2}}$. % = \|{E_k}' G_{\text{undet},k}\|_2
\een
\end{lem}
\begin{proof} The proof is given in Appendix \ref{proof_lem_bound_R}. \end{proof}
%The first claim essentially follows by using the fact that $\hat{G}_1,\cdots,\hat{G}_{k-1}$ are mutually orthonormal and the triangle inequality. Recall that $\Psi_{k-1} = (I - \hat{G}_{\text{det},k} {\hat{G}_{\text{det},k}}')$. The last three claims use this and the first claim and apply Lemma \ref{lemma0}. The last claim also uses the definition of $D_k$ and its QR decomposition. The complete proof is given in Appendix \ref{proof_lem_bound_R}. \end{proof}
 % and so $\Psi_{k-1} = \prod_{k_2=1}^{k-1}(I - \hat{G}_{k_2}{\hat{G}_{k_2}}')$

\subsection{Bounding the subspace error, $\tilde\zeta_k$}

\begin{lem}[High probability bound on $\tilde{\zeta_k}$]\label{tilde_zeta}
Assume that the conditions of Theorem \ref{thm2} hold. Then,
\beq
\mathbf{P} (\tilde{\zeta}_k \leq \tilde{c}_k \zeta \ | \tilde\Gamma_{j,k-1}^e) \geq  \tilde{p}(\tilde{\alpha},\zeta)
\nn
\eeq
 where $\tilde{p}(.)$ is defined in Lemma \ref{lem_bound_terms}.
\end{lem}
\begin{proof} This follows by combining Lemma \ref{bnd_tzetakp} and the last claim of Lemma \ref{lem_bound_terms}, both of which are given below. 
\end{proof}

\begin{lem}[Bounding $\tilde{\zeta_k}^+$] \label{bnd_tzetakp}
%Define
%\bea
%F(\tilde{g},\tilde{h}) &:=& f_{dec}(\tilde{g},\tilde{h}) - \frac{f_{inc}(\tilde{g},\tilde{h})}{\zeta}  \label{Func_del}
%\eea
%If $F(\tilde{g}_{\max},\tilde{h}_{\max}) >0$, then $f_{dec}(\tilde{g}_k,\tilde{h}_k) >0$ and $\tilde\zeta_k^+ \le \tilde{c}_k \zeta$
If
\begin{align*}
f_{dec}(\tilde{g}_{\max},\tilde{h}_{\max}, &\kappa_{s,e}^+,\kappa_{s,*}^+ + r \zeta) - \\ &\frac{f_{inc}(\tilde{g}_{\max},\tilde{h}_{\max},\kappa_{s,e}^+,\kappa_{s,*}^+ + r \zeta)}{\tilde{c}_{\min} \zeta} > 0  \label{Func_del}
\end{align*}
then $f_{dec}(\tilde{g}_k,\tilde{h}_k, \kappa_{s,e}^+,\kappa_{s,*}^+ + r \zeta) >0$ and $\tilde\zeta_k^+ \le \tilde{c}_k \zeta$.
\end{lem}
\begin{proof} Recall from Definition \ref{def_zeta} that $\tilde{\zeta_k}^+ := \frac{f_{inc}(\tilde{g}_k,\tilde{h}_k,\kappa_{s,e}^+,\kappa_{s,*}^+ + r \zeta)}{f_{dec}(\tilde{g}_k,\tilde{h}_k,\kappa_{s,e}^+,\kappa_{s,*}^+ + r \zeta)}$. %that $f_{inc}(.)$, $f_{dec}(.)$ are defined in Definition \ref{def alpha_del} and
Notice that $f_{inc}(.)$ is an increasing function of $\tilde{g},\tilde{h}$, and $f_{dec}(.)$ is a decreasing function. Using the definition of $\tilde{g}_{\max},\tilde{h}_{\max}, \tilde{c}_{\min}$ given in Assumption \ref{clusterass}, the result follows.
\end{proof}
%by definition, $\tilde{g}_{\max} = \max_j \max_k \tilde{g}_{j,k}$ and $\tilde{h}_{\max} = \max_j \max_k \tilde{h}_{j,k}$.
 %$\tilde{g}_k \leq \tilde{g}_{\max}$, $\tilde{h}_k \leq \tilde{h}_{\max}$ and $\tilde{c}_k \ge \tilde{c}_{\min}$,

%\begin{remark} \label{f_inc_rem}
%If we ignore the small terms of $f_{inc}(.)$ and $f_{dec}(.)$, the above condition simplifies to requiring that $\frac{3\kappa_{s,e}^+ \phi^+ \tilde{g}_{\max} + \kappa_{s,e}^+ \phi^+ \tilde{h}_{\max}}{1-\tilde{h}_{\max}} \le \frac{\tilde{c}_{\min}}{r+c}$. Since $\tilde{g}_{\max} \ge 1$, the first term of the numerator is the largest one. To ensure that this condition holds we need $\kappa_{s,e}^+$ to be very small. However, as explained in Sec \ref{proof_lem_bound_terms}, if we also assume denseness of $D_{k}$, i.e. if we assume $\kappa_s(D_{k}) \le \kappa_{s,D}^+$ for a small enough $\kappa_{s,D}^+$, then the first term of the numerator can be replaced by $\max( 3\kappa_{s,e}^+ \kappa_{s,D}^+  \phi^+ \tilde{g}_{\max}, \kappa_{s,e}^+ \phi^+ \tilde{h}_{\max})$. This will relax the requirement on $\kappa_{s,e}^+$, e.g. now $\kappa_{s,e}^+ =\kappa_{s,D}^+= 0.3$ will work.
%\end{remark}

\begin{lem}[Bounding $\tilde{\zeta_k}$] \label{defnPCA}
If $\lambda_{\min}(\tilde{A}_k) - \lambda_{\max}(\tilde{A}_{k,\perp}) - \|\tilde{\mathcal{H}}_k\|_2 >0$, then
\beq
\tilde{\zeta_k} \leq  \frac{\|\tilde{\mathcal{H}}_k\|_2}{\lambda_{\min} (\tilde{A}_k) - \lambda_{\max} (\tilde{A}_{k,\perp}) -  \|\tilde{\mathcal{H}}_k\|_2}
\label{zetakbnd_del}
\eeq
\end{lem}
\begin{proof}
The proof is the same as that of Lemma \ref{zetakbnd}.

%Recall that $\tilde{A}_k$, $\tilde{A}_{k,\perp}$, $\tilde{\mathcal{H}}_k$ are defined in Definition \ref{defHk_del}. The result follows by using the fact that $\tilde{\zeta}_k = \|(I - \hat{G}_k \hat{G}_k') D_{j,k}\|_2 = \|(I - \hat{G}_k \hat{G}_k') E_{k} R_{k}\|_2 \le \|(I - \hat{G}_k \hat{G}_k') E_{k}\|_2$ and applying Lemma \ref{sin_theta_weyl} with $E \equiv E_k$ and $F \equiv \hat{G}_k$.
\end{proof}

\begin{lem}[High probability bounds for each of the terms in the $\tilde{\zeta}_k$ bound and for $\tilde\zeta_k$]\label{lem_bound_terms} % for all $1 \leq k \leq \vartheta_j$
 %Let $\kappa_{s,e} := \kappa_s(\Phi_K P_{j})$.
 Assume that the conditions of Theorem \ref{thm2} hold. Also, assume that $\mathbf{P}(\tilde\Gamma_{j,k-1}^e)>0$. Then,  for all $1 \leq k \leq \vartheta_j$,
\ben
\item $\mathbf{P}(\lambda_{\min}(\tilde{A}_{k}) \geq \lambda_{k}^-(1-r^2 \zeta^2  - 0.1 \zeta) | \tilde\Gamma_{j,k-1}^e) >1- \tilde{p}_1(\tilde{\alpha},\zeta)$
with $\tilde{p}_1(\tilde{\alpha},\zeta)$ given in (\ref{prob1}).
\item $\mathbf{P}(\lambda_{\max}(\tilde{A}_{k,\perp}) \leq \lambda_k^- (\tilde{h}_k+r^2 \zeta^2 f+0.1 \zeta) | \tilde\Gamma_{j,k-1}^e) > 1-\tilde{p}_2(\tilde{\alpha},\zeta)$ with $\tilde{p}_2(\tilde{\alpha},\zeta)$ given in (\ref{prob2}).
\item $\mathbf{P}(\|\tilde{\mathcal{H}}_{k}\|_2 \leq \lambda_k^- f_{inc}(\tilde{g}_k,\tilde{h}_k,\kappa_{s,e}^+,\kappa_{s,*}^+ + r \zeta) \ |\tilde\Gamma_{j,k-1}^e) \geq 1 - \tilde{p}_3(\tilde{\alpha},\zeta)$ with $\tilde{p}_3(\tilde{\alpha},\zeta)$ given in (\ref{prob3}).
%\lambda_k^-( 3 (r+c) \zeta  \kappa_{s,e}  \phi^+ \tilde{g}_k +  [(r+c)\zeta \kappa_{s,e} \phi^+ + (r+c)\zeta \kappa_{s,e} (1+ 2\phi^+)\frac{r^2\zeta^2}{\sqrt{1-r^2\zeta^2}}] \tilde{h}_k +  [r^2\zeta^2 + 4(r+c)r \zeta^2 \kappa_{s,e} \phi^+ + 2(r+c)^2 \zeta^2(1+ \kappa_{s,e}^2) {\phi^+}^2] f) + 0.2 \zeta \lambda^- \ |\tilde\Gamma_{j,k-1}^e) \geq 1 - \tilde{p}_3(\tilde{\alpha},\zeta)$ with $\tilde{p}_3(\tilde{\alpha},\zeta)$ given in (\ref{prob3}).
\item $\mathbf{P}( \lambda_{\min} (\tilde{A}_k) -   \lambda_{\max} (\tilde{A}_{k,\perp})  - \|\tilde{\mathcal{H}}_k\|_2 \ge \lambda_k^- f_{dec}(\tilde{g}_k,\tilde{h}_k,\kappa_{s,e}^+,\kappa_{s,*}^+ + r \zeta) \ |\tilde\Gamma_{j,k-1}^e) \geq \tilde{p}(\tilde{\alpha},\zeta) :=1- \tilde{p}_{1}(\tilde{\alpha},\zeta) - \tilde{p}_{2}(\tilde{\alpha},\zeta) - \tilde{p}_{3}(\tilde{\alpha},\zeta)$.
\item If $f_{dec}(\tilde{g}_k,\tilde{h}_k) >0$, then $\mathbf{P}(\tilde{\zeta}_k \leq \tilde\zeta_k^+ \ | \tilde\Gamma_{j,k-1}^e) \geq \tilde{p} (\tilde{\alpha},\zeta)$ \nn
\een
\end{lem}
\begin{proof}
Recall that $f_{inc}(.)$, $f_{dec}(.)$ and $\tilde\zeta_k^+$ are defined in Definition \ref{def_zeta}. The proof of the first three claims is given in Appendix \ref{proof_lem_bound_terms}. This proof uses Lemmas \ref{bound_R} and \ref{cslem}, Remark \ref{rem_kappa_D}, and the Hoeffing corollaries. The fourth claim follows directly from the first three using the union bound on probabilities. The fifth claim follows from the fourth using Lemma \ref{defnPCA}.
\end{proof}

\section{Model Verification and Simulation Experiments}
\label{model_expts}
We first discuss model verification for real data in Sec \ref{model_verify}.  We then describe simulation experiments in Sec \ref{sims}.

\subsection{Model Verification for real data} % {\cal L}_{1:d}$  {\cal L}_{d+1:2d}$;
\label{model_verify}

We experimented with two background image sequence datasets. The first was a video of lake water motion. The second was a video of window curtains moving due to the wind. The curtain sequence is available at \url{http://home.engineering.iastate.edu/~chenlu/ReProCS/Fig2.mp4}. For this sequence, the image size was $n=5120$ and the number of images, $t_{\max}=1755$.
The lake sequence is available at \url{http://home.engineering.iastate.edu/~chenlu/ReProCS/ReProCS.htm} (sequence 3). For this sequence, $n=6480$ and the number of images, $t_{\max}=1500$.
Any given background image sequence will never be exactly low rank, but only approximately so.
Let the data matrix with its empirical mean subtracted be ${\cal L}_{full}$. Thus ${\cal L}_{full}$ is a $n \times t_{\max}$ matrix.
We first ``low-rankified" this dataset by computing the EVD of $(1/t_{\max}) {\cal L}_{full} {\cal L}_{full}'$; retaining the 90\% eigenvectors' set (i.e. sorting eigenvalues in non-increasing order and retaining all eigenvectors until the sum of the corresponding eigenvalues exceeded 90\% of the sum of all eigenvalues); and projecting the dataset into this subspace. To be precise, we computed $P_{full}$ as the matrix containing these eigenvectors and we computed the low-rank matrix ${\cal L} = P_{full} P_{full}' {\cal L}_{full}$.
 Thus ${\cal L}$ is a $n \times t_{\max}$ matrix with $\rank({\cal L}) < \min(n, t_{\max})$.
The curtains dataset is of size $5120 \times 1755$, but 90\% of the energy is contained in only $34$ directions, i.e. $\rank({\cal L})=34$. The lake dataset is of size $6480 \times 1500$ but 90\% of the energy is contained in only $14$ directions, i.e. $\rank({\cal L})=14$. This indicates that both datasets are indeed approximately low rank.%r_{full} =

%Call the resulting matrix ${\cal L}$. 90\% of the energy is contained in only $14$ directions, i.e.

In practical data, the subspace does not just change as simply as in the model given in Sec. \ref{model}. There are also rotations of the new and existing eigen-directions at each time which have not been modeled there. Moreover, with just one training sequence of a given type, it is not possible to compute $\text{Cov}(L_t)$ at each time $t$. Thus it is not possible to compute the delay between subspace change times. The only thing we can do is to assume that there may be a change every $d$ frames, and that during these $d$ frames the data is stationary and ergodic, and then estimate  $\text{Cov}(L_t)$ for this period using a time average. We proceeded as follows.
We took the first set of $d$ frames, ${\cal L}_{1:d} := [L_1, L_2 \dots L_d]$, estimated its covariance matrix as $(1/d) {\cal L}_{1:d} {\cal L}_{1:d}'$ and computed $P_0$ as the 99.99\% eigenvectors' set. Also, we stored the lowest retained eigenvalue and called it $\lambda^-$. It is assumed that all directions with eigenvalues below $\lambda^-$ are due to noise. Next, we picked the next set of $d$ frames, ${\cal L}_{d+1:2d}: = [L_{d+1}, L_{d+2}, \dots L_{2d}]$; projected them perpendicular to $P_0$, i.e. computed ${\cal L}_{1,p}=(I - P_0 P_0'){\cal L}_{d+1:2d}$; and computed $P_{1,\new}$ as the eigenvectors of $(1/d) {\cal L}_{1,p} {\cal L}_{1,p}'$ with eigenvalues equal to or above $\lambda^-$. Then,  $P_1 = [P_0, P_{1,\new}]$. For the third set of $d$ frames, we repeated the above procedure, but with $P_0$ replaced by $P_1$ and obtained $P_2$. A similar approach was repeated for each batch. %It is difficult to compute deletions, so we assume a model with $c_{j,\old}=0$.

% and plot $\|a_{t,\new}\|_\infty/\|a_{t,*}\|_\infty$
 We used $d=150$ for both the datasets. In each case, we computed $r_0 := \rank(P_0)$, and $c_{\max} := \max_j \rank(P_{j,\new})$. For each batch of $d$ frames, we also computed $a_{t,\new}: = P_{j,\new}' L_t$, $a_{t,*}: = P_{j-1}' L_t$ and $\gamma_*: = \max_t \|a_{t}\|_\infty$. We got $c_{\max}=3$ and $r_0=8$ for the lake sequence and $c_{\max}=5$ and $r_0=29$ for the curtain sequence. Thus the ratio $c_{\max}/r_0$ is sufficiently small in both cases.
In Fig \ref{model_verification}, we plot $\|a_{t,\new}\|_{\infty}/\gamma_*$ for one 150-frame period of the curtain sequence and for three 150-frame change periods of the lake sequence. If we take  $\alpha=40$, we observe that $\gamma_\new: = \max_j \max_{t_j \le t < t_j+ \alpha} ||a_{t,\new}||_\infty = 0.125 \gamma_*$ for the curtain sequence and $\gamma_\new = 0.06 \gamma_*$ for the lake sequence, i.e. the projection along the new directions is small for the initial $\alpha$ frames. Also, clearly, it increases slowly. In fact $\|a_{t,\new}\|_{\infty} \le \max(v^{k-1} \gamma_\new,\gamma_*)$ for all $t \in \mathcal{I}_{j,k}$ also holds with $v=1.5$ for the curtain sequence and $v=1.8$ for the lake sequence.

{\em Verifying the clustering assumption. }
We verified the clustering assumption for the lake video as follows. We first ``low-rankified" it to 90\% energy as explained above. Note that, with one sequence, it is not possible to estimate $\Lambda_t$ (this would require an ensemble of sequences) and thus it is not possible to check if all $\Lambda_t$'s in $[\tilde{t}_j, t_{j+1}-1]$ are similar enough. However, by assuming that $\Lambda_t$ is the same for a long enough sequence, one can estimate it using a time average and then verify if its eigenvalues are sufficiently clustered. When this was done, we observed that the clustering assumption holds with $\tilde{g}_{\max} = 7.2$, $\tilde{h}_{\max} = 0.34$ and $\vartheta_{\max} = 7$

%We should point out that the key requirement for ReProCS to work is that $\|a_{t,\new}\|_2$ be initially small and increase gradually. Notice that both the sparse recovery error bound and the lower bound required on $S_{\min}$ depend on $\sqrt{c} \gamma_\new$ which is the upper bound on $\|a_{t,\new}\|_2$ for the first $\alpha$ frames. To get a simple model, we instead bound $\|a_{t,\new}\|_2$ by $\sqrt{c} \|a_{t,\new}\|_\infty$ and then assume that $c$ should be small and $\|a_{t,\new}\|_\infty$ should be small and should increase gradually. In Fig \ref{model_verification_l2}, we plot $\|a_{t,\new}\|_2/\|a_{t,*}\|_2$. Notice that this is also is initially small and increases slowly. ?? Delete this para and the figure ??
%Because in our model, we assume a bound on

%The plots for $\|a_{t,\new}\|_\infty/\|a_{t,*}\|_\infty$ for two change periods and for both the lake and the curtains' sequences are shown in Fig \ref{model_verification}. Clearly this ratio is initially small and increases gradually as required by our assumption. We also show $\|a_{t,\new}\|_2/\|a_{t,*}\|_2$ and a similar pattern is seen for this as well.

\begin{figure}
\centerline{
\includegraphics[height=5cm]{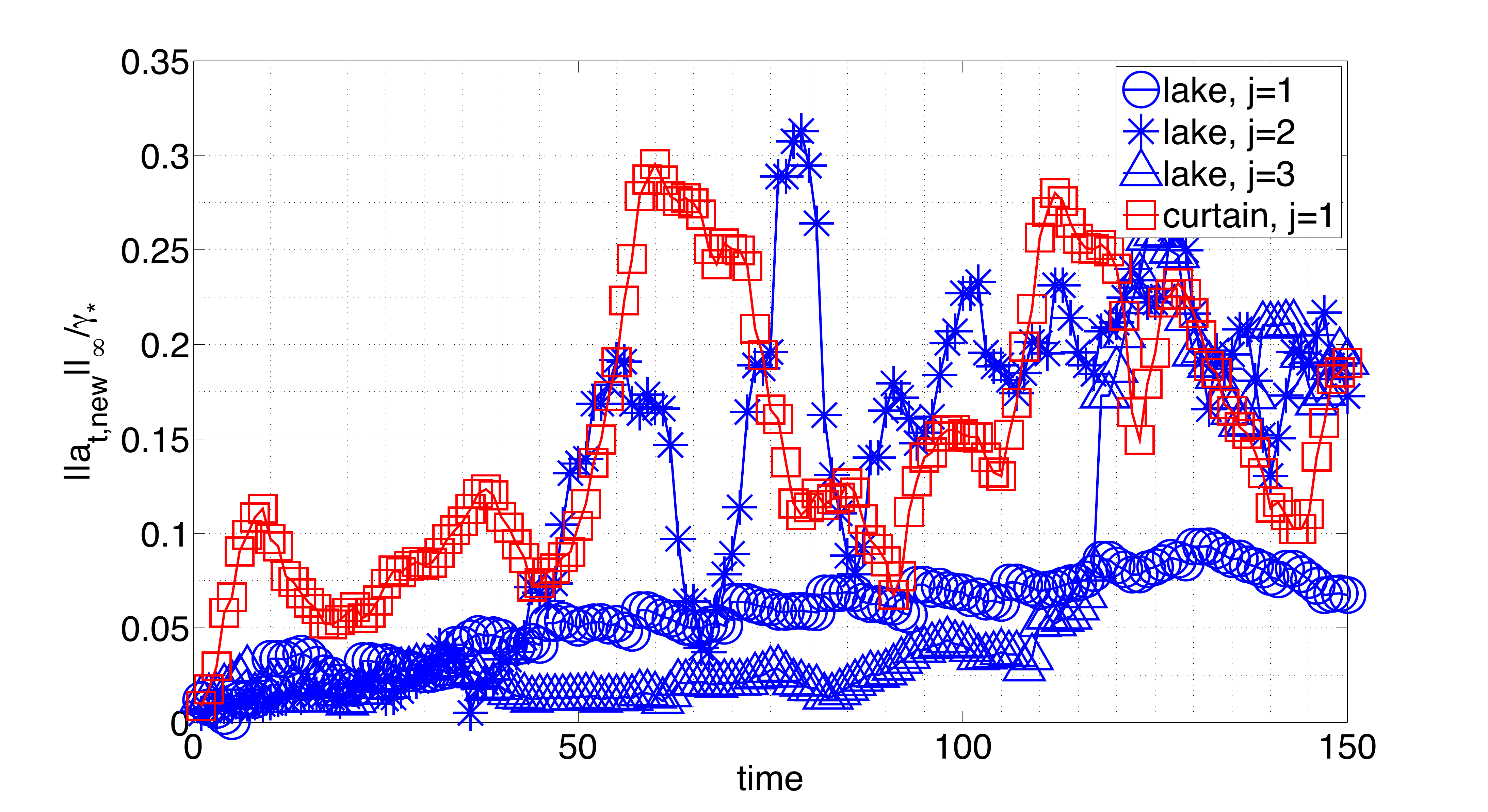} %width =8cm,
}
\caption{Verification of slow subspace change. The figure is discussed in Sec \ref{model_verify}.}
\label{model_verification}
\end{figure}

\subsection{Simulation Experiments} \label{sims}

%\subsubsection{Data Generation}\label{gendata}
The simulated data is generated as follows.
The measurement matrix $\mathcal{M}_t := [M_1, M_2,\cdots, M_t]$ is of size $2048 \times 4200$. It can be decomposed as a sparse matrix $\mathcal{S}_t:= [S_1, S_2,\cdots, S_t]$ plus a low rank matrix $\mathcal{L}_t:= [L_1, L_2,\cdots, L_t]$.

The sparse matrix $\mathcal{S}_t := [S_1,S_2,\cdots,S_t]$ is generated as follows.
\ben
\item For $1 \leq t \leq t_{\train} = 200$, $S_t = 0$.
\item For $t_{\train} < t \leq 5200$, $S_t$ has $s$ nonzero elements. The initial support $T_0 = \{1,2,\dots s\}$. Every $\Delta$ time instants we increment the support indices by 1. For example, for $t \in [t_\train +1, t_\train + \Delta-1]$, $T_t = T_0$, for $t \in [t_\train + \Delta, t_\train + 2\Delta-1]$. $T_t = \{2,3,\dots s+1\}$ and so on. Thus, the support set changes in a highly correlated fashion over time and this results in the matrix ${\cal S}_t$ being low rank. The larger the value of $\Delta$, the smaller will be the rank of ${\cal S}_t$ (for $t > t_\train+\Delta$).

\item The signs of the nonzero elements of $S_t$ are $\pm 1$ with equal probability  and the magnitudes are uniformly distributed between $2$ and $3$. Thus, $S_{\min} = 2$.
\een

The low rank matrix $\mathcal{L}_t := [L_1,L_2,\cdots,L_t]$ where $L_t := P_{(t)} a_t$ is generated as follows:
\ben
\item There are a total of $J=2$ subspace change times,  $t_1=301$ and $t_2= 2701$. Let $U$ be an $2048 \times (r_0 + c_{1,\new} + c_{2,\new})$ orthonormalized random Gaussian matrix.
\ben
\item For $1\leq t \leq t_1 -1$, $P_{(t)} = P_0$ has rank $r_0$ with  $P_0 = U_{[1,2,\cdots,r_0]}$. %$\text{rank}(\mathcal{L}_{t}) = r_0$ for all $t \le t_1-1$
\item For $t_1  \leq t \leq t_2-1$, $P_{(t)} = P_1 = [P_0 \ P_{1,\new}]$ has rank $r_1 = r_0 + c_{1,\new}$ with $ P_{1,\new} = U_{[r_0+1,\cdots,r_0+c_{1,\new}]}$. %with $c_{1,\new} =1$ and $\text{rank}(\mathcal{L}_{t}) = r_1 = r_0+1$ for this period.
\item For $t\geq t_2$, $P_{(t)} = P_2 = [P_1 \ P_{2,\new}]$  has rank $r_2 = r_1 + c_{2,\new}$ with $ P_{2,\new} = U_{[r_0+c_{1,\new}+1,\cdots,r_0+c_{1,\new}+c_{2,\new}]}$.%with $c_{2,\new}=1$ and $\text{rank}(\mathcal{L}_{t}) =  r_2  = r_0+2$ for this period.
\een
\item $a_t$ is independent over $t$.  The various $(a_t)_i$'s are also mutually independent for different $i$. %We let $(a_t)_i$ be uniformly distributed between $-\gamma_{i,t}$ and $\gamma_{i,t}$ as Let $(a_t)_i$ denote the $i$th element of $a_t$.
    \ben
    \item For $1\leq t < t_1$, we let $(a_t)_i$ be uniformly distributed between $-\gamma_{i,t}$ and $\gamma_{i,t}$, where  %$a_t:= {P_0}' L_t$. W
\begin{multline*}
\gamma_{i,t}=\\
\begin{cases} 400 & \text{if $i=1,2,\cdots,r_0/4, \forall t$,}
\\
30 &\text{if $i=r_0/4+1,r_0/4+2,\cdots,r_0/2, \forall t$.}
\\
2 &\text{if $i=r_0/2+1,r_0/2+2,\cdots,3r_0/4, \forall t$.}
\\
1 &\text{if $i=3r_0/4+1,3r_0/4+2,\cdots,r_0, \forall t$.}
\end{cases} 
\end{multline*}

% $a_{t,*} := {P_0}'L_t$ and  := {P_{1,\new}}'L_t  defined in (\ref{gamma0}).
\item For $t_1  \leq t < t_2$, $a_{t,*}$ is an $r_0$ length vector, $a_{t,\new}$ is a $c_{1,\new}$ length vector and $L_t := P_{(t)} a_t = P_1 a_t = P_0 a_{t,*} + P_{1,\new} a_{t,\new}$.
$(a_{t,*})_i$ is uniformly distributed between $-\gamma_{i,t}$ and $\gamma_{i,t}$ and $a_{t,\new}$ is uniformly distributed between $-\gamma_{r_1,t}$ and $\gamma_{r_1,t}$, where
\[
        \gamma_{r_1,t}=
\begin{cases} 1.1^{k-1} \quad \text{ if } t_1 + (k-1) \alpha \leq t \leq \\ 
\hspace{1 in}t_1 +k\alpha-1 \\
\hspace{1in} k=1,2,3,4 \\
1.1^{4-1} = 1.331  \quad \text{if $t \geq t_1 + 4\alpha$.}
\end{cases}  %\label{gamma1}
\]

%        $-1.1^{k-1}$ and $1.1^{k-1}$ for $t_1 + (k-1) \alpha \leq t \leq t_1 +k\alpha-1$ with $\alpha=100$ and $k=1,2,3,4$. For $t_1 + 4\alpha \leq t < t_2$, $\gamma_{r_0+1,t} = 1.1^{4-1} = 1.331$.
\item For $t\geq t_2$, $a_{t,*}$ is an $r_1=r_0 + c_{1,\new}$ length vector,   $a_{t,\new}$ is a $c_{2,\new}$ length vector and $L_t := P_{(t)} a_t = P_2 a_t = [P_0 \ P_{1,\new}] a_{t,*} + P_{2,\new} a_{t,\new}$. Also, $(a_{t,*})_i$ is uniformly distributed between $-\gamma_{i,t}$ and $\gamma_{i,t}$ for $i=1,2,\cdots,r_0$ and is uniformly distributed between $-\gamma_{r_1,t}$ and $\gamma_{r_1,t}$ for $i=r_0+1, \dots r_1$. $a_{t,\new}$ is uniformly distributed between $-\gamma_{r_2,t}$ and $\gamma_{r_2,t}$, where
    \[
        \gamma_{r_2,t}=
\begin{cases} 1.1^{k-1}  \quad \text{if } t_2 + (k-1) \alpha \leq t \leq \\
\hspace{1 in} t_2 +k\alpha-1, \\
\hspace{1 in} k=1,2,\cdots,7
\\
1.1^{7-1} = 1.7716 \quad  \text{if $t \geq t_2 + 7\alpha$.}
\end{cases}  %\label{gamma2}
\]
        %$-1.1^{k-1}$ and $1.1^{k-1}$ for $t_2 + (k-1) \alpha \leq t \leq t_2 +k\alpha-1$ with $k=1,2,\cdots,7$. For $t \geq t_2 + 7\alpha$, $\gamma_{r_1+1,t} = 1.1^{7-1} =1.7716$.
    \een
\een
Thus for the above model, $\gamma_* = 400$, $\gamma_{\new} =1$, $\lambda^+ = 53333$, $\lambda^- = 0.3333$ and $f:=\frac{\lambda^+}{\lambda^-} = 1.6 \times 10^{5}$.
Also, $S_{\min}=2$.

We used $\mathcal{L}_{t_{\train}} + \mathcal{N}_{t_{\train}}$ as the training sequence to estimate $\Phat_0$. Here  $\mathcal{N}_{t_{\train}} = [N_1, N_2 ,\cdots, N_{t_{\train}}]$ is i.i.d. random noise with each $(N_t)_i$ uniformly distributed between $-10^{-3}$ and $10^{-3}$. This is done to ensure that $\Span(\Phat_0) \neq \Span(P_0)$ but only approximates it.
% where $\mathcal{L}_{t_{\train}} = [L_1, L_2 ,\cdots, L_{t_{\train}}]$.

\begin{figure*}
\subfigure[$\Delta = 2$]
{\includegraphics[width = 8 cm, height = 4 cm]{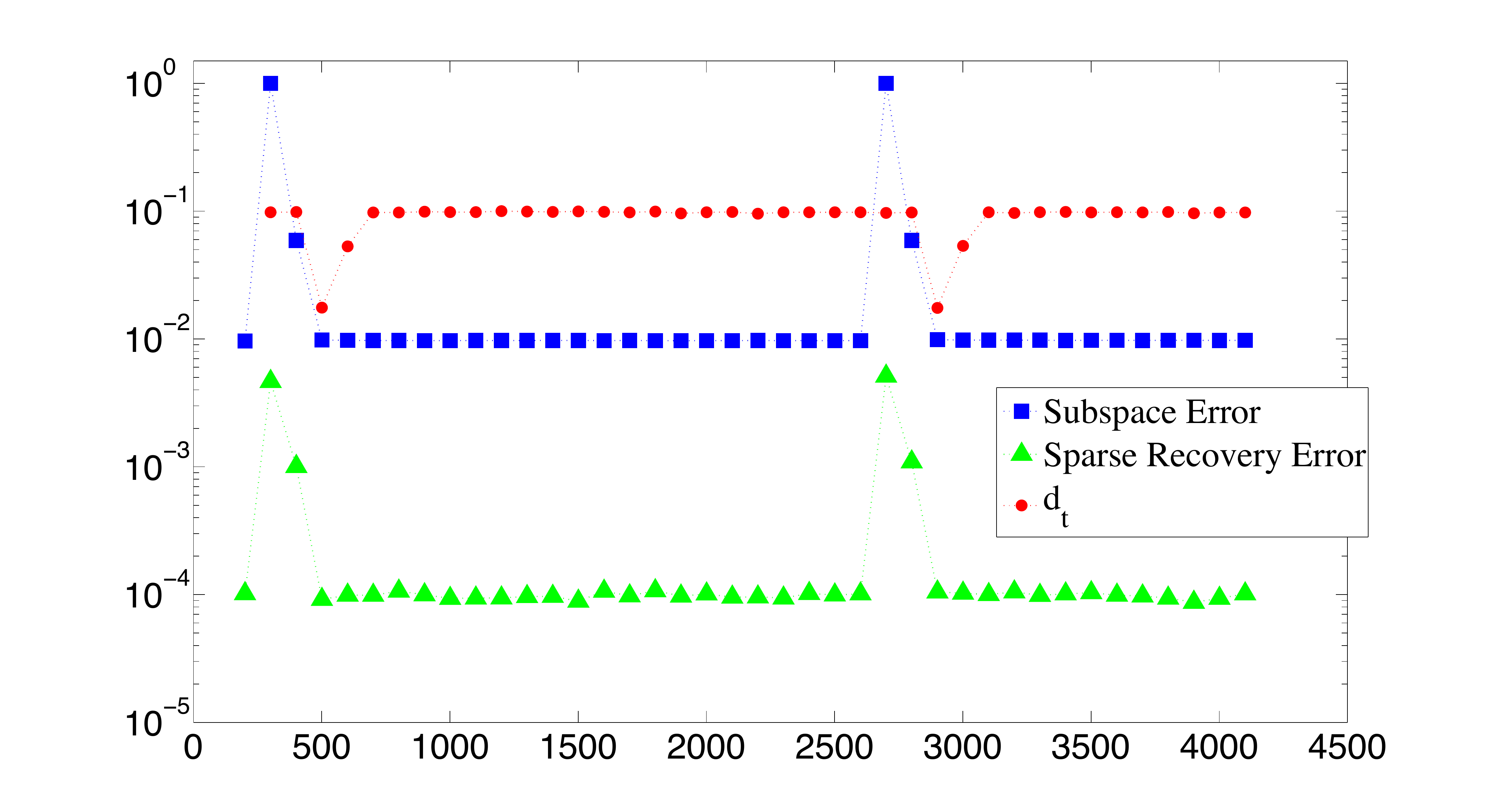}\label{D2}}
\subfigure[$\Delta = 10$]
{\includegraphics[width = 8 cm, height = 4cm]{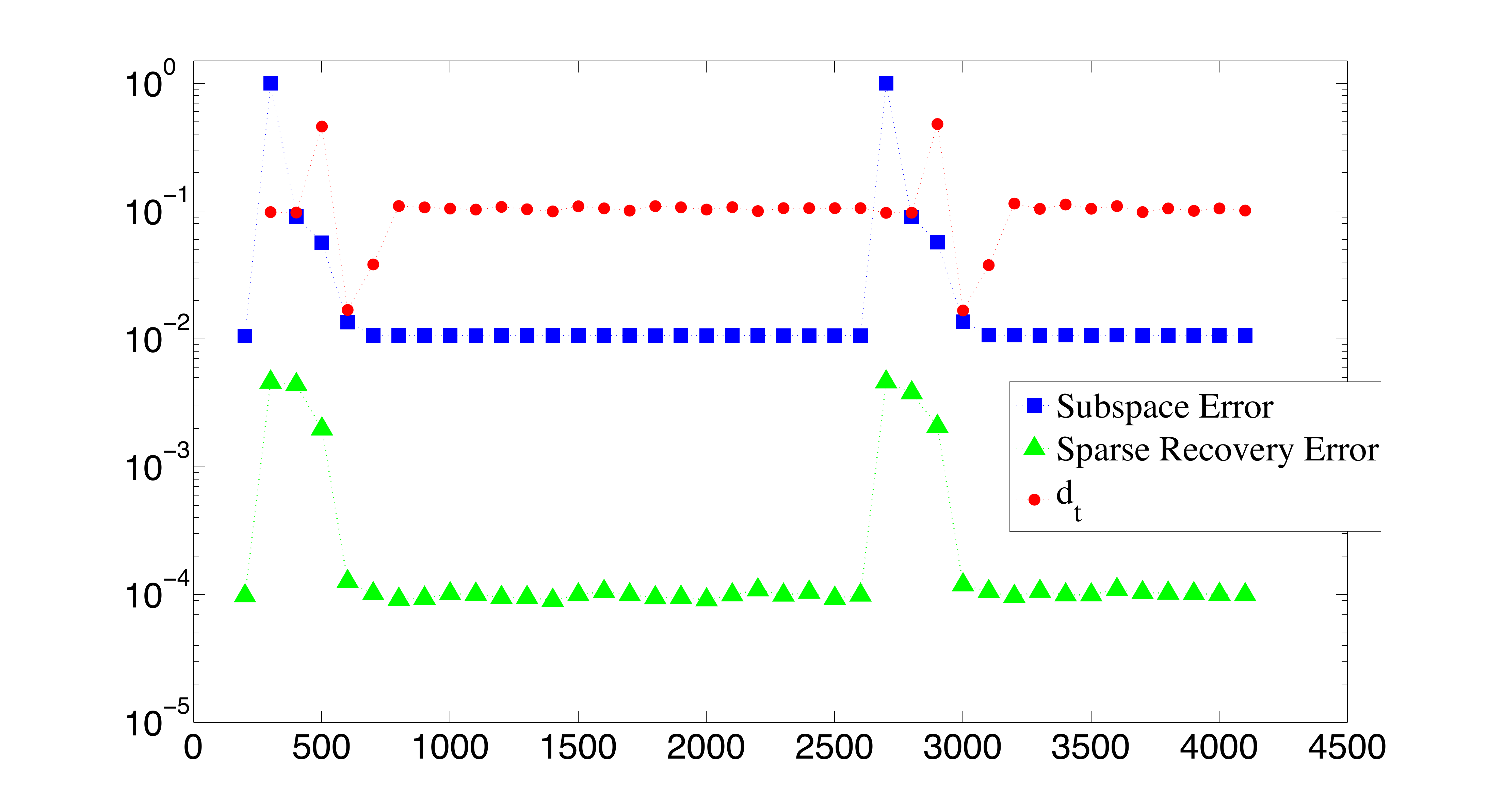}\label{D10}}\\
\subfigure[$\Delta = 50$]
{\includegraphics[width = 8 cm, height = 4 cm]{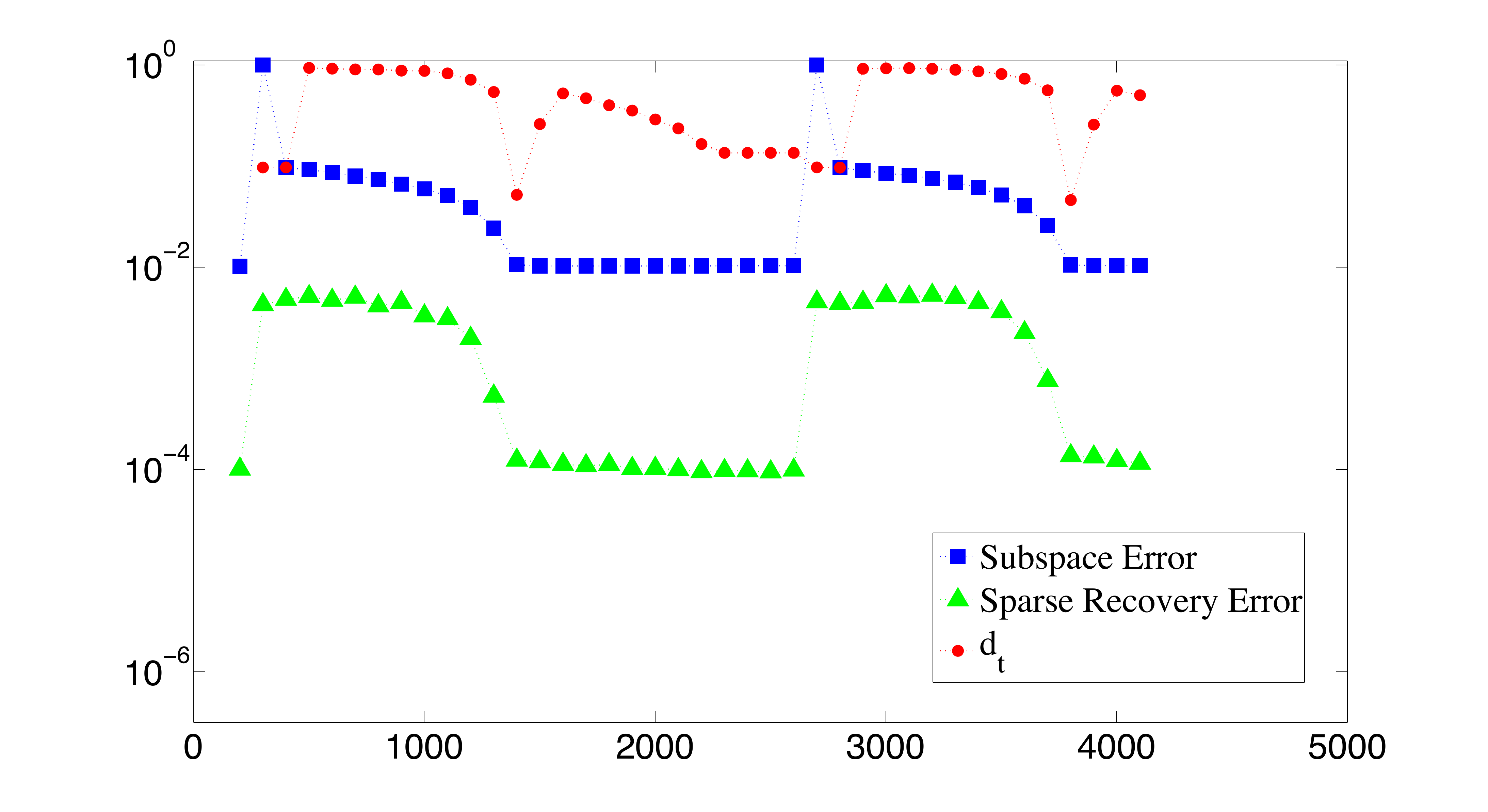}\label{D50}}
\subfigure[$\Delta = 100$]
{\includegraphics[width = 8 cm, height = 4 cm]{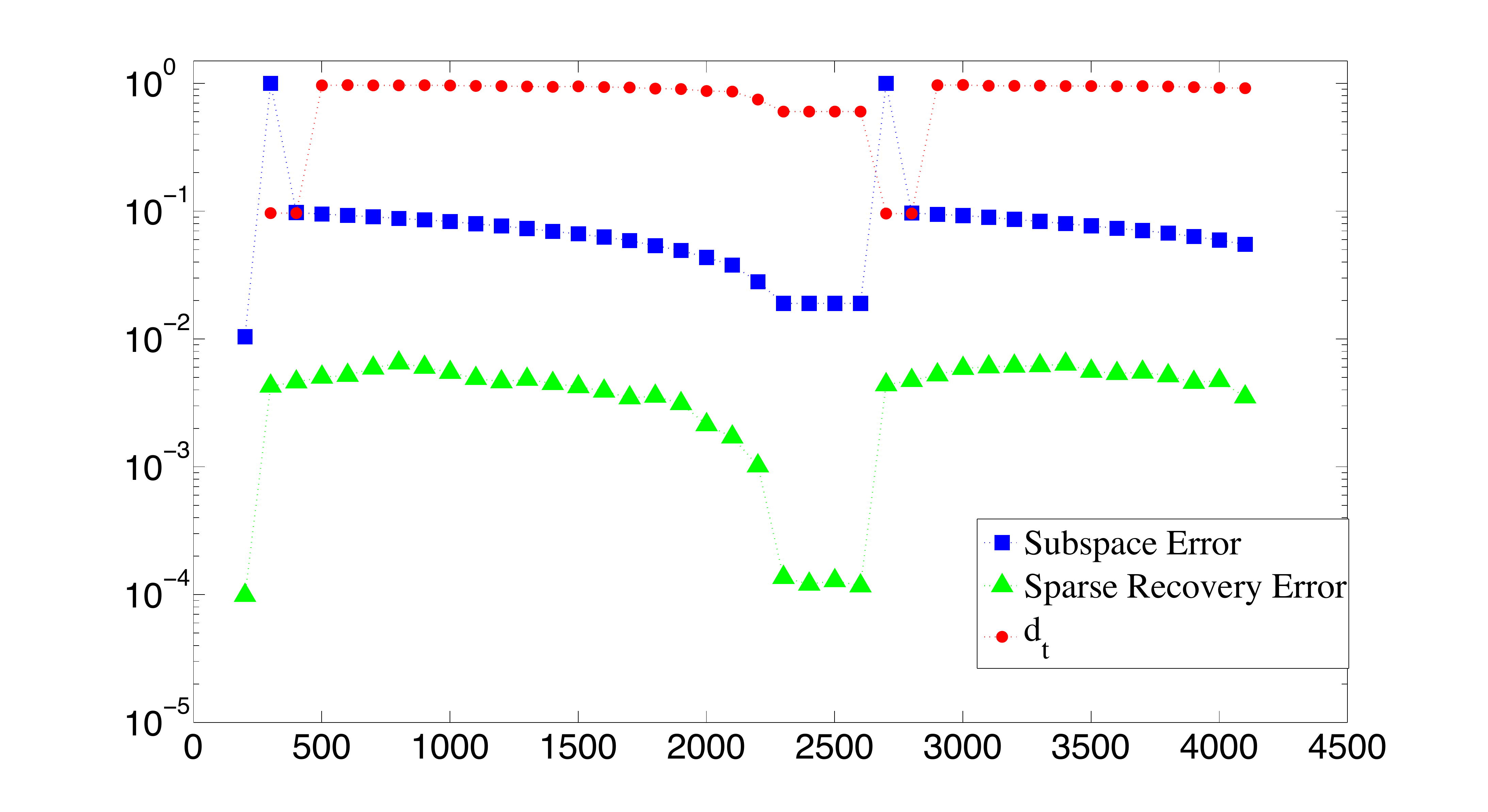}\label{D100}}
\caption{Plots of $d_t$, $SE$ and $e_t$ for simulated data with $r_0 =36$, $s = \max_t |T_t| = 20$ \label{add sim}}
\end{figure*}

\begin{figure}
\centerline{
{\label{del sim a}
\includegraphics[width =\columnwidth]{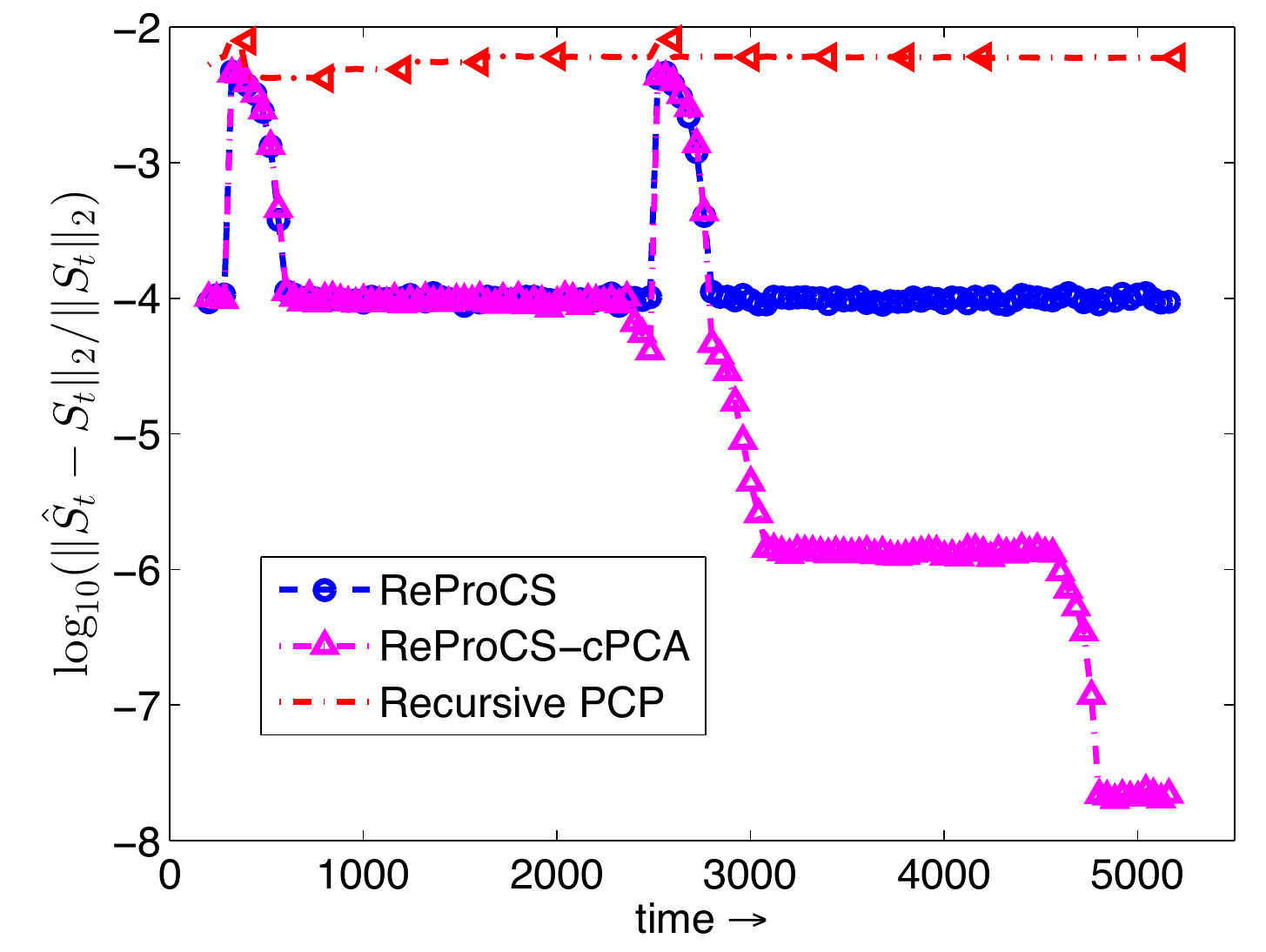}
}
%\subfigure[$\Delta = 50$]{\label{prac sim}
%\includegraphics[width =6cm, height=5cm]{r36_s20_Delta50_reconS}
%%\includegraphics[width =8cm, height=5cm]{r36_s20_Delta50_SE.pdf}
%}
%\subfigure[$\Delta = 50$, $\SE_{(t)}$]{\label{prac sim}
%\includegraphics[width =6cm, height=5cm]{r36_s20_Delta50_SE}
%}
}
\caption{Reconstruction errors of $S_t$ with $r_0 =36$, $s = \max_t |T_t| = 20$.  The times at which PCP is done are marked by red triangles. 
$\Delta:10$, comparing PCP with ReProCS and ReProCS-cPCA. %(b), (c): $\Delta=50$, comparing PCP with ReProCS and practical ReProCS (Algorithm \ref{reprocs_practical}). 
\label{del sim}}
\end{figure}

%\subsubsection{Results}
 Figure \ref{add sim} shows the results of applying Algorithm \ref{reprocs} (ReProCS) to data generated according to the above model. %The algorithm parameters were set as explained in Sec \ref{prac_rep}. %The algorithm parameters were set as follows. We set $\xi_t$ to be the norm of the projection residual at the previous time i.e. $\xi_t = \|\Phi_{(t)} \hat{L}_{t-1} \|_2$.  The support threshold $\omega$ was chosen to be
The model parameters used were $s=20$, $r_0=36$ and $c_{1,\new}= c_{2,\new} = 1$, and each subfigure corresponds to a different value of $\Delta$. Because of the correlated support change, the $2048 \times t$ sparse matrix $\mathcal{S}_t = [S_1 ,S_2,\cdots,S_t]$ is rank deficient in either case, e.g. for Fig. \ref{D2}, $\mathcal{S}_t$ has rank $69,119,169,1219$ at $t=300,400,500,2600$; for Fig. \ref{D10}, $\mathcal{S}_t$ has rank $29,39,49,259$ at $t=300,400,500,2600$. We plot the subspace error $\SE_{(t)}$ and the normalized error for $S_t$, $\frac{\|\hat{S}_t-S_t\|_2}{\|S_t\|_2}$ averaged over 100 Monte Carlo simulations. We also plot the ratio $d_t :=\frac{\|{I_{T_t}}' D_{j,\new,k}\|_2}{\|D_{j,\new,k}\|_2}$. This serves as a proxy for $\kappa_s(D_{j,\new,k})$ (which has exponential computational complexity). In fact, in our proofs, we only need this ratio to be small. %at every $t = t_j + k \alpha-1$.

% We compared against PCP \cite{rpca}. At every $t = t_j + 4 k\alpha$, we solved (\ref{pcp_prob}) with $\lambda = 1/\sqrt{\max(n,t)}$ to recover ${\cal S}_t$ and ${\cal L}_t$. We used the estimates of $S_t$ for the last $4 \alpha$ frames as the final estimates of $\Shat_t$. So, the $\Shat_t$ for $t=t_j+1, \dots t_j + 4 \alpha$ is obtained from PCP done at $t=t_j + 4 \alpha$, the $\Shat_t$ for $t=t_j+4\alpha + 1, \dots t_j + 8 \alpha$ is obtained from PCP done at $t=t_j + 8 \alpha$ and so on.

 As can be seen from Figs. \ref{D2} and \ref{D10}, the subspace error $\SE_{(t)}$ of ReProCS decreased exponentially and stabilized after about $4$ projection PCA update steps. The averaged normalized error for $S_t$ followed a similar trend. %ReProCS(practical) performed similar to ReProCS but stabilized in about $6$ projection PCA update steps.
In Fig. \ref{D10} where $\Delta=10$, the subspace error $\SE_{(t)}$ also decreased but the decrease was a bit slower as compared to Fig. \ref{D2} where $\Delta =2$. %Also, the ratio $\frac{\|{I_{T_t}}' D_{j,\new,k}\|_2}{\|D_{j,\new,k}\|_2}$ was now larger.
 %Therefore, more projection PCA steps were needed till the subsapce error gets stabilized.
%Because of the correlated support change, the error of PCP was larger in both cases. The difference in performance between ReProCS and PCP is larger when $\Delta=50$.

%For Fig. \ref{Delta_10}, we increased $s$ to $100$ and we used $\Delta=10$. A larger $s$ results in a larger $\frac{\|{I_{T_t}}' D_{j,\new,k}\|_2}{\|D_{j,\new,k}\|_2}$  (and larger $\kappa_s(D_{j,\new,k})$). Thus, the rate of decrease of $\SE_{(t)}$ is smaller than that for the previous two figures. %Fig \ref{s20_Delta_10}. %It took $12$ projection PCA update steps in ReProCS and 14 in ReProCS(practical) for the error to stabilize.
%The error of $S_t$ followed a similar trend.% The error for PCP is significantly larger. % (or equivalently larger $\kappa_s(D_{j,\new,k})$).

In Fig. \ref{D100} we set $\Delta = 100.$  In this case $\mathcal{S}_t$ is very low rank.  The rank of $\mathcal{S}_t$ at $t=300,1000,2600$  is $20,27,43$.  We can see here that the subspace error decays rather slowly and does not return all the way to $.01$ within the $K\alpha$ frames.

Finally, if we set $\Delta = \infty$, the ratio $\frac{\|{I_{T_t}}' D_{j,\new,k}\|_2}{\|D_{j,\new,k}\|_2}$ was $1$ always. As a result, the subspace error and hence the reconstruction error of ReProCS did not decrease from its initial value at the subspace change time. %For ReProCS, the average error $\frac{1}{5200} \sum_{t=201}^{5200} \frac{\|\hat{S}_t - S_t\|_2}{\|S_t\|_2}=8.4 \times 10^{-3}$. %The error of PCP was also very high: $\frac{1}{5200} \sum_{t=201}^{5200} \frac{\|\hat{S}_t - S_t\|_2}{\|S_t\|_2}=0.43$.

%In all the above experiments, $T_t$ is highly correlated.
We also did one experiment in which we generated $T_t$ of size $s=100$ uniformly at random from all possible $s$-size subsets of $\{1,2,\dots n\}$. $T_t$ at different times $t$ was also generated independently.  In this case, the reconstruction error of ReProCS is $\frac{1}{5000} \sum_{t=201}^{5200} \frac{\|\hat{S}_t - S_t\|_2}{\|S_t\|_2}=2.8472\times 10^{-4}$. The error for PCP was $3.5 \times 10^{-3}$ which is also quite small.

The data for figure \ref{del sim} was generated the same as above except that we use the more general subspace model that allows for deletion of directions.  Here, for $1\leq t \leq t_1 -1$, $P_{(t)} = P_0$ has rank $r_0$ with  $P_0 = U_{[1,2,\cdots,36]}$. %$\text{rank}(\mathcal{L}_{t}) = r_0$ for all $t \le t_1-1$
 For $t_1  \leq t \leq t_2-1$, $P_{(t)} = P_1 = [P_0\setminus P_{1,\old} \ P_{1,\new}]$ has rank $r_1 = r_0 + c_{1,\new} - c_{1,\old} = 34$ with $ P_{1,\new} = U_{[37]}$ and $P_{1,\old} = U _{[9,18,36]}$.
For $t\geq t_2$, $P_{(t)} = P_2 = [P_1\setminus P_{2,\old} \ P_{2,\new}]$ has rank $r_2 = r_1 + c_{2,\new} - c_{2,\old} = 32$ with $ P_{2,\new} = U_{[38]}$ and $P_{1\old} = U _{[8,17,35]}$.  Again, we average over 100 Monte Carlo simulations.

As can be seen from Figure \ref{del sim}, the normalized sparse recovery error of ReProCS and ReProCS-cPCA decreased exponentially and stabilized. Furthermore, ReProCS-cPCA outperforms over ReProCS greatly when deletion steps are done.  %Figure \ref{prac sim} shows that  ReProCS-practical performed quite similarly to ReProCS.

We also compared against PCP \cite{rpca}. At every $t = t_j + 4 k\alpha$, we solved (\ref{pcp_prob}) with $\lambda = 1/\sqrt{\max(n,t)}$ as suggested in \cite{rpca} to recover ${\cal S}_t$ and ${\cal L}_t$. We used the estimates of $S_t$ for the last $4 \alpha$ frames as the final estimates of $\Shat_t$. So, the $\Shat_t$ for $t=t_j+1, \dots t_j + 4 \alpha$ is obtained from PCP done at $t=t_j + 4 \alpha$, the $\Shat_t$ for $t=t_j+4\alpha + 1, \dots t_j + 8 \alpha$ is obtained from PCP done at $t=t_j + 8 \alpha$ and so on. Because of the correlated support change, the error of PCP was larger in both cases.

\section{Conclusions and Future Work} \label{conc}
In this work, we studied the recursive (online) robust PCA problem, which can also be interpreted as a problem of recursive sparse recovery in the presence of large but structured noise (noise that is dense and lies in a ``slowly changing" low dimensional subspace). We analyzed a novel solution approach called Recursive Projected CS or ReProCS that was introduced in our earlier work \cite{rrpcp_allerton,rrpcp_allerton11,han_tsp}. The ReProCS algorithm that we analyze assumes knowledge of the subspace change model on the $L_t$'s. We showed that, under mild assumptions and a denseness assumption on the currently unestimated subspace, $\Span(D_{j,\new,k})$ (this assumption depends on algorithm estimates), w.h.p., ReProCS can exactly recover the support set of $S_t$ at all times; the reconstruction errors of both $S_t$ and $L_t$ are upper bounded by a time-invariant and small value; and after every subspace change time, w.h.p., the subspace recovery error decays to a small enough value within a finite delay.
%The assumption on the currently unestimated subspace depends on
The most important open question that is being addressed in ongoing work is how to make our result a correctness result, i.e. how to remove the denseness assumption on $D_{j,\new,k}$ (see a forthcoming paper). Two other issues being studied are (i) how to get a result for the correlated $L_t$'s case \cite{zhan_ISIT}, and (ii) how to analyze the ReProCS algorithm when subspace change times are not known. Finally, an open question is how to to bound the sparse recovery error even when the support set is not exactly recovered. The undersampled measurements' case is also being studied \cite{rrpcp_globalsip}.
%In the current work, we needed $\That_t = T_t$ for $t < t_j+k\alpha-1$ to ensure that the $e_t$'s for $t \in \mathcal{I}_{j,k}$ are conditionally independent given $X_{j,k-1}=[a_1, a_2, \dots a_{t_j + (k-1) \alpha - 1}]$. 

\appendices
\renewcommand\thetheorem{\thesection.\arabic{theorem}}

\section{Proofs of Preliminary Lemmas}
\label{appendix prelim}

\textbf{Proof of Lemma \ref{lemma0}}
\begin{proof} %$\|P\|_2 =1$
Because $P$, $Q$ and $\Phat$ are basis matrix, $P'P=I$, ${Q}'Q = I$ and $\Phat'\Phat=I$.
\ben
\item Using $P'P = I$ and $\|M\|_2^2 = \|MM'\|_2$, $\|(I-\Phat{\Phat}')P P'\|_2 =\|(I-\Phat{\Phat}')P\|_2$. Similarly, $\|(I-PP')\Phat {\Phat}'\|_2=\|(I-PP')\Phat\|_2$. Let $D_1 = (I-\Phat{\Phat}')P P'$ and let $D_2=(I-PP')\Phat {\Phat}'$. Notice that $\|D_1\|_2 = \sqrt{\lambda_{\max}(D_1'D_1)} = \sqrt{\|D_1'D_1\|_2}$ and $\|D_2\|_2 = \sqrt{\lambda_{\max}(D_2'D_2)} = \sqrt{\|D_2'D_2\|_2}$. So, in order to show $\|D_1\|_2 = \|D_2\|_2$, it suffices to show that $\|D_1' D_1\|_2 = \|D_2'D_2\|_2$.
    Let $P'\Phat\overset{SVD}{=} U\Sigma V'$. Then, $D_1'D_1 = P (I - P'\Phat{\Phat}'P)P' = P U (I-\Sigma^2) U' P' $ and $D_2'D_2 = \Phat ( I - {\Phat}'PP'\Phat){\Phat}' = \Phat V (I-\Sigma^2) V'{\Phat}'$ are the compact SVD's of $D_1' D_1$ and $D_2' D_2$ respectively. Therefore, $\|D_1' D_1\| = \|D_2'D_2\|_2 = \|I-\Sigma^2\|_2$ and hence $\|(I-\Phat{\Phat}')PP'\|_2 =\|(I - P{P}')\Phat{\Phat}'\|_2$.

\item $\|P{P}' -\Phat {\Phat}'\|_2 = \| PP' - \Phat{\Phat}'PP' + \Phat{\Phat}'PP'-\Phat {\Phat}'\|_2 \leq \| (I- \Phat{\Phat}')PP'\|_2 +
    \|(I-PP')\Phat {\Phat}'\|_2 = 2 \zeta_*$.

\item Since ${Q}'P = 0$, then $\|{Q}'\Phat\|_2 = \|{Q}'(I-P P')\Phat\|_2 \leq \|(I-P P')\Phat\|_2  = \zeta_*$.

\item %Note that $(I-\Phat {\Phat}')' (I-\Phat {\Phat}') =(I-\Phat {\Phat}')$.
Let $M = (I-\Phat {\Phat}') Q)$. Then $M'M = Q'(I-\Phat {\Phat}')Q$ and so $\sigma_i ((I-\Phat {\Phat}') Q) = \sqrt{\lambda_i (Q'(I-\Phat {\Phat}')Q)}$. Clearly, $\lambda_{\max} (Q'(I-\Phat {\Phat}')Q) \leq 1$. By Weyl's Theorem, $\lambda_{\min} (Q'(I-\Phat {\Phat}')Q) \geq 1 - \lambda_{\max} (Q'\Phat {\Phat}'Q) = 1- \|{Q}'\Phat\|_2^2 \geq 1-\zeta_*^2$. Therefore, $\sqrt{1-\zeta_*^2} \leq \sigma_{i}((I-\Phat {\Phat}') Q) \leq 1$.

\een

For the case when $P$ and $\Phat$ are not the same size, the proof of 1 is used, but $\Sigma^2$ becomes $\Sigma\Sigma'$ for $D_1$ and $\Sigma'\Sigma$ for $D_2$.  Since $\Sigma$ is of size $r_1 \times r_2$, $\Sigma\Sigma'$ will be of size $r_1\times r_1$ and $\Sigma'\Sigma$ will be of size $r_2\times r_2$.  Because $r_1 \leq r_2$, every singular value of $D_1'D_1$ will be a singualr value of $D_2'D_2$ (using the SVD as in the proof of 1 above ).  Using the characterization of the matrix 2-norm as the largest singluar value, $\|D_1'D_1\|_2 \leq \|D_2'D_2\|$.

\end{proof}

\textbf{Proof of Lemma \ref{rem_prob}}
\begin{proof}
It is easy to see that $\mathbf{P}(\ecalb, \ecalc) = \E[\mathbb{I}_\calb(X,Y) \mathbb{I}_\calc(X)].$
%$\E[\mathbb{I}_{\calb,\calc}(X,Y)] = \E[\mathbb{I}_\calb(X,Y) \mathbb{I}_\calc(X)].$ := \E[\mathbb{I}_{\calb,\calc}(X,Y)] =
If $\E[\mathbb{I}_{\calb}(X,Y)|X] \ge p$ for all $X \in \calc$, this means that $\E[\mathbb{I}_{\calb}(X,Y)|X] \mathbb{I}_{\calc}(X) \ge p  \mathbb{I}_{\calc}(X) $. This, in turn, implies that
\begin{align*}
\mathbf{P}(\ecalb, \ecalc) = \E[\mathbb{I}_\calb(X,Y) \mathbb{I}_\calc(X)] &= \E[ \E[\mathbb{I}_{\calb}(X,Y)|X] \mathbb{I}_{\calc}(X) ] \\
&\ge p \E[\mathbb{I}_{\calc}(X) ].
\end{align*}
Recall from Definition \ref{probdefs} that $\mathbf{P}(\ecalb|X) = \E[\mathbb{I}_{\calb}(X,Y)|X]$ and $\mathbf{P}(\ecalc)= \E[\mathbb{I}_{\calc}(X) ]$.
Thus, we conclude that if $\mathbf{P}(\ecalb|X) \ge p$ for all $X \in \calc$, then $\mathbf{P}(\ecalb, \ecalc) \ge p \mathbf{P}(\ecalc)$. Using the definition of $\mathbf{P}(\ecalb|\ecalc)$, the claim follows. %Thus, if, $\mathbf{P}(\calc)> 0$, $\mathbf{P}(\calb|\calc) \ge p$. %\E[I_{\calb}(X,Y) I_{\calc}(X) ]
%\ \\ In summary, if $\mathbf{P}(\calb|X) \ge p$ for all $X \in \calc$ and $\mathbf{P}(\calc)> 0$, then $\mathbf{P}(\calb|\calc) \ge p$. %The same is also true for any subset of $\calc$, i.e. if $\calb_3 \subseteq \calc$, then $\mathbf{P}(\calb|\calb_3) \ge p$. = \E[I_{\calb}(X,Y)|X]
\end{proof}

\textbf{Proof of Corollary \ref{hoeffding_nonzero}}

\begin{proof}
\ben
\item Since, for any $X \in {\cal C}$, conditioned on $X$, the $Z_t$'s are independent, the same is also true for $Z_t - g(X)$ for any function of $X$.
Let $Y_t := Z_t - \E(Z_t|X)$. Thus, for any $X \in {\cal C}$, conditioned on $X$, the $Y_t$'s are independent.
Also, clearly $\E(Y_t|X) = 0$. Since for all $X \in \calc$, $\mathbf{P}(b_1 I \preceq Z_t \preceq b_2 I|X)=1$ and since $\lambda_{\max}(.)$ is a convex function, and $\lambda_{\min}(.)$ is a concave function, of a Hermitian matrix, thus $b_1 I \preceq \E(Z_t|X) \preceq b_2 I$ w.p. one for all $X \in \calc$. Therefore, $\mathbf{P}(Y_t^2 \preceq (b_2 -b_1)^2 I|X) = 1$ for all $X \in \calc$.
Thus, for Theorem \ref{hoeffding}, $\sigma^2 = \|\sum_t (b_2 - b_1)^2I\|_2 = \alpha (b_2-b_1)^2$. For any $X \in \calc$, applying Theorem \ref{hoeffding} for $\{Y_t\}$'s conditioned on $X$, we get that, for any $\eps > 0$,
    %$$\mathbf{P}( \lambda_{\max} (\sum_t Y_t) \leq \alpha \epsilon|X) > 1- n\exp(-\frac{\alpha^2 \epsilon^2}{8 \alpha (b_2 -b_1)^2}) \ \text{for all} \ X \in \calc$$
%    which is equivalent to
\begin{multline*}  
\mathbf{P}\left( \lambda_{\max} \left(\frac{1}{\alpha} \sum_t Y_t \right) \leq \epsilon \Big | X \right) > \\1- n\exp\left(\frac{-\alpha \epsilon^2}{8  (b_2 -b_1)^2}\right) \ \text{for all} \ X \in \calc
\end{multline*}
    By Weyl's theorem, $\lambda_{\max} (\frac{1}{\alpha} \sum_t Y_t) = \lambda_{\max} (\frac{1}{\alpha} \sum_t (Z_t - \E(Z_t|X))
    \geq \lambda_{\max} (\frac{1}{\alpha} \sum_t Z_t) +  \lambda_{\min} (\frac{1}{\alpha} \sum_t -\E(Z_t|X))$.
    Since $\lambda_{\min} (\frac{1}{\alpha} \sum_t -\E(Z_t|X)) = - \lambda_{\max} (\frac{1}{\alpha} \sum_t \E(Z_t|X))\geq -b_4$, thus $ \lambda_{\max} (\frac{1}{\alpha} \sum_t Y_t) \geq \lambda_{\max} (\frac{1}{\alpha} \sum_t Z_t) - b_4$.
    Therefore,
 \begin{multline*}
\mathbf{P}\left( \lambda_{\max} \left(\frac{1}{\alpha} \sum_t Z_t \right) \leq b_4 + \epsilon\Big | X \right) > \\ 1- n\exp\left(\frac{-\alpha \epsilon^2}{8  (b_2 -b_1)^2}\right) \ \text{for all} \ X \in \calc
\end{multline*}

%$Y_t^2 = (E(Z_t) - Z_t)^2 \preceq (b_2 -b_1)^2 I$ w.p. one for all $X \in \calc$.  for all $X \in \calc$
\item Let $Y_t = \E(Z_t|X) - Z_t$. As before, $\E(Y_t|X) = 0$ and conditioned on any $X \in {\cal C}$, the $Y_t$'s are independent and $\mathbf{P}(Y_t^2 \preceq (b_2 -b_1)^2 I|X) = 1$.  As before, applying Theorem \ref{hoeffding}, we get that for any $\epsilon >0$,
 \begin{multline*}
\mathbf{P}\left( \lambda_{\max} \left(\frac{1}{\alpha} \sum_t Y_t \right) \leq \epsilon \Big | X \right) > \\
1- n\exp\left(\frac{-\alpha \epsilon^2}{8  (b_2 -b_1)^2} \right) \ \text{for all} \ X \in \calc
\end{multline*}
    By Weyl's theorem, $\lambda_{\max}(\frac{1}{\alpha}\sum_t Y_t) = \lambda_{\max}(\frac{1}{\alpha} \sum_t(\E(Z_t|X) - Z_t)) \geq \lambda_{\min} (\frac{1}{\alpha} \sum_t \E(Z_t|X)) + \lambda_{\max} (\frac{1}{\alpha} \sum_t -Z_t) = \lambda_{\min} (\frac{1}{\alpha} \sum_t \E(Z_t|X)) - \lambda_{\min} (\frac{1}{\alpha} \sum_t Z_t) \ge b_3 - \lambda_{\min} (\frac{1}{\alpha} \sum_t Z_t)$
Therefore, for any $\epsilon >0$,
 \begin{multline*}
\mathbf{P} \left(\lambda_{\min}\left(\frac{1}{\alpha}\sum_t Z_t\right) \geq b_3 -\epsilon \Big| X \right) \\
\geq  1- n \exp\left(\frac{-\alpha \epsilon^2}{8(b_2-b_1)^2}\right) \ \text{for all} \ X \in \calc
\end{multline*}
\een
\end{proof}

\textbf{Proof of Corollary \ref{hoeffding_rec}}

\begin{proof}
Define the dilation of an $n_1 \times n_2$ matrix $M$ as $\text{dilation} (M) := \left[\begin{array}{cc}0 & {M}' \\ M & 0 \\\end{array} \right]$. Notice that this is an $(n_1+n_2) \times (n_1 +n_2)$ Hermitian matrix \cite{tail_bound}. As shown in \cite[equation 2.12]{tail_bound}, %a quick computation gives $ [\text{dilation} (M)]^2 = \left[\begin{array}{cc}{M}'M & 0 \\ 0 & M{M}' \\\end{array} \right]$ and
\bea
\lambda_{\max}(\text{dilation}(M)) = \|\text{dilation} (M)\|_2 = \|M\|_2
\label{dilM}
\eea
Thus, the corollary assumptions imply that $\mathbf{P}(\|\text{dilation} (Z_t)\|_2 \leq b_1 |X) = 1$ for all $X \in \calc$. Thus, $\mathbf{P}(-b_1 I \preceq \text{dilation} (Z_t) \preceq b_1 I | X) = 1$ for all $X \in \calc$.
%Since
%$$\|\text{dilation} (Z_t)\|_2 :=  \sqrt{\lambda_{\max}({\text{dilation} (Z_t)}' \text{dilation} (Z_t))} = \sqrt{\max (\lambda_{\max}^2(\text{dilation} (Z_t)), \lambda_{\min}^2(\text{dilation} (Z_t)))},$$
%therefore, under the same conditioning, $ -b_1 I \preceq \text{dilation} (Z_t) \preceq b_1 I$ w.p. one.
%
Using (\ref{dilM}), the corollary assumptions also imply that $\frac{1}{\alpha}\sum_t \E( \text{dilation}(Z_t) |X) = \text{dilation} ( \frac{1}{\alpha}\sum_t \E(Z_t|X)) \preceq b_2 I$ for all $X \in \calc$.
Finally, $Z_t$'s conditionally independent given $X$,  for any $X \in \calc$, implies that the same thing also holds for $\text{dilation} (Z_t)$'s.
Thus, applying Corollary \ref{hoeffding_nonzero} for the sequence $\{\text{dilation} (Z_t)\}$, we get that, %for any $\epsilon >0$,
\begin{multline*}
\mathbf{P} \left(\lambda_{\max}\left(\frac{1}{\alpha}\sum_t \text{dilation}(Z_t)\right) \leq b_2 + \epsilon \Big| X \right) \geq \\ 1- (n_1+n_2) \exp\left(\frac{-\alpha \epsilon^2}{32 b_1^2}\right) \ \text{for all} \ X \in \calc
\end{multline*}
Using (\ref{dilM}), $\lambda_{\max}(\frac{1}{\alpha}\sum_t \text{dilation}(Z_t)) = \lambda_{\max}(\text{dilation}(\frac{1}{\alpha}\sum_t Z_t))  = \|\frac{1}{\alpha}\sum_t Z_t\|_2$ and this gives the final result. %and so, $$\mathbf{P} (\|\frac{1}{\alpha}\sum_t Z_t\|_2 \leq b_2 + \epsilon |X) \geq 1-(n_1+n_2) \exp(-\frac{\alpha \epsilon^2}{32 b_1^2}) \ \text{for all} \ X \in \calc$$
\end{proof}

%with $ \mathbf{P}(-b_1 I \preceq \text{dilation} (Z_t) \preceq b_1I |X) = 1$ and $\frac{1}{\alpha}\sum_t \E( \text{dilation} (Z_t)|X) \preceq b_2 I$, for all $X \in \calc$. We get that for any $\epsilon >0$,

\textbf{Proof of Lemma \ref{delta_kappa}}
\begin{proof}%[Proof of Lemma \ref{delta_kappa}] \lambda_{\min}(A_T'A_T)  \leq
Let $A = I - PP'$. By definition, $\delta_s(A) := \max\{ \max_{|T| \leq s}(\lambda_{\max}(A_T'A_T) -1),\max_{|T| \leq s} ( 1 - \lambda_{\min} (A_T' A_T))) \} $.
Notice that $A_T'A_T = I - I_T' PP'I_T$. Since $I_T' PP'I_T$ is p.s.d., by Weyl's theorem, $\lambda_{\max}(A_T'A_T) \leq1$.
Since $\lambda_{\max}(A_T'A_T)- 1\leq 0$ while $1 - \lambda_{\min}(A_T'A_T) \geq 0$, thus, %for any set $T$, $\lambda_{\max}(A_T'A_T) -1  \leq  1 - \lambda_{\min} (A_T' A_T)$ and hence
\beq
\delta_s(I - PP') = \max_{|T| \leq s}\Big(1 - \lambda_{\min} ( I - I_T' PP'I_T)\Big) \label{defn_kappa_1}
\eeq
By Definition, $\kappa_s(P) = \max_{|T| \leq s} \frac{\|I_T' P\|_2}{\|P\|_2} =\max_{|T| \leq s} \|I_T' P\|_2$.
Notice that $\|I_T' P\|_2^2 = \lambda_{\max} (I_T' PP'I_T)  =  1-\lambda_{\min} (I - I_T'PP'I_T)$ \footnote{This follows because $B=I_T'PP'I_T$ is a Hermitian matrix.  Let $B = U \Sigma U'$ be its EVD. Since $UU'=I$, $\lambda_{\min}(I-B) = \lambda_{\min}(U(I - \Sigma)U') =\lambda_{\min}(I - \Sigma) = 1 - \lambda_{\max}(\Sigma) = 1-\lambda_{\max}(B)$.}, and so
\beq
\kappa_s^2(P) =\max_{|T| \leq s} \Big(1 - \lambda_{\min} (I - I_T'PP'I_T)\Big)\label{defn_kappa_2}
\eeq
From (\ref{defn_kappa_1}) and (\ref{defn_kappa_2}), we get $ \delta_s(I-PP') = \kappa_s^2 (P) $.
\end{proof}

\section{The Need for Projection PCA}
\label{projpca}

%In definition of $\alpha_\add$, show dependence on $f$: For most distributions, $\gamma_*^2$ is proportional to $\lambda^+$, e.g. in the case of certain $(a_t)_i$'s being $uniform(-\gamma_*, \gamma^*)$, $\gamma_*^2 = 4\lambda^+$. Thus, in this case, $\alpha_\add = ??$.  For most distributions, $\gamma_*^2$ is proportional to $\lambda^+$, e.g. in the case of certain $(a_t)_i$'s being $uniform(-\gamma^*, \gamma^*)$, $\gamma_*^2 = 4\lambda^+$. Thus, for $\gamma_*$ large enough, $\alpha_\add \approx \left\lceil C(\log 6 (K+1) J + 11 \log n) \frac{f}{\zeta^2}$ where $C=4608$.

%?? add as a remark earlier somewhere:  Since our problem definition allows large noise, $L_t$, but assumes slow subspace change, thus the maximum condition number of $Cov[L_t]$, $f$, cannot be small. The reason is as follows. Slow subspace change implies that the projection of $L_t$ along the new directions is initially small. This also means that the minimum eigenvalue of $Cov[L_t]$ will be small for sometime after a subspace change and so $\lambda^-$ needs to be small. Since $\E[\|L_t\|^2] \le r_{\max} \lambda^+ $ and $r_{\max}$ is bounded (low-dimensional), thus, large $L_t$ means $\lambda^+$  needs to be large. As a result $f=\lambda^+/\lambda^-$ will be large.

\subsection{Projection-PCA vs Standard PCA}
%
%in the observed data vector ($\Lhat_t$)  %of $Cov[L_t]$, $f$,
The reason that we cannot use standard PCA for subspace update in our work is because, in our case, the error $e_t= L_t - \Lhat_t$ in the observed data vector $\Lhat_t$ is correlated with the true data vector $L_t$; and the condition number of $\operatorname{Cov}[L_t]$ is large (see Remark \ref{large_f}). In other works that study finite sample PCA, e.g. \cite{nadler} and references therein, the large condition number does not cause a problem because they assume that the error/noise ($e_t$) is uncorrelated with the true data vector ($L_t$). Moreover, $e_t$ or $L_t$ or both are zero mean (which we have too). Thus, the dominant term in the perturbation of the estimated covariance matrix, $(1/\alpha) \sum_t \Lhat_t \Lhat_t'$ w.r.t. the true one is $(1/\alpha) \sum_t e_t e_t'$. For $\alpha$ large enough, the other two terms $(1/\alpha) \sum_t L_t e_t'$ and its transpose are close to zero w.h.p. due to law or large numbers. Thus, the subspace error bound obtained using the $\sin \theta$ theorem and the matrix Hoeffding inequality, will depend, w.h.p., only on the ratio of the maximum eigenvalue of $\operatorname{Cov}[e_t]$ to the smallest eigenvalue of $\operatorname{Cov}[L_t]$. The probability with which this bound holds depends on $f$, however the probability can be made large by increasing the number of data points $\alpha$.
However, in our case, because $e_t$ and $L_t$ are correlated, this strategy does not work. %This issue is explained in detail in Appendix \ref{projpca}.
We explain this below.

In this discussion, we remove the subscript $j$. Also, let $P_*:= P_{j-1}$, $\Phat_*:= \Phat_{j-1}$, $r_* = \rank(P_*)$.
Consider $t=t_j + k\alpha-1$ when the $k^{th}$ projection PCA or PCA is done.
Since the error $e_t = L_t - \Lhat_t$ is correlated with $L_t$, the dominant terms in the perturbation matrix seen by PCA are $(1/ (t_j+ k\alpha)) \sum_{t=1}^{t_j+k\alpha-1} L_t e_t'$ and its transpose, while for projection PCA, they are  $(1/ \alpha) \Phi_0  \sum_{t \in \mathcal{I}_{j,k}} L_t e_t' \Phi_0$ and its transpose.
The magnitude of $L_t$ can be large. The magnitude of $e_t$ is smaller than a constant times that of $L_t$. The constant is less than one but, at $t=t_j+\alpha-1$, it is not negligible. Thus, the norm of the perturbation seen by PCA at this time may not be small. As a result, the bound on the subspace error, $\SE_{(t)}$, obtained by applying the $\sin \theta$ theorem may be more than one (and hence meaningless since by definition $\SE_{(t)} \le 1$). For projection PCA, because of $\Phi_0$, the perturbation is much smaller and hence so is the bound on $\SE_{(t)}$. %, especially at the first PCA step,

Let $\SE_k: = \SE_{(t_j+k\alpha-1)}= \SE_{(t)}$ denote the subspace error for $t \in \mathcal{I}_{j,k}$. Consider $k=1$ first. For PCA, we can show that $\SE_1 \lesssim \check{C} \kappa_s^+ g^+  + \check{C}' f \zeta_*^+ $  for constants $\check{C}, \check{C}'$ that are more than one but not too large. Here $g^+$ is the upper bound on the condition number of $\text{Cov}(a_{t,\new}))$ and it is valid to assume that $g^+$ is small so that $\check{C} \kappa_s^+ g^+ < 1$. However, $f$ is a bound on the maximum condition number of $\text{Cov}(a_t)=\text{Cov}(L_t)$ and this can be large. When it is, the second term may not be less than one. On the other hand, for projection PCA, we have $\SE_k \le \zeta_k + \zeta_* \le \zeta_k^+ + \zeta_*^+$ with $\zeta_*^+  = r \zeta$, and  $\zeta_k^+ \approx \check{C} \kappa_s^+ g^+ \zeta_{k-1}^+  + \check{C}'  f  (\zeta_*^+)^2$ and $\zeta_0^+=1$. Thus $\SE_1 \lesssim \check{C} \kappa_s^+ g^+  + \check{C}' f (\zeta_*^+)^2 +  \zeta_*^+$. The first term in this bound is similar to that of PCA, but the second term is much smaller. The third term is negligibly small. Thus, in this case, it is easier to ensure that the bound is less than one.

Moreover, our goal is to show that within a finite delay after a subspace change time, the subspace error decays down from one to a value proportional to $\zeta$. For projection PCA, this can be done because we can separately bound the subspace error of the existing subspace, $\zeta_*$, and of the newly added one, $\zeta_k$, and then bound the total subspace error, $\SE_{(t)}$, by $\zeta_* + \zeta_k$ for $t \in \mathcal{I}_{j,k}$. Assuming that, by $t=t_j$, $\zeta_*$ is small enough, i.e. $\zeta_* \le r_* \zeta$ with $\zeta < 0.00015/r^2f$, we can show that within $K$ iterations, $\zeta_k$ also becomes small enough so that $\SE_{(t)} \le (r_*+c)\zeta$. However, for PCA, it is not possible to separate the subspace error in this fashion. For $k > 1$, all we can claim is that $\SE_k \lesssim \check{C} \kappa_s^+ f \ \SE_{k-1}$. Since $f$ can be large (larger than $1/\kappa_s^+$), this cannot be used to show that $\SE_k$ decreases with $k$.

%
%\subsection{Why multiple iterations to get a final estimate of $P_{j,\new}$}
%The reason is again because $e_t$ and $L_t$ are correlated and so the dominant perturbation terms are $(1/ \alpha) \sum_t \Phi_0 L_t e_t' \Phi_0$ and its transpose.  The magnitude of $e_t$ is smaller than a constant times that of $L_t$. The constant is less than one, but, at the first projection PCA time, it is not of the order of $\zeta$. With the first projection PCA, we get the first estimate of $P_\new$, $\Phat_{\new,1}$. Using this in the projected CS steps ensures a smaller $e_t$ for the second interval and thus a smaller perturbation seen by the second projection PCA. This results in a reduced perturbation seen by the second projection PCA step and thus an improved estimate $\Phat_{j,\new,2}$. This makes $e_t$ even smaller for the third interval and so on. Within $K$ steps, we can show that it is small enough to ensure that $\SE_{(t)}$ is proportional to $\zeta$.
%
%%If $e_t$ and $L_t$ were independent, this would not happen. An improved estimate $\Phat_{\new,k}$ would not reduce $e_t$ and thus would not reduce the dominant perturbation term for the next projection PCA estimate.
%
%%Within $K$ updates ($K$ chosen as given in Theorem \ref{thm1}), under mild assumptions, it can be shown that the error drops down to a small enough value.

\subsection{Why not use all $k \alpha$ frames at $t=t_j+ k \alpha-1$}
Another possible way to implement projection PCA is to use the past $k \alpha$ estimates $\Lhat_t$ at the $k^{th}$ projection PCA time, $t=t_j+ k \alpha-1$. This may actually result in an improved algorithm. We believe that it can also be analyzed using the approaches developed in this paper.  However, the analysis will be more complicated. We briefly try to explain why. The perturbation seen at $t=t_j+ k \alpha-1$, $\mathcal{H}_k$, will now satisfy $\mathcal{H}_k \approx (1/ (k \alpha) )\sum_{k'=1}^k \sum_{t \in \mathcal{I}_{j,k'}} \Phi_0 (- L_t e_t' - e_t L_t' + e_t e_t' ) \Phi_0$ instead of just being approximately equal to the last ($k'=k$) term. Bounds on each of these terms will hold with a different probability. Thus, proving a lemma similar to Lemma \ref{termbnds} will be more complicated.

\section{Proof of Lemma \ref{termbnds}}
\label{appendix termbnds}

For convenience, we will use $\frac{1}{\alpha}\sum_t$ to denote $\frac{1}{\alpha} \sum_{t \in \mathcal{I}_{j,k}}$. The proof follows using the following key facts and the Hoeffding corollaries.

\begin{fact}\label{keyfacts}
Under the assumptions of Theorem \ref{thm1} the following are true.

\begin{enumerate}

\item The matrices $D_{\new}$, $R_{\new}$, $E_{\new}$, $D_{*}, D_{\new,k-1}$, $\Phi_{k-1}$ are functions of the r.v. $X_{j,k-1}$.  Since $X_{j,k-1}$ is independent of any $a_{t}$ for $t \in  \mathcal{I}_{j,k}$ the same is true for the matrices  $D_{\new}$, $R_{\new}$, $E_{\new}$, $D_{*}, D_{\new,k-1}$, $\Phi_{k-1}$.
    \\
    All terms that we bound for the first two claims of the lemma are of the form $\frac{1}{\alpha} \sum_{t \in \mathcal{I}_{j,k}} Z_t$ where $Z_t= f_1(X_{j,k-1}) Y_t f_2(X_{j,k-1})$, $Y_t$ is a sub-matrix of $a_t a_t'$ and $f_1(.)$ and $f_2(.)$ are functions of $X_{j,k-1}$. Thus, conditioned on  $X_{j,k-1}$, the $Z_t$'s are mutually independent. (Recall that we assume independence of the $a_t$'s.  \label{X_at_indep} \label{zt_indep}
    \\
    All the terms that we bound for the third claim contain $e_t$. Using Lemma \ref{cslem}, conditioned on $X_{j,k-1}$, $e_t$ satisfies (\ref{etdef0}) w.p. one whenever $X_{j,k-1} \in \Gamma_{j,k-1}$. Using (\ref{etdef0}), it is easy to see that all these terms are also of the above form whenever $X_{j,k-1} \in \Gamma_{j,k-1}$.
    \\
    Thus, conditioned on $X_{j,k-1}$, the $Z_t$'s for all the above terms are mutually independent, whenever $X_{j,k-1} \in \Gamma_{j,k-1}$.

%\item  $X_{j,k-1}$ is independent of any $a_{t}$ for $t \in  \mathcal{I}_{j,k}$ , and hence the same is true for the matrices  $D_{\new}$, $R_{\new}$, $E_{\new}$, $D_{*}, D_{\new,k-1}$, $\Phi_{k-1}$. Also, $a_{t}$'s for different $t \in \mathcal{I}_{j,k}$ are mutually independent. Thus, conditioned on  $X_{j,k-1}$, the $Z_t$'s defined above are mutually independent.

\item It is easy to see that $\|\Phi_{k-1} P_*\|_2 \le \zeta_*$, $\zeta_0 = \|D_\new\|_2 \le 1$,  $\Phi_0 D_\new = \Phi_0'D_\new = D_\new$,  $\|R_\new\| \le  1$, $\|(R_\new)^{-1}\| \le 1/\sqrt{1 - \zeta_*^2}$, ${E_{\new,\perp}}' D_\new = 0$, and $\|{E_{\new}}' \Phi_0 e_t\| =\|(R_\new')^{-1} D_\new' \Phi_0 e_t\| = \|(R_\new)^{-1} D_\new' e_t\| \le \|(R_\new')^{-1} D_\new' I_{T_t}\| \|e_t\| \le \frac{\kappa_s(D_\new)}{\sqrt{1 - \zeta_*^2}}\|e_t\|$. The bounds on $\|R_\new\|$ and $\|(R_\new)^{-1}\|$ follow using Lemma \ref{lemma0} and the fact that $\sigma_{i}(R_\new) = \sigma_{i}(D_\new)$.
 \label{rnew}

\item $X_{j,k-1} \in \Gamma_{j,k-1}$ implies that
\begin{enumerate}
\item $\zeta_{j,*} \le \zeta_*^+$ (By definition of $\Gamma_{j,k-1}$ (Definition \ref{Gamma_def}))
\item $\zeta_{k-1}\leq \zeta_{k-1}^+ \leq 0.6^{k-1} + 0.4c\zeta$  (This follows by the definition of $\Gamma_{j,k-1}$ and Lemma \ref{expzeta}.)
\end{enumerate}
\label{leqzeta}

\item Item \ref{leqzeta} implies that conditioned on $X_{j,k-1} \in \Gamma_{j,k-1}$
\ben
\item $\kappa_s(D_\new) \le \kappa_s^+$ (follows by Lemma \ref{Dnew0_lem}),
\item $\lambda_{\min}(R_{\new}{R_{\new}}') \geq 1-(\zeta_*^+)^2$ (follows from Lemma \ref{lemma0} and the fact that $\sigma_{\min}(R_\new) = \sigma_{\min}(D_\new)$),
\item $\|{I_{T_t}}' \Phi_{k-1} P_* \|_2 \leq \|\Phi_{k-1} P_* \|_2 \leq \zeta_{j,*} \leq \zeta_{j,*}^+$,
\item $\|{I_{T_t}}'D_{\new,k-1}\|_2 \leq \kappa_{s}(D_{\new,k-1}) \zeta_{k-1} \leq \kappa_s^+ \zeta_{k-1}^+$.
\een
\label{X_in_Gamma}

\item By Weyl's theorem (Theorem \ref{weyl}),  for a sequence of matrices $B_t$, $\lambda_{\min}(\sum_t B_t) \ge \sum_t \lambda_{\min}(B_t)$ and $\lambda_{\max}(\sum_t B_t) \le \sum_t \lambda_{\max}(B_t)$.

\end{enumerate}

\end{fact}

\begin{proof}

Consider $A_k := \frac{1}{\alpha} \sum_t {E_{\new}}' \Phi_{0} L_t {L_t}' \Phi_{0} E_{\new}$. Notice that ${E_{\new}}' \Phi_{0} L_t = R_{\new} a_{t,\new} + {E_{\new}}' D_* a_{t,*}$. Let $Z_t = R_{\new} a_{t,\new} {a_{t,\new}}' {R_{\new}}'$ and let $Y_t = R_{\new} a_{t,\new}{a_{t,*}}' {D_*}' {E_{\new}}' +  {E_{\new}}' D_* a_{t,*}{a_{t,\new}}' {R_{\new}}'$, then
\beq
A_k  \succeq \frac{1}{\alpha} \sum_t Z_t + \frac{1}{\alpha} \sum_t Y_t \label{lemmabound_1}
\eeq

Consider $\sum_t Z_t = \sum_t R_{\new} a_{t,\new} {a_{t,\new}}' R_{\new}'$.
\begin{enumerate}
\item Using item \ref{zt_indep} of Fact \ref{keyfacts}, the $Z_t$'s are conditionally independent given $X_{j,k-1}$.

\item Using item \ref{X_at_indep}, Ostrowoski's theorem (Theorem \ref{ost}),  and item \ref{X_in_Gamma},  for all $X_{j,k-1} \in \Gamma_{j,k-1}$, $\lambda_{\min}\left( \E(\frac{1}{\alpha}\sum_t Z_t|X_{j,k-1})\right) = \lambda_{\min}\left( R_{\new} \frac{1}{\alpha}\sum_t \E(a_{t,\new}{a_{t,\new}}') {R_{\new}}'\right) \ge \lambda_{\min} \left(R_{\new} {R_{\new}}'\right)\lambda_{\min} \left(\frac{1}{\alpha}\sum_t \E(a_{t,\new}{a_{t,\new}}')\right) \geq (1-(\zeta_{j,*}^+)^2)\lambda_{\new,k}^-$.

\item Finally, using items \ref{rnew} and the bound on $\| a_t \|_{\infty}$ from the model, conditioned on $X_{j,k-1}$,  $0 \preceq Z_t  \preceq c \gamma_{\new,k}^2 I \preceq c  \max\left((1.2)^{2k} \gamma_{\new}^2, \gamma_*^2 \right) I$ holds w.p. one for all $X_{j,k-1} \in \Gamma_{j,k-1}$.
\end{enumerate}

Thus, applying Corollary \ref{hoeffding_nonzero} with $\epsilon = \frac{c\zeta\lambda^-}{24}$, we get
\begin{multline}
\mathbf{P}\left(\lambda_{\min} \left(\frac{1}{\alpha} \sum_t Z_t\right)
\geq   (1-(\zeta_*^+)^2)\lambda_{\new,k}^-  \right.\\
\left. - \frac{c\zeta\lambda^-}{24} \bigg| X_{j,k-1}\right) \geq  \\
1- c \exp \left(\frac{-\alpha \zeta^2 (\lambda^-)^2}{8 \cdot 24^2 \cdot \min(1.2^{4k} \gamma_{\new}^4, \gamma_*^4)}\right)
 \label{lemma_add_A1}
\end{multline}
 for all $X_{j,k-1} \in  \Gamma_{j,k-1}$.

Consider $Y_t = R_{\new} a_{t,\new}{a_{t,*}}' {D_*}' {E_{\new}} +  {E_{\new}}' D_* a_{t,*}{a_{t,\new}}' {R_{\new}}'$.
\begin{enumerate}
\item  Using item \ref{zt_indep}, the $Y_t$'s are conditionally independent given $X_{j,k-1}$.

\item Using item \ref{X_at_indep} and the fact that $a_{t,\new}$ and $a_{t,*}$ are mutually uncorrelated,  $\E\left(\frac{1}{\alpha}\sum_t Y_t|X_{j,k-1}\right) = 0$ for all $X_{j,k-1} \in \Gamma_{j,k-1}$.

\item Using  the bound on $\| a_t \|_{\infty}$, items \ref{rnew}, \ref{X_in_Gamma}, and Fact \ref{constants}, conditioned on $X_{j,k-1}$, $\|Y_t\| \le 2\sqrt{c r} \zeta_*^+ \gamma_* \gamma_{\new,k}  \leq 2\sqrt{c r} \zeta_*^+ \gamma_*^2 \le  2$ holds w.p. one for all $X_{j,k-1} \in \Gamma_{j,k-1}$.
\end{enumerate}

Thus, under the same conditioning, $-b I \preceq Y_t  \preceq b I$ with $b = 2$ w.p. one.

Thus, applying Corollary \ref{hoeffding_nonzero} with $\epsilon = \frac{c\zeta\lambda^-}{24}$, we get
\begin{multline}
\mathbf{P}\left(\lambda_{\min} \left(\frac{1}{\alpha} \sum_t Y_t \right) \geq  \frac{-c\zeta\lambda^-}{24} \Big| X_{j,k-1} \right) \geq \\
1- c \exp \left( \frac{-\alpha c^2 \zeta^2(\lambda^-)^2} {8 \cdot 24^2 \cdot (2b)^2}\right)  \ \text{for all $X_{j,k-1} \in  \Gamma_{j,k-1}$}
\label{lemma_add_A2}
\end{multline}

Combining (\ref{lemmabound_1}), (\ref{lemma_add_A1}) and (\ref{lemma_add_A2}) and using the union bound, $\mathbf{P} (\lambda_{\min}(A_k) \geq \lambda_{\new,k}^-(1 - (\zeta_*^+)^2) - \frac{c\zeta\lambda^-}{12}| X_{j,k-1}) \geq 1-p_a(\alpha,\zeta)  \ \text{for all $X_{j,k-1} \in  \Gamma_{j,k-1}$}$.
The first claim of the lemma follows by using $\lambda_{\new,k}^- \ge \lambda^-$ and then applying Lemma \ref{rem_prob} with $X \equiv X_{j,k-1}$ and $\calc \equiv \Gamma_{j,k-1}$.

Now consider $A_{k,\perp} := \frac{1}{\alpha} \sum_t {E_{\new,\perp}}' \Phi_{0} L_t {L_t}' \Phi_{0} E_{\new,\perp}$. Using item \ref{rnew}, ${E_{\new,\perp}}' \Phi_{0} L_t = {E_{\new,\perp}}' D_* a_{t,*}$. Thus, $A_{k,\perp} = \frac{1}{\alpha} \sum_t Z_t$ with  $Z_t={E_{\new,\perp}}' D_* a_{t,*} {a_{t,*}}' {D_*}' E_{\new,\perp}$ which is of size $(n-c)\times (n-c)$.
Using the same ideas as above we can show that $0 \preceq Z_t \preceq r (\zeta_*^+)^2 \gamma_*^2 I \preceq \zeta I$ and $\E\left(\frac{1}{\alpha}\sum_t Z_t|X_{j,k-1}\right) \preceq (\zeta_*^+)^2 \lambda^+ I$.  Thus by Corollary \ref{hoeffding_nonzero} with $\epsilon = \frac{c\zeta\lambda^-}{24}$ and Lemma \ref{rem_prob} the second claim follows.

Using the expression for $\mathcal{H}_k$ given in Definition \ref{defHk}, it is easy to see that
\begin{align}
\|\mathcal{H}_k \|_2 &\leq  \max\{ \|H_k\|_2, \|H_{k,\perp}\|_2 \} + \|B_k\|_2 \nn \\
&\leq \Big\|\frac{1}{\alpha} \sum_t e_t {e_t}'\Big\|_2 +  \max(\|T2\|_2, \|T4\|_2) + \|B_k\|_2
\label{add_calH1}
\end{align}
where $T2:= \frac{1}{\alpha} \sum_t  {E_{\new}}' \Phi_{0}( L_t {e_t}' + e_t {L_t}')\Phi_{0} E_{\new}$ and $T4 :=\frac{1}{\alpha} \sum_t {E_{\new,\perp}}'\Phi_{0} (L_t {e_t}' + {e_t}'L_t)\Phi_{0} E_{\new,\perp}$. The second inequality follows by using the facts that (i) $H_k = T1 - T2$ where $T1 := \frac{1}{\alpha} \sum_t {E_{\new}}' \Phi_{0} e_t {e_t}'\Phi_{0} E_{\new}$, (ii) $H_{k,\perp} = T3 - T4$ where $T3 := \frac{1}{\alpha} \sum_t {E_{\new,\perp}}'\Phi_0 e_t {e_t}'\Phi_0  E_{\new,\perp}$, and (iii) $\max(\|T1\|_2, \|T3\|_2) \le \|\frac{1}{\alpha} \sum_t e_t {e_t}'\|_2$.
Next, we obtain high probability bounds on each of the terms on the RHS of (\ref{add_calH1}) using the Hoeffding corollaries.

Consider $\|\frac{1}{\alpha} \sum_t e_t {e_t}'\|_2$. Let $Z_t = e_t {e_t}'$. %and the bound on $\| a_t \|_{\infty}$
\begin{enumerate}
\item Using item \ref{zt_indep}, conditioned on $X_{j,k-1}$, the various $Z_t$'s in the summation are independent, for all $X_{j,k-1} \in \Gamma_{j,k-1}$.
\item Using item \ref{X_in_Gamma}, and the bound on $\| a_t \|_{\infty}$,
conditioned on $X_{j,k-1}$, $0 \preceq Z_t \preceq b_1 I$ w.p. one for all $X_{j,k-1} \in  \Gamma_{j,k-1}$. Here $b_1:=(\kappa_s^+ \zeta_{k-1}^+ \phi^+ \sqrt{c} \gamma_{\new,k} + \zeta_*^+ \phi^+ \sqrt{r} \gamma_*)^2$.
\item Also using item \ref{X_in_Gamma}, $0 \preceq \frac{1}{\alpha} \sum_t \E(Z_t|X_{j,k-1}) \preceq b_2I$, with  $b_2:= (\kappa_s^+)^2 (\zeta_{k-1}^+)^2 (\phi^+)^2 \lambda_{\new,k}^+ + (\zeta_*^+)^2 (\phi^+)^2 \lambda^+$ for all $X_{j,k-1} \in  \Gamma_{j,k-1}$.
\end{enumerate}
Thus, applying Corollary \ref{hoeffding_nonzero} with $\epsilon = \frac{c\zeta\lambda^-}{24}$,
\begin{multline}
\mathbf{P} \left( \Big\|\frac{1}{\alpha} \sum_t  e_t {e_t}' \Big\|_2 \leq b_2  + \frac{c\zeta\lambda^-}{24} \Big| X_{j,k-1} \right) \geq \\
1- n \exp\left(\frac{-\alpha c^2 \zeta^2 (\lambda^-)^2}{ 8 \cdot 24^2 b_1^2}\right)  \ \text{for all $X_{j,k-1} \in  \Gamma_{j,k-1}$}
\label{add_etet}
\end{multline}

Consider $T2$. Let $Z_t: = {E_{\new}}' \Phi_{0} (L_t {e_t}' + e_t{L_t}')\Phi_{0} E_{\new}$ which is of size $c \times c$. Then $T2 = \frac{1}{\alpha} \sum_t Z_t$.
\begin{enumerate}
\item Using item \ref{zt_indep}, conditioned on $X_{j,k-1}$, the various $Z_t$'s used in the summation are mutually independent, for all $X_{j,k-1} \in \Gamma_{j,k-1}$.
Using item \ref{rnew}, ${E_{\new}}'\Phi_{0} L_t  = R_{\new} a_{t,\new} + {E_\new}' D_* a_{t,*}$ and ${E_{\new}}'\Phi_{0} e_t = ({R_{\new}}')^{-1} {D_{\new}}' e_t$.

\item Thus, using items \ref{rnew}, \ref{X_in_Gamma}, and the bound on $\| a_t \|_{\infty}$, it follows that conditioned on $X_{j,k-1}$, $\|Z_t\|_2 \leq 2 \tilde{b}_3 \leq 2 b_3$ w.p. one for all $X_{j,k-1} \in  \Gamma_{j,k-1}$. Here, $\tilde{b}_3:=  \frac{\kappa_s^+}{\sqrt{1-(\zeta_*^+)^2}} \phi^+( \kappa_s^+ \zeta_{k-1}^+ \sqrt{c} \gamma_{\new,k} +   \sqrt{r} \zeta_*^+ \gamma_*)(\sqrt{c} \gamma_{\new,k} + \sqrt{r} \zeta_*^+ \gamma_*)$ and
$b_3:= \frac{1}{\sqrt{1-(\zeta_*^+)^2}}(
\phi^+  c {\kappa_s^+}^2 \zeta_{k-1}^+  \gamma_{\new,k}^2
+ \phi^+ \sqrt{rc} {\kappa_s^+}^2   \zeta_{k-1}^+ \zeta_*^+  \gamma_{\new,k} \gamma_*
+ \phi^+ \sqrt{rc} \kappa_s^+  \zeta_*^+ \gamma_* \gamma_{\new,k}
+ \phi^+  r {\zeta_*^+}^2 \gamma_*^2)$.

\item Also, $\|\frac{1}{\alpha} \sum_t \E(Z_t|X_{j,k-1})\|_2 \leq 2 \tilde{b}_4 \leq 2 b_4$ where $\tilde{b}_4: = \frac{\kappa_s^+}{\sqrt{1-(\zeta_*^+)^2}}\phi^+ \kappa_s^+ \zeta_{k-1}^+ \lambda_{\new,k}^+ +  \frac{\kappa_s^+}{\sqrt{1-(\zeta_*^+)^2}} \phi^+ (\zeta_*^+)^2  \lambda^+$ and $b_4:=\frac{\kappa_s^+}{\sqrt{1-(\zeta_*^+)^2}}\phi^+ \kappa_s^+ \zeta_{k-1}^+ \lambda_{\new,k}^+ +  \frac{1}{\sqrt{1-(\zeta_*^+)^2}} \phi^+ (\zeta_*^+)^2  \lambda^+$.
\end{enumerate}
Thus, applying Corollary \ref{hoeffding_rec} with $\epsilon = \frac{c\zeta\lambda^-}{24}$,
\begin{multline}
\mathbf{P}\left( \|T2\|_2 \leq 2 b_4  + \frac{c\zeta\lambda^-}{24}\Big|X_{j,k-1}\right) \\
\geq 1- c \exp\left(\frac{-\alpha c^2 \zeta^2 (\lambda^-)^2}{32 \cdot 24^2 \cdot 4 b_3^2}\right) \ \text{for all $X_{j,k-1} \in   \Gamma_{j,k-1}$} \nn
\end{multline}

Consider $T4$. Let $Z_t: = {E_{\new,\perp}}'\Phi_{0} (L_t {e_t}' + e_t{L_t}')\Phi_{0} E_{\new,\perp}$ which is of size $(n-c)\times (n-c)$. Then $T4 = \frac{1}{\alpha} \sum_t Z_t$.
\begin{enumerate}
\item Using item \ref{zt_indep}, conditioned on $X_{j,k-1}$, the various $Z_t$'s used in the summation are mutually independent, for all $X_{j,k-1} \in  \Gamma_{j,k-1}$.
 Using item \ref{rnew}, ${E_{\new,\perp}}'\Phi_{0} L_t ={E_{\new,\perp}}' D_* a_{t,*}$.
\item Thus, conditioned on $X_{j,k-1}$, $\|Z_t\|_2 \leq  2b_5$ w.p. one for all $X_{j,k-1} \in  \Gamma_{j,k-1}$. Here $b_5:=   \phi^+ r(\zeta_*^+)^2 \gamma_*^2    +  \phi^+ \sqrt{rc} \kappa_s^+ \zeta_*^+  \zeta_{k-1}^+ \gamma_* \gamma_{\new,k}$
This follows using items \ref{X_in_Gamma} and the bound on $\| a_t \|_{\infty}$.
\item  Also, $\|\frac{1}{\alpha} \sum_t \E(Z_t|X_{j,k-1})\|_2 \leq 2 b_6, \ b_6:= \phi^+ (\zeta_*^+)^2  \lambda^+$.
\end{enumerate}

Applying Corollary \ref{hoeffding_rec} with $\epsilon = \frac{c\zeta\lambda^-}{24}$,
\begin{multline}
\mathbf{P} \left( \|T4 \|_2 \leq 2b_6  + \frac{c\zeta\lambda^-}{24} \Big| X_{j,k-1}\right) \geq \\
1- (n-c) \exp\left(\frac{-\alpha c^2 \zeta^2 (\lambda^-)^2}{32 \cdot 24^2 \cdot 4 b_5^2}\right) \ \text{for all $X_{j,k-1} \in   \Gamma_{j,k-1}$} \nn
\end{multline}

Consider $\max(\|T2 \|_2,\|T4 \|_2)$. Since $b_3 > b_5$ (follows because $\zeta_{k-1}^+ \le 1$) and $b_4 > b_6$, so $2b_6  + \frac{c\zeta\lambda^-}{24} < 2b_4  + \frac{c\zeta\lambda^-}{24}$ and $1- (n-c) \exp\left(\frac{-\alpha c^2 \zeta^2 (\lambda^-)^2}{8\cdot 24^2 \cdot 4 b_5^2}\right) > 1- (n-c) \exp\left(\frac{-\alpha c^2 \zeta^2 (\lambda^-)^2}{8 \cdot 24^2 \cdot 4 b_3^2}\right)$. Therefore, for all $X_{j,k-1} \in   \Gamma_{j,k-1}$, $\mathbf{P}\left( \|T4 \|_2 \leq 2 b_4 + \frac{c\zeta\lambda^-}{24} \Big| X_{j,k-1} \right) \geq 1- (n-c) \exp\left(\frac{-\alpha c^2 \zeta^2 (\lambda^-)^2}{32 \cdot 24^2 \cdot 4 b_3^2}\right)$.

By the union bound, for all $X_{j,k-1} \in  \Gamma_{j,k-1}$,
\begin{multline}
\mathbf{P} \left( \max(\|T2 \|_2,\|T4 \|_2)\leq 2b_4 +  \frac{c\zeta\lambda^-}{24} \Big|X_{j,k-1}\right) \geq \\
1- n \exp\left(\frac{-\alpha c^2\zeta^2 (\lambda^-)^2}{32 \cdot 24^2 \cdot 4b_3^2}\right)
\label{add_maxT}
\end{multline}

Consider $\|B_k\|_2$. Let $Z_t := {E_{\new,\perp}}'\Phi_{0} (L_t-e_t)({L_t}'-{e_t}')\Phi_{0} E_{\new}$ which is of size $(n-c)\times c$. Then $B_k = \frac{1}{\alpha} \sum_t Z_t$.
Using item \ref{rnew}, ${E_{\new,\perp}}'\Phi_{0} (L_t-e_t) = {E_{\new,\perp}}'( D_{*} a_{t,*} - \Phi_{0} e_t)$,  ${E_{\new}}' \Phi_{0} (L_t - e_t) = R_{\new} a_{t,\new}+ {E_{\new}}' D_* a_{t,*} + (R_\new')^{-1} D_\new' e_t$. Also, $\|Z_t\|_2 \leq b_7$ w.p. one for all $X_{j,k-1} \in  \Gamma_{j,k-1}$ and  $\|\frac{1}{\alpha} \sum_t \E(Z_t|X_{j,k-1})\|_2 \leq b_8$ for all $X_{j,k-1} \in  \Gamma_{j,k-1}$.
 Here
\begin{align*}
b_7 := &(\sqrt{r} \zeta_*^+  (1+ \phi^+)\gamma_*  + (\kappa_s^+) \zeta_{k-1}^+ \phi^+ \sqrt{c} \gamma_{\new,k})\cdot \\
&\left( \sqrt{c} \gamma_{\new,k} + \sqrt{r} \zeta_*^+ \left(1+\frac{1}{\sqrt{1-(\zeta_*^+)^2}} \kappa_s^+ \phi^+\right) \gamma_*  + \right.\\
& \left. \frac{1}{\sqrt{1-(\zeta_*^+)^2}} {\kappa_s^+}^2 \zeta_{k-1}^+ \phi^+ \sqrt{c} \gamma_{\new,k}\right)
\end{align*}
and
\begin{align*}
b_8 := &\left(\kappa_s^+ \zeta_{k-1}^+ \phi^+ + \frac{1}{\sqrt{1-(\zeta_*^+)^2}}(\kappa_s^+)^3 (\zeta_{k-1}^+)^2 (\phi^+)^2\right) \lambda_{\new,k}^+  \\ 
&+ (\zeta_*^+)^2 \left(1 + \phi^+ + \frac{1}{\sqrt{1-(\zeta_*^+)^2}}\kappa_s^+ \phi^+ + \right. \\
 &\hspace{1.7in}\left. \frac{1}{\sqrt{1-(\zeta_*^+)^2}}\kappa_s^+(\phi^+)^2 \right) \lambda^+
\end{align*}

Thus, applying Corollary \ref{hoeffding_rec} with $\epsilon=\frac{c\zeta\lambda^-}{24}$,
\begin{multline}
\mathbf{P} \left(\|B_k\|_2 \leq b_8 + \frac{c\zeta\lambda^-}{24} \Big| X_{j,k-1}\right) \geq \\
 1 - n \exp\left(\frac{-\alpha c^2 \zeta^2 (\lambda^-)^2}{32 \cdot 24^2 b_7^2}\right) \ \text{for all $X_{j,k-1} \in   \Gamma_{j,k-1}$}
\label{Bk}
\end{multline}

Using (\ref{add_calH1}), (\ref{add_etet}), (\ref{add_maxT}) and (\ref{Bk}) and the union bound,  for any $X_{j,k-1} \in  \Gamma_{j,k-1}$,
\begin{align*}
&\mathbf{P} \left(\|\mathcal{H}_k\|_2 \leq b_9 + \frac{c\zeta\lambda^-}{8} \Big|X_{j,k-1}\right)\geq \\
&  1-n \exp\left(\frac{-\alpha c^2 \zeta^2 (\lambda^-)^2}{8 \cdot 24^2 b_1^2}\right)- n \exp\left(\frac{-\alpha c^2\zeta^2 (\lambda^-)^2}{32\cdot 24^2 \cdot 4 b_3^2}\right) \\
&\hspace{1.7in} - n\exp\left(\frac{-\alpha c^2 \zeta^2 (\lambda^-)^2 }{32 \cdot 24^2 b_7^2}\right)
\end{align*}
where
\begin{align*}
b_9 &:= b_2 +2b_4+ b_8 \\
&= \left( (\frac{2(\kappa_s^+)^2 \phi^+}{\sqrt{1-(\zeta_*^+)^2}} + \kappa_s^+ \phi^+ )\zeta_{k-1}^+  + \right.\\
&\hspace{.8in}\left. ( (\kappa_s^+)^2 (\phi^+)^2  +  \frac{(\kappa_s^+)^3 (\phi^+)^2 }{\sqrt{1-(\zeta_*^+)^2}})  (\zeta_{k-1}^+)^2
\right)  \lambda_{\new,k}^+ \  \\
& + \left((\phi^+)^2 +  \frac{2\phi^+ }{\sqrt{1-(\zeta_*^+)^2}}
+ 1 + \phi^+ + \right.\\
&\hspace{1in}\left.\frac{\kappa_s^+ \phi^+}{\sqrt{1-(\zeta_*^+)^2}} + \frac{\kappa_s^+(\phi^+)^2 }{\sqrt{1-(\zeta_*^+)^2}} \right) (\zeta_*^+)^2   \lambda^+
\end{align*}

Using $\lambda_{\new,k}^- \ge \lambda^-$ and $f := \lambda^+/\lambda^-$, $b_9 + \frac{c\zeta\lambda^-}{8} \le \lambda_{\new,k}^- (b + 0.125c\zeta)$ where $b$ is defined in Definition \ref{zetakplus}.
Using Fact \ref{constants} and substituting $\kappa_s^+ = 0.15$, $\phi^+=1.2$, one can upper bound $b_1$, $b_3$ and $b_7$ and show that the above probability is lower bounded by $1- p_c(\alpha,\zeta)$. Finally, applying Lemma \ref{rem_prob}, the third claim of the lemma follows.

\end{proof}

\section{Proof of Lemma \ref{bound_R}}\label{proof_lem_bound_R}
\begin{proof} [Proof of Lemma \ref{bound_R}]
\ben
\item The first claim follows because $\|D_{\text{det},k}\|_2 = \|\Psi_{k-1} G_{\text{det},k}\|_2  = \| \Psi_{k-1} [G_1 G_2 \cdots G_{k-1}]\|_2 \leq \sum_{k_1=1}^{k-1}\|\Psi_{k-1} G_{k_1}\|_2
\leq \sum_{k_1=1}^{k-1} \|\Psi_{k_1} G_{k_1}\|_2 =  \sum_{k_1=1}^{k-1} \tilde{\zeta}_{k_1} \leq  \sum_{k_1=1}^{k-1} \tilde{c}_{k_1} \zeta \leq r\zeta$. The first inequality follows by triangle inequality. The second one follows because $\hat{G}_1,\cdots,\hat{G}_{k-1}$ are mutually orthonormal and so $\Psi_{k-1} = \prod_{k_2=1}^{k-1}(I - \hat{G}_{k_2}{\hat{G}_{k_2}}')$.
%  $$
%    \|D_{\text{det},k}\|_2 \leq \sum_{k_1=1}^{k-1} \| \prod_{k_2=1}^{k_1} ( I -\hat{G}_{k_2}{\hat{G}_{k_2}}') G_{k_1}\|_2
%    = \sum_{k_1=1}^{k-1} \|\Psi_{k_1} G_{k_1}\|_2
%    = \sum_{k_1=1}^{k-1} \tilde{\zeta}_{k_1} \leq  \sum_{k_1=1}^{k-1} \tilde{c}_{k_1} \zeta \leq r\zeta
%  $$

\item By the first claim, $\|(I - \hat{G}_{\text{det},k} {\hat{G}_{\text{det},k}}') G_{\text{det},k} \|_2 =\|\Psi_{k-1} G_{\text{det},k}\|_2  \leq r\zeta$. By item 2) of Lemma \ref{lemma0} with $P = G_{\text{det},k}$ and $\hat{P} = \hat{G}_{\text{det},k}$, the result $\|G_{\text{det},k} {G_{\text{det},k}}' - \hat{G}_{\text{det},k}{\hat{G}_{\text{det},k}}'\|_2 \leq 2 r\zeta$ follows.

\item Recall that $D_k \overset{QR}{=} E_k R_k$ is a QR decomposition where $E_k$ is orthonormal and $R_k$ is upper triangular. Therefore, $\sigma_i(D_k)  = \sigma_i (R_k)$. Since $\|(I - \hat{G}_{\text{det},k} {\hat{G}_{\text{det},k}}')G_{\text{det},k}\|_2 =\|\Psi_{k-1} G_{\text{det},k}\|_2 \leq r\zeta$ and $G_k' G_{\text{det},k} = 0$, by item 4) of Lemma \ref{lemma0} with $P=G_{\text{det},k}$, $\Phat=\hat{G}_{\text{det},k}$ and $Q=G_k$, we have $\sqrt{1-r^2\zeta^2} \leq \sigma_i((I - \hat{G}_{\text{det},k} {\hat{G}_{\text{det},k}}')G_k)=\sigma_i(D_k)\leq 1$.

\item Since $D_k \overset{QR}{=} E_k R_k$, so  $\|{D_{\text{undet},k}}'E_k \|_2 =\|{D_{\text{undet},k}}'D_k R_k^{-1} \|_2 = \|{G_{\text{undet},k}}'\Psi_{k-1}' \Psi_{k-1} G_k R_k^{-1} \|_2 =  \|{G_{\text{undet},k}}'\Psi_{k-1} G_k R_k^{-1} \|_2 = \|{G_{\text{undet},k}}'D_k R_k^{-1} \|_2 = \|{G_{\text{undet},k}}'E_k\|_2$.
Since $E_k = D_k R_k^{-1} = ( I -\hat{G}_{\text{det},k} {\hat{G}_{\text{det},k}}') G_k R_k^{-1}$,
\bea
\|{G_{\text{undet},k}}'E_k \|_2  &=& \| {G_{\text{undet},k}}' ( I -\hat{G}_{\text{det},k} {\hat{G}_{\text{det},k}}') G_k R_k^{-1}\|_2 \nn \\
%&\leq& \| {G_{\text{undet},k}}' ( I -\hat{G}_{\text{det},k} {\hat{G}_{\text{det},k}}') G_k \|_2 \| R_k^{-1}\|_2 \nn \\
&\leq& \frac{\| {G_{\text{undet},k}}' ( I -\hat{G}_{\text{det},k} {\hat{G}_{\text{det},k}}') G_k \|_2}{ \sqrt{1-r^2 \zeta^2})} \nn \\ 
&=&  \frac{\| {G_{\text{undet},k}}' \hat{G}_{\text{det},k} {\hat{G}_{\text{det},k}}' G_k \|_2}{\sqrt{1-r^2 \zeta^2})} \nn
%&\leq& \| G_{\text{undet},k} {G_{\text{undet},k}}' \hat{G}_{\text{det},k} {\hat{G}_{\text{det},k}}'\|_2 (1/\sqrt{1-r^2 \zeta^2})\nn \\
%&\leq& \| (I - \hat{G}_{\text{det},k} {\hat{G}_{\text{det},k}}') G_{\text{det},k} {G_{\text{det},k}}'\|_2 (1/\sqrt{1-r^2 \zeta^2}) \nn \\
%&\leq& \frac{r\zeta}{\sqrt{1-r^2 \zeta^2}} \nn
\eea
By item 3) of Lemma \ref{lemma0} with $P = {G}_{\text{det},k}$, $\Phat = \hat{G}_{\text{det},k}$ and $Q= G_{\text{undet},k}$, we get $\| {G_{\text{undet},k}}' \hat{G}_{\text{det},k}\|_2 \leq r\zeta$. By item 3) of Lemma \ref{lemma0} with $\Phat = \hat{G}_{\text{det},k}$ and $Q= G_k$, we get $ \| {\hat{G}_{\text{det},k}}'G_{k} \|_2 \leq r\zeta$.  Therefore, $\|{G_{\text{undet},k}}'E_k \|_2 = \|{E_k}' G_{\text{undet},k}\|_2 \leq \frac{r^2 \zeta^2}{\sqrt{1-r^2\zeta^2}}$.
\een
\end{proof}

\section{Proof of Lemma \ref{lem_bound_terms}} \label{proof_lem_bound_terms}

\begin{proof}
We use $\frac{1}{\tilde{\alpha}}\sum_t$ to denote $\frac{1}{\tilde{\alpha}} \sum_{t \in \tilde{\mathcal{I}}_{j,k}}$.

For $ t \in \tilde{\mathcal{I}}_{j,k}$, let $a_{t,k} := {G_{j,k}}'L_t$, $a_{t,\text{det}} := {G_{\text{det},k}}'L_t = [G_{j,1},\cdots G_{j,k-1}]'L_t$ and $a_{t,\text{undet}}:= {G_{\text{undet},k}}'L_t = [G_{j,k+1}\cdots G_{j,\vartheta_j}]'L_t$. Then $a_t:= P_j'L_t$ can be split as
$a_t = [ a_{t,\text{det}}' \ a_{t,k}' \ a_{t,\text{undet}}']'$. %Under the assumption of these lemmas, the following are true. {P_{(t)}}

This lemma follows using the following facts and the Hoeffding corollaries, Corollary \ref{hoeffding_nonzero} and \ref{hoeffding_rec}.
\ben
%\item The statement {\em ``conditioned on r.v. $X$, the event ${\cal E}^e$ holds w.p. one for all $X \in \Gamma$" is equivalent to ``$\mathbf{P}({\cal E}^e|X)=1, \ \text{for all} \ X \in \Gamma$".} We often use the former statement in our proofs since it is often easier to interpret. %\item Thus, the second claim of Lemma \ref{lem_et} is equivalent to  ``$e_t$ satisfies (\ref{etdef}), $\|e_t\|_2 \le {\phi^+} \sqrt{\zeta}$ and $\| \Phi_K P_j \|_2 \le (r+c) \zeta$ holds  w.p. one, for all $\tilde X_{j,k-1} \in \tilde\Gamma_{j,k-1}$". \label{et_eq}

\item The matrices $D_k$, $R_k$, $E_k$, $D_{\text{det},k}, D_{\text{undet},k}$, $\Psi_{k-1}$, $\Phi_K$ are functions of the r.v. $\tilde X_{j,k-1}$. All terms that we bound for the first two claims of the lemma are of the form $\frac{1}{\alpha} \sum_{t \in \mathcal{\tilde{I}}_{j,k}} Z_t$ where $Z_t= f_1(\tilde X_{j,k-1}) Y_t f_2(\tilde X_{j,k-1})$, $Y_t$ is a sub-matrix of $a_t a_t'$ and $f_1(.)$ and $f_2(.)$ are functions of $\tilde X_{j,k-1}$. For instance, one of the terms while bounding $\lambda_{\min}(\mathcal{A}_k)$ is $\frac{1}{\tilde{\alpha}} \sum_t R_k a_{t,k} {a_{t,k}}'{R_k}'$.
 $\tilde X_{j,k-1}$ is independent of any $a_{t}$ for $t \in  \mathcal{\tilde{I}}_{j,k}$ , and hence the same is true for the matrices  $D_k$, $R_k$, $E_k$, $D_{\text{det},k}, D_{\text{undet},k}$, $\Psi_{k-1}$, $\Phi_K$. Also, $a_{t}$'s for different $t \in \mathcal{\tilde{I}}_{j,k}$ are mutually independent. Thus, conditioned on  $\tilde X_{j,k-1}$, the $Z_t$'s defined above are mutually independent.
\label{X_at_indep}

\item All the terms that we bound for the third claim contain $e_t$. Using Lemma \ref{cslem}, conditioned on $\tilde X_{j,k-1}$, $e_t$ satisfies (\ref{etdef0}) w.p. one whenever $\tilde X_{j,k-1} \in \tilde\Gamma_{j,k-1}$. Conditioned on $\tilde X_{j,k-1}$, all these terms are also of the form $\frac{1}{\alpha} \sum_{t \in \mathcal{\tilde{I}}_{j,k}} Z_t$ with $Z_t$ as defined above, whenever $\tilde X_{j,k-1} \in \tilde\Gamma_{j,k-1}$. Thus, conditioned on  $\tilde X_{j,k-1}$, the $Z_t$'s for these terms are mutually independent, whenever $\tilde X_{j,k-1} \in \tilde\Gamma_{j,k-1}$.
%\label{zt_indep}

\item By Remark \ref{Gamma_rem} and the definition of $\tilde\Gamma_{j,k-1}$, $\tilde X_{j,k-1} \in \tilde\Gamma_{j,k-1}$ implies that $\zeta_{*} \le r \zeta$, $\tilde\zeta_{k'} \le c_{k'} \zeta, \ \text{for all} \ k'=1,2,\dots k-1$, $\zeta_K \le \zeta_K^+ \le c \zeta$, (iv) $\phi_K \le \phi^+$ (by Lemma \ref{cslem}); (v) $\|\Phi_K P_j\|_2 \le (r+c)\zeta$; and (vi) all conclusions of Lemma \ref{bound_R} hold. %\le \|\Phi_K [P_*, P_\new]\|_2 \le \zeta_* + \zeta_K

%\item Using Lemma \ref{RIC_bnd} and the definition of $\phi_K$, $\phi_K \le \phi^+ := 1.1735$.

\item By the clustering assumption, $ \lambda_k^- \le \lambda_{\min}(\E(a_{t,k}{a_{t,k}}')) \le \lambda_{\max}(\E(a_{t,k}{a_{t,k}}')) \le  \lambda_k^+$; $\lambda_{\max}(\E(a_{t,\text{det}}{a_{t,\text{det}}}')) \le \lambda_1^+ = \lambda^+$; and $\lambda_{\max}(\E(a_{t,\text{undet}}{a_{t,\text{undet}}}')) \le \lambda_{k+1}^+$. Also, $\lambda_{\max}(\E(a_t a_t')) \le \lambda^+$.

\item By Weyl's theorem, for a sequence of matrices $B_t$, $\lambda_{\min}(\sum_t B_t) \ge \sum_t \lambda_{\min}(B_t)$ and $\lambda_{\max}(\sum_t B_t) \le \sum_t \lambda_{\max}(B_t)$.
\een

%Given $\tilde X_{j,k-1} \in \tilde\Gamma_{j,k-1}$, $e_t$ satisfies (\ref{etdef_del}) by Lemma \ref{lem_et}.

 Consider $\tilde{A}_k = \frac{1}{\tilde{\alpha}} \sum_t {E_k}' \Psi_{k-1} L_t{L_t}' \Psi_{k-1} E_k$. Notice that
${E_k}' \Psi_{k-1} L_t = R_k a_{t,k} + {E_k}'(D_{\text{det},k} a_{t,\text{det}} + D_{\text{undet},k} a_{t,\text{undet}})$. Let $Z_t = R_k a_{t,k} {a_{t,k}}'{R_k}'$ and let $Y_t = R_k a_{t,k}  ({a_{t,\text{det}}}'{D_{\text{det},k}}' + {a_{t,\text{undet}}}'{D_{\text{undet},k}}')E_k +  E_k'(D_{\text{det},k} a_{t,\text{det}} + D_{\text{undet},k} a_{t,\text{undet}}) {a_{t,k}}'{R_k}'$. Then
\beq
\tilde{A}_k \succeq \frac{1}{\tilde{\alpha}} \sum_t Z_t + \frac{1}{\tilde{\alpha}} \sum_t Y_t\label{lemmabound_1_pt2}
\eeq

%the fact that $\|R_k\|_2 \geq \sqrt{1-r^2\zeta^2}$ (follows from
Consider $\frac{1}{\tilde{\alpha}} \sum_t Z_t = \frac{1}{\tilde{\alpha}} \sum_t R_k a_{t,k} {a_{t,k}}'{R_k}'$. (a) As explained above, the $Z_t$'s are conditionally independent given $\tilde X_{j,k-1}$. (b) Using Ostrowoski's theorem and Lemma \ref{bound_R}, for all $\tilde X_{j,k-1} \in \tilde\Gamma_{j,k-1}$, $\lambda_{\min}( \E(\frac{1}{\tilde{\alpha}}\sum_t Z_t|\tilde X_{j,k-1})) = \lambda_{\min}( R_{k} \frac{1}{\tilde{\alpha}}\sum_t \E(a_{t,k}{a_{t,k}}') {R_{k}}') \ge \lambda_{\min} (R_{k} {R_{k}}')\lambda_{\min} (\frac{1}{\tilde{\alpha}}\sum_t \E(a_{t,k}{a_{t,k}}')) \geq (1- r^2 \zeta^2)\lambda_{k}^-$.
(c) Finally, using $\|R_k\|_2 \leq 1$ and $\|a_{t,k}\|_2 \leq \sqrt{\tilde{c}_k} \gamma_* $, conditioned on $\tilde X_{j,k-1}$, $0 \preceq Z_t  \preceq \tilde{c}_k \gamma_{*}^2 I  $ holds w.p. one for all $\tilde X_{j,k-1} \in \tilde\Gamma_{j,k-1}$.

Thus, applying Corollary \ref{hoeffding_nonzero} with $\epsilon = 0.1 \zeta \lambda^-$, and using $\tilde{c}_k \leq r$, for all $\tilde X_{j,k-1} \in \tilde\Gamma_{j,k-1}$,
\begin{multline}
\mathbf{P}\left(\lambda_{\min} \Big(\frac{1}{\tilde{\alpha}} \sum_t Z_t\Big)
\geq   (1- r^2\zeta^2)\lambda_{k}^-  - 0.1 \zeta \lambda^- \Big| \tilde X_{j,k-1}\right)
\geq  \\
1- \tilde{c}_k \exp \left(\frac{-\tilde{\alpha} \epsilon^2 }{8  (\tilde{c}_k  \gamma_{*}^2)^2}\right)
\geq 1 - r \exp \left(\frac{-\tilde{\alpha} \cdot (0.1 \zeta \lambda^-)^2 }{8  r^2 \gamma_{*}^4}\right)
  \label{lemma_add_A1_pt2}
\end{multline}

%. The last inequality follows as $\tilde{c}_k \leq r$.

Consider $Y_t = R_k a_{t,k}  ({a_{t,\text{det}}}'{D_{\text{det},k}}' + {a_{t,\text{undet}}}'{D_{\text{undet},k}}')E_k +  E_k'(D_{\text{det},k} a_{t,\text{det}} + D_{\text{undet},k} a_{t,\text{undet}}) {a_{t,k}}'{R_k}'$. (a) As before, the $Y_t$'s are conditionally independent given $\tilde X_{j,k-1}$. (b) Since $\E[a_t]=0$ and $\text{Cov}[a_t]=\Lambda_t$ is diagonal, $\E(\frac{1}{\alpha}\sum_t Y_t|\tilde X_{j,k-1}) = 0$ whenever $\tilde X_{j,k-1} \in \tilde\Gamma_{j,k-1}$. (c) Conditioned on $\tilde X_{j,k-1}$, $\|Y_t\|_2 \le 2 \sqrt{\tilde{c}_k r} \gamma_*^2 r\zeta(1+ \frac{r\zeta}{\sqrt{1-r^2\zeta^2}}) \leq 2 r^2 \zeta \gamma_*^2 ( 1+ \frac{10^{-4}}{\sqrt{1-10^{-4}}}) \leq \frac{2}{r} ( 1+ \frac{10^{-4}}{\sqrt{1-10^{-4}}}) < 2.1$ holds w.p. one for all $\tilde X_{j,k-1} \in \tilde\Gamma_{j,k-1}$. This follows because $\tilde X_{j,k-1} \in \tilde\Gamma_{j,k-1}$ implies that $\|D_{\text{det},k}\|_2 \leq r\zeta$, $\|{E_k}' D_{\text{undet},k}\|_2 = \| {E_k}' G_{\text{undet},k}\|_2 \leq \frac{r^2 \zeta^2}{\sqrt{1-r^2\zeta^2}}$.
Thus, under the same conditioning, $-b I \preceq Y_t  \preceq b I$ with $b =2.1$ w.p. one.
Thus, applying Corollary \ref{hoeffding_nonzero} with $\epsilon = 0.1 \zeta \lambda^-$, we get
\begin{multline}
\mathbf{P}\left(\lambda_{\min} \Big(\frac{1}{\tilde{\alpha}} \sum_t Y_t\Big) \geq - 0.1 \zeta \lambda^- \Big| \tilde X_{j,k-1} \right) \geq \\
1- r \exp \left( \frac{-\tilde{\alpha} (0.1 \zeta \lambda^-)^2} {8 \dot (4.2)^2 }\right)  \ \text{for all $\tilde X_{j,k-1} \in  \tilde\Gamma_{j,k-1}$}
\label{lemma_add_A2_pt2}
\end{multline}

Combining (\ref{lemmabound_1_pt2}), (\ref{lemma_add_A1_pt2}) and (\ref{lemma_add_A2_pt2}) and using the union bound, $\mathbf{P} (\lambda_{\min}(\tilde{A}_k) \geq \lambda_{k}^-(1 - r^2\zeta^2) - 0.2 \zeta \lambda^-| \tilde X_{j,k-1}) \geq 1-\tilde{p}_1(\tilde{\alpha},\zeta) \ \text{for all $\tilde X_{j,k-1} \in  \tilde\Gamma_{j,k-1}$}$ where
\beq
\tilde{p}_1 (\tilde{\alpha},\zeta) := r \exp \left(\frac{-\tilde{\alpha} \cdot (0.1 \zeta \lambda^-)^2 }{8  r^2 \gamma_{*}^4}\right) + r \exp \left( \frac{-\tilde{\alpha} (0.1 \zeta \lambda^-)^2} {8 \dot (4.2)^2}\right) \label{prob1}
\eeq
The first claim of the lemma follows by using $\lambda_{k}^- \ge \lambda^-$ and applying Lemma \ref{rem_prob} with $X \equiv \tilde X_{j,k-1}$ and $\calc \equiv \tilde\Gamma_{j,k-1}$.

Consider $\tilde{A}_{k,\perp} := \frac{1}{\alpha} \sum_t {E_{k,\perp}}' \Psi_{k-1} L_t {L_t}' \Psi_{k-1} E_{k,\perp}$. Notice that ${E_{k,\perp}}' \Psi_{k-1} L_t = {E_{k,\perp}}' (D_{\text{det},k} a_{t,\text{det}} + D_{\text{undet},k} a_{t,\text{undet}})$. Thus, $\tilde{A}_{k,\perp} = \frac{1}{\tilde{\alpha}} \sum_t Z_t$ with $Z_t={E_{k,\perp}}' (D_{\text{det},k} a_{t,\text{det}} + D_{\text{undet},k} a_{t,\text{undet}})(D_{\text{det},k} a_{t,\text{det}} + D_{\text{undet},k} a_{t,\text{undet}})'E_{k,\perp}$ which is of size $(n-\tilde{c}_k)\times (n-\tilde{c}_k)$.
(a) As before, given $\tilde X_{j,k-1}$, the $Z_t$'s are independent.
(b) Conditioned on $\tilde X_{j,k-1}$, $0 \preceq Z_t  \preceq r \gamma_*^2 I$ w.p. one for all $\tilde X_{j,k-1} \in  \tilde\Gamma_{j,k-1}$.
(c) $\E(\frac{1}{\alpha}\sum_t Z_t|\tilde X_{j,k-1}) \preceq (\lambda_{k+1}^+ + r^2\zeta^2 \lambda^+)I$ for all $\tilde X_{j,k-1} \in  \tilde\Gamma_{j,k-1}$.

Thus applying Corollary \ref{hoeffding_nonzero} with $\epsilon =0.1 \zeta \lambda^-$ and using $\tilde{c}_k \geq \tilde{c}_{\min}$, we get
 \begin{multline}
\mathbf{P}(\lambda_{\max}(\tilde{A}_{k,\perp})  \leq \lambda_{k+1}^+ + r^2 \zeta^2 \lambda^+ + 0.1 \zeta \lambda^- | \tilde X_{j,k-1}) \geq\\
 1- \tilde{p}_2(\tilde{\alpha},\zeta) \ \text{for all $\tilde X_{j,k-1} \in  \tilde\Gamma_{j,k-1}$}
\end{multline} %(n-\tilde{c}_k) \exp (-\frac{\tilde{\alpha}(0.1 \zeta \lambda^-)^2}{8  r \gamma_*^2 })
where
\beq
\tilde{p}_2(\tilde{\alpha},\zeta) := (n-\tilde{c}_{\min}) \exp \left(\frac{-\tilde{\alpha} (0.1 \zeta \lambda^-)^2}{8   r^2 \gamma_*^4 }\right) \label{prob2}
\eeq
The second claim follows using $\lambda_{k}^- \ge \lambda^-$, $f:=\lambda^+/\lambda^-$, $\tilde{h}_k := {\lambda_{k+1}}^+ / {\lambda_{k}}^-$ in the above expression and applying Lemma \ref{rem_prob}.

Consider the third claim. Using the expression for $\tilde{\mathcal{H}}_k$ given in Definition \ref{defHk_del}, it is easy to see that
\begin{align}
\|\tilde{\mathcal{H}}_k \|_2 &\leq  \max\{ \|\tilde{H}_k\|_2, \|\tilde{H}_{k,\perp}\|_2 \} + \|\tilde{B}_k\|_2 \nn \\
&\leq \Big\|\frac{1}{\tilde{\alpha}} \sum_t e_t {e_t}' \Big\|_2 +  \max(\|T2\|_2, \|T4\|_2) + \|\tilde{B}_k\|_2
\label{add2_calH1}
\end{align}
where $T2:= \frac{1}{\tilde{\alpha}} \sum_t  {E_{k}}' \Psi_{k-1}( L_t {e_t}' + e_t {L_t}')\Psi_{k-1} E_{k}$ and $T4 :=\frac{1}{\tilde{\alpha}} \sum_t {E_{k,\perp}}'\Psi_{k-1} (L_t {e_t}' + {e_t}'L_t)\Psi_{k-1} E_{k,\perp}$. The second inequality follows by using the facts that (i) $\tilde{H}_k = T1 - T2$ where $T1 := \frac{1}{\tilde{\alpha}} \sum_t {E_{k}}' \Psi_{k-1} e_t {e_t}'\Psi_{k-1} E_{k}$, (ii) $\tilde{H}_{k,\perp} = T3 - T4$ where $T3 := \frac{1}{\tilde{\alpha}} \sum_t {E_{k,\perp}}'\Psi_{k-1} e_t {e_t}'\Psi_{k-1}  E_{k,\perp}$, and (iii) $\max(\|T1\|_2, \|T3\|_2) \le \|\frac{1}{\tilde{\alpha}} \sum_t e_t {e_t}'\|_2$.

Next, we obtain high probability bounds on each of the terms on the RHS of (\ref{add_calH1}) using the Hoeffding corollaries. %First notice that each of these terms contains $e_t$. Using the second claim of Lemma \ref{lem_et}, $e_t$ satisfies (\ref{etdef}) w.p. one, for all $\tilde X_{j,k-1} \in \tilde\Gamma_{j,k-1}$. Also, $\| \Phi_K P_j \|_2 \le (r+c) \zeta$ w.p. one, for all $\tilde X_{j,k-1} \in \tilde\Gamma_{j,k-1}$.

Consider $\|\frac{1}{\tilde{\alpha}} \sum_t e_t {e_t}'\|_2$. Let $Z_t = e_t {e_t}'$.
(a) As explained in the beginning of the proof, conditioned on $\tilde X_{j,k-1}$, the various $Z_t$'s in the summation are independent whenever $\tilde X_{j,k-1} \in \tilde\Gamma_{j,k-1}$. Also, by Lemma \ref{cslem}, under this conditioning, $\That_t=T_t$ for all $t \in \tilde{I}_{j,k}$ and hence $e_t$ satisfies (\ref{etdef0}) in this interval. Recall also that in this interval, $\Phi_{(t)} = \Phi_K$. Thus, using $\|\Phi_K P_j\|_2 \leq (r+c) \zeta$,
$$\|e_t\|_2 \le \phi^+ \sqrt{\zeta}$$
(b) Conditioned on $\tilde X_{j,k-1}$, $0 \preceq Z_t \preceq b_1 I$ w.p. one for all $\tilde X_{j,k-1} \in  \tilde\Gamma_{j,k-1}$. Here $b_1:={\phi^+}^2 \zeta$.
(c) Using $\|\Phi_K P_j\|_2 \leq (r+c) \zeta$, $0 \preceq \frac{1}{\alpha} \sum_t \E(Z_t|\tilde X_{j,k-1}) \preceq b_2I, \ b_2:= (r+c)^2 \zeta^2 {\phi^+}^2 \lambda^+$ for all $\tilde X_{j,k-1} \in  \tilde\Gamma_{j,k-1}$.

Thus, applying Corollary \ref{hoeffding_nonzero} with $\epsilon = 0.1 \zeta \lambda^-$,
\begin{multline}
\mathbf{P} \left( \Big\|\frac{1}{\tilde{\alpha}} \sum_t  e_t {e_t}' \Big\|_2 \leq b_2  +  0.1 \zeta \lambda^- \Big| \tilde X_{j,k-1} \right) \\
\geq 1- n \exp\left(\frac{-\tilde{\alpha}( 0.1 \zeta \lambda^-)^2}{ 8 \cdot b_1^2}\right)  \ \text{for all $\tilde X_{j,k-1} \in  \tilde\Gamma_{j,k-1}$}
\label{add2_etet}
\end{multline}

Consider $T2$. Let $Z_t: = {E_{k}}' \Psi_{k-1} (L_t {e_t}' + e_t{L_t}')\Psi_{k-1} E_{k}$ which is of size $\tilde{c}_k \times \tilde{c}_k$. Then $T2 = \frac{1}{\tilde{\alpha}} \sum_t Z_t$.
(a) Conditioned on $\tilde X_{j,k-1}$, the various $Z_t$'s used in the summation are mutually independent whenever $\tilde X_{j,k-1} \in \tilde\Gamma_{j,k-1}$.
(b) Notice that ${E_{k}}'\Psi_{k-1} L_t  = R_{k} a_{t,k} + {E_k}' (D_{\text{det},k} a_{t,\text{det}} + D_{\text{undet},k}a_{t,\text{undet}})$ and
${E_{k}}'\Psi_{k-1} e_t = (R_{k}^{-1})' D_k' e_t = (R_{k}^{-1})' D_k' I_{T_t} [(\Phi_K)_{T_t}' (\Phi_K)_{T_t}]^{-1} {I_{T_t}}' \Phi_K P_j a_{t}$.
Thus conditioned on $\tilde X_{j,k-1}$, $\|Z_t\|_2 \leq  2 b_3$ w.p. one for all $\tilde X_{j,k-1} \in  \tilde\Gamma_{j,k-1}$. Here,
$b_3:= \frac{\sqrt{r\zeta}}{\sqrt{1-r^2\zeta^2}} \phi^+ \gamma_*$. This follows using $\|(R_{k}^{-1})'\|_2\leq 1/\sqrt{1-r^2\zeta^2}$, $\|e_t\|_2\leq \phi^+ \sqrt{\zeta}$ and $\|E_k'\Psi_{k-1}L_t\|_2 \le \|L_t\|_2 \leq \sqrt{r} \gamma_*$.
(c) Also, $\|\frac{1}{\alpha} \sum_t \E(Z_t|\tilde X_{j,k-1})\|_2 \leq  2 b_4$ where $b_4:= \kappa_{s,D}^+ \kappa_{s,e}^+ (r+c) \zeta \phi^+ ( \lambda_k^+ + r \zeta \lambda^+ + \frac{r^2\zeta^2}{\sqrt{1-r^2\zeta^2}} \lambda_{k+1}^+)$. Here $\kappa_{s,D}^+ = \kappa_{s,*}^+ + r\zeta$ defined in Remark \ref{rem_kappa_D} is the bound on $\max_j \max_k \kappa_s( D_{j,k})$.% that follows from the denseness of $P_{j}$ and the event $\tilde\Gamma_{j,k-1}^e$.
%and $\kappa_{s,e} = \max_j \max_k \kappa_s(\Phi_K P_j)$.

Thus, applying Corollary \ref{hoeffding_rec} with $\epsilon = 0.1 \zeta \lambda^-$, for all $\tilde X_{j,k-1} \in \tilde\Gamma_{j,k-1}$,
\begin{multline*}
\mathbf{P}( \|T2\|_2 \leq 2 b_4  + 0.1 \zeta \lambda^- | \tilde X_{j,k-1}) \\
\geq 1- \tilde{c}_k\exp\left(\frac{-\tilde{\alpha}(0.1 \zeta \lambda^-)^2}{32 \cdot 4 b_3^2}\right) \nn
%\geq 1- r\exp(-\frac{\tilde{\alpha}(0.1 \zeta \lambda^-)^2}{32 \cdot 4 b_3^2})  \nn
\end{multline*}

Consider $T4$. Let $Z_t: = {E_{k,\perp}}'\Psi_{k-1} (L_t {e_t}' + e_t{L_t}')\Psi_{k-1} E_{k,\perp}$ which is of size $(n-\tilde{c}_k)\times (n-\tilde{c}_k)$. Then $T4 = \frac{1}{\tilde{\alpha}} \sum_t Z_t$. (a) conditioned on $\tilde X_{j,k-1}$, the various $Z_t$'s used in the summation are mutually independent whenever $\tilde X_{j,k-1} \in  \tilde\Gamma_{j,k-1}$.
(b) Notice that ${E_{k,\perp}}'\Psi_{k-1} L_t ={E_{k,\perp}}' (D_{\text{det},k} a_{t,\text{det}} + D_{\text{undet},k}a_{t,\text{undet}})$. Thus, conditioned on $\tilde X_{j,k-1}$, $\|Z_t\|_2 \leq  2b_5$ w.p. one for all $\tilde X_{j,k-1} \in  \tilde\Gamma_{j,k-1}$. Here $b_5:=   \sqrt{r\zeta} \phi^+ \gamma_*$.
(c) Also, for all $\tilde X_{j,k-1} \in \tilde\Gamma_{j,k-1}$, $\|\frac{1}{\tilde{\alpha}} \sum_t \E(Z_t|\tilde X_{j,k-1})\|_2 \leq 2 b_6, \ b_6:= \kappa_{s,e}^+  (r+c) \zeta \phi^+ (\lambda_{k+1}^+ + r \zeta \lambda^+)$.
Applying Corollary \ref{hoeffding_rec} with $\epsilon = 0.1 \zeta \lambda^-$, for all $\tilde X_{j,k-1} \in \tilde\Gamma_{j,k-1}$,
\begin{align*}
\mathbf{P}( \|T4\|_2 \leq 2b_6  + & 0.1 \zeta \lambda^-| \tilde X_{j,k-1}) \geq\\
& 1- (n-\tilde{c}_k) \exp\left(\frac{-\tilde{\alpha} (0.1 \zeta \lambda^-)^2}{32 \cdot4 b_5^2}\right) \\
\geq & 1- (n-\tilde{c}_{\min}) \exp\left(\frac{-\tilde{\alpha} (0.1 \zeta \lambda^-)^2}{32 \cdot4 b_5^2}\right).\nn
\end{align*}
%for all $\tilde X_{j,k-1} \in   \tilde\Gamma_{j,k-1}$.

Consider $\max(\|T2\|_2,\|T4\|_2)$. %Since $b_3 = b_5$ and $b_4 > b_6$, so $2b_6  + \epsilon < 2b_4  + \epsilon$. Therefore, for all $\tilde X_{j,k-1} \in  \tilde\Gamma_{j,k-1}$,
%%
%$$\mathbf{P}( \|T4 \|_2 \leq 2 b_4 + 0.1 \zeta \lambda^- | \tilde X_{j,k-1} ) \geq 1- (n-\tilde{c}_{k}) \exp\left(\frac{-\tilde{\alpha}(0.1 \zeta \lambda^-)^2}{32 \cdot 4 b_3^2}\right)$$
%%
By union bound and using $b_3 > b_5$, for all $\tilde X_{j,k-1} \in  \tilde\Gamma_{j,k-1}$,
\begin{multline}
\mathbf{P}( \max(\|T2 \|_2,\|T4 \|_2)\leq 2\max(b_4,b_6) +  0.1 \zeta \lambda^- |\tilde X_{j,k-1}) \geq \\
1- n \exp\left(\frac{-\tilde{\alpha} (0.1 \zeta \lambda^-)^2}{32 \cdot 4b_3^2}\right)
\label{add2_maxT}
\end{multline}
%where $b_4:=  \kappa_{s,e} (r+c) \zeta \phi^+ ( \lambda_k^+ + r \zeta \lambda^+ + \frac{r^2\zeta^2}{\sqrt{1-r^2\zeta^2}} \lambda_{k+1}^+)$.

%then $\|\frac{1}{\alpha} \sum_t \E(Z_t|\tilde X_{j,k-1})\|_2\leq  2\kappa_{s,D} b_4$.
%??? include in this change. Notice that if we also introduce an extra denseness coefficient $\kappa_{s,D} := \max_j \max_k \kappa_s(D_k)$, then $\mathbf{P}( \|T2\|_2 \leq 2 \kappa_{s,D} b_4  + 0.1 \zeta \lambda^- | \tilde X_{j,k-1}) \geq 1- \tilde{c}_k\exp\left(\frac{-\tilde{\alpha}(0.1 \zeta \lambda^-)^2}{32 \cdot 4 b_3^2}\right)$. Thus, $\mathbf{P}( \max(\|T2 \|_2,\|T4 \|_2)\leq 2\max(\kappa_{s,D} b_4, b_6) +  0.1 \zeta \lambda^- |\tilde X_{j,k-1}) \geq 1- n \exp\left(\frac{-\tilde{\alpha} (0.1 \zeta \lambda^-)^2}{32 \cdot 4b_3^2}\right)$. This would help to get a looser bounds on $\tilde{g}_{\max}$ and $\tilde{h}_{\max}$ in Theorem \ref{thm2}.

Consider $\|\tilde{B}_k\|_2$. Let $Z_t := {E_{k,\perp}}'\Psi_{k-1} (L_t-e_t)({L_t}'-{e_t}')\Psi_{k-1} E_{k}$ which is of size $(n-\tilde{c}_k)\times \tilde{c}_k$. Then $\tilde{B}_k = \frac{1}{\tilde{\alpha}} \sum_t Z_t$.
(a) conditioned on $\tilde X_{j,k-1}$, the various $Z_t$'s used in the summation are mutually independent whenever $\tilde X_{j,k-1} \in  \tilde\Gamma_{j,k-1}$.
(b) Notice that ${E_{k,\perp}}'\Psi_{k-1} (L_t-e_t) = {E_{k,\perp}}'( D_{\text{det},k}a_{t,\text{det}} + D_{\text{undet},k} a_{t,\text{undet}} - \Psi_{k-1} e_t)$ and  ${E_{k}}' \Psi_{k-1} (L_t - e_t) = R_{k} a_{t,k}+ {E_{k}}'( D_{\text{det},k} a_{t,\text{det}} + D_{\text{undet},k} a_{t,\text{undet}} - \Psi_{k-1} e_t)$. Thus,  conditioned on $\tilde X_{j,k-1}$,
$ \|Z_t\|_2 \leq b_7$
w.p. one for all $X_{j,K, k-1} \in  \Gamma_{j,K, k-1}$. Here $b_7 := (\sqrt{r}\gamma_* + \phi^+ \sqrt{\zeta})^2$.
(c)  $\|\frac{1}{\tilde{\alpha}} \sum_t \E(Z_t|\tilde X_{j,k-1})\|_2 \leq b_8$ for all $\tilde X_{j,k-1} \in  \tilde\Gamma_{j,k-1}$ where
\begin{align*}
b_8 := & (r+c) \zeta  \kappa_{s,e}^+  \phi^+ \lambda_k^+  \\
& + \left[(r+c)\zeta \kappa_{s,e}^+ \phi^+ + (r+c)\zeta \kappa_{s,e}^+ \frac{r^2\zeta^2}{\sqrt{1-r^2\zeta^2}}\right]\lambda_{k+1}^+ \\
& + [r^2\zeta^2 + 2(r+c)r \zeta^2 \kappa_{s,e}^+ \phi^+ + (r+c)^2 \zeta^2 {\kappa_{s,e}^+}^2 {\phi^+}^2] \lambda^+  \nn
%\label{b8}
\end{align*}
Thus, applying Corollary \ref{hoeffding_rec} with $\epsilon=0.1 \zeta \lambda^-$,
\begin{multline}\label{Bktil}
\mathbf{P} (\|\tilde{B}_k\|_2 \leq b_8 + 0.1 \zeta \lambda^- | \tilde X_{j,k-1}) \geq \\
1 - n \exp\left(\frac{-\tilde{\alpha} (0.1 \zeta \lambda^-)^2}{32 \cdot b_7^2}\right) \ \text{for all $\tilde X_{j,k-1} \in   \tilde\Gamma_{j,k-1}$}
\end{multline}

Using (\ref{add2_calH1}), (\ref{add2_etet}), (\ref{add2_maxT}) and (\ref{Bktil}) and the union bound, for any $\tilde X_{j,k-1} \in  \tilde\Gamma_{j,k-1}$,
$$\mathbf{P} (\|\tilde{\mathcal{H}}_k\|_2 \leq b_9 + 0.2 \zeta \lambda^-|\tilde X_{j,k-1}) \geq 1- \tilde{p}_3(\tilde{\alpha},\zeta)$$
where $b_9 := b_2 +2b_4+ b_8$ and
\begin{multline}
\tilde{p}_3(\tilde{\alpha},\zeta):= n \exp\left(\frac{-\tilde{\alpha} \epsilon^2}{8 \cdot b_1^2}\right) + n \exp\left(\frac{-\tilde{\alpha} \epsilon^2}{32\cdot 4 b_3^2}\right) \\
+ n \exp\left(\frac{-\tilde{\alpha} \epsilon^2}{32 \cdot b_7^2}\right) \label{prob3}
\end{multline}
with $b_1 = {\phi^+}^2 \zeta$, $b_3 := \sqrt{r\zeta} \phi^+ \gamma_*$, $b_7 := (\sqrt{r} \gamma_* + \phi^+ \sqrt{\zeta})^2$.
%\bea
%b_9 &:=& b_2 +2b_4+ b_8
%\nn \\
%&=& 3 (r+c) \zeta  \kappa_{s,e}  \phi^+ \lambda_k^+
%+  [(r+c)\zeta \kappa_{s,e} \phi^+ + (r+c)\zeta \kappa_{s,e} (1+ 2\phi^+)\frac{r^2\zeta^2}{\sqrt{1-r^2\zeta^2}}]\lambda_{k+1}^+ \nn\\
%&&+  [r^2\zeta^2 + 4(r+c)r \zeta^2 \kappa_{s,e} \phi^+ + 2(r+c)^2 \zeta^2(1+ \kappa_{s,e}^2) \phi_K^2] \lambda^+ \label{b9}
%\eea
Using $\lambda_{k}^- \ge \lambda^-$, $f:= \lambda^+/ \lambda^-$, $\tilde{g}_k := \lambda_k^+/\lambda_k^-$ and $\tilde{h}_k := \lambda_{k+1}^+ / \lambda_k^-$,  and then applying Lemma \ref{rem_prob}, the third claim of the lemma follows.
\end{proof}

\bibliographystyle{IEEEtran}
\bibliography{tipnewpfmt}

\end{document}